\documentclass{mscs}

\title[Dynamic Game Semantics]
  {Dynamic Game Semantics}
\author[N. Yamada]
  {N\ls O\ls R\ls I\ls H\ls I\ls R\ls O\ns
   Y\ls A\ls M\ls A\ls D\ls A$^\dagger$\ns and\ns S\ls A\ls M\ls S\ls O\ls N\ns A\ls B\ls
   R\ls A\ls M\ls S\ls K\ls Y$^\dagger$ \\
   $^\dagger$Department of Computer Science, University of Oxford \\ Wolfson Building, Parks Rd, OX1 3QD
   \if0 Email: \texttt{norihiro.yamada@cs.ox.ac.uk} \\ Email: \texttt{samson.abramsky@cs.ox.ac.uk} \fi }
\date{October 2018}

\pubyear{2001}
\pagerange{\pageref{firstpage}--\pageref{lastpage}}
\volume{00}
\doi{S0960129501003425}

\usepackage{amsmath,amsfonts,amssymb,stmaryrd,mathrsfs,mathdots,amsbsy,float}
\usepackage[small]{diagrams}
\usepackage{hyperref}
\usepackage{bm}
\usepackage{tikz}
\usetikzlibrary{tikzmark}
\usepackage[small]{diagrams}
\usepackage{float}
\usetikzlibrary{positioning,calc}

\newtheorem{definition}{Definition}[section]
\newtheorem{example}{Example}[section]
\newtheorem{theorem}{Theorem}[section]
\newtheorem{lemma}{Lemma}[section]
\newtheorem{proposition}{Proposition}[section]
\newtheorem{corollary}{Corollary}[section]

\newtheorem{notation}{Notation}

\newarrow{Corresponds}<--->

\begin{document}

\label{firstpage}
\maketitle

\begin{abstract}
The present paper gives a mathematical, in particular, \emph{syntax-independent}, formulation of \emph{intensionality} and \emph{dynamics} of computation in terms of games and strategies.
Specifically, we give a game semantics for a higher-order programming language that distinguishes programs with the same value yet different algorithms (or intensionality), equipped with the \emph{hiding operation} on strategies that precisely corresponds to the (small-step) operational semantics (or dynamics) of the language.
Categorically, our games and strategies give rise to a \emph{cartesian closed bicategory}, and our game semantics forms an instance of a generalization of the standard interpretation of functional programming languages in cartesian closed categories.
This work is intended to be the first step towards a mathematical (both categorical and game-semantic) foundation of intensional and dynamic aspects of logic and computation; our approach should be applicable to a wide range of logics and computations.
\end{abstract}


\ifprodtf \newpage \else \vspace*{-1\baselineskip}\fi

\section{Introduction}
\label{Introduction}
In \cite{girard1989proofs}, J.-Y. Girard mentions the dichotomy between the \emph{static} and the \emph{dynamic} viewpoints in logic and computation; the former identifies terms (i.e., proofs or programs) with their \emph{denotations} (i.e., results of their computations in an ideal sense), while the latter focuses on their \emph{senses} (i.e., algorithms or intensionality) and dynamics (i.e., proof-normalization or reduction).
This distinction has been certainly reflected in the two mutually complementary semantics of programming languages: \emph{denotational} and \emph{operational} ones \cite{amadio1998domains,winskel1993formal,gunter1992semantics}.
He points out that a \emph{mathematical} formulation of the former has been well-developed, based on Scott's beautiful \emph{domain theory} \cite{scott1976data,gierz2003continuous,abramsky1994domain}, but it is not the case for the latter; the treatment of senses has been based on ad-hoc \emph{syntactic manipulation}.
He then emphasizes the importance of \emph{mathematics of senses}: 
\begin{quote}
The establishment  of  a  truly  operational  semantics  of  algorithms  is  perhaps  the  most important problem in computer science \cite{girard1989proofs}.
\end{quote}

The present work addresses this problem; specifically, it gives an interpretation $\llbracket \_ \rrbracket_{\mathcal{D}}$ of a programming language $\mathcal{L}$ with a small-step operational semantics $\to$ and a syntax-independent operation $\mathcal{H}$ that satisfy the following \emph{\bfseries dynamic correspondence property (DCP)}: $\mathsf{M_1 \to M_2}$ if and only if (a.k.a. iff) $\llbracket \mathsf{M_1} \rrbracket_{\mathcal{D}} \neq \llbracket \mathsf{M_2} \rrbracket_{\mathcal{D}}$ and $\mathcal{H}(\llbracket \mathsf{M_1} \rrbracket_{\mathcal{D}}) = \llbracket \mathsf{M_2} \rrbracket_{\mathcal{D}}$ for any programs $\mathsf{M_1}$ and $\mathsf{M_2}$ of $\mathcal{L}$.
Note that the `only if' and `if' directions correspond respectively to certain soundness and completeness properties of the interpretation $\mathcal{H}$ of $\rightarrow$. 
Note also that the interpretation $\llbracket \_ \rrbracket_{\mathcal{D}}$ is \emph{finer} than the usual (sound) denotational semantics because $\mathsf{M_1 \to M_2}$ implies $\llbracket \mathsf{M_1} \rrbracket_{\mathcal{D}} \neq \llbracket \mathsf{M_2} \rrbracket_{\mathcal{D}}$.
Thus, the interpretation $\llbracket \_ \rrbracket_{\mathcal{D}}$ and the operation $\mathcal{H}$ capture \emph{intensionality} and \emph{dynamics} of computation, respectively.

Although our framework is intended to be a \emph{general} approach, being applicable to a wide range of logics and computations, as the first step, we focus on a finite fragment of the programming language \emph{PCF} \cite{scott1993type,plotkin1977lcf} customized for our aim.

\subsection{Game Semantics}
Our approach is based on \emph{game semantics} \cite{abramsky1997semantics,hyland1997game}, a particular kind of denotational semantics of logic and computation, in which formulas (or types) and proofs (or programs) are interpreted  as \emph{games} and \emph{strategies}, respectively. 

We employ game semantics for its conceptual naturality and mathematical precision, which has been demonstrated by various \emph{full completeness} and \emph{full abstraction} results \cite{curien2007definability} in the literature, leading to a conceptually and mathematically deeper understanding of logic and computation.  
Also, game semantics is very flexible: It has modeled a wide range of formal systems and programming languages by simply varying constraints on strategies \cite{abramsky1999game}, which enables us to compare and relate various concepts \emph{syntax-independently}.
We hope that these advantages of game semantics are true also for intensionality and dynamics of logic and computation. 

A \emph{game}, roughly, is a certain kind of a rooted forest whose branches represent possible `developments' or \emph{(valid) positions} of a `game in the usual sense' (such as chess, poker, etc.).
\emph{Moves} of a game are nodes of the game, where some moves are distinguished and called \emph{initial}; only initial moves can be the first element (or occurrence) of a position of the game. 
\emph{Plays} of a game are (finitely or infinitely) increasing sequences $(\boldsymbol{\epsilon}, m_1, m_1 m_2, \dots)$ of positions of the game, where $\boldsymbol{\epsilon}$ is the \emph{empty sequence}. 
For our purpose, it suffices to focus on rather standard \emph{sequential} (as opposed to \emph{concurrent} \cite{abramsky1999concurrent}), \emph{unpolarized} (as opposed to \emph{polarized }\cite{laurent2004polarized}) games played by two participants, \emph{Player}, representing a `computational agent', and \emph{Opponent}, representing an `environment', in each of which Opponent always starts a play (i.e., unpolarized), and then they alternately and separately perform moves (i.e., sequential) allowed by the rules of the game.
Strictly speaking, a position of each game is not just a finite sequence of moves: Each occurrence $m$ of Opponent's or O- (resp. Player's or P-) non-initial move in a position is assigned or \emph{points to} a previous occurrence $m'$ of P- (resp. O-) move in the position, representing that $m$ is performed specifically as a response to $m'$. 

A \emph{strategy} on a game, on the other hand, is what tells Player which move (together with a pointer) she should make at each of her turns in the game.
Hence, a game semantics $\llbracket \_ \rrbracket_{\mathcal{G}}$ of a programming language $\mathcal{L}$ interprets a type $\mathsf{A}$ of $\mathcal{L}$ as a game $\llbracket \mathsf{A} \rrbracket_{\mathcal{G}}$ that specifies possible plays between Player and Opponent, and a term $\mathsf{M : A}$\footnote{For simplicity, here we focus on \emph{closed} terms, i.e., ones with the \emph{empty context}.} of $\mathcal{L}$ as a strategy $\llbracket \mathsf{M} \rrbracket_{\mathcal{G}}$ on the game $\llbracket \mathsf{A} \rrbracket_{\mathcal{G}}$ that describes for Player how to play on $\llbracket \mathsf{A} \rrbracket_{\mathcal{G}}$; an execution of the term $\mathsf{M}$ is then modeled as a play of the game $\llbracket \mathsf{A} \rrbracket_{\mathcal{G}}$ in which Player follows $\llbracket \mathsf{M} \rrbracket_{\mathcal{G}}$.

Let us consider a simple example. 
The simplest game is the \emph{terminal game} $T$, which has no moves, and thus it has only the trivial position $\boldsymbol{\epsilon}$.

As another example, consider the \emph{natural number game} $N$, which is the following rooted tree (which is infinite in width):
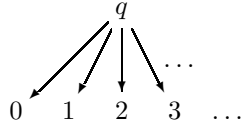
\begin{diagram}
& & q & & \\
& \ldTo(2, 2) \ldTo(1, 2) & \dTo & \rdTo(1, 2) \ \dots & \\
0 & 1 & 2 & 3 & \dots
\end{diagram}
in which a play starts with Opponent's question $q$ (`What is your number?') and ends with Player's answer $n \in \mathbb{N}$ (`My number is $n$!'), where $\mathbb{N}$ is the set of all natural numbers, and $n$ points to $q$ (though this pointer is omitted in the above diagram). A strategy $\underline{10}$ on $N$, for instance, that corresponds to $10 \in \mathbb{N}$ can be represented by the map $q \mapsto 10$ equipped with a pointer from $10$ to $q$ (though it is the only choice).
In the following, pointers of most strategies are obvious, and thus we often omit them.

As yet another example, consider the game $N \multimap N$ of \emph{linear functions} \cite{girard1987linear} (also written informally $N_{[0]} \multimap N_{[1]}$) on natural numbers, whose typical maximal position is $q_{[1]} q_{[0]} n_{[0]} m_{[1]}$, where $n, m \in \mathbb{N}$, and $(\_)_{[i]}$ for $i = 0, 1$ are arbitrary, unspecified `tags' to distinguish the two copies of $N$ (in the rest of the paper, we employ a similar notation for three or more copies of a game in the obvious manner too), or diagrammatically\footnote{The diagram is depicted as above only to clarify which component game each move belongs to; it should be read just as a finite sequence, namely, $q_{[1]} q_{[0]} n_{[0]} m_{[1]}$, equipped with the pointers represented by the arrows in the diagram.}:
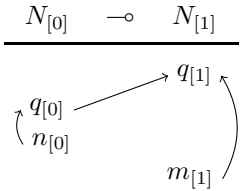
\begin{center}
\begin{tabular}{ccc}
$N_{[0]}$ & $\multimap$ & $N_{[1]}$ \\ \hline 
&&\tikzmark{cmultimap1} $q_{[1]}$ \tikzmark{cmultimap3} \\
\tikzmark{cmultimap2} $q_{[0]}$ \tikzmark{dmultimap1}&& \\
\tikzmark{dmultimap2} $n_{[0]}$&& \\
&&$m_{[1]}$ \tikzmark{dmultimap3}
\end{tabular}
\begin{tikzpicture}[overlay, remember picture, yshift=.25\baselineskip]
\draw [->] ({pic cs:dmultimap1}) to ({pic cs:cmultimap1});
\draw [->] ({pic cs:dmultimap2}) [bend left] to ({pic cs:cmultimap2});
\draw [->] ({pic cs:dmultimap3}) [bend right] to ({pic cs:cmultimap3});
\end{tikzpicture}
\end{center}
which can be read as follows:
\begin{enumerate}
\item Opponent's question $q_{[1]}$ for an output (`What is your output?');
\item Player's question $q_{[0]}$ for an input (`Wait, what is your input?);
\item Opponent's answer, say, $n_{[0]}$ to $q_{[0]}$ (`OK, here is an input $n$.');
\item Player's answer, say, $m_{[1]}$ to $q_{[1]}$ (`Alright, the output is then $m$.').
\end{enumerate}
A strategy $\mathit{succ}$ on this game that corresponds to the (linear) successor function can be represented by the map $q_{[1]} \mapsto q_{[0]}, q_{[1]} q_{[0]} n_{[0]} \mapsto n+1_{[1]}$, where $n$ ranges over $\mathbb{N}$, or diagrammatically:
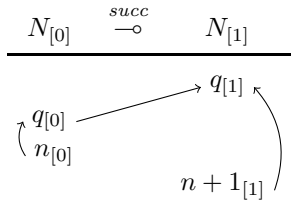
\begin{center}
\begin{tabular}{ccc}
$N_{[0]}$ & $\stackrel{\mathit{succ}}{\multimap}$ & $N_{[1]}$ \\ \hline 
&&\tikzmark{cmultimap11} $q_{[1]}$ \tikzmark{cmultimap13} \\
\tikzmark{cmultimap12} $q_{[0]}$ \tikzmark{dmultimap11}&& \\
\tikzmark{dmultimap12} $n_{[0]}$&& \\
&&$n+1_{[1]}$ \tikzmark{dmultimap13}
\end{tabular}
\begin{tikzpicture}[overlay, remember picture, yshift=.25\baselineskip]
\draw [->] ({pic cs:dmultimap11}) to ({pic cs:cmultimap11});
\draw [->] ({pic cs:dmultimap12}) [bend left] to ({pic cs:cmultimap12});
\draw [->] ({pic cs:dmultimap13}) [bend right] to ({pic cs:cmultimap13});
\end{tikzpicture}
\end{center}

\subsection{Static Game Semantics}
Game semantics is often said to be an \emph{intensional}, \emph{dynamic} semantics for a category of games and strategies is usually not well-pointed, and plays of a game may be regarded as `intensional, dynamic interactions' between the participants of the game.
However, it has been employed as denotational semantics, and thus it is in particular \emph{sound}: If two programs evaluate to the same value, then their denotations in conventional game semantics are identical.
Consequently, conventional game semantics $\llbracket \_ \rrbracket_{\mathcal{G}}$ is actually \emph{extensional} and \emph{static} in the sense that if there is a reduction $\mathsf{M_1 \to M_2}$ in syntax, then the equation $\llbracket \mathsf{M_1} \rrbracket_{\mathcal{G}} = \llbracket \mathsf{M_2} \rrbracket_{\mathcal{G}}$ holds in the semantics (i.e., it does not capture the dynamics $\mathsf{M_1 \to M_2}$ or the intensional difference between $\mathsf{M_1}$ and $\mathsf{M_2}$).
In other words, it is \emph{not} intensional or dynamic in the sense that it does not satisfy DCPs.

Therefore, to establish mathematics of senses, we need to introduce a more dynamic, intensional refinement of games and strategies so that it satisfies DCPs for logical systems and programming languages. 
To get some insights to develop such games and strategies, let us see how conventional game semantics fails to be dynamic or intensional.
The point in a word is that `internal communication' between strategies for their composition is \emph{a priori} `hidden', and thus the resulting strategy is always in `normal form'.
For instance, the composition $\mathit{succ} ; \mathit{double} : N \multimap N$ of strategies $\mathit{succ} : N \multimap N$ and $\mathit{double} : N \multimap N$, implementing the successor and the doubling (linear) functions, respectively,
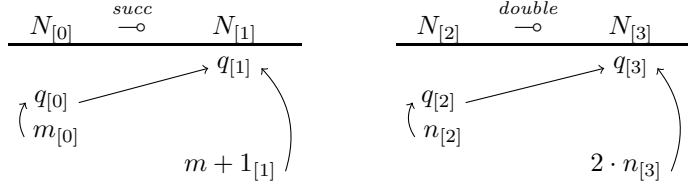
\begin{center}
\begin{tabular}{ccccccccc}
$N_{[0]}$ & $\stackrel{\mathit{succ}}{\multimap}$ & $N_{[1]}$ &&&& $N_{[2]}$ & $\stackrel{\mathit{double}}{\multimap}$ & $N_{[3]}$ \\ \cline{1-3} \cline{7-9}
&&\tikzmark{csucc31} $q_{[1]}$ \tikzmark{csucc33}&&&&&&\tikzmark{cdouble31} $q_{[3]}$ \tikzmark{cdouble33} \\
\tikzmark{csucc32} $q_{[0]}$ \tikzmark{dsucc31}&&&&&&\tikzmark{cdouble32} $q_{[2]}$ \tikzmark{ddouble31}&& \\
\tikzmark{dsucc32} $m_{[0]}$ &&&&&&\tikzmark{ddouble32} $n_{[2]}$&& \\
&&$m+1_{[1]}$ \tikzmark{dsucc33} &&&&&&$2 \cdot n_{[3]}$ \tikzmark{ddouble33} 
\end{tabular}
\begin{tikzpicture}[overlay, remember picture, yshift=.25\baselineskip]
\draw [->] ({pic cs:dsucc31}) to ({pic cs:csucc31});
\draw [->] ({pic cs:dsucc32}) [bend left] to ({pic cs:csucc32});
\draw [->] ({pic cs:dsucc33}) [bend right] to ({pic cs:csucc33});
\draw [->] ({pic cs:ddouble31}) to ({pic cs:cdouble31});
\draw [->] ({pic cs:ddouble32}) [bend left] to ({pic cs:cdouble32});
\draw [->] ({pic cs:ddouble33}) [bend right] to ({pic cs:cdouble33});
\end{tikzpicture}
\end{center}
is formed as follows. First, by `internal communication', we mean that Player plays the role of Opponent in the intermediate component games $N_{[1]}$ and $N_{[2]}$ just by `copy-catting' her last moves, resulting in the following play:
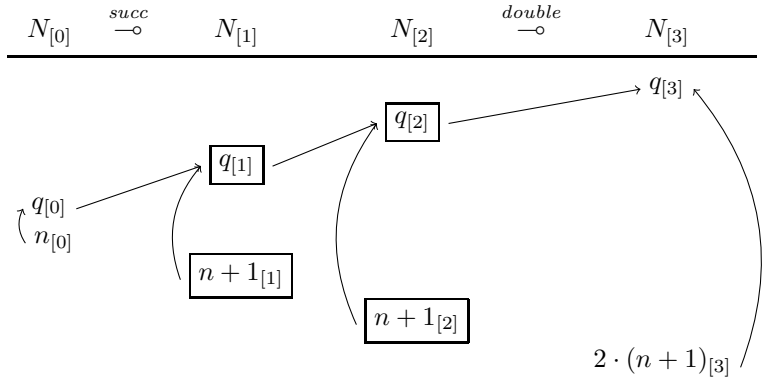
\begin{center}
\begin{tabular}{ccccccc}
$N_{[0]}$ & $\stackrel{\mathit{succ}}{\multimap}$ & $N_{[1]}$ && $N_{[2]}$ & $\stackrel{\mathit{double}}{\multimap}$ & $N_{[3]}$ \\ \hline
&&&&&& \tikzmark{cPC2} $q_{[3]}$ \tikzmark{cPC1} \\
&&&&\tikzmark{cPC3} \fbox{$q_{[2]}$} \tikzmark{dPC2} && \\
&& \tikzmark{cPC4} \fbox{$q_{[1]}$} \tikzmark{dPC3} &&&& \\
\tikzmark{cPC5} $q_{[0]}$ \tikzmark{dPC4} &&&&&& \\
\tikzmark{dPC5} $n_{[0]}$ &&&&&& \\
&&\tikzmark{dPC6} \fbox{$n+1_{[1]}$} &&&& \\
&&&&\tikzmark{dPC7} \fbox{$n+1_{[2]}$} && \\
&&&&&&$2 \cdot (n+1)_{[3]}$ \tikzmark{dPC1}
\end{tabular}
\begin{tikzpicture}[overlay, remember picture, yshift=.25\baselineskip]
\draw [->] ({pic cs:dPC1}) [bend right] to ({pic cs:cPC1});
\draw [->] ({pic cs:dPC2}) to ({pic cs:cPC2});
\draw [->] ({pic cs:dPC3}) to ({pic cs:cPC3});
\draw [->] ({pic cs:dPC4}) to ({pic cs:cPC4});
\draw [->] ({pic cs:dPC5}) [bend left] to ({pic cs:cPC5});
\draw [->] ({pic cs:dPC6}) [bend left] to ({pic cs:cPC4});
\draw [->] ({pic cs:dPC7}) [bend left] to ({pic cs:cPC3});
\end{tikzpicture}
\end{center}
where each move for `internal communication' is marked by a square box just for clarity, and the pointer from $q_{[1]}$ to $q_{[2]}$ is added because the move $q_{[1]}$ is no longer initial. 
Importantly, it is assumed that Opponent plays on the game $N_{[0]} \multimap N_{[3]}$, `seeing' only moves of $N_{[0]}$ or $N_{[3]}$.
The resulting play is to be read as follows:
\begin{enumerate}

\item Opponent's question $q_{[3]}$ for an output in $N_{[0]} \multimap N_{[3]}$ (`What is your output?');

\item Player's question \fbox{$q_{[2]}$} by $\mathit{double}$ for an input in $N_{[2]} \multimap N_{[3]}$ (`What is an input?');

\item \fbox{$q_{[2]}$} triggers the question \fbox{$q_{[1]}$} for an output in $N_{[0]} \multimap N_{[1]}$ (`What is an output?');

\item Player's question $q_{[0]}$ by $\mathit{succ}$ for an input in $N_{[0]} \multimap N_{[1]}$ (`Wait, what is an input?');

\item Opponent's answer, say, $n_{[0]}$ to $q_{[0]}$ in $N_{[0]} \multimap \  N_{[3]}$ (`Here is an input $n$.');

\item Player's answer \fbox{$n+1_{[1]}$} to \fbox{$q_{[1]}$} by $\mathit{succ}$ in $N_{[0]} \multimap N_{[1]}$ (`The output is then $n+1$.');

\item \fbox{$n+1_{[1]}$} triggers the answer \fbox{$n+1_{[2]}$} to \fbox{$q_{[2]}$} in $N_{[2]} \multimap N_{[3]}$ (`Here is the input $n+1$.');

\item Player's answer $2 \cdot (n+1)_{[3]}$ to $q_{[3]}$ by $\mathit{double}$ in $N_{[2]} \multimap N_{[3]}$ (`The output is then $2 \cdot (n+1)$!').

\end{enumerate}

Next, `hiding' means to hide or delete every move with a square box from the play, resulting in the strategy for the (linear) function $n \mapsto 2 \cdot (n + 1)$ as expected:
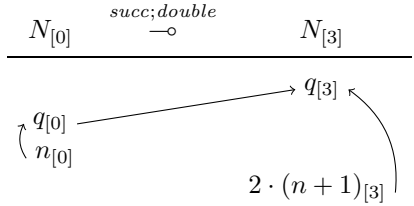
\begin{center}
\begin{tabular}{ccc}
$N_{[0]}$ & $\stackrel{\mathit{succ} ; \mathit{double}}{\multimap}$ & $N_{[3]}$ \\ \hline
&&\tikzmark{cPCH1} $q_{[3]}$ \tikzmark{cPCH3} \\
\tikzmark{cPCH2} $q_{[0]}$ \tikzmark{dPCH1}&& \\
\tikzmark{dPCH2} $n_{[0]}$&& \\
&&$2 \cdot (n + 1)_{[3]}$ \tikzmark{dPCH3} 
\end{tabular}
\begin{tikzpicture}[overlay, remember picture, yshift=.25\baselineskip]
\draw [->] ({pic cs:dPCH1}) to ({pic cs:cPCH1});
\draw [->] ({pic cs:dPCH2}) [bend left] to ({pic cs:cPCH2});
\draw [->] ({pic cs:dPCH3}) [bend right] to ({pic cs:cPCH3});
\end{tikzpicture}
\end{center}
Note that it is `hiding' that makes the resulting play a valid one on the game $N \multimap N$. 

Now, let us plug in the strategy $\underline{5}^T : q_{[5]} \mapsto 5_{[5]}$ on the game $T_{[4]} \multimap N_{[5]}$, which coincides with $N$ up to `tags'. The composition $\underline{5}^T ; \mathit{succ} ; \mathit{double} : T \multimap N$\footnote{Composition of strategies is \emph{associative} \cite{abramsky1997semantics,hyland1997game,abramsky1999game}; thus, the order of applying composition does not matter.} is computed again by `internal communication': 
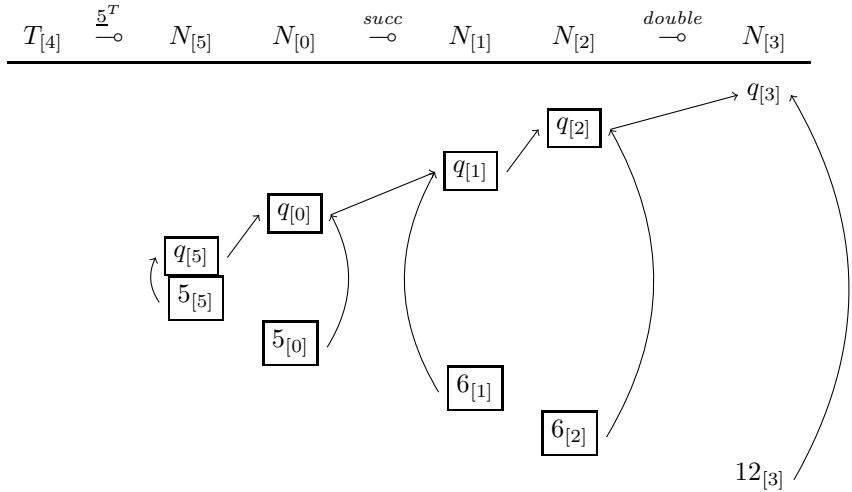
\begin{center}
\begin{tabular}{ccccccccc}
$T_{[4]}$& $\stackrel{\underline{5}^T}{\multimap}$ & $N_{[5]}$ & $N_{[0]}$ & $\stackrel{\mathit{succ}}{\multimap}$ & $N_{[1]}$ & $N_{[2]}$ & $\stackrel{\mathit{double}}{\multimap}$ & $N_{[3]}$ \\ \hline
&&&&&&&& \tikzmark{cfive2} $q_{[3]}$ \tikzmark{cfive1} \\
&&&&&& \tikzmark{cfive3} \fbox{$q_{[2]}$} \tikzmark{dfive2} && \\
&&&&& \tikzmark{cfive4} \fbox{$q_{[1]}$} \tikzmark{dfive3} &&& \\
&&&\tikzmark{cfive5}  \fbox{$q_{[0]}$} \tikzmark{dfive4} &&&&& \\
&& \tikzmark{cfive6} \fbox{$q_{[5]}$} \tikzmark{dfive5} &&&&&& \\
&& \tikzmark{dfive6} \fbox{$5_{[5]}$} &&&&&& \\
&&& \fbox{$5_{[0]}$} \tikzmark{dfive7} &&&&& \\
&&&&& \tikzmark{dfive9} \fbox{$6_{[1]}$} &&& \\
&&&&&& \fbox{$6_{[2]}$} \tikzmark{dfive8} && \\
&&&&&&&&$12_{[3]}$ \tikzmark{dfive1}
\end{tabular}
\begin{tikzpicture}[overlay, remember picture, yshift=.25\baselineskip]
\draw [->] ({pic cs:dfive1}) [bend right] to ({pic cs:cfive1});
\draw [->] ({pic cs:dfive2}) to ({pic cs:cfive2});
\draw [->] ({pic cs:dfive3}) to ({pic cs:cfive3});
\draw [->] ({pic cs:dfive4}) to ({pic cs:cfive4});
\draw [->] ({pic cs:dfive5}) to ({pic cs:cfive5});
\draw [->] ({pic cs:dfive6}) [bend left] to ({pic cs:cfive6});
\draw [->] ({pic cs:dfive7}) [bend right] to ({pic cs:dfive4});
\draw [->] ({pic cs:dfive8}) [bend right] to ({pic cs:dfive2});
\draw [->] ({pic cs:dfive9}) [bend left] to ({pic cs:cfive4});
\end{tikzpicture}
\end{center}
plus `hiding':
\begin{center}
\begin{tabular}{ccc}
$T_{[4]}$ & $\stackrel{\underline{5}^T ; \mathit{succ} ; \mathit{double}}{\multimap}$ &$N_{[3]}$ \\ \hline
&&$q_{[3]}$ \tikzmark{cfive11} \\
&&$12_{[3]}$ \tikzmark{dfive11}
\end{tabular}
\begin{tikzpicture}[overlay, remember picture, yshift=.25\baselineskip]
\draw [->] ({pic cs:dfive11}) [bend right] to ({pic cs:cfive11});
\end{tikzpicture}
\end{center}

In syntax, on the other hand, assuming that there are a (ground) type $\iota$ of natural numbers, a numeral $\underline{\mathsf{n}}$ of type $\iota$ for each $n \in \mathbb{N}$, and constants $\mathsf{succ}$ and $\mathsf{double}$ of type $\iota$ for the successor and the doubling functions, respectively, equipped with the operational semantics $\mathsf{succ} \ \! \underline{\mathsf{n}} \to \underline{\mathsf{n+1}}$ and $\mathsf{double} \ \! \underline{\mathsf{n}} \to \underline{\mathsf{2 \cdot n}}$ for all $n \in \mathbb{N}$ in an arbitrary functional programming language, the program $\mathsf{p_1 \stackrel{\mathrm{df. }}{\equiv} \lambda x . (\lambda y . \ \! double \ \! y) ((\lambda z . \ \! succ \ \! z) \ \! x)}$ represents the syntactic composition $\mathsf{succ} ; \mathsf{double}$. When it is applied to the numeral $\underline{\mathsf{5}}$, we have the following chain of reductions:
\begin{align*}
\mathsf{p_1} \ \! \underline{\mathsf{5}} &\to^\ast \mathsf{(\lambda x. \ \! \mathsf{double} \ \! (\mathsf{succ} \ \! x)) \ \! \underline{\mathsf{5}}} \\
&\to^\ast \mathsf{\mathsf{double} \ \! (\mathsf{succ} \ \! \underline{\mathsf{5}})} \\
&\to^\ast \mathsf{\mathsf{double} \ \! \underline{\mathsf{6}}} \\
&\to^\ast \underline{\mathsf{12}}.
\end{align*}
Therefore, it seems that reduction in syntax corresponds in game semantics to `hiding internal communication'. 
As seen in the above example, however, this game-semantic normalization is \emph{a priori} executed and thus invisible in conventional game semantics $\llbracket \_ \rrbracket_{\mathcal{G}}$. As a result, the two programs $\mathsf{p_1} \ \! \underline{\mathsf{5}}$ and $\underline{\mathsf{12}}$ are interpreted by $\llbracket \_ \rrbracket_{\mathcal{G}}$ as the same strategy. 
Moreover, observe that moves with a square box describe \emph{intensionality} or \emph{step-by-step processes} to compute an output from an input, but they are invisible after `hiding'. 
Thus, e.g., a program $\mathsf{p_2 \stackrel{\mathrm{df. }}{\equiv} \lambda x . (\lambda y . \ \! \mathsf{succ} \ \! y)(\lambda v . (\lambda z . \ \!\mathsf{succ} \ \! z)((\lambda w . \ \! \mathsf{double} \ \! w) \ \! v) \ \! x)}$, representing the same function as $\mathsf{p_1}$ yet a different algorithm $\mathsf{double} ; \mathsf{succ} ; \mathsf{succ}$ is modeled as:
\begin{equation*}
\llbracket \mathsf{p_2} \rrbracket_{\mathcal{G}} = \llbracket \mathsf{double} ; \mathsf{succ} ; \mathsf{succ} \rrbracket_{\mathcal{G}} = \llbracket \mathsf{succ} ; \mathsf{double} \rrbracket_{\mathcal{G}} = \llbracket \mathsf{p_1} \rrbracket_{\mathcal{G}}.
\end{equation*}

To sum up, we have observed the following:
\begin{enumerate}

\item \textsc{(Reduction as hiding).}
Reduction in syntax corresponds in game semantics to `hiding intermediate moves (i.e., moves with a square box)';

\item \textsc{(A priori normalization).}
However, the `hiding' process is \emph{a priori} executed in conventional game semantics, and thus strategies are always in `normal form';

\item \textsc{(Intermediate moves as intensionality).}
Also, `intermediate moves' constitute \emph{intensionality} of computation; however, they are not captured in conventional game semantics again due to the a priori execution of the `hiding' operation.

\end{enumerate}

\subsection{Dynamic Games and Strategies}
From these observations, we have obtained a promising solution: to define a variant of games and strategies, in which `intermediate moves' are not a priori `hidden', representing intensionality of logic and computation, and the \emph{hiding operations} $\mathcal{H}$ on the games and strategies that `hide intermediate moves' in a step-by-step fashion, interpreting dynamics of logic and computation.
Let us call such a variant of games (resp. strategies) \emph{\bfseries dynamic games} (resp. \emph{\bfseries dynamic strategies}).

In doing so, we shall develop mathematical structures that are conceptually natural and mathematically elegant. This effort is to inherit the natural, intuitive nature of conventional game semantics so that the resulting interpretation would be insightful, convincing and useful.
Also, mathematics often leads to a `correct' formulation: If a definition gives rise to neat mathematical structures, then it is likely to succeed in capturing the essence of concepts and phenomena of concern, and subsume various instances (n.b., recall that our aim is to establish \emph{mathematics} of senses).
In fact, dynamic games and strategies are a natural generalization of conventional games and strategies, and they satisfy beautiful algebraic laws; as a consequence, they form a \emph{cartesian closed bicategory (CCB)} in the sense of \cite{ouaknine1997two}\footnote{N.b., for the present work, it suffices to know that a CCB is a generalized CCC in the sense that the equational axioms of CCCs are required to hold only \emph{up to 2-cell isomorphisms}.} $\mathcal{LDG}$ (Definition~\ref{DefCCBoCCDG}), in which 0- (resp. 1-) cells are certain dynamic games (resp. dynamic strategies), and 2-cells are the \emph{extensional equivalence} between 1-cells; the countably-infinite iteration of the hiding operations $\mathcal{H}$ on dynamic games and strategies induces the 2-functor $\mathcal{H}^\omega : \mathcal{LDG} \to \mathcal{LMG}$, where the CCC $\mathcal{LMG}$ of conventional games and strategies can be seen as an `extensionally collapsed' $\mathcal{LDG}$.

\subsection{Dynamic Game Semantics}
We then give, as the main result of the present work, a game semantics $\llbracket \_ \rrbracket_{\mathcal{DG}}$ of \emph{finitary PCF} (i.e., the simply-typed $\lambda$-calculus equipped with the boolean type) in $\mathcal{LDG}$ that together with the hiding operation $\mathcal{H}$ satisfies the DCP (Corollary~\ref{CoroFirstDynamicGameSemantics}), which we call \emph{\bfseries dynamic game semantics} as it captures dynamics and intensionality of computation better than conventional ones.
We select finitary PCF as our target language since a simple language would be appropriate for the first work on dynamic game semantics.

Note that it does not make much sense to ask whether full abstraction holds for dynamic game semantics as its aim is rather to capture intensionality of computation.

Also, the semantics does not satisfy faithfulness: The semantic equation is of course \emph{finer than $\beta$-equivalence} but also \emph{coarser than $\alpha$-equivalence}, e.g., non-$\alpha$-equivalent terms $\mathsf{(\lambda x . \ \! \underline{0}) \ \! \underline{1}}$ and $\mathsf{(\lambda x . \ \! \underline{0}) \ \! \underline{2}}$ are interpreted to be the same in dynamic game semantics, which is because the semantic equation captures \emph{algorithmic difference} of terms, while $\alpha$-equivalence distinguishes how they are constructed even if their algorithms coincide (n.b., this point calls for (syntax-independent) mathematics of senses).

On the other hand, it makes sense to ask if full completeness holds for dynamic game semantics.
In fact, we shall establish \emph{dynamic full completeness} (Corollary~\ref{CoroDynamicFullCompleteness}).

\subsection{Our Contribution and Related Work}
To the best of our knowledge, the present work is the first syntax-independent characterization of dynamics and intensionality of computation in the sense of DCPs.

The work closest in spirit is Girard's \emph{geometry of interaction (GoI)} \cite{girard1989geometry,girard1990geometry,girard1995geometry,girard2003geometry,girard2011geometry,girard2013geometry}. However, GoI appears mathematically \emph{ad-hoc} for it does not conform to the standard categorical semantics of type theories \cite{lambek1988introduction,pitts2001categorical,crole1993categories,jacobs1999categorical}; also, it does not capture the \emph{step-by-step} process of reduction in the sense of DCPs.
In contrast, dynamic game semantics refines the standard semantics and does satisfy a DCP.

Next, the idea of exhibiting `intermediate moves' in the composition of strategies is \emph{nothing new}; there are game-semantic approaches \cite{dimovski2005data,greenland2005game,ong2006model} that give such moves an official status. However, because their aims are rather to develop a tool for program analysis and verification, they do not study in depth mathematical structures thereof, give an intensional game semantics that refines the standard categorical semantics or formulate a \emph{step-by-step} `hiding' process.
Therefore, our contribution for this point is to study algebraic structures of games and strategies when we do not a priori `hide intermediate moves' and refine the standard categorical semantics in such a way that satisfies DCPs.


Also, there are several approaches to model dynamics of computation by \emph{2-categories} \cite{seely1987modelling,hilken1996towards,mellies2005axiomatic}. In these papers, however, the horizontal composition of 1-cells is the \emph{normalizing} one, which is why the structure is 2-categories rather than bicategories.\footnote{N.b., the unit law \emph{on the nose} does not hold if the composition is non-normalizing.}
In addition, the 2-cells of their 2-categories are rewriting, while the 2-cells of our bicategory are the external equivalence between 1-cells; note that 2-cells in a bicategory cannot interpret rewriting unless the horizontal composition is normalizing since associativity of non-normalizing composition with respect to such 2-cells does not hold.\footnote{N.b., there is no rewriting between 1-cells $(f;g);h$ and $f;(g;h)$ if the composition is non-normalizing.} 
Thus, although their motivations are similar to ours, our \emph{bicategorical} approach seems novel, interpreting an application of terms by non-normalizing composition, the extensional equivalence of terms by 2-cells and rewriting by the hiding operation $\mathcal{H}$.
Moreover, their frameworks are categorical, while we instantiate our categorical model by game semantics.
Furthermore, neither of the previous work establishes a DCP. 

Finally, note that the present work has some implications from theoretical as well as practical viewpoints.
From the theoretical perspective, it enables us to study dynamics and intensionality of computation as purely \emph{mathematical} (or \emph{semantic}) concepts, just like any concepts in pure mathematics such as differentiation and integration in calculus, homotopy in topology, etc. 
Thus, we would be able to rigorously analyze the essence of these concepts, ignoring superfluous syntactic details. 
From the practical point, on the other hand, it might become a useful tool for language analysis and design, e.g., our variant of finitary PCF would not exist without the present work.

\subsection{Structure of the paper}
The rest of the present paper proceeds as follows. 
This introduction ends with fixing some notations.
Then, Section~\ref{DynamicBicategoricalSemantics} formulates our target programming language and its bicategorical semantics that satisfies the DCP so that it remains to establish its game-semantic instance. 
Next, Section~\ref{DynamicGamesAndStrategies} introduces dynamic games and strategies and studies their basic algebraic structures, and Section~\ref{DynamicGameSemantics} gives dynamic game semantics of the language. 
Finally, Section~\ref{ConclusionAndFutureWorkOfDynamicGameSemantics} draws a conclusion and proposes some future work.

\begin{notation*}
We use the following notations throughout the paper:
\begin{itemize}


\item We use bold letters $\boldsymbol{s}, \boldsymbol{t}, \boldsymbol{u}, \boldsymbol{v}$, etc. for sequences, in particular $\boldsymbol{\epsilon}$ for the \emph{empty sequence}, and letters $a, b, c, d, m, n, x, y, z$, etc. for elements of sequences;

\item Given $k \in \mathbb{N}$, we write $\overline{k}$ for the finite set $\{ 1, 2, \dots, k \ \! \} \subseteq \mathbb{N}$ (n.b., $\overline{0} = \emptyset$);

\item We often abbreviate a finite sequence $\boldsymbol{s} = (x_1, x_2, \dots, x_{|\boldsymbol{s}|})$ as $x_1 x_2 \dots x_{|\boldsymbol{s}|}$, where $|\boldsymbol{s}|$ denotes the \emph{length} (i.e., the number of elements) of $\boldsymbol{s}$, and write $\boldsymbol{s}(i)$, where $i \in \overline{|\boldsymbol{s}|}$, as another notation for $x_i$;


\item A \emph{concatenation} of sequences is represented by the juxtaposition of them, but we often write $a \boldsymbol{s}$, $\boldsymbol{t} b$, $\boldsymbol{u} c \boldsymbol{v}$ for $(a) \boldsymbol{s}$, $\boldsymbol{t} (b)$, $\boldsymbol{u} (c) \boldsymbol{v}$, etc., and also write $\boldsymbol{s} . \boldsymbol{t}$ for $\boldsymbol{s t}$;

\item We define $\boldsymbol{s}^n \stackrel{\mathrm{df. }}{=} \underbrace{\boldsymbol{s} \boldsymbol{s} \cdots \boldsymbol{s}}_n$ for a sequence $\boldsymbol{s}$ and a natural number $n \in \mathbb{N}$;

\item We write $\mathsf{Even}(\boldsymbol{s})$ (resp. $\mathsf{Odd}(\boldsymbol{s})$) iff $\boldsymbol{s}$ is of even-length (resp. odd-length);

\item We define $S^\mathsf{P} \stackrel{\mathrm{df. }}{=} \{ \boldsymbol{s} \in S \mid \mathsf{P}(\boldsymbol{s}) \}$ for a set $S$ of sequences and $\mathsf{P} \in \{ \mathsf{Even}, \mathsf{Odd} \}$;

\item $\boldsymbol{s} \preceq \boldsymbol{t}$ means $\boldsymbol{s}$ is a \emph{prefix} of $\boldsymbol{t}$, i.e., $\boldsymbol{t} = \boldsymbol{s} . \boldsymbol{u}$ for some sequence $\boldsymbol{u}$, and given a set $S$ of sequences, we define $\mathsf{Pref}(S) \stackrel{\mathrm{df. }}{=} \{ \boldsymbol{s} \mid \exists \boldsymbol{t} \in S . \ \! \boldsymbol{s} \preceq \boldsymbol{t} \ \! \}$;

\item For a poset $P$ and a subset $S \subseteq P$, $\mathsf{Sup}(S)$ denotes the \emph{supremum} of $S$;


\item $X^* \stackrel{\mathrm{df. }}{=} \{ x_1 x_2 \dots x_n \mid n \in \mathbb{N}, \forall i \in \overline{n} . \ \! x_i \in X \ \! \}$ for each set $X$;

\item For a function $f : A \to B$ and a subset $S \subseteq A$, we define $f \upharpoonright S : S \to B$ to be the \emph{restriction} of $f$ to $S$, and $f^\ast : A^\ast \to B^\ast$ by $f^\ast(a_1 a_2 \dots a_n) \stackrel{\mathrm{df. }}{=} f(a_1) f(a_2) \dots f(a_n) \in B^\ast$ for all $a_1 a_2 \dots a_n \in A^\ast$;


\item Given sets $X_1, X_2, \dots, X_n$, and $i \in \overline{n}$, we write $\pi_i$ (or $\pi_i^{(n)}$) for the \emph{$i^{\text{th}}$-projection function} $X_1 \times X_2 \times \dots \times X_n \to X_i$ that maps $(x_1, x_2, \dots, x_n) \mapsto x_i$;

\item $\simeq$ denote the \emph{Kleene equality}, i.e., $x \simeq y \stackrel{\mathrm{df. }}{\Leftrightarrow} (x \downarrow \wedge \ y \downarrow \wedge \ x = y) \vee (x \uparrow \wedge \ y \uparrow)$, where we write $x \downarrow$ if an element $x$ is defined, and $x \uparrow$ otherwise.

\end{itemize}
\end{notation*}

\section{Dynamic Bicategorical Semantics}
\label{DynamicBicategoricalSemantics}
Let us first present a \emph{categorical} description of how dynamic games and strategies capture dynamics and intensionality of logic and computation, and show that it is a \emph{refinement} of the standard categorical semantics of type theories \cite{lambek1988introduction,pitts2001categorical,crole1993categories,jacobs1999categorical}.

\subsection{Beta-Categories of Computation}
The categorical structure for our interpretation of logic and computation is \emph{$\beta$-categories of computation (BoCs)}, a certain kind of bicategories whose 2-cells are the \emph{extensional equivalence} between 1-cells, equipped with an \emph{evaluation} satisfying certain axioms.

Let us first introduce a more general notion of \emph{$\beta$-categories}, which are categories \emph{up to an equivalence relation on morphisms}:
\begin{definition}[$\beta$-categories]
A \emph{\bfseries $\boldsymbol{\beta}$-category} is a pair $\mathcal{C} = (\mathcal{C}, \simeq)$ that consists of:
\begin{itemize}

\item A class $\mathsf{ob}(\mathcal{C})$ of \emph{\bfseries objects}, where we usually write $A \in \mathcal{C}$ for $A \in \mathsf{ob}(\mathcal{C})$;

\item A class $\mathcal{C}(A, B)$ of \emph{\bfseries $\boldsymbol{\beta}$-morphisms} from $A$ to $B$ for each pair $A, B \in \mathcal{C}$, where we often write $f : A \rightarrow B$ for $f \in \mathcal{C}(A, B)$ if $\mathcal{C}$ is obvious from the context;

\item A (class) function $\mathcal{C}(A, B) \times \mathcal{C}(B, C) \stackrel{;_{A, B, C}}{\rightarrow} \mathcal{C}(A, C)$, called the \emph{\bfseries $\boldsymbol{\beta}$-composition} on $\beta$-morphisms from $A$ to $B$ and from $B$ to $C$, for each triple $A, B, C \in \mathcal{C}$;

\item A $\beta$-morphism $\mathit{id}_A \in \mathcal{C}(A, A)$, called the \emph{\bfseries $\boldsymbol{\beta}$-identity} on $A$, for each $A \in \mathcal{C}$;

\item An equivalence (class) relation $\simeq_{A, B}$ on $\mathcal{C}(A, B)$, called the \emph{\bfseries equivalence} on $\beta$-morphisms from $A$ to $B$, for each pair $A, B \in \mathcal{C}$

\end{itemize}
where we also write $\mathcal{C}(B, C) \times \mathcal{C}(A, B) \stackrel{\circ_{A, B, C}}{\rightarrow} \mathcal{C}(A, C)$ for the $\beta$-composition $;_{A, B, C}$ and often omit the subscripts on $;_{A, B, C}$, $\circ_{A, B, C}$ and $\simeq_{A, B}$, such that it satisfies:
\begin{align*}
(f ; g) ; h &\simeq f ; (g ; h) \\
f ; \mathit{id}_B &\simeq f \\
\mathit{id}_A ; f &\simeq f \\
f \simeq f' \wedge g \simeq g' &\Rightarrow f ; g \simeq f' ; g' 
\end{align*}
for any $A, B, C, D \in \mathcal{C}$, $f, f' : A \rightarrow B$, $g, g' : B \rightarrow C$ and $h : C \rightarrow D$.
Moreover, it is \emph{\bfseries cartesian closed} iff:
\begin{itemize}

\item There is an object $T \in \mathcal{C}$, called a \emph{\bfseries $\boldsymbol{\beta}$-terminal object}, equipped with a $\beta$-morphism $!_A : A \rightarrow T$, called the \emph{\bfseries canonical $\boldsymbol{\beta}$-morphism} on $A$, for each $A \in \mathcal{C}$ that satisfies:
\begin{equation*} 
!_A \simeq t \ \text{for any $t : A \rightarrow T$;}
\end{equation*}

\item There is an object $A \times B \in \mathcal{C}$ for each pair $A, B \in \mathcal{C}$, called a \emph{\bfseries $\boldsymbol{\beta}$-(binary) product} of $ A$ and $B$, equipped with $\beta$-morphisms $\pi^{A, B}_1 : A \times B \rightarrow A$ and $\pi^{A, B}_2 : A \times B \rightarrow B$, called the first and the second \emph{\bfseries $\boldsymbol{\beta}$-projections} of $A \times B$, respectively, and an assignment $\langle \_, \_ \rangle$ of a $\beta$-morphism $\langle a, b \rangle^C_{A, B} : C \rightarrow A \times B$, called the \emph{\bfseries $\boldsymbol{\beta}$-pairing} of $a$ and $b$, to given $C \in \mathcal{C}$, $a : C \rightarrow A$ and $b : C \rightarrow B$, that satisfies:
\begin{align*}
\langle a, b \rangle^C_{A, B} ; \pi^{A, B}_1 &\simeq a \\
\langle a, b \rangle^C_{A, B} ; \pi^{A, B}_2 &\simeq b \\
\langle h ; \pi_1^{A, B}, h ; \pi_2^{A, B} \rangle^C_{A, B} &\simeq h \ \text{for any $h : C \rightarrow A \times B$} \\
a \simeq a' \wedge b \simeq b' \Rightarrow \langle a, b \rangle &\simeq \langle a', b' \rangle \ \text{for any $a' : C \rightarrow A$ and $b' : C \rightarrow B$;} 
\end{align*}

\item There are an object $C^B \in \mathcal{C}$ and a $\beta$-morphism $\mathit{ev}_{B, C} : C^B \times B \rightarrow C$, called the \emph{\bfseries $\boldsymbol{\beta}$-exponential} and the \emph{\bfseries $\boldsymbol{\beta}$-evaluation} of $B$ and $C$, respectively, for each pair $B, C \in \mathcal{C}$, equipped with an assignment $\Lambda$ of a $\beta$-morphism $\Lambda_{A, B, C}(k) : A \rightarrow C^B$ (also written $\Lambda^{B, C}_A(k)$ or $\Lambda_A(k)$), called the \emph{\bfseries $\boldsymbol{\beta}$-currying} of $k$, to given $A \in \mathcal{C}$ and $k : A \times B \rightarrow C$, that satisfies:
\begin{align*}
\langle \pi_1^{A, B} ; \Lambda_{A, B, C}(k), \pi_2^{A, B} \rangle^{A \times B}_{C^B, B} ; \mathit{ev}_{B, C} &\simeq k \\
\Lambda_{A, B, C}(\langle \pi_1^{A, B} ; l, \pi_2^{A, B} \rangle^{A \times B}_{C^B, B} ; \mathit{ev}_{B, C}) &\simeq l \ \text{for any $l : A \rightarrow C^B$} \\
k \simeq k' \Rightarrow \Lambda_{A, B, C}(k) &\simeq \Lambda_{A, B, C}(k') \ \text{for any $k' : A \times B \rightarrow C$} 
\end{align*}

\end{itemize}
where we often omit the sub/superscripts on $\pi_i^{A, B}$, $\langle \_, \_ \rangle^C_{A, B}$, $\mathit{ev}_{B, C}$ and $\Lambda_{A, B, C}$.
\end{definition}

That is, a (resp. cartesian closed) $\beta$-category $\mathcal{C} = (\mathcal{C}, \simeq)$ is a (resp. cartesian closed) category \emph{up to $\simeq$} (i.e., the equation $=$ on morphisms is replaced with the equivalence relation $\simeq$ on 1-cells), where the prefix `$\beta$-' signifies the compromise `up to $\simeq$'. 
Alternatively, regarding objects and $\beta$-morphisms of $\mathcal{C}$ as 0-cells and 1-cells, respectively, and defining 2-cells by $\mathcal{C}(A, B)(d, c) \stackrel{\mathrm{df. }}{=} \begin{cases} \{ \simeq \} &\text{if $d \simeq c$;} \\ \emptyset &\text{otherwise} \end{cases}$ for any $A, B \in \mathcal{C}$ and $d, c : A \rightarrow B$, where $\{ \simeq \}$ is any singleton set, we may identify $\mathcal{C}$ with a (resp. cartesian closed \cite{ouaknine1997two}) \emph{bicategory} whose 2-cells are only the trivial one.

We are now ready to define \emph{$\beta$-categories of computation (BoCs)}:
\begin{definition}[BoCs]
\label{DefBoCs}
A \emph{\bfseries $\boldsymbol{\beta}$-category of computation (BoC)} is a $\beta$-category $\mathcal{C} = (\mathcal{C}, \simeq)$ equipped with a (class) function $\mathcal{E}$ on $\beta$-morphisms of $\mathcal{C}$, called the \emph{\bfseries evaluation} (of computation), that satisfies:
\begin{itemize}

\item \textsc{(Subject reduction).} $\mathcal{E}(f) : A \rightarrow B$ for all $A, B \in \mathcal{C}$ and $f : A \rightarrow B$;

\item \textsc{(Termination).} $f \downarrow$ for all $A, B \in \mathcal{C}$ and $f : A \rightarrow B$;

\item \textsc{($\beta$-identities).} $\mathit{id}_A \in \mathcal{V}_{\mathcal{C}}(A, A)$ for all $A \in \mathcal{C}$; 

\item \textsc{(Evaluation).} $f \simeq f' \Leftrightarrow \exists v \in \mathcal{V}_{\mathcal{C}}(A, B) . \ \! f \downarrow v \wedge f' \downarrow v$ for all $A, B \in \mathcal{C}$ and $f, f' : A \rightarrow B$


\end{itemize}
where $\mathcal{V}_{\mathcal{C}}(A, B) \stackrel{\mathrm{df. }}{=} \{ v \in \mathcal{C}(A, B) \mid \mathcal{E}(v) = v \ \! \}$, whose elements are called \emph{\bfseries values} from $A$ to $B$, and we write $f \downarrow$, or specifically $f \downarrow \mathcal{E}^{n}(f)$, if $\mathcal{E}^{n}(f) \in \mathcal{V}_{\mathcal{C}}(A, B)$ for some $n \in \mathbb{N}$.\footnote{Note that if $\mathcal{E}^{n_1}(f), \mathcal{E}^{n_2}(f) \in \mathcal{V}_{\mathcal{C}}(A, B)$ for any $n_1, n_2 \in \mathbb{N}$, then clearly $\mathcal{E}^{n_1}(f) = \mathcal{E}^{n_2}(f)$, where $\mathcal{E}^n$ denotes the $n$-times iteration of $\mathcal{E}$ for all $n \in \mathbb{N}$.} 
It is \emph{\bfseries cartesian closed}, which we call a \emph{\bfseries cartesian closed BoC (CCBoC)}, iff so is $\mathcal{C}$ as a $\beta$-category, all the canonical $\beta$-morphisms, the $\beta$-projections and the $\beta$-evaluations of $\mathcal{C}$ are values, and all the $\beta$-pairing and the $\beta$-currying of $\mathcal{C}$ preserve values. 
\end{definition}

\begin{convention*}
Since the equivalence $\simeq$ of a BoC $\mathcal{C}$ may be completely recovered from the evaluation $\mathcal{E}$, we usually specify the BoC by a pair $\mathcal{C} = (\mathcal{C}, \mathcal{E})$.
If $f \downarrow \mathcal{E}^{n}(f)$ for some $n \in \mathbb{N}$, then we call $\mathcal{E}^{n}(f)$ the \emph{\bfseries value} of $f$ and also write $\mathcal{E}^\omega(f)$ for it.
\end{convention*}

The intuition behind Definition~\ref{DefBoCs} is as follows. 
In a BoC $\mathcal{C} = (\mathcal{C}, \mathcal{E})$, $\beta$-morphisms are (possibly \emph{intensional} but not necessarily `effective') \emph{computations} with the domain and the codomain (objects) specified, and values are \emph{extensional} computations such as functions (as graphs). 
The $\beta$-composition is `non-normalizing composition' or \emph{concatenation} of computations, and $\beta$-identities are \emph{unit computations} (they are just like identity functions). 
The execution of a computation $f$ is achieved by \emph{evaluating} it into a unique value $\mathcal{E}^\omega(f)$, which corresponds to \emph{dynamics} of computation.\footnote{In the present work, every dynamic strategy (or $\beta$-morphism) becomes a value by a \emph{finite} iteration of the hiding operation (or evaluation) due to the axiom on labeling functions (Definition~\ref{DefDynamicArenas}), and thus the axiom Termination (Definition~\ref{DefBoCs}) makes sense. Of course, if we consider another, in particular finer, evaluation of computations (which is left as future work), then this point may no longer hold.}
In addition, the equivalence relation $\simeq$ witnesses the \emph{extensional equivalence} between $\beta$-morphisms modulo $\mathcal{E}^\omega$.
The four axioms then should make sense from this perspective. 
In this way, a BoC provides a `universe' of dynamic, intensional computations.

It is easy to see that a BoC $\mathcal{C} = (\mathcal{C}, \mathcal{E})$ induces the category $\mathcal{V}_{\mathcal{C}}$ given by:
\begin{itemize}

\item Objects are those of $\mathcal{C}$;

\item Morphisms $A \to B$ are elements in $\mathcal{V}_{\mathcal{C}}(A, B)$, i.e., values from $A$ to $B$ in $\mathcal{C}$;

\item The composition of morphisms $u : A \to B$ and $v : B \to C$ is $\mathcal{E}^\omega(u ; v) : A \to C$;

\item Identities are $\beta$-identities in $\mathcal{C}$.

\end{itemize}
Regarding the BoC $\mathcal{C}$ as the trivial bicategory as already specified above, and the category $\mathcal{V}_{\mathcal{C}}$ as the trivial 2-category, the evaluation $\mathcal{E}$ induces the 2-functor $\mathcal{E}^\omega : \mathcal{C} \to \mathcal{V}_{\mathcal{C}}$ that maps $A \mapsto A$ for 0-cells $A$, $f \mapsto \mathcal{E}^\omega(f)$ for 1-cells $f$, and $\simeq \ \mapsto \ =$ for 2-cells $\simeq$.
Clearly, $\mathcal{V}_{\mathcal{C}}$ is cartesian closed if so is $\mathcal{C}$, where canonical morphisms into a terminal object, projections, evaluations, pairing and currying of $\mathcal{V}_{\mathcal{C}}$ are respectively the corresponding `$\beta$-ones' in $\mathcal{C}$. 

The point here is that we may decompose the standard interpretation $\llbracket \_ \rrbracket_{\mathcal{S}}$ of functional programming languages in a CCC $\mathcal{V}_{\mathcal{C}}$ \cite{lambek1988introduction,pitts2001categorical,crole1993categories,jacobs1999categorical} as a more intensional interpretation $\llbracket \_ \rrbracket_{\mathcal{D}}$ in a CCBoC $\mathcal{C} = (\mathcal{C}, \mathcal{E})$ and the full evaluation $\mathcal{E}^\omega : \mathcal{C} \to \mathcal{V}_{\mathcal{C}}$, i.e., $\llbracket \_ \rrbracket_{\mathcal{S}} = \mathcal{E}^\omega(\llbracket \_ \rrbracket_{\mathcal{D}})$, and talk about \emph{intensional difference} between computations: Terms $\mathsf{M}$ and $\mathsf{M'}$ are interpreted to be \emph{intensionally equal} if $\llbracket \mathsf{M} \rrbracket_{\mathcal{D}} = \llbracket \mathsf{M'} \rrbracket_{\mathcal{D}}$ and \emph{extensionally equal} if $\llbracket \mathsf{M} \rrbracket_{\mathcal{D}} \simeq \llbracket \mathsf{M'} \rrbracket_{\mathcal{D}}$.
Also, the \emph{one-step} evaluation $\mathcal{E}$ is to capture the small-step operational semantics of the target language, i.e., to satisfy the DCP (see Definition~\ref{DefDCP} for the precise definition specialized to our target language). 


\subsection{Finitary PCF}
\label{FPCF}
Next, let us introduce our target programming language for dynamic game semantics. 

First, recall that there is a one-to-one correspondence between \emph{PCF B\"{o}hm trees} (i.e., terms of PCF in \emph{$\eta$-long normal form}) \cite{amadio1998domains} and innocent, well-bracketed strategies \cite{hyland2000full,abramsky1999game,curien2006notes}; this highlight in the literature of game semantics is called \emph{strong definability}.
Naturally, we would like to exploit the strong definability result to establish the first instance of dynamic game semantics as the task would be easier than otherwise. 

On the other hand, the higher-order functional programming language \emph{PCF} \cite{scott1993type,plotkin1977lcf} has the \emph{natural number type} and the \emph{fixed-point combinators}, which make PCF B\"{o}hm trees \emph{infinitary} in width and depth, respectively. 
However, we would like to select, as the first target language for dynamic game semantics, the simplest one possible because then the idea and the mechanism would be most visible. 
For this reason, let us choose \emph{finitary PCF}, i.e., the finite fragment of PCF that has only the \emph{boolean type} as the ground type (or equivalently, the \emph{simply-typed $\lambda$-calculus} \cite{church1940formulation,sorensen2006lectures} equipped with the boolean type).

We then define a simple small-step operational semantics (or reduction strategy) of finitary PCF whose execution order is obvious from types and has an immediate counterpart in dynamic game semantics.

\begin{remark*}
Note that an execution of \emph{linear head reduction (LHR)} \cite{danos2004head} corresponds in a step-by-step fashion to an `internal communication' between strategies \cite{danos1996game}.
Hence, one may wonder if it would be better to employ LHR as the operational semantics of finitary PCF; however, note that:
\begin{itemize}

\item The correspondence is \emph{not} between terms and strategies;

\item LHR is executed by \emph{linear substitution}, which makes the calculus very different from the usual $\lambda$-calculus with $\beta$-reduction.


\end{itemize}
By these two points, we have conjectured that it would require significantly more work than the present work to establish a game-semantic DCP with respect to LHR, and therefore we leave it as future work.
\end{remark*}

In the following, we give the precise definition of the resulting target programming language (viz., finitary PCF equipped with the small-step operational semantics).

\begin{notation*}
We employ the following notations:

\begin{itemize}

\item Let $\mathscr{V}$ be a countably infinite set of \emph{variables}, written $\mathsf{x}$, $\mathsf{y}$, $\mathsf{z}$, etc., for which we assume the \emph{variable convention} (or \emph{Barendregt's convention} \cite{Hankin1994-HANLCA-2}\footnote{I.e., we assume that in any term of concern every bound variable is chosen to be different from any free variable occurring in that mathematical context.});

\item We use sans-serif letters such as $\mathsf{\Gamma}$, $\mathsf{A}$ and $\mathsf{a}$ for syntactic objects and $\equiv$ for syntactic equality up to \emph{$\alpha$-equivalence}, i.e., up to renaming of bound variables.

\end{itemize}
\end{notation*}

\begin{definition}[FPCF]
\label{DefFPCF}
The \emph{\bfseries finitary PCF (FPCF)} is a functional programming language defined as follows: 
\begin{itemize}

\item \textsc{(Types).} A \emph{\bfseries type} $\mathsf{A}$ is an expression generated by the grammar:
\begin{equation*}
\mathsf{A} \stackrel{\mathrm{df. }}{\equiv} o \ | \ \mathsf{A_1} \Rightarrow \mathsf{A_2}
\end{equation*}
where $o$ is the \emph{\bfseries boolean type} and $\mathsf{A_1 \Rightarrow A_2}$ is the \emph{\bfseries function type} from $\mathsf{A_1}$ to $\mathsf{A_2}$ ($\mathsf{\Rightarrow}$ is right associative).
We write $\mathsf{A, B, C}$, etc. for types.
Note that each type $\mathsf{A}$ may be written uniquely of the form $\mathsf{A_1} \Rightarrow \mathsf{A_2} \Rightarrow \dots \Rightarrow \mathsf{A_k} \Rightarrow o$, where $k \in \mathbb{N}$.

\item \textsc{(Raw-terms).} A \emph{\bfseries raw-term} $\mathsf{M}$ is an expression generated by the grammar:
\begin{equation*}
\mathsf{M} \stackrel{\mathrm{df. }}{\equiv} \mathsf{x} \mid \mathsf{tt} \mid \mathsf{ff} \mid \mathsf{case(M)[M_1 ; M_2]} \mid \lambda \mathsf{x^A . M} \mid \mathsf{M_1 M_2}
\end{equation*}
where $\mathsf{x}$ ranges over variables, and $\mathsf{A}$ over types. 
We call $\mathsf{tt}$, $\mathsf{ff}$, $\lambda \mathsf{x^A . M}$ and $\mathsf{M_1 M_2}$ respectively the \emph{\bfseries true constant}, the \emph{\bfseries false constant}, an \emph{\bfseries abstraction} and an \emph{\bfseries application}.
We write $\mathsf{M, P, Q, R}$, etc. for raw-terms and often omit $\mathsf{A}$ in an abstraction $\lambda \mathsf{x^A}$; an application is always left-associative, e.g., $\mathsf{M_1M_2M_3}$ may be written informally $\mathsf{(M_1M_2)M_3}$.
The set $\mathscr{FV}(\mathsf{M}) \subseteq \mathscr{V}$ of all \emph{\bfseries free variables} occurring in a raw-term $\mathsf{M}$ is defined by the following induction on $\mathsf{M}$:
\begin{align*}
\mathscr{FV}(\mathsf{x}) &\stackrel{\mathrm{df. }}{=} \{ \mathsf{x} \} \\
\mathscr{FV}(\mathsf{tt}) &\stackrel{\mathrm{df. }}{=} \mathscr{FV}(\mathsf{ff}) \stackrel{\mathrm{df. }}{=} \emptyset \\
\mathscr{FV}(\mathsf{case(M)[M_1 ; M_2]}) &\stackrel{\mathrm{df. }}{=} \mathscr{FV}(\mathsf{M}) \cup \mathscr{FV}(\mathsf{M_1}) \cup \mathscr{FV}(\mathsf{M_2}) \\ 
\mathscr{FV}(\lambda \mathsf{x . M}) &\stackrel{\mathrm{df. }}{=} \mathscr{FV}(\mathsf{M}) \setminus \{ \mathsf{x} \} \\
\mathscr{FV}(\mathsf{M_1 M_2}) &\stackrel{\mathrm{df. }}{=} \mathscr{FV}(\mathsf{M_1}) \cup \mathscr{FV}(\mathsf{M_2}).
\end{align*}

\item \textsc{(Contexts).} A \emph{\bfseries context} is a finite sequence $\mathsf{x_1 : A_1, x_2 : A_2, \dots, x_k : A_k}$ of (variable : type)-pairs with $\mathsf{x_i \neq x_j}$ if $i \neq j$, where $i, j \in \overline{k}$.
We write $\mathsf{\Gamma}$, $\mathsf{\Delta}$, $\mathsf{\Theta}$, etc. for contexts.

\item \textsc{(Terms).} A \emph{\bfseries term} is an expression of the form $\mathsf{\Gamma \vdash M : B}$, where $\mathsf{\Gamma}$ is a \emph{\bfseries context}, $\mathsf{M}$ is a raw-term, and $\mathsf{B}$ is a type, generated by the following \emph{\bfseries typing rules}:
\begin{align*}
&\textsc{(B)} \frac{ \ \mathsf{b} \in \{ \mathsf{tt}, \mathsf{ff} \} \ }{ \ \mathsf{\Gamma \vdash b} : o \ } \ \ \ \textsc{(C1)} \frac{ \ \begin{aligned} &\mathsf{A} \equiv \mathsf{A_1} \Rightarrow \mathsf{A_2} \Rightarrow \dots \Rightarrow \mathsf{A_k} \Rightarrow o \ \ \mathsf{\Gamma} \equiv \mathsf{\Delta, \Theta} \\
&\forall i \in \overline{k} . \ \! \mathsf{\Gamma \vdash V_i : A_i} \wedge \sharp(\mathsf{V_i}) = 0 \wedge \mathsf{x} \not \in \mathscr{FV}(\mathsf{V_i}) \\ &\forall j \in \overline{2} . \ \! \mathsf{\Gamma \vdash W_j} : o \wedge \sharp{(\mathsf{W_j})} = 0 \wedge \mathsf{x} \not \in \mathscr{FV}(\mathsf{W_j}) \end{aligned} \ }{ \ \mathsf{\Delta, x : A, \Theta \vdash case(x V_1 V_2 \dots V_k)[W_1 ;  W_2]} : o \ } \\ \\
&\textsc{(C2)} \frac{ \ \mathsf{\Gamma \vdash M} : o \ \ \forall j \in \overline{2} . \ \! \mathsf{\Gamma \vdash P_j} : o \ }{ \ \mathsf{\Gamma \vdash case(M)[P_1 ; P_2]} : o \ } \ \ \ \textsc{(L)} \frac{ \ \mathsf{\Gamma, x : A \vdash M : B} \ }{ \ \mathsf{\Gamma} \vdash \lambda \mathsf{x^A . \ \! M} : \mathsf{A \Rightarrow B} \ }  \\ \\
&\textsc{(A)} \frac{ \ \mathsf{\Gamma \vdash M_1 : A \Rightarrow B \ \ \Gamma \vdash M_2 : A} \ }{ \ \mathsf{\Gamma \vdash M_1 M_2 : B} \ } 
\end{align*}
where $\sharp(\mathsf{\Gamma \vdash M : B}) \in \mathbb{N}$, often abbreviated as $\sharp(\mathsf{M})$, is the \emph{\bfseries execution number} of each term $\mathsf{\Gamma \vdash M : B}$ defined by the following induction on $\mathsf{\Gamma \vdash M : B}$:
\begin{itemize}

\item $\sharp(\mathsf{b}) \stackrel{\mathrm{df. }}{=} 0$ if $\mathsf{b} \in \{ \mathsf{tt}, \mathsf{ff} \}$;

\item $\sharp(\mathsf{case(x V_1 V_2 \dots V_k)[W_1 ; W_2]}) \stackrel{\mathrm{df. }}{=} 0$;

\item $\sharp(\mathsf{case(M)[P_1 ; P_2]}) \stackrel{\mathrm{df. }}{=} 0$;

\item $\sharp(\lambda \mathsf{x^A . M}) \stackrel{\mathrm{df. }}{=} \sharp(\mathsf{M})$;

\item $\sharp(\mathsf{M_1 M_2}) \stackrel{\mathrm{df. }}{=} \max (\sharp(\mathsf{M_1}), \sharp(\mathsf{M_1})) + 1$. 

\end{itemize}

We write $\mathsf{\Gamma} \vdash \{ \mathsf{M} \}_e : \mathsf{B}$ for the term $\mathsf{\Gamma} \vdash \mathsf{M} : \mathsf{B}$ such that $\sharp(\mathsf{M}) = e$.
Also, we often omit the context and/or the type of a term if it does not bring confusion. 
A \emph{\bfseries program} (resp. a \emph{\bfseries value}) is a term generated by the rules B, C1, L and A (resp. B, C1 and L).
A \emph{\bfseries subterm} of a term $\mathsf{\Gamma \vdash M : B}$ is a term that occurs in the deduction of $\mathsf{\Gamma \vdash M : B}$, where note that a deduction (tree) of each term of FPCF is clearly \emph{unique}.

\begin{remark*}
The rules C2 above and $\vartheta_4$ below are necessary for `intermediate terms' during an evaluation of a program into a value.
\end{remark*}

\item \textsc{($\beta \vartheta$-reduction).} The \emph{\bfseries $\boldsymbol{\beta \vartheta}$-reduction} $\to_{\beta \vartheta}$ on terms is the \emph{contextual closure}, i.e., the closure with respect to the typing rules, of the union of the following five rules:
\begin{align*}
(\lambda \mathsf{x . \ \! M}) \mathsf{P} &\to_{\beta} \mathsf{M [P / x]} \\
\mathsf{case (tt) [M_1 ; M_2]} &\to_{\vartheta_1} \! \mathsf{M_1} \\
\mathsf{case (ff) [M_1 ; M_2]} &\to_{\vartheta_2} \! \mathsf{M_2} \\
\mathsf{case (case (x \boldsymbol{\mathsf{V}})[W_1 ; W_2]) [M_1 ; M_2]} &\to_{\vartheta_3} \!\mathsf{case (x \boldsymbol{\mathsf{V}})[case(W_1)[M_1 ; M_2] ; case(W_2)[M_1 ; M_2]]} \\
\mathsf{case (case (M)[P_1 ; P_2]) [Q_1 ; Q_2]} &\to_{\vartheta_4} \! \mathsf{case (M)[case(P_1)[Q_1 ; Q_2] ; case(P_2)[Q_1 ; Q_2]]}
\end{align*}
where $\mathsf{M[P/x]}$ denotes the \emph{capture-free substitution} \cite{Hankin1994-HANLCA-2} of $\mathsf{P}$ for $\mathsf{x}$ in $\mathsf{M}$, and $\mathsf{x \boldsymbol{\mathsf{V}}}$ abbreviates $\mathsf{x V_1 V_2 \dots V_k}$ of the rule C1.
We write $\mathit{nf}(\mathsf{M})$ for the \emph{normal form} of each term $\mathsf{M}$ with respect to $\rightarrow_{\beta \vartheta}$, i.e., $\mathit{nf}(\mathsf{M})$ is a term such that $\mathsf{M} \rightarrow_{\beta \vartheta}^{\ast} \mathit{nf}(\mathsf{M})$ and $\mathit{nf}(\mathsf{M}) \not \rightarrow_{\beta \vartheta} \mathsf{M'}$ for any term $\mathsf{M'}$, which uniquely exists by Theorems~\ref{ThmCR} and \ref{ThmNormalization} given below.
The \emph{\bfseries parallel $\boldsymbol{\beta \vartheta}$-reduction} $\rightrightarrows_{\beta \vartheta}$ on terms evaluates each term $\mathsf{M}$ in a single-step to its normal form $\mathit{nf}(\mathsf{M})$.

\item \textsc{(Operational semantics).} The \emph{\bfseries (small-step) operational semantics} (or the \emph{\bfseries reduction strategy}) $\to$ on programs $\mathsf{M}$ is the `simultaneous execution' of $\rightrightarrows_{\beta \vartheta}$ on all subterms of $\mathsf{M}$ with the execution number $1$, or more precisely $\rightarrow$ is defined by:
\begin{equation*}
\mathsf{M} \rightarrow \begin{cases} \mathsf{V} &\text{if $\mathsf{M \equiv M_1 M_2}$, $\sharp(\mathsf{M_1M_2}) = 1$ and $\mathsf{M_1 M_2} \rightrightarrows_{\beta \vartheta} \mathsf{V}$;} \\ \mathsf{M'_1 M'_2} &\text{if $\mathsf{M \equiv M_1 M_2}$, $\sharp(\mathsf{M_1 M_2) \geqslant 2}$ and $\mathsf{M_i} \rightarrow \mathsf{M'_i}$ for $i = 1, 2$;} \\ \mathsf{\lambda x^A . \ \! \tilde{M}'} &\text{if $\mathsf{M \equiv \lambda x^A . \ \! \tilde{M}}$ and $\mathsf{\tilde{M}} \rightarrow \mathsf{\tilde{M}'}$.} \end{cases}
\end{equation*}
\end{itemize}

\begin{remark*}
The operational semantics $\rightarrow$ of FPCF might appear a bit unusual, but as we shall see, it has a natural game-semantic counterpart, i.e., it makes sense from the game-semantic point of view. 
\end{remark*}

$\boldsymbol{\mathsf{Eq(FPCF)}}$ is the equational theory that consists of judgements $\mathsf{\Gamma \vdash M = M' : B}$, where $\mathsf{\Gamma \vdash M : B}$ and $\mathsf{\Gamma \vdash M' : B}$ are terms of FPCF such that $\mathit{nf}(\mathsf{M}) \equiv \mathit{nf}(\mathsf{M'})$.

\end{definition}



Note that values of FPCF are PCF B\"{o}hm trees except that the `bottom term' $\mathsf{\bot}$ and the \emph{natural number type} $\iota$ are excluded; the $\beta \vartheta$-reduction $\rightarrow_{\beta \vartheta}$ is essentially taken from Section 6 of the book \cite{amadio1998domains}.

\begin{remark*}
Let $\mathsf{A} \equiv \mathsf{A_1} \Rightarrow \mathsf{A_2} \Rightarrow \dots \Rightarrow \mathsf{A_k} \Rightarrow o$ be an arbitrary type of FPCF. 
Note that an expression of the form $\mathsf{\Delta, x:A, \Theta} \vdash \mathsf{x} : \mathsf{A}$ is \emph{not} a term of FPCF, but instead there is another $\mathsf{\Delta, x:A, \Theta} \vdash \underline{\mathsf{x}}^{\mathsf{A}} : \mathsf{A}$, where $\mathsf{\underline{x}^A \stackrel{\mathrm{df. }}{\equiv} \lambda x_1^{A_1} x_2^{A_2} \dots x_k^{A_k} . \ \! \mathsf{case} (x \underline{x_1}^{A_1} \underline{x_2}^{A_2} \dots \underline{x_k}^{A_k}) [tt ; ff]}$, which \emph{is} a term of FPCF.
We often write $\mathsf{\underline{x}}$ for $\mathsf{\underline{x}^A}$ if it does not bring confusion. 
\end{remark*}

Thus, FPCF computes as follows. 
Given a program $\mathsf{\Gamma} \vdash \{ \mathsf{M} \}_e : \mathsf{B}$, it produces a \emph{finite} chain of \emph{finitary} rewriting
\begin{equation}
\label{ComputationOfSystemT}
\mathsf{M} \to \mathsf{M_1} \to \mathsf{M_2} \to \dotsm \to \mathsf{M_e}
\end{equation}
where $\mathsf{M_e}$ is a value.
Note that the program $\mathsf{M}$ is constructed from values by a finite number of applications, and the computation (\ref{ComputationOfSystemT}) is executed in the \emph{first-applications-first-evaluated} fashion, e.g., if $\mathsf{M \equiv (V_1 V_2) ( (V_3 V_4) (V_5 V_6) )}$ and $e = 3$, where $\mathsf{V_1, V_2, \dots, V_6}$ are values, then the computation (\ref{ComputationOfSystemT}) would be of the form
\begin{equation*}
\mathsf{(V_1 V_2) ( (V_3 V_4) (V_5 V_6) ) \to V_7 (V_8 V_9) \to V_7 V_{10} \to V_{11}}
\end{equation*}
where $\mathsf{V_7 \equiv \mathit{nf}(V_1 V_2)}$, $\mathsf{V_8 \equiv \mathit{nf}(V_3 V_4)}$, $\mathsf{V_9 \equiv \mathit{nf}(V_5 V_6)}$, $\mathsf{V_{10} \equiv \mathit{nf}(V_8 V_9)}$ and $\mathsf{V_{11} \equiv \mathit{nf}(V_7 V_{10})}$.

The rest of the present section is devoted to showing that the computation (\ref{ComputationOfSystemT}) of FPCF in fact correctly works (Corollary~\ref{CoroCorrectness}). 

First, by the following Proposition~\ref{PropUniqueTyping} and Theorem~\ref{ThmSubjectReduction}, it makes sense that $\to_{\beta \vartheta}$ is defined \emph{on terms} (not on raw-terms):
\begin{proposition}[Unique typing]
\label{PropUniqueTyping}
If $\mathsf{\Gamma} \vdash \{ \mathsf{M} \}_e : \mathsf{B}$ and $\mathsf{\Gamma} \vdash \{ \mathsf{M} \}_{e'} : \mathsf{B'}$, then $e = e'$ and $\mathsf{B \equiv B'}$.
\end{proposition}
\begin{proof}
By induction on the construction of $\mathsf{\Gamma \vdash M : B}$.
\end{proof}

\begin{lemma}[Free variable lemma]
\label{LemFreeVariables}
If $\mathsf{\Gamma \vdash M : B}$, and $\mathsf{x} \in \mathscr{V}$ occurs free in $\mathsf{M}$, then $\mathsf{x : A}$ occurs in $\mathsf{\Gamma}$ for some type $\mathsf{A}$.
\end{lemma}
\begin{proof}
By induction on the construction of $\mathsf{\Gamma \vdash M : B}$.
\end{proof}
 
\begin{lemma}[EW-lemma]
\label{LemExchangeAndWeakening}
If $\mathsf{x_1 : A_1, x_2 : A_2, \dots, x_k : A_k} \vdash \{ \mathsf{M} \}_e : \mathsf{B}$, then:
\begin{enumerate}

\item $\mathsf{x_{\sigma(1)} : A_{\sigma (1)}, x_{\sigma(2)} : A_{\sigma(2)}, \dots, x_{\sigma(k)} : A_{\sigma (k)}} \vdash \{ \mathsf{M} \}_e : \mathsf{B}$ for any permutation $\sigma$ of $\overline{k}$;

\item $\mathsf{x_1 : A_1, x_2 : A_2, \dots, x_k : A_k, x_{k+1} : A_{k+1}} \vdash \{ \mathsf{M} \}_{e} : \mathsf{B}$ for any variable $\mathsf{x_{k+1}} \in \mathscr{V}$ and type $\mathsf{A_{k+1}}$ such that $\mathsf{x_{k+1} \not \equiv x_i}$ for $i = 1, 2, \dots, k$.

\end{enumerate}
\end{lemma}
\begin{proof}
By induction on the construction of $\mathsf{x_1 : A_1, x_2 : A_2, \dots, x_k : A_k \vdash M : B}$.
\end{proof}

\begin{lemma}[Substitution lemma]
\label{LemSubstitution}
If $\mathsf{\Gamma, x : A} \vdash \{ \mathsf{P} \}_{e} : \mathsf{B}$ and $\mathsf{\Gamma} \vdash \mathsf{Q} : \mathsf{A}$, then $\mathsf{\Gamma} \vdash \{ \mathsf{P[Q/x]} \}_{e} : \mathsf{B}$.
\end{lemma}
\begin{proof}
By induction on $|\mathsf{P}|$ with the help of Lemmata~\ref{LemFreeVariables} and \ref{LemExchangeAndWeakening}.
\end{proof}

\begin{theorem}[Subject reduction]
\label{ThmSubjectReduction}
If $\mathsf{\Gamma \vdash M : B}$ and $\mathsf{M \to_{\beta \vartheta} R}$, then $\mathsf{\Gamma \vdash R : B}$.
\end{theorem}
\begin{proof}
By induction on the structure $\mathsf{M \to_{\beta \vartheta} R}$ with the help of Lemma~\ref{LemSubstitution}.
\end{proof}

Next, we show that $\rightrightarrows_{\beta \vartheta}$ is well-defined (Theorems~\ref{ThmCR} and \ref{ThmNormalization}). 
\begin{lemma}[Hindley-Rosen]
\label{LemHR}
Let $R_1$ and $R_2$ be binary relations on the set $\mathcal{T}$ of all terms, and let us write $\to_{R_i}$ for the contextual closure of $R_i$ for $i = 1, 2$. If $\to_{R_1}$ and $\to_{R_2}$ are Church-Rosser, and satisfy $\forall \mathsf{M, P, Q} \in \mathcal{T} . \ \! \mathsf{M \to_{R_1}^* \! P} \wedge \mathsf{M \to_{R_2}^* \! Q \Rightarrow \exists R \in \mathcal{T} . \ \! P \to_{R_2}^* \! R} \wedge \mathsf{Q \to_{R_1}^* \! R}$, then $\to_{R_1 \cup R_2}$ is Church-Rosser.
\end{lemma}
\begin{proof}
By simple `diagram chase'; see \cite{Hankin1994-HANLCA-2} for the details.
\end{proof}

\begin{theorem}[Church-Rosser]
\label{ThmCR}
The $\beta \vartheta$-reduction $\to_{\beta \vartheta}$ is Church-Rosser.
\end{theorem}
\begin{proof}
First, it is easy to see that the \emph{$\vartheta$-reduction} $\to_{\vartheta} \ \! \stackrel{\mathrm{df. }}{=} \ \! \bigcup_{i=1}^4 \! \to_{\vartheta_i}$ satisfies the diamond-property, and thus it is Church-Rosser. 

Also, we may show that:
\begin{equation}
\label{semidiamond}
\mathsf{M \to_\beta P \wedge M \to_\vartheta Q} \Rightarrow \exists \mathsf{R} . \ \! \mathsf{P \to_\vartheta^\ast R \wedge Q \to_\beta R}
\end{equation}
for all terms $\mathsf{M}$, $\mathsf{P}$ and $\mathsf{Q}$, where note the asymmetry of $\to_\vartheta$ and $\to_\beta$, by a case analysis on the relation between $\beta$- and $\vartheta$-redexes in $\mathsf{M}$:
\begin{itemize}

\item If the $\beta$-redex is inside the $\vartheta$-redex, then it is easy to see that (\ref{semidiamond}) holds;

\item If the $\vartheta$-redex is inside the body of the function subterm of the $\beta$-redex, then it suffices to show that $\to_{\vartheta}$ commutes with substitution, but it is straightforward;

\item If $\vartheta$-redex is inside the argument of the $\beta$-redex, then it may be duplicated by a finite number $n$, but whatever the number $n$ is, (\ref{semidiamond}) clearly holds;

\item If the $\beta$- and $\vartheta$-redexes are disjoint, then (\ref{semidiamond}) trivially holds.
 
\end{itemize}

It then follows from (\ref{semidiamond}) that:
\begin{equation}
\label{SemiChurchRosser}
\mathsf{M \to_\beta P \wedge M \to_\vartheta^\ast Q} \Rightarrow \exists \mathsf{R} . \ \! \mathsf{P \to_\vartheta^\ast R \wedge Q \to_\beta R}
\end{equation}
which in turn implies that:
\begin{equation}
\label{ChurchRosser}
\mathsf{M \to_\beta^\ast P \wedge M \to_\vartheta^\ast Q} \Rightarrow \exists \mathsf{R} . \ \! \mathsf{P \to_\vartheta^\ast R \wedge Q \to_\beta^\ast R}
\end{equation}
for all terms $\mathsf{M}$, $\mathsf{P}$ and $\mathsf{Q}$. 
Applying Lemma~\ref{LemHR} to (\ref{ChurchRosser}) (or equivalently by the well-known `diagram chase' argument on $\to_\beta^\ast$ and $\to_\vartheta^\ast$), we may conclude that the $\beta \vartheta$-reduction $\to_{\beta \vartheta} \ \! = \ \! \to_{\beta} \cup \to_{\vartheta}$ is Church-Rosser, completing the proof.
\end{proof}

Now, we show \emph{strong normalization} of $\rightarrow_{\beta \vartheta}$, i.e., there is no infinite chain of $\rightarrow_{\beta \vartheta}$:
\begin{theorem}[SN]
\label{ThmNormalization}
The $\beta \vartheta$-reduction $\rightarrow_{\beta \vartheta}$ is strongly normalizing (SN).
\end{theorem}
\begin{proof}
By a slight, straightforward modification of the proof of strong normalization of the simply-typed $\lambda$-calculus in \cite{Hankin1994-HANLCA-2}.
\end{proof}

Thus, it follows from Theorems~\ref{ThmCR} and \ref{ThmNormalization} that the normal form $\mathit{nf}(\mathsf{M})$ of each term $\mathsf{M}$ of FPCF (with respect to $\rightarrow_{\beta \vartheta}$) uniquely exists. 
Moreover, we have:
\begin{theorem}[Normal forms are values]
\label{ThmNormalFormsAreValues}
The normal form $\mathit{nf}(\mathsf{M})$ of every program $\mathsf{M}$ (with respect to $\rightarrow_{\beta \vartheta}$) is a value.
\end{theorem}
\begin{proof}
It has been shown in \cite{amadio1998domains} during the proof to show that PCF B\"{o}hm trees are closed under composition.
\end{proof}

Therefore, we have shown that the operational semantics $\rightarrow$ is well-defined:
\begin{corollary}[Correctness of operational semantics]
\label{CoroCorrectness}
If $\mathsf{\Gamma} \vdash \{ \mathsf{M} \}_{e} : \mathsf{B}$ is a program, and $e > 1$ (resp. $e = 1$), then there exists a unique program (resp. value) $\mathsf{\Gamma} \vdash \{ \mathsf{M'} \}_{e-1} : \mathsf{B}$ that satisfies $\mathsf{M} \rightarrow \mathsf{M'}$.
\end{corollary}
\begin{proof}
By Theorems~\ref{ThmSubjectReduction}, \ref{ThmCR}, \ref{ThmNormalization} and \ref{ThmNormalFormsAreValues}.
\end{proof}

\subsection{Dynamic Bicategorical Semantics of Finitary PCF}
Next, we present a general, categorical recipe to give semantics of FPCF in a CCBoC in such a way that satisfies the DCP.

\begin{definition}[Structures for FPCF]
\label{DefStructuresForFPCF}
A \emph{\bfseries structure} for FPCF in a CCBoC $\mathcal{C} = (\mathcal{C}, \mathcal{E})$ is a tuple $\mathcal{S} = (\mathscr{B}, 1, \times, \pi, \Rightarrow, \mathit{ev}, \underline{\mathit{tt}}, \underline{\mathit{ff}}, \vartheta)$ such that:
\begin{itemize}

\item $\mathscr{B} \in \mathcal{C}$;

\item $1$, $(\times, \pi_1, \pi_2)$ and $(\Rightarrow, \mathit{ev})$ are respectively a $\beta$-terminal object, a $\beta$-product (with $\beta$-projections) and a $\beta$-exponential (with $\beta$-evaluations) in $\mathcal{C}$; 

\item $\underline{\mathit{tt}}, \underline{\mathit{ff}}: 1 \rightarrow \mathscr{B}$ and  $\vartheta : \mathscr{B} \times (\mathscr{B} \times \mathscr{B}) \rightarrow \mathscr{B}$ are values in $\mathcal{C}$.

\end{itemize}
The \emph{\bfseries interpretation} $\llbracket \_ \rrbracket_{\mathcal{C}}^{\mathcal{S}}$ of FPCF induced by $\mathcal{S}$ in $\mathcal{C}$ assigns an object $\llbracket \mathsf{A} \rrbracket_{\mathcal{C}}^{\mathcal{S}} \in \mathcal{C}$ to each type $\mathsf{A}$, an object $\llbracket \mathsf{\Gamma} \rrbracket_{\mathcal{C}}^{\mathcal{S}} \in \mathcal{C}$ to each context $\mathsf{\Gamma}$, and a $\beta$-morphism $\llbracket \mathsf{M} \rrbracket_{\mathcal{C}}^{\mathcal{S}} : \llbracket \mathsf{\Gamma} \rrbracket_{\mathcal{C}}^{\mathcal{S}} \to \llbracket \mathsf{B} \rrbracket_{\mathcal{C}}^{\mathcal{S}}$ to each term $\mathsf{\Gamma \vdash M : B}$ as follows:
\begin{itemize}

\item \textsc{(Types).}  $\llbracket o \rrbracket_{\mathcal{C}}^{\mathcal{S}} \stackrel{\mathrm{df. }}{=} \mathscr{B}$ and $\llbracket \mathsf{A \Rightarrow B} \rrbracket_{\mathcal{C}}^{\mathcal{S}} \stackrel{\mathrm{df. }}{=} \llbracket \mathsf{A} \rrbracket_{\mathcal{C}}^{\mathcal{S}} \Rightarrow \llbracket \mathsf{B} \rrbracket_{\mathcal{C}}^{\mathcal{S}}$;

\item \textsc{(Contexts).} $\llbracket \boldsymbol{\epsilon} \rrbracket_{\mathcal{C}}^{\mathcal{S}} \stackrel{\mathrm{df. }}{=} 1$ and $\llbracket \mathsf{\Gamma, x : A} \rrbracket_{\mathcal{C}}^{\mathcal{S}} \stackrel{\mathrm{df. }}{=} \llbracket \mathsf{\Gamma} \rrbracket_{\mathcal{C}}^{\mathcal{S}} \times \llbracket \mathsf{A} \rrbracket_{\mathcal{C}}^{\mathcal{S}}$; 

\item \textsc{(Terms).} 
\begin{align*}
\llbracket \mathsf{\Gamma} \vdash \mathsf{tt} : o \rrbracket_{\mathcal{C}}^{\mathcal{S}} &\stackrel{\mathrm{df. }}{=} \mathcal{E}^\omega(!_{\llbracket \mathsf{\Gamma} \rrbracket_{\mathcal{C}}^{\mathcal{S}}} ; \underline{\mathit{tt}}) \\
\llbracket \mathsf{\Gamma} \vdash \mathsf{ff} : o \rrbracket_{\mathcal{C}}^{\mathcal{S}} &\stackrel{\mathrm{df. }}{=} \mathcal{E}^\omega(!_{\llbracket \mathsf{\Gamma} \rrbracket_{\mathcal{C}}^{\mathcal{S}}} ; \underline{\mathit{ff}}) \\
\llbracket \mathsf{\Gamma \vdash \lambda x . M : A \Rightarrow B} \rrbracket_{\mathcal{C}}^{\mathcal{S}} &\stackrel{\mathrm{df. }}{=} \Lambda_{\llbracket \mathsf{\Gamma} \rrbracket_{\mathcal{C}}^{\mathcal{S}}, \llbracket \mathsf{A} \rrbracket_{\mathcal{C}}^{\mathcal{S}}, \llbracket \mathsf{B} \rrbracket_{\mathcal{C}}^{\mathcal{S}}} (\llbracket \mathsf{\Gamma, x : A \vdash M : B} \rrbracket_{\mathcal{C}}^{\mathcal{S}}) \\
\llbracket \mathsf{\Gamma \vdash M N : B} \rrbracket_{\mathcal{C}}^{\mathcal{S}} &\stackrel{\mathrm{df. }}{=} \langle \llbracket \mathsf{\Gamma \vdash M : A \Rightarrow B} \rrbracket_{\mathcal{C}}^{\mathcal{S}}, \llbracket \mathsf{\Gamma \vdash N : A} \rrbracket_{\mathcal{C}}^{\mathcal{S}} \rangle_{\llbracket \mathsf{A \Rightarrow B} \rrbracket_{\mathcal{C}}^{\mathcal{S}}, \llbracket \mathsf{A} \rrbracket_{\mathcal{C}}^{\mathcal{S}}}^{\llbracket \mathsf{\Gamma} \rrbracket_{\mathcal{C}}^{\mathcal{S}}} ; \mathit{ev}_{\llbracket \mathsf{A} \rrbracket_{\mathcal{C}}^{\mathcal{S}}, \llbracket \mathsf{B} \rrbracket_{\mathcal{C}}^{\mathcal{S}}} \\
\llbracket \mathsf{\Gamma} \vdash \mathsf{case(x \boldsymbol{\mathsf{V}})[\mathsf{W_1} ; \mathsf{W_2}]} : o \rrbracket_{\mathcal{C}}^{\mathcal{S}} &\stackrel{\mathrm{df. }}{=} \mathcal{E}^\omega(\langle \llbracket \mathsf{\Gamma} \vdash \mathsf{x \boldsymbol{\mathsf{V}}} : o \rrbracket_{\mathcal{C}}^{\mathcal{S}}, \langle \llbracket \mathsf{\Gamma} \vdash \mathsf{W_1} : o \rrbracket_{\mathcal{C}}^{\mathcal{S}}, \llbracket \mathsf{\Gamma} \vdash \mathsf{W_2} : o \rrbracket_{\mathcal{C}}^{\mathcal{S}} \rangle \rangle ; \vartheta) \\
\llbracket \mathsf{\Gamma} \vdash \mathsf{case(M)[P_1 ; P_2]} : o \rrbracket_{\mathcal{C}}^{\mathcal{S}} &\stackrel{\mathrm{df. }}{=} \mathcal{E}^\omega(\langle \llbracket \mathsf{\Gamma} \vdash \mathsf{M} : o \rrbracket_{\mathcal{C}}^{\mathcal{S}}, \langle \llbracket \mathsf{\Gamma} \vdash \mathsf{P_1} : o \rrbracket_{\mathcal{C}}^{\mathcal{S}}, \llbracket \mathsf{\Gamma} \vdash \mathsf{P_2} : o \rrbracket_{\mathcal{C}}^{\mathcal{S}} \rangle \rangle ; \vartheta) 
\end{align*}
where $\llbracket \mathsf{\Gamma \vdash x : A} \rrbracket_{\mathcal{C}}^{\mathcal{S}} : \llbracket \mathsf{\Gamma} \rrbracket_{\mathcal{C}}^{\mathcal{S}} \rightarrow \llbracket \mathsf{A} \rrbracket_{\mathcal{C}}^{\mathcal{S}}$ (n.b., $\mathsf{\Gamma \vdash x : A}$ is not a term of FPCF, but we need it for the application $\mathsf{x \boldsymbol{\mathsf{V}}}$) is the obvious (possibly iterated) $\beta$-projection.
\end{itemize}
Moreover, the structure $\mathcal{S}$ is \emph{\bfseries standard} iff it satisfies the following five axioms:
\begin{enumerate}

\item The maps $\Lambda_{A, B, C}$ and $\langle \_, \_ \rangle_{A, B}^C$ in $\mathcal{C}$ are bijections for each triple $A, B, C \in \mathcal{C}$; 

\item The object $\mathscr{B}$, a $\beta$-product and a $\beta$-exponential of $\mathcal{C}$ are pairwise distinct; 

\item Each $\beta$-composition that occurs as the interpretation of a term is not a value;

\item A $\beta$-currying and a $\beta$-composition of $\mathcal{C}$ that occur as the interpretations of terms never coincide;

\item The $\beta$-evaluation $\mathit{ev}_{A, B}$ for any $A, B \in \mathcal{C}$ is a mono with respect to the $\beta$-composition, i.e., $f ; \mathit{ev}_{A, B} = f' ; \mathit{ev}_{A, B} \Rightarrow f = f'$ for any $C \in \mathcal{C}$ and $f, f' : C \rightarrow B^A \times A$ in $\mathcal{C}$.

\end{enumerate}
\end{definition}

Clearly, the interpretation $\llbracket \_ \rrbracket_{\mathcal{C}}^{\mathcal{S}}$ followed by $\mathcal{E}^\omega$, i.e., $\mathcal{E}^\omega(\llbracket \_ \rrbracket_{\mathcal{C}}^{\mathcal{S}})$, coincides with the standard categorical interpretation of the equational theory $\mathsf{Eq(FPCF)}$ in the CCC $\mathcal{V}_{\mathcal{C}}$ \cite{lambek1988introduction,pitts2001categorical,crole1993categories,jacobs1999categorical}. 
In this sense, we have refined the standard categorical semantics of type theories.

At this point, let us recall the DCP (see Section~\ref{Introduction}) specifically for the interpretation of FPCF induced by a structure in a CCBoC: 
\begin{definition}[DCP for FPCF]
\label{DefDCP}
The interpretation $\llbracket \_ \rrbracket_{\mathcal{C}}^{\mathcal{S}}$ of FPCF induced by a structure $\mathcal{S}$ for FPCF in a CCBoC $\mathcal{C} = (\mathcal{C}, \mathcal{E})$ satisfies the \emph{\bfseries dynamic correspondence property (DCP)} iff for any programs $\mathsf{M_1}$ and $\mathsf{M_2}$ of FPCF we have:
\begin{equation*}
\mathsf{M_1} \rightarrow \mathsf{M_2} \Leftrightarrow \llbracket \mathsf{M_1} \rrbracket_{\mathcal{C}}^{\mathcal{S}} \neq \llbracket \mathsf{M_2} \rrbracket_{\mathcal{C}}^{\mathcal{S}} \wedge \mathcal{E}(\llbracket \mathsf{M_1} \rrbracket_{\mathcal{C}}^{\mathcal{S}}) = \llbracket \mathsf{M_2} \rrbracket_{\mathcal{C}}^{\mathcal{S}}.
\end{equation*}
\end{definition}

Now, we reduce the DCP for FPCF to the following:
\begin{definition}[PDCP for FPCF]
\label{DefPDCP}
The interpretation $\llbracket \_ \rrbracket_{\mathcal{C}}^{\mathcal{S}}$ of FPCF induced by a structure $\mathcal{S}$ for FPCF in a CCBoC $\mathcal{C} = (\mathcal{C}, \mathcal{E})$ satisfies the \emph{\bfseries pointwise dynamic correspondence property (PDCP)} iff for each term $\mathsf{\Gamma} \vdash \{ \mathsf{M} \}_{e} : \mathsf{B}$ it satisfies:
\begin{equation*}
\mathcal{E}(\llbracket \mathsf{M} \rrbracket_{\mathcal{C}}^{\mathcal{S}}) = \begin{cases} \Lambda \circ \mathcal{E}(\llbracket \mathsf{P} \rrbracket_{\mathcal{C}}^{\mathcal{S}}) &\text{if $\mathsf{M \equiv \lambda x . P}$;} \\ 
\llbracket \mathsf{W} \rrbracket_{\mathcal{C}}^{\mathcal{S}} \ \text{such that $\llbracket \mathsf{W} \rrbracket_{\mathcal{C}}^{\mathcal{S}} \neq \llbracket \mathsf{M} \rrbracket_{\mathcal{C}}^{\mathcal{S}}$} &\text{if $\mathsf{M \equiv UV}$, $e = 1$ and $\mathsf{UV} \rightarrow \mathsf{W}$;} \\
\langle \mathcal{E}(\llbracket \mathsf{L} \rrbracket_{\mathcal{C}}^{\mathcal{S}}), \mathcal{E}(\llbracket \mathsf{R} \rrbracket_{\mathcal{C}}^{\mathcal{S}}) \rangle ; \mathit{ev} &\text{if $\mathsf{M \equiv LR}$ and $e > 1$;} \\
\llbracket \mathsf{M} \rrbracket_{\mathcal{C}}^{\mathcal{S}} &\text{otherwise.} 
\end{cases}
\end{equation*}
\end{definition}

\begin{lemma}[P-lemma]
\label{LemPDCPLemma}
If the interpretation $\llbracket \_ \rrbracket_{\mathcal{C}}^{\mathcal{S}}$ induced by a standard structure $\mathcal{S}$ for FPCF in a CCBoC $\mathcal{C} = (\mathcal{C}, \mathcal{E})$ satisfies the PDCP, then $\mathcal{E}(\llbracket \mathsf{M} \rrbracket_{\mathcal{C}}^{\mathcal{S}}) \neq \llbracket \mathsf{M} \rrbracket_{\mathcal{C}}^{\mathcal{S}} \Leftrightarrow \sharp(\mathsf{M}) \geqslant 1$ for all terms $\mathsf{M}$.
\end{lemma}
\begin{proof}
By induction on the construction of $\mathsf{M}$, where the first and the fifth axioms on standardness of $\mathcal{S}$ is essential. 
\end{proof}

\begin{theorem}[Standard bicategorical semantics of FPCF]
\label{ThmDynamicSemanticsOfSystemTVartheta}
The interpretation $\llbracket \_ \rrbracket_{\mathcal{C}}^{\mathcal{S}}$ of FPCF induced by a standard structure $\mathcal{S}$ for FPCF in a CCBoC $\mathcal{C} = (\mathcal{C}, \mathcal{E})$ satisfies the DCP if it satisfies the PDCP.
\end{theorem}
\begin{proof}
In the following, we abbreviate $\llbracket \_ \rrbracket_{\mathcal{C}}^{\mathcal{S}}$ as $\llbracket \_ \rrbracket$.
Assume that $\llbracket \_ \rrbracket$ satisfies the PDCP.
We show $\mathsf{M} \rightarrow \mathsf{M'} \Leftrightarrow \llbracket \mathsf{M} \rrbracket \neq \llbracket \mathsf{M'} \rrbracket \wedge \mathcal{E}(\llbracket \mathsf{M} \rrbracket) = \llbracket \mathsf{M'} \rrbracket$ for any programs $\mathsf{\Gamma} \vdash \{ \mathsf{M} \}_{e} : \mathsf{B}$ and $\mathsf{\Gamma} \vdash \{ \mathsf{M'} \}_{e'} : \mathsf{B}$ of FPCF by induction on the construction of $\mathsf{M}$:
\begin{itemize}

\item If $\mathsf{M} \equiv \mathsf{tt}$ or $\mathsf{M} \equiv \mathsf{ff}$, then there is no term $\mathsf{M'}$ such that $\mathsf{M} \rightarrow \mathsf{M'}$, and there is no $\beta$-morphism $f'$ in $\mathcal{C}$ such that $\llbracket \mathsf{M} \rrbracket \neq f' \wedge \mathcal{E}(\llbracket \mathsf{M} \rrbracket) = f'$ because $\mathcal{E}(\llbracket \mathsf{M} \rrbracket) = \llbracket \mathsf{M} \rrbracket$.

\item If $\mathsf{\Gamma} \vdash \mathsf{M} \equiv \mathsf{case(x V_1 V_2 \dots V_k)[W_1 ; W_2]} : o$, then it can be handled in the same manner as the above case.

\item If $\mathsf{\Gamma \vdash M \equiv \lambda x^A . P : A \Rightarrow C}$, then we have:
\begin{align*}
\mathsf{M} \to \mathsf{M'} &\Leftrightarrow \mathsf{M' \equiv \lambda x . P'} \wedge \mathsf{P} \to \mathsf{P'} \ \text{for some program $\mathsf{P'}$ and variable $\mathsf{x}$} \\
&\Leftrightarrow \mathsf{M' \equiv \lambda x . P'} \wedge \llbracket \mathsf{P} \rrbracket \neq \llbracket \mathsf{P'} \rrbracket \wedge \mathcal{E}(\llbracket \mathsf{P} \rrbracket) = \llbracket \mathsf{P'} \rrbracket \ \text{for some $\mathsf{P'}$ and $\mathsf{x}$} \\ 
&\text{(by the induction hypothesis)} \\
&\Leftrightarrow \llbracket \mathsf{P} \rrbracket \neq \Lambda^{-1}(\llbracket \mathsf{M'} \rrbracket) \wedge \mathcal{E}(\llbracket \mathsf{P} \rrbracket) = \Lambda^{-1}(\llbracket \mathsf{M'} \rrbracket) \\ 
&\text{(n.b., for $\Leftarrow$, $\Lambda^{-1}(\llbracket \mathsf{M'} \rrbracket) \downarrow$ implies that $\mathsf{M'}$ must be a currying as $\mathcal{S}$ is standard)} \\
&\Leftrightarrow \Lambda^{-1}(\llbracket \mathsf{M} \rrbracket) \neq \Lambda^{-1}(\llbracket \mathsf{M'} \rrbracket) \wedge \Lambda^{-1} \circ \mathcal{E}(\llbracket \mathsf{M} \rrbracket) = \Lambda^{-1}(\llbracket \mathsf{M'} \rrbracket) \ \text{(as $\mathcal{E}(\llbracket \mathsf{M} \rrbracket) = \Lambda \circ \mathcal{E} (\llbracket \mathsf{P} \rrbracket)$)} \\
&\Leftrightarrow \llbracket \mathsf{M} \rrbracket \neq \llbracket \mathsf{M'} \rrbracket \wedge \mathcal{E}(\llbracket \mathsf{M} \rrbracket) = \llbracket \mathsf{M'} \rrbracket \ \text{(by the bijectivity of $\Lambda$).}
\end{align*}

\item If $\mathsf{M \equiv LR}$, $\sharp(\mathsf{L}) \geqslant 1$ and  $\sharp(\mathsf{R}) \geqslant 1$, then we have:
\begin{align*}
\mathsf{M} \to \mathsf{M'} &\Leftrightarrow \mathsf{M'} \equiv \mathsf{L'R'} \wedge \mathsf{L} \to \mathsf{L'} \wedge \mathsf{R} \rightarrow \mathsf{R'} \ \text{for some programs $\mathsf{L'}$ and $\mathsf{R'}$} \\
&\Leftrightarrow \mathsf{M'} \equiv \mathsf{L'R'} \wedge \llbracket \mathsf{L} \rrbracket \neq \llbracket \mathsf{L'} \rrbracket \wedge \mathcal{E}(\llbracket \mathsf{L} \rrbracket) = \llbracket \mathsf{L'} \rrbracket \wedge \llbracket \mathsf{R} \rrbracket \neq \llbracket \mathsf{R'} \rrbracket \wedge \mathcal{E}(\llbracket \mathsf{R} \rrbracket) = \llbracket \mathsf{R'} \rrbracket \\
&\text{for some $\mathsf{L'}$ and $\mathsf{R'}$ (by the induction hypothesis)} \\
&\Leftrightarrow \llbracket \mathsf{M'} \rrbracket = \langle \mathcal{E}(\llbracket \mathsf{L} \rrbracket), \mathcal{E}(\llbracket \mathsf{R} \rrbracket) \rangle ; \mathit{ev} \wedge \llbracket \mathsf{L} \rrbracket \neq \mathcal{E}(\llbracket \mathsf{L} \rrbracket) \wedge \llbracket \mathsf{R} \rrbracket \neq \mathcal{E}(\llbracket \mathsf{R} \rrbracket) \\
&\text{(n.b., $\Leftarrow$ holds by the third and the fourth axioms on standardness of $\mathcal{S}$)} \\
&\Leftrightarrow \llbracket \mathsf{M'} \rrbracket = \mathcal{E}(\llbracket \mathsf{LR} \rrbracket) \wedge \llbracket \mathsf{L} \rrbracket \neq \mathcal{E}(\llbracket \mathsf{L} \rrbracket) \wedge \llbracket \mathsf{R} \rrbracket \neq \mathcal{E}(\llbracket \mathsf{R} \rrbracket) \\
&\text{(because the interpretation $\llbracket \_ \rrbracket$ satisfies the PDCP)} \\
&\Leftrightarrow \llbracket \mathsf{M'} \rrbracket = \mathcal{E}(\llbracket \mathsf{M} \rrbracket) \ \text{(by Lemma~\ref{LemPDCPLemma})} \\
&\Leftrightarrow \llbracket \mathsf{M'} \rrbracket = \mathcal{E}(\llbracket \mathsf{M} \rrbracket) \wedge \mathcal{E}(\llbracket \mathsf{M} \rrbracket) \neq \llbracket \mathsf{M} \rrbracket \ \text{(again by Lemma~\ref{LemPDCPLemma}).}
\end{align*}

\item If $\mathsf{M \equiv LR}$, $\sharp(\mathsf{L}) = 0$ and  $\sharp(\mathsf{R}) \geqslant 1$, then we have:
\begin{align*}
\mathsf{M} \rightarrow \mathsf{M'} &\Leftrightarrow \mathsf{M'} \equiv \mathsf{LR'} \wedge \mathsf{R} \rightarrow \mathsf{R'} \ \text{for some program $\mathsf{R'}$} \\
&\Leftrightarrow \mathsf{M'} \equiv \mathsf{LR'} \wedge \llbracket \mathsf{R} \rrbracket \neq \llbracket \mathsf{R'} \rrbracket \wedge \mathcal{E}(\llbracket \mathsf{R} \rrbracket) = \llbracket \mathsf{R'} \rrbracket \ \text{for some $\mathsf{R'}$} \\
&\text{(by the induction hypothesis)} \\
&\Leftrightarrow \llbracket \mathsf{M'} \rrbracket = \langle \llbracket \mathsf{L} \rrbracket, \mathcal{E}(\llbracket \mathsf{R} \rrbracket) \rangle ; \mathit{ev} \wedge \llbracket \mathsf{R} \rrbracket \neq \mathcal{E}(\llbracket \mathsf{R} \rrbracket) \\
&\text{(n.b., $\Leftarrow$ holds as in the above case)} \\
&\Leftrightarrow \llbracket \mathsf{M'} \rrbracket = \langle \llbracket \mathsf{L} \rrbracket, \mathcal{E}(\llbracket \mathsf{R} \rrbracket) \rangle ; \mathit{ev} \ \text{(by Lemma~\ref{LemPDCPLemma})} \\
&\Leftrightarrow \llbracket \mathsf{M'} \rrbracket = \mathcal{E}(\llbracket \mathsf{M} \rrbracket) \ \text{(by Lemma~\ref{LemPDCPLemma} and the PDCP of the interpretation $\llbracket \_ \rrbracket$)} \\
&\Leftrightarrow \llbracket \mathsf{M'} \rrbracket = \mathcal{E}(\llbracket \mathsf{M} \rrbracket) \wedge \mathcal{E}(\llbracket \mathsf{M} \rrbracket) \neq \llbracket \mathsf{M} \rrbracket \ \text{(again by Lemma~\ref{LemPDCPLemma}).}
\end{align*}

\item If $\mathsf{M \equiv LR}$, $\sharp(\mathsf{L}) \geqslant 1$ and  $\sharp(\mathsf{R}) = 0$, then it is handled similarly to the above case.

\item If $\mathsf{M \equiv LR}$, $\sharp(\mathsf{L}) = 0$ and  $\sharp(\mathsf{R}) = 0$, then we have:
\begin{align*}
\mathsf{M} \rightarrow \mathsf{M'} &\Leftrightarrow \mathcal{E}(\llbracket \mathsf{M} \rrbracket) = \llbracket \mathsf{M'} \rrbracket \ \text{(since the interpretation $\llbracket \_ \rrbracket$ satisfies the PDCP)} \\
&\Leftrightarrow \mathcal{E}(\llbracket \mathsf{M} \rrbracket) = \llbracket \mathsf{M'} \rrbracket \wedge \mathcal{E}(\llbracket \mathsf{M} \rrbracket) \neq \llbracket \mathsf{M} \rrbracket \ \text{(by Lemma~\ref{LemPDCPLemma})}
\end{align*}

\end{itemize}
which completes the proof.
\end{proof}


To summarize the present section, we have defined bicategorical `universes' of dynamic, intensional computations, viz., (CC)BoCs, presented the simple functional programming language FPCF, and given an interpretation of the latter in the former as well as a sufficient condition, namely, the PDCP, for the interpretation to satisfy the DCP.
Hence, our research problem (described in Section~\ref{Introduction}) has been reduced to giving a standard structure for FPCF in a game-semantic CCBoC that satisfies the PDCP.

\section{Dynamic Games and Strategies}
\label{DynamicGamesAndStrategies}
The present section introduces dynamic games and strategies and studies their algebraic structures.
The main idea of dynamic games and strategies is to introduce the distinction between \emph{internal} and \emph{external} moves to conventional games and strategies; internal moves constitute `internal communication' between dynamic strategies, representing \emph{intensionality} of computation, and they are to be \emph{a posteriori} `hidden' by the \emph{hiding operation}, capturing \emph{dynamics} of computation.
Conceptually, external moves are `official' ones for the underlying game, while internal moves are supposed to be `invisible' to Opponent for they represent how Player `internally' computes the next external move.

Dynamic games and strategies are based on the variant given in \cite{abramsky1999game}, which we call \emph{static} games and strategies (more generally, to distinguish our `dynamic concepts' from conventional ones, we add the word \emph{static} in front of the corresponding notions in \cite{abramsky1999game}, e.g., static arenas, static legal positions, etc.); this choice is because the variant combines good points of the two best-known variants: \emph{AJM-games} \cite{abramsky2000full} and \emph{HO-games} \cite{hyland2000full}: It interprets the \emph{linear decomposition} of implication \cite{girard1987linear}, and it is \emph{flexible} enough to model a wide range of programming features \cite{abramsky1999game}. 
We have chosen this variant with the hope that our framework is also applicable to various formal systems and programming languages.

\subsection{Dynamic Arenas and Legal Positions}
Just like static games \cite{abramsky1999game}, dynamic games are based on (the `dynamic generalizations' of) \emph{arenas} and \emph{legal positions}.
An arena defines the basic components of a game, which in turn induces a set of legal positions that specifies the basic rules of the game. 
Let us first introduce these preliminary concepts.

\begin{definition}[Dynamic arenas]
\label{DefDynamicArenas}
A \emph{\bfseries dynamic arena} is a triple 
\begin{equation*}
G = (M_G, \lambda_G, \vdash_G) 
\end{equation*}
such that:
\begin{itemize}

\item $M_G$ is a set, whose elements are called \emph{\bfseries moves};

\item $\lambda_G$ is a function $M_G \to \{ \mathsf{O}, \mathsf{P} \} \times \{ \mathsf{Q}, \mathsf{A} \} \times \mathbb{N}$, called the \emph{\bfseries labeling function}, that satisfies $\mu(G) \stackrel{\mathrm{df. }}{=} \mathsf{Sup}(\{ \lambda_G^{\mathbb{N}}(m) \mid m \in M_G \ \! \}) \in \mathbb{N}$;

\item $\vdash_G$ is a subset of $(\{ \star \} \cup M_G) \times M_G$, where $\star$ is an arbitrary element such that $\star \not \in M_G$, called the \emph{\bfseries enabling relation}, that satisfies:
\begin{itemize}

\item \textsc{(E1).} If $\star \vdash_G m$, then $\lambda_G(m) = \mathsf{O}\mathsf{Q}0$ and $n = \star$ whenever $n \vdash_G m$;

\item \textsc{(E2).} If $m \vdash_G n$ and $\lambda_G^{\mathsf{QA}}(n) = \mathsf{A}$, then $\lambda_G^{\mathsf{QA}}(m) = \mathsf{Q}$ and $\lambda_G^{\mathbb{N}}(m) = \lambda_G^{\mathbb{N}}(n)$;

\item \textsc{(E3).} If $m \vdash_G n$ and $m \neq \star$, then $\lambda_G^{\mathsf{OP}}(m) \neq \lambda_G^{\mathsf{OP}}(n)$;

\item \textsc{(E4).} If $m \vdash_G n$, $m \neq \star$ and $\lambda_G^{\mathbb{N}}(m) \neq \lambda_G^{\mathbb{N}}(n)$, then $\lambda_G^{\mathsf{OP}}(m) = \mathsf{O}$ 

\end{itemize}
\end{itemize}
in which $\lambda_G^{\mathsf{OP}} \stackrel{\mathrm{df. }}{=} \pi_1 \circ \lambda_G : M_G \to \{ \mathsf{O}, \mathsf{P} \}$, $\lambda_G^{\mathsf{QA}} \stackrel{\mathrm{df. }}{=} \pi_2 \circ \lambda_G : M_G \to \{ \mathsf{Q}, \mathsf{A} \}$ and $\lambda_G^{\mathbb{N}} \stackrel{\mathrm{df. }}{=} \pi_3 \circ \lambda_G : M_G \to \mathbb{N}$.
A move $m \in M_G$ is \emph{\bfseries initial} if $\star \vdash_G m$, an \emph{\bfseries O-move} (resp. a \emph{\bfseries P-move}) if $\lambda_G^{\mathsf{OP}}(m) = \mathsf{O}$ (resp. if $\lambda_G^{\mathsf{OP}}(m) = \mathsf{P}$), a \emph{\bfseries question} (resp. an \emph{\bfseries answer}) if $\lambda_G^{\mathsf{QA}}(m) = \mathsf{Q}$ (resp. if $\lambda_G^{\mathsf{QA}}(m) = \mathsf{A}$), and \emph{\bfseries internal} or \emph{\bfseries $\boldsymbol{\lambda_G^{\mathbb{N}}(m)}$-internal}  (resp. \emph{\bfseries external}) if $\lambda_G^{\mathbb{N}}(m) > 0$ (resp. if $\lambda_G^{\mathbb{N}}(m) = 0$).
Any $\boldsymbol{s} \in M_G^\ast$ is \emph{\bfseries $\boldsymbol{d}$-complete} if it ends with a move $m$ such that $\lambda_G^{\mathbb{N}}(m) = 0 \vee \lambda_G^{\mathbb{N}}(m) > d$, where $d \in \mathbb{N} \cup \{ \omega \}$, and $\omega$ is the least transfinite ordinal.
\end{definition}

Recall that a static arena $G$ \cite{abramsky1999game} determines possible \emph{moves} of a game, each of which is Opponent's/Player's question/answer, where the third parity $\lambda_G^{\mathbb{N}}$ is not included, and specifies which move $n$ can be performed for each move $m$ by the relation $m \vdash_G n$ (and $\star \vdash_G m$ means that $m$ can \emph{initiate} a play). 
The axioms on a static arena are the following:
\begin{itemize}

\item \textsc{(E1).} An initial move must be Opponent's question, and an initial move cannot be enabled by any move;

\item \textsc{(The first point of E2).} An answer must be performed for a question;

\item \textsc{(E3).} An O-move must be performed for a P-move, and vice versa.

\end{itemize}

Thus, a dynamic arena is a static arena equipped with the \emph{\bfseries priority order} $\lambda_G^{\mathbb{N}}$ on moves that satisfies additional axioms on the priority order; it is called so for it determines the `priority order' of moves to be `hidden' by the hiding operations on dynamic games (Definition~\ref{DefHidingOperationOnDynamicGames}) and on dynamic strategies (Definition~\ref{DefHidingOperationOnDynamicStrategies}).
We need all natural numbers for $\lambda_G^{\mathbb{N}}$, not only the internal/external (I/E) distinction, to define a \emph{step-by-step} execution of the hiding operations.
Conversely, dynamic arenas are generalized static arenas: A static arena is equivalent to a dynamic arena whose moves are all external.

The additional axioms for dynamic arenas $G$ are intuitively natural ones:
\begin{itemize}

\item We require a \emph{finite} upper bound $\mu(G)$ of the priority orders for it is conceptually natural and technically necessary for concatenation of dynamic games (Definition~\ref{DefConcatenationOfDynamicGames}) to be well-defined and for the hiding operation on dynamic games to terminate;

\item The axiom E1 adds the equation $\lambda_G^{\mathbb{N}}(m_0) = 0$ for all $m_0 \in M_G^{\mathsf{Init}}  \stackrel{\mathrm{df. }}{=} \{ m \in M_G \mid \star \vdash m \ \! \}$ since Opponent cannot `see' internal moves;

\item The second requirement of the axiom E2 states that the priority orders between a `QA-pair' must coincide, which is intuitively reasonable;

\item The additional axiom E4 states that only Player can make a move for a previous one if they have different priority orders for internal moves are `invisible' to Opponent (as we shall see, if $\lambda_G^{\mathbb{N}}(m_1) = k_1 < k_2 = \lambda_G^{\mathbb{N}}(m_2)$, then after the $k_1$-many iteration of the hiding operation, $m_1$ and $m_2$ become external and internal, respectively, i.e., the I/E-parity of moves is \emph{relative}, which is why E4 is not only concerned with I/E-parity but more fine-grained priority orders).

\end{itemize}

\begin{convention*}
Henceforth, an \emph{\bfseries arena} refers to a dynamic arena by default. 
\end{convention*}

\begin{example}
\label{ExTerminalArena}
The \emph{\bfseries terminal arena} $T$ is given by $T \stackrel{\mathrm{df. }}{=} (\emptyset, \emptyset, \emptyset)$.
\end{example}

\begin{example}
\label{ExFlatArenas}
The \emph{\bfseries flat arena} $\mathit{flat}(S)$ on a given set $S$ is given by $M_{\mathit{flat}(S)} \stackrel{\mathrm{df. }}{=} \{ q \} \cup S$, where $q$ is any element with $q \not \in S$; $\lambda_{\mathit{flat}(S)} : q \mapsto \mathsf{OQ}0, (m \in S) \mapsto \mathsf{PA}0$; $\vdash_{\mathit{flat}(S)} \ \stackrel{\mathrm{df. }}{=} \{ (\star, q) \} \cup \{ (q, m) \mid m \in S \ \! \}$.
For instance, $N \stackrel{\mathrm{df. }}{=} \mathit{flat}(\mathbb{N})$ is the arena of natural numbers, and $\boldsymbol{2} \stackrel{\mathrm{df. }}{=} \mathit{flat}(\mathbb{B})$, where $\mathbb{B} \stackrel{\mathrm{df. }}{=} \{ \mathit{tt}, \mathit{ff} \}$, is the arena of booleans.
\end{example}

As already mentioned, interactions between Opponent and Player in a (dynamic or static) game are represented by certain finite sequences of moves of the underlying arena, equipped with \emph{pointers} (Definition~\ref{DefJSequences}) that specify the occurrence of a move in the sequence for which each occurrence of a non-initial move in the sequence is performed.
Technically, pointers are to distinguish similar but different computations; see \cite{abramsky1999game,curien2006notes} for this point.

\begin{definition}[Occurrences of moves]
\label{DefMoveOccurrences}
Given a finite sequence $\boldsymbol{s} \in M_G^\ast$ of moves of an arena $G$, an \emph{\bfseries occurrence (of a move)} in $\boldsymbol{s}$ is a pair $(\boldsymbol{s}(i), i)$ such that $i \in \overline{|\boldsymbol{s}|}$. 
More specifically, we call the pair $(\boldsymbol{s}(i), i)$ an \emph{\bfseries initial occurrence} (resp. a \emph{\bfseries non-initial occurrence}) in $\boldsymbol{s}$ if $\star \vdash_G \boldsymbol{s}(i)$ (resp. otherwise).
\end{definition}

\begin{definition}[J-sequences \cite{hyland2000full,abramsky1999game}]
\label{DefJSequences}
A \emph{\bfseries justified (j-) sequence} of an arena $G$ is a pair $\boldsymbol{s} = (\boldsymbol{s}, \mathcal{J}_{\boldsymbol{s}})$ of a finite sequence $\boldsymbol{s} \in M_G^\ast$ and a map $\mathcal{J}_{\boldsymbol{s}} : \overline{|\boldsymbol{s}|} \rightarrow \{ 0 \} \cup \overline{|\boldsymbol{s}|-1}$ such that for all $i \in \overline{|\boldsymbol{s}|}$ $\mathcal{J}_{\boldsymbol{s}}(i) = 0$ if $\star \vdash_G \boldsymbol{s}(i)$, and $0 < \mathcal{J}_{\boldsymbol{s}}(i) < i \wedge \boldsymbol{s}({\mathcal{J}_{\boldsymbol{s}}(i)}) \vdash_G \boldsymbol{s}(i)$ otherwise.
The occurrence $(\boldsymbol{s}({\mathcal{J}_{\boldsymbol{s}}(i)}), \mathcal{J}_{\boldsymbol{s}}(i))$ is called the \emph{\bfseries justifier} of a non-initial occurrence $(\boldsymbol{s}(i), i)$ in $\boldsymbol{s}$.
We also say that $(\boldsymbol{s}(i), i)$ is \emph{\bfseries justified} by $(\boldsymbol{s}({\mathcal{J}_{\boldsymbol{s}}(i)}), \mathcal{J}_{\boldsymbol{s}}(i))$, or there is a \emph{\bfseries pointer} from the former to the latter.
\end{definition}

The idea is that each non-initial occurrence in a j-sequence must be performed for a specific previous occurrence, viz., its justifier, in the j-sequence. 

\begin{convention*}
By abuse of notation, we usually keep the pointer structure $\mathcal{J}_{\boldsymbol{s}}$ of each j-sequence $\boldsymbol{s} = (\boldsymbol{s}, \mathcal{J}_{\boldsymbol{s}})$ implicit and often abbreviate occurrences $(\boldsymbol{s}(i), i)$ in $\boldsymbol{s}$ as $\boldsymbol{s}(i)$.
Also, we usually write $\mathcal{J}_{\boldsymbol{s}}(\boldsymbol{s}(i)) = \boldsymbol{s}(j)$ if $\mathcal{J}_{\boldsymbol{s}}(i) = j$.
This convention is mathematically imprecise, but it does not bring any serious confusion in practice. 
\end{convention*}

\begin{notation*}
We write $\mathscr{J}_G$ for the set of all j-sequences of an arena $G$.
We write $\boldsymbol{s} = \boldsymbol{t}$ for any $\boldsymbol{s}, \boldsymbol{t} \in \mathscr{J}_G$ if $\boldsymbol{s}$ and $\boldsymbol{t}$ are the same j-sequence of $G$, i.e., $\boldsymbol{s} = \boldsymbol{t}$ and $\mathcal{J}_{\boldsymbol{s}} = \mathcal{J}_{\boldsymbol{t}}$.
\end{notation*}

\begin{definition}[J-subsequences]
\label{DefJSubsequences}
Given an arena $G$ and a j-sequence $\boldsymbol{s} \in \mathscr{J}_G$, a \emph{\bfseries j-subsequence} of $\boldsymbol{s}$ is a j-sequence $\boldsymbol{t} \in \mathscr{J}_G$ that satisfies:
\begin{itemize}

\item $\boldsymbol{t}$ is a subsequence of $\boldsymbol{s}$, for which we write $\boldsymbol{t} = (\boldsymbol{s}(i_1), \boldsymbol{s}(i_2), \dots, \boldsymbol{s}(i_{|\boldsymbol{t}|}))$; 

\item $\mathcal{J}_{\boldsymbol{t}}(\boldsymbol{s}(i_r)) = \boldsymbol{s}(i_l)$ iff there are occurrences $\boldsymbol{s}(j_1), \boldsymbol{s}(j_2), \dots, \boldsymbol{s}(j_{k})$ in $\boldsymbol{s}$ eliminated in $\boldsymbol{t}$, where $l, r, k \in \mathbb{N}$ and  $1 \leqslant l < r \leqslant |\boldsymbol{t}|$, such that $\mathcal{J}_{\boldsymbol{s}}(\boldsymbol{s}(i_r)) = \boldsymbol{s}(j_1) \wedge \mathcal{J}_{\boldsymbol{s}}(\boldsymbol{s}(j_1)) = \boldsymbol{s}(j_2) \wedge \dots \wedge \mathcal{J}_{\boldsymbol{s}}(\boldsymbol{s}(j_{k-1})) = \boldsymbol{s}(j_{k}) \wedge \mathcal{J}_{\boldsymbol{s}}(\boldsymbol{s}(j_{k})) = \boldsymbol{s}(i_l)$.

\end{itemize}
\end{definition}

We now consider justifiers, j-sequences and arenas from the `external point of view':
\begin{definition}[External justifiers]
Let $G$ be an arena, and assume $\boldsymbol{s} \in \mathscr{J}_G$ and $d \in \mathbb{N} \cup \{ \omega \}$.
Each non-initial occurrence $n$ in $\boldsymbol{s}$ has a unique sequence of justifiers $m m_1 m_2 \dots m_k n$ $(k \geqslant 0)$, i.e., $\mathcal{J}_{\boldsymbol{s}}(n) = m_k$, $\mathcal{J}_{\boldsymbol{s}}(m_k) = m_{k-1}$, \dots, $\mathcal{J}_{\boldsymbol{s}}(m_{2}) = m_1$ and $\mathcal{J}_{\boldsymbol{s}}(m_1) = m$, such that $\lambda_G^{\mathbb{N}}(m) = 0 \vee \lambda_G^{\mathbb{N}}(m) > d$ and $0 < \lambda_G^{\mathbb{N}}(m_i) \leqslant d$ for $i = 1, 2, \dots, k$. 
We call $m$ the \emph{\bfseries $\boldsymbol{d}$-external justifier} of $n$ in $\boldsymbol{s}$.
\end{definition}

\begin{notation*}
We write $\mathcal{J}_{\boldsymbol{s}}^{\circleddash d}(n)$ for the $d$-external justifier of $n$ in a j-sequence $\boldsymbol{s}$.
\end{notation*}

Note that $d$-external justifiers are a simple generalization of justifiers because 0-external justifiers coincide with justifiers (as there is no `0-internal' move).
More generally, $d$-external justifiers are justifiers after the $d$-times iteration of the hiding operation, as we shall see shortly.

\begin{definition}[External j-subsequences]
\label{DefHidingOnJSequences}
Let $G$ be an arena, $\boldsymbol{s} \in \mathscr{J}_G$ and $d \in \mathbb{N} \cup \{ \omega \}$. 
The \emph{\bfseries $\boldsymbol{d}$-external j-subsequence} $\mathcal{H}^d_G(\boldsymbol{s})$ of $\boldsymbol{s}$ is obtained from $\boldsymbol{s}$ by deleting occurrences of internal moves $m$ such that $0 < \lambda_G^{\mathbb{N}}(m) \leqslant d$ and equipping it with the pointers $\mathcal{J}_{\mathcal{H}^d_G(\boldsymbol{s})} : n \mapsto \mathcal{J}_{\boldsymbol{s}}^{\circleddash d}(n)$ (more precisely, $\mathcal{J}_{\mathcal{H}^d_G(\boldsymbol{s})}$ is the obvious \emph{restriction} of $\mathcal{J}_{\boldsymbol{s}}^{\circleddash d}$).
\end{definition}

\begin{definition}[External arenas]
\label{DefExternalDynamicArenas}
Let $G$ be an arena, and $d \in \mathbb{N} \cup \{ \omega \}$. 
The \emph{\bfseries $\boldsymbol{d}$-external arena} $\mathcal{H}^d(G)$ of $G$ is given by:
\begin{itemize}

\item $M_{\mathcal{H}^d(G)} \stackrel{\mathrm{df. }}{=} \{ m \in M_G \mid \lambda_G^{\mathbb{N}}(m) = 0 \vee \lambda_G^{\mathbb{N}}(m) > d \ \! \}$;

\item $\lambda_{\mathcal{H}^d(G)} \stackrel{\mathrm{df. }}{=} \lambda_G^{\circleddash d} \upharpoonright M_{\mathcal{H}^d(G)}$, where $\lambda_G^{\circleddash d} \stackrel{\mathrm{df. }}{=} \langle \lambda_G^{\mathsf{OP}}, \lambda_G^{\mathsf{QA}}, n \mapsto \lambda_G^{\mathbb{N}} (n) \circleddash d \rangle$, and $n \circleddash d \stackrel{\mathrm{df. }}{=} \begin{cases} n - d &\text{if $n \geqslant d$;} \\ 0 &\text{otherwise} \end{cases}$ for all $n \in \mathbb{N}$;

\item $m \vdash_{\mathcal{H}^d(G)} n \stackrel{\mathrm{df. }}{\Leftrightarrow} \exists k \in \mathbb{N}, m_1, m_2, \dots, m_{2k-1}, m_{2k} \in M_G \setminus M_{\mathcal{H}^d(G)} . \ \! m \vdash_G m_1 \wedge \forall i \in \overline{k} . \ \! m_{2i-1} \vdash_G m_{2i} \wedge m_{2k} \vdash_G n$ ($\Leftrightarrow m \vdash_G n$ if $k = 0$).

\end{itemize}
\end{definition}

That is, the $d$-external arena $\mathcal{H}^d(G)$ is obtained from the arena $G$ by deleting internal moves $m$ such that $0 < \lambda_G^{\mathbb{N}}(m) \leqslant d$, decreasing by $d$ the priority orders of the remaining moves and `concatenating' the enabling relation to form the `$d$-external' one. 

\begin{convention*}
Given $d \in \mathbb{N} \cup \{ \omega \}$, we regard $\mathcal{H}^d$ as an operation on dynamic arenas $G$, and $\mathcal{H}_G^d$ as an operation on j-sequences $\boldsymbol{s} \in \mathscr{J}_G$.
\end{convention*}

Now, let us establish:

\begin{lemma}[External closure lemma]
\label{LemExternalClosureLemma}
If $G$ is an arena, then, for all $d \in \mathbb{N} \cup \{ \omega \}$, so is $\mathcal{H}^d(G)$, and $\mathcal{H}_G^d(\boldsymbol{s}) \in \mathscr{J}_{\mathcal{H}^d(G)}$ for all $\boldsymbol{s} \in \mathscr{J}_G$. 
\end{lemma}
\begin{proof}
The case $d = 0$ is trivial; thus, assume $d > 0$. Clearly, the set $M_{\mathcal{H}^d(G)}$ of moves and the labeling function $\lambda_{\mathcal{H}^d(G)}$ are well-defined. Now, let us verify the axioms for the enabling relation $\vdash_{\mathcal{H}^d(G)}$:
\begin{itemize}

\item \textsc{(E1).} Note that $\star \vdash_{\mathcal{H}^d(G)} m \Leftrightarrow \star \vdash_G m$ (because $\Leftarrow$ is immediate, and $\Rightarrow$ holds by E4 on $G$ as initial moves are all external).
Thus, if $\star \vdash_{\mathcal{H}^d(G)} m$, then $\lambda_{\mathcal{H}^d(G)} (m) = \lambda_G^{\circleddash d}(m) =  \mathsf{O}\mathsf{Q}0$, and $n \vdash_{\mathcal{H}^d(G)} m \Rightarrow n = \star$.

\item \textsc{(E2).} Assume $m \vdash_{\mathcal{H}^d(G)} n$ and $\lambda_{\mathcal{H}^d(G)}^\mathsf{QA} (n) = \mathsf{A}$.  
If $m \vdash_G n$, then $\lambda_{\mathcal{H}^d(G)}^\mathsf{QA} (m) =\lambda_G^\mathsf{QA} (m) = \mathsf{Q}$ and
$\lambda_{\mathcal{H}^d(G)}^\mathbb{N} (m) =\lambda_G^\mathbb{N} (m) \circleddash d = \lambda_G^\mathbb{N} (n) \circleddash d = \lambda_{\mathcal{H}^d(G)}^\mathbb{N} (n)$.
Otherwise, i.e., there are some $k \in \mathbb{N}^+ \stackrel{\mathrm{df. }}{=} \{ n \in \mathbb{N} \mid n > 0 \ \! \}$ and $m_1, m_2, \dots, m_{2k} \in M_G \setminus M_{\mathcal{H}^d(G)}$ such that $m \vdash_G m_1 \wedge \forall i \in \overline{k} . \ \! m_{2i-1} \vdash_G m_{2i} \wedge m_{2k} \vdash_G n$, then in particular $m_{2k} \vdash_G n$ with $\lambda_{G}^\mathsf{QA} (n) = \mathsf{A}$, but $\lambda_G^\mathbb{N} (m_{2k}) \neq \lambda_G^\mathbb{N} (n)$, a contradiction. 

\item \textsc{(E3).} Assume $m \vdash_{\mathcal{H}^d(G)} n$ and $m \neq \star$. 
If $m \vdash_G n$, then $\lambda_{\mathcal{H}^d(G)}^\mathsf{OP} (m) = \lambda_G^\mathsf{OP} (m) \neq \lambda_G^\mathsf{OP} (n) = \lambda_{\mathcal{H}^d(G)}^\mathsf{OP} (n)$.
If $\exists k \in \mathbb{N}^+, m_1, m_2, \dots, m_{2k} \in M_G \setminus M_{\mathcal{H}^d(G)}. \ \! m \vdash_G m_1 \wedge \forall i \in \overline{k} . \ \! m_{2i-1} \vdash_G m_{2i} \wedge m_{2k} \vdash_G n$, then $\lambda_{\mathcal{H}^d(G)}^\mathsf{OP} (m) = \lambda_G^\mathsf{OP} (m) = \lambda_G^\mathsf{OP} (m_2) = \lambda_G^\mathsf{OP} (m_4) = \dots = \lambda_G^\mathsf{OP} (m_{2k}) \neq \lambda_G^\mathsf{OP} (n) = \lambda_{\mathcal{H}^d(G)}^\mathsf{OP} (n)$.

\item \textsc{(E4).} Assume $m \vdash_{\mathcal{H}^d(G)} n$, $m \neq \star$ and $\lambda_{\mathcal{H}^d(G)}^\mathbb{N}(m) \neq \lambda_{\mathcal{H}^d(G)}^\mathbb{N}(n)$. 
Then, we have $\lambda_{G}^\mathbb{N}(m) \neq \lambda_{G}^\mathbb{N}(n)$. 
If $m \vdash_G n$, then it is trivial; otherwise, i.e., there are some $k \in \mathbb{N}^+$, $m_1, m_2, \dots, m_{2k} \in M_G \setminus M_{\mathcal{H}^d(G)}$ with the same property as in the case of E3 above, $\lambda^{\mathsf{OP}}_{\mathcal{H}^d(G)}(m) = \lambda^{\mathsf{OP}}_{G}(m) = \mathsf{O}$ by E3 on $G$ since $\lambda^{\mathbb{N}}_G(m) \neq \lambda^{\mathbb{N}}_G(m_1)$.
\end{itemize}
Hence, we have shown that the structure $\mathcal{H}^d(G)$ forms a well-defined arena.

Next, let $\boldsymbol{s} \in \mathscr{J}_G$; we have to show $\mathcal{H}_G^d(\boldsymbol{s}) \in \mathscr{J}_{\mathcal{H}^d(G)}$.
Assume that $m$ is a non-initial occurrence in $\mathcal{H}_G^d(\boldsymbol{s})$.
By the definition, the $d$-external justifier $m_0 \stackrel{\mathrm{df. }}{=} \mathcal{J}_{\mathcal{H}_G^d(\boldsymbol{s})}(m)$ occurs in $\mathcal{H}_G^d(\boldsymbol{s})$. 
If $m$ is a P-move, then the sequence of justifiers $m_0 \vdash_G m_1 \vdash_G \dots \vdash_G m_k \vdash m$ satisfies $\mathsf{Even}(k)$ by the axioms E3 and E4 on $G$, so that $m_0 \vdash_{\mathcal{H}^d(G)} m$ by the definition.
If $m$ is an O-move, then the justifier $m'_0 \stackrel{\mathrm{df. }}{=} \mathcal{J}_{\boldsymbol{s}}(m)$ satisfies $\lambda_G^{\mathbb{N}}(m'_0) = \lambda_G^{\mathbb{N}}(m)$ by the axiom E4 on $G$, and so $m'_0 \vdash_{\mathcal{H}^d(G)} m$ by the definition. 
Since $m$ is arbitrary, we have shown that $\mathcal{H}_G^d(\boldsymbol{s}) \in \mathscr{J}_{\mathcal{H}^d(G)}$, completing the proof.
\end{proof}

Next, let us introduce a useful lemma:
\begin{lemma}[Stepwise hiding on arenas]
\label{LemStepwiseHidingOnDynamicArenas}
Given an arena $G$, we have $\widetilde{\mathcal{H}}^i(G) = \mathcal{H}^i(G)$ for all $i \in \mathbb{N}$, where $\widetilde{\mathcal{H}}^i$ denotes the $i$-times iteration of $\mathcal{H}^1$.
\end{lemma}
\begin{proof}
By induction on $i$.
\end{proof}

Thus, we may just focus on $\mathcal{H}^1$: Henceforth, we write $\mathcal{H}$ for $\mathcal{H}^1$ and call it the \emph{\bfseries hiding operation (on arenas)}; $\mathcal{H}^i$ for each $i \in \mathbb{N}$ denotes the $i$-times iteration of $\mathcal{H}$.

We may establish a similar inductive property for j-sequences:
\begin{lemma}[Stepwise hiding on j-sequences]
\label{LemStepwiseHidingOnJSequences}
Given a j-sequence $\boldsymbol{s} \in \mathscr{J}_G$ of an arena $G$, we have $\mathcal{H}_G^{i+1}(\boldsymbol{s}) = \mathcal{H}^1_{\mathcal{H}^i(G)}(\mathcal{H}_G^i(\boldsymbol{s}))$ for all $i \in \mathbb{N}$.
\end{lemma}
\begin{proof}
By induction on $i$, where note that $\mathcal{H}_G^{i+1}(\boldsymbol{s}), \mathcal{H}^1_{\mathcal{H}^i(G)}(\mathcal{H}_G^i(\boldsymbol{s})) \in \mathscr{J}_{\mathcal{H}^{i+1}(G)}$ by Lemmata~\ref{LemExternalClosureLemma} and \ref{LemStepwiseHidingOnDynamicArenas}.
\end{proof}

Lemma~\ref{LemStepwiseHidingOnJSequences} implies that the equation 
\begin{equation}
\label{HidingEquation}
\mathcal{H}_G^i(\boldsymbol{s}) = \mathcal{H}^1_{\mathcal{H}^{i-1}(G)} \circ \mathcal{H}^1_{\mathcal{H}^{i-2}(G)} \circ \dots \circ \mathcal{H}^1_{\mathcal{H}^1(G)} \circ \mathcal{H}^1_G(\boldsymbol{s})
\end{equation}
holds for any arena $G$, $\boldsymbol{s} \in \mathscr{J}_G$ and $i \in \mathbb{N}$ (n.b., the equation (\ref{HidingEquation}) means $\boldsymbol{s} = \boldsymbol{s}$ if $i = 0$).
Thus, we may focus on the operation $\mathcal{H}_G^1$ on j-sequences, where $G$ ranges over all arenas. 
Henceforth, we write $\mathcal{H}_G$ for $\mathcal{H}_G^1$ and call it the \emph{\bfseries hiding operation on j-sequences of $\boldsymbol{G}$}; $\mathcal{H}_G^i$ for each $i \in \mathbb{N}$ denotes the operation on the right-hand side of (\ref{HidingEquation}).  

Now, to deal with external j-subsequences in a mathematically rigorous manner, let us extend the hiding operation on j-sequences to that on j-subsequences (Definition~\ref{DefJSubsequences}):
\begin{definition}[Point-wise hiding on j-sequences]
Let $\boldsymbol{s} \in \mathscr{J}_G$ be a j-sequence of an arena $G$. 
Given an occurrence $m$ in $\boldsymbol{s}$, we define $\widehat{\mathcal{H}}_G^m(\boldsymbol{s})$ to be the j-subsequence of $\boldsymbol{s}$ that consists of occurrences in $\boldsymbol{s}$ different from $m$ if $m$ is 1-internal, and $\boldsymbol{s}$ otherwise.
Moreover, given a subsequence $\boldsymbol{t} = m_1 m_2 \dots m_k$ of (the underlying finite sequence of) $\boldsymbol{s}$ and a permutation $\sigma$ on $\overline{k}$, we define $\widehat{\mathcal{H}}_G^{{\boldsymbol{t}}, \sigma}(\boldsymbol{s}) \stackrel{\mathrm{df. }}{=} \widehat{\mathcal{H}}_G^{m_{\sigma(k)}} \circ \cdots \circ \widehat{\mathcal{H}}_G^{m_{\sigma(2)}} \circ \widehat{\mathcal{H}}_G^{m_{\sigma(1)}}(\boldsymbol{s})$.
\end{definition}

The point here is that the hiding operation on j-sequences can be executed in the `move-wise' fashion in any order: 
\begin{lemma}[Move-wise lemma]
\label{LemMoveWiseLemmaForHidingOnJSequences}
Let $G$ be an arena, and $\boldsymbol{s} \in \mathscr{J}_G$.
\begin{enumerate}

\item $\widehat{\mathcal{H}}_G^{\boldsymbol{t}, \sigma_1}(\boldsymbol{s}) = \widehat{\mathcal{H}}_G^{\boldsymbol{t}, \sigma_2}(\boldsymbol{s})$ for any subsequence $\boldsymbol{t}$ of $\boldsymbol{s}$ and permutations $\sigma_1$ and $\sigma_2$ on $\overline{|\boldsymbol{t}|}$;

\item $\widehat{\mathcal{H}}_G^{\boldsymbol{s}, \sigma}(\boldsymbol{s}) = \mathcal{H}_G(\boldsymbol{s})$ for any permutation $\sigma$ on $\overline{|\boldsymbol{s}|}$.

\end{enumerate}
\end{lemma}
\begin{proof}
Immediate from the definition. 
\end{proof}

By Lemma~\ref{LemMoveWiseLemmaForHidingOnJSequences}, we have established the `move-wise' procedure to execute the hiding operation $\mathcal{H}_G$ on j-sequences of a given arena $G$, where the order of deleting moves is irrelevant.
Then, e.g., it follows that $\mathcal{H}_G(\boldsymbol{stuv}) = \widehat{\mathcal{H}}_G^{\boldsymbol{v}, \nu} \circ \widehat{\mathcal{H}}_G^{\boldsymbol{u}, \mu} \circ \widehat{\mathcal{H}}_G^{\boldsymbol{t}, \tau} \circ \widehat{\mathcal{H}}_G^{\boldsymbol{s}, \sigma}(\boldsymbol{stuv})$ for any arena $G$ and $\boldsymbol{stuv} \in \mathscr{J}_G$, where $\sigma$, $\tau$, $\mu$ and $\nu$ are arbitrary permutations on $\overline{|\boldsymbol{s}|}$, $\overline{|\boldsymbol{t}|}$, $\overline{|\boldsymbol{u}|}$ and $\overline{|\boldsymbol{v}|}$, respectively, which will be useful in the rest of the paper.

\begin{convention*}
Thanks to Lemma~\ref{LemMoveWiseLemmaForHidingOnJSequences}, we henceforth dispense with the notation $\widehat{\mathcal{H}}_G^{\boldsymbol{s}, \sigma}$, where $G$ ranges over arenas, $\boldsymbol{s}$ over j-sequences of $G$, and $\sigma$ over permutations on $\overline{|\boldsymbol{s}|}$, implicitly admitting any order of `move-wise' execution of the operation $\mathcal{H}_G$.
Also, we write, abusing notation, $\mathcal{H}_G(\boldsymbol{s}) . \mathcal{H}_G(\boldsymbol{t}) . \mathcal{H}_G(\boldsymbol{u}) . \mathcal{H}_G(\boldsymbol{v})$ for $\widehat{\mathcal{H}}_G^{\boldsymbol{v}, \nu} \circ \widehat{\mathcal{H}}_G^{\boldsymbol{u}, \mu} \circ \widehat{\mathcal{H}}_G^{\boldsymbol{t}, \tau} \circ \widehat{\mathcal{H}}_G^{\boldsymbol{s}, \sigma}(\boldsymbol{stuv})$ given above, so that $\mathcal{H}_G(\boldsymbol{stuv}) = \mathcal{H}_G(\boldsymbol{s}) . \mathcal{H}_G(\boldsymbol{t}) . \mathcal{H}_G(\boldsymbol{u}) . \mathcal{H}_G(\boldsymbol{v})$.
\end{convention*}

Next, let us recall the notion of `relevant part' of previous moves, called \emph{views}:
\begin{definition}[Views \cite{abramsky1999game}] 
\label{DefViews}
Given a j-sequence $\boldsymbol{s}$ of an arena $G$, the \emph{\bfseries Player (P-) view} $\lceil \boldsymbol{s} \rceil_G$ and the \emph{\bfseries Opponent (O-) view} $\lfloor \boldsymbol{s} \rfloor_G$ (we often omit the subscript $G$) are given by the following induction on $|\boldsymbol{s}|$: 
\begin{itemize}

\item$\lceil \boldsymbol{\epsilon} \rceil_G \stackrel{\mathrm{df. }}{=} \boldsymbol{\epsilon}$;

\item $\lceil \boldsymbol{s} m \rceil_G \stackrel{\mathrm{df. }}{=} \lceil \boldsymbol{s} \rceil_G . m$ if $m$ is a P-move;

\item $\lceil \boldsymbol{s} m \rceil_G \stackrel{\mathrm{df. }}{=} m$ if $m$ is initial;

\item $\lceil \boldsymbol{s} m \boldsymbol{t} n \rceil_G \stackrel{\mathrm{df. }}{=} \lceil \boldsymbol{s} \rceil_G . m n$ if $n$ is an O-move with $\mathcal{J}_{\boldsymbol{s} m \boldsymbol{t} n}(n) = m$;

\item $\lfloor \boldsymbol{\epsilon} \rfloor_G \stackrel{\mathrm{df. }}{=} \boldsymbol{\epsilon}$;

\item $\lfloor \boldsymbol{s} m \rfloor_G \stackrel{\mathrm{df. }}{=} \lfloor \boldsymbol{s} \rfloor_G . m$ if $m$ is an O-move;

\item $\lfloor \boldsymbol{s} m \boldsymbol{t} n \rfloor_G \stackrel{\mathrm{df. }}{=} \lfloor \boldsymbol{s} \rfloor_G . m n$ if $n$ is a P-move with $\mathcal{J}_{\boldsymbol{s} m \boldsymbol{t} n}(n) = m$

\end{itemize}
where the justifiers of the remaining occurrences in $\lceil \boldsymbol{s} \rceil$ (resp. $\lfloor \boldsymbol{s} \rfloor$) are unchanged if they occur in $\lceil \boldsymbol{s} \rceil$ (resp. $\lfloor \boldsymbol{s} \rfloor$), and undefined otherwise. 
A \emph{\bfseries view} is a P- or O-view. 
\end{definition}

The idea behind Definition~\ref{DefViews} is as follows.
For a j-sequence $\boldsymbol{t}m$ of an arena $G$ such that $m$ is a P-move (resp. an O-move), the P-view $\lceil \boldsymbol{t} \rceil$ (resp. the O-view $\lfloor \boldsymbol{t} \rfloor$) is intended to be the currently `relevant part' of $\boldsymbol{t}$ for Player (resp. Opponent). That is, Player (resp. Opponent) is concerned only with the last O-move (resp. P-move), its justifier and that justifier's P-view (resp. O-view), which then recursively proceeds.

We are now ready to introduce a `dynamic generalization' of static legal positions:
\begin{definition}[Dynamic legal positions]
\label{DefDynamicLegalPositions}
Given an arena $G$, a \emph{\bfseries dynamic legal position} of $G$ is a j-sequence $\boldsymbol{s} \in \mathscr{J}_G$ that satisfies:
\begin{itemize}

\item \textsc{(Alternation).} If $\boldsymbol{s} = \boldsymbol{s_1} m n \boldsymbol{s_2}$, then $\lambda_G^\mathsf{OP} (m) \neq \lambda_G^\mathsf{OP} (n)$;


\item \textsc{(Generalized visibility).} If $\boldsymbol{s} = \boldsymbol{t} m \boldsymbol{u}$ with $m$ non-initial, and $d \in \mathbb{N} \cup \{ \omega \}$ satisfy $\lambda_G^{\mathbb{N}}(m) = 0 \vee \lambda_G^{\mathbb{N}}(m) > d$, then $\mathcal{J}^{\circleddash d}_{\boldsymbol{s}}(m)$ occurs in $\lceil \mathcal{H}_G^d(\boldsymbol{t}) \rceil_{\mathcal{H}^d(G)}$ if $m$ is a P-move, and it occurs in $\lfloor \mathcal{H}_G^d(\boldsymbol{t}) \rfloor_{\mathcal{H}^d(G)}$ if $m$ is an O-move;

\item \textsc{(IE-switch).} If $\boldsymbol{s} = \boldsymbol{s_1} m n \boldsymbol{s_2}$ with $\lambda_G^{\mathbb{N}}(m) \neq \lambda_G^{\mathbb{N}}(n)$, then $m$ is an O-move.

\end{itemize}
\end{definition}

\begin{notation*}
$\mathscr{L}_G$ denotes the set of all dynamic legal positions of a dynamic arena $G$. 
\end{notation*}

Recall that a static legal position \cite{abramsky1999game} of a static arena is a j-sequence of the arena that satisfies alternation and \emph{visibility} (i.e., generalized visibility only for $d = 0$).
It specifies the basic rules of a static game in the sense that every `development' or \emph{(valid) position} of the game must be a legal position of the underlying arena (but the converse does not necessarily hold):
\begin{itemize}

\item In a position of the static game, Opponent always makes the first move by a question, and then Player and Opponent alternately play (by alternation), in which every non-initial move must be made for a specific previous move;

\item The justifier of each non-initial move occurring in the position must belong to the `relevant' part of previous moves occurring in the position (by visibility).
\end{itemize}

The additional axioms on dynamic legal positions are conceptually natural ones:
\begin{itemize}

\item Generalized visibility is a generalization of visibility, which requires that visibility must hold after any iteration of the hiding operation on j-sequences;

\item IE-switch states that only Player can change a priority order during a play as internal moves are `invisible' to Opponent, where the same remark as the one in the axiom E4 is applied for the finer distinction of priority orders than the I/E-parity.

\end{itemize}

Note that a dynamic legal position of a static arena, seen as a dynamic arena whose moves are all external, is clearly a static legal position, and vice versa. 
Hence, dynamic legal positions are in fact a generalization of static legal positions. 

\begin{convention*}
Henceforth, a \emph{\bfseries legal position} refers to a dynamic legal position by default. 
\end{convention*}

\subsection{Dynamic Games}
We are now ready to define the central notion of \emph{dynamic games}:
\begin{definition}[Dynamic games]
\label{DefDynamicGames}
A \emph{\bfseries dynamic game} is a quintuple 
\begin{equation*}
G = (M_G, \lambda_G, \vdash_G, P_G, \simeq_G)
\end{equation*}
such that:
\begin{itemize}

\item The triple $(M_G, \lambda_G, \vdash_G)$ forms an arena (Definition~\ref{DefDynamicArenas});

\item $P_G$ is a subset of $\mathscr{L}_G$, whose elements are called \emph{\bfseries (valid) positions} of $G$, that satisfies:
\begin{itemize}

\item \textsc{(P1).} $P_G$ is non-empty and prefix-closed;

\item \textsc{(DP2).} If $\boldsymbol{s} m n \in P_G^{\mathsf{Even}}$ and $\lambda_G^{\mathbb{N}}(n) > 0$, then $\exists r \in M_G . \ \! \boldsymbol{s} m n r \in P_G$;

\item \textsc{(DP3).} Given $\boldsymbol{t}r, \boldsymbol{t'} r' \in P_G^{\mathsf{Odd}}$ and $i \in \mathbb{N}$ such that $i < \lambda_G^{\mathbb{N}}(r) = \lambda_G^{\mathbb{N}}(r')$, if $\mathcal{H}_G^i(\boldsymbol{t}) = \mathcal{H}_G^i(\boldsymbol{t'})$, then $\mathcal{H}_G^i(\boldsymbol{t} r) = \mathcal{H}_G^i(\boldsymbol{t'} r')$;

\end{itemize}

\item $\simeq_G$ is an equivalence relation on $P_G$, called the \emph{\bfseries identification of (valid) positions}, that satisfies:
\begin{itemize}

\item \textsc{(I1).} $\boldsymbol{s} \simeq_G \boldsymbol{t} \Rightarrow |\boldsymbol{s}| = |\boldsymbol{t}|$;

\item \textsc{(I2).} $\boldsymbol{s} m \simeq_G \boldsymbol{t} n \Rightarrow \boldsymbol{s} \simeq_G \boldsymbol{t} \wedge \lambda_G(m) = \lambda_G(n) \wedge (m, n \in M_G^{\mathsf{Init}} \vee (\exists i \in \overline{|\boldsymbol{s}|} . \ \! \mathcal{J}_{\boldsymbol{s} m}(m) = \boldsymbol{s}(i) \wedge \mathcal{J}_{\boldsymbol{t} n}(n) = \boldsymbol{t}(i)))$;

\item \textsc{(DI3).} $\forall d \in \mathbb{N} \cup \{ \omega \} . \ \! \boldsymbol{s} \simeq_G^d \boldsymbol{t} \wedge \boldsymbol{s} m \in P_G \Rightarrow \exists \boldsymbol{t} n \in P_G . \ \! \boldsymbol{s} m \simeq_G^d \boldsymbol{t}n$, where $\boldsymbol{u} \simeq_G^d \boldsymbol{v} \stackrel{\mathrm{df. }}{\Leftrightarrow} \exists \boldsymbol{u'}, \boldsymbol{v'} \in P_G . \ \! \boldsymbol{u'} \simeq_G \boldsymbol{v'} \wedge \mathcal{H}_G^d(\boldsymbol{u'}) = \mathcal{H}_G^d(\boldsymbol{u}) \wedge \mathcal{H}_G^d(\boldsymbol{v'}) = \mathcal{H}_G^d(\boldsymbol{v})$ for all $\boldsymbol{u}, \boldsymbol{v} \in P_G$.

\end{itemize}

\end{itemize}
A \emph{\bfseries play} of $G$ is an finitely or infinitely increasing sequence of positions $(\boldsymbol{\epsilon}, m_1, m_1m_2, \dots)$ of $G$.
A dynamic game whose moves are all external is said to be \emph{\bfseries normalized}.
\end{definition}

Recall that a static game \cite{abramsky1999game} is a quintuple similar to a dynamic game except that the underlying arena is static, and it only satisfies the axioms P1, I1, I2 and I3 (i.e., DI3 only for $d = 0$).
The axiom P1 corresponds to the natural phenomenon that a non-empty `moment' or position of a game must have the previous `moment'. 
Identifications of positions are originally introduced in \cite{abramsky2000full} and also employed in Section 3.6 of \cite{mccusker1998games}.
They are to identify positions up to inessential details of `tags' for disjoint union, particularly for \emph{exponential} ! (Definition~\ref{DefExponential}); each position $\boldsymbol{s} \in P_G$ of a game $G$ is a representative of the equivalence class $[\boldsymbol{s}] \stackrel{\mathrm{df. }}{=} \{ \boldsymbol{t} \in P_G \mid \boldsymbol{t} \simeq_G \boldsymbol{s} \ \! \} \in P_G / \! \simeq_G$ which we take as primary. 
For this underlying idea, the three axioms I1, I2 and I3 should make sense.

The additional axioms DP2 and DP3 are in order to enable Player to `play alone', i.e., Opponent does not have to choose odd-length positions, for the internal part of a play since conceptually Opponent cannot `see' internal moves; technically, the axiom DP2 is to preserve totality of dynamic strategies under the hiding operation (Corollary~\ref{CoroPreservationOfConstraintsOnDynamicStrategiesUnderHiding}), and the axiom DP3 is for \emph{external consistency} of dynamic strategies: A dynamic strategy behaves always in the same manner from the viewpoint of Opponent, i.e., the external part of a play by a dynamic strategy does not depend on the internal part (Theorem~\ref{ThmExternalConsistency}).
Note that the axiom DP2 is slightly involved to be preserved under the hiding operation (Theorem~\ref{ThmExternalClosureOfDynamicGames}); it is necessary to generalize the axiom I3 to DI3 for the same reason.

\begin{remark*}
It is certainly simpler to dispense with the identification $\simeq_G$ of positions for each game $G$ by adopting a simpler formulation of exponential $!$ as in \cite{mccusker1998games}; however, it would be mathematically \emph{ad-hoc} because the cartesian closed structure of games and strategies would not arise via the standard Girard translation.
Recall that the aim of the present work is to establish \emph{mathematics} of dynamics and intensionality of logic and computation, where `good' mathematics should be robust and general, not ad-hoc; also, it is interesting as future work to extend the present work to \emph{linear} logic and computation.
For these reasons, we have decided to retain $\simeq_G$ as a structure of each game $G$.
Moreover, we shall establish various reasonable properties on identification of positions, which adds credibility of the notions of dynamic games and strategies. 
\end{remark*}

\begin{convention*}
Henceforth, a \emph{\bfseries game} refers to a dynamic game by default. 
\end{convention*}

\begin{example}
\label{ExTerminalGame}
The \emph{\bfseries terminal game} $T \stackrel{\mathrm{df. }}{=} (\emptyset, \emptyset, \emptyset, \{ \boldsymbol{\epsilon} \}, \{ (\boldsymbol{\epsilon}, \boldsymbol{\epsilon}) \})$ is the simplest game.
\end{example}

\begin{example}
\label{ExFlatGames}
The \emph{\bfseries flat game} $\mathit{flat}(S)$ on a given set $S$ is defined as follows. 
The triple $\mathit{flat}(S) = (M_{\mathit{flat}(S)}, \lambda_{\mathit{flat}(S)}, \vdash_{\mathit{flat}(S)})$ is the flat arena in Example~\ref{ExFlatArenas}, $P_{\mathit{flat}(S)} \stackrel{\mathrm{df. }}{=} \{ \boldsymbol{\epsilon}, q \ \! \} \cup \{ q m \mid m \in S \ \! \}$, and $\simeq_{\mathit{flat}(S)} \ \stackrel{\mathrm{df. }}{=} \{ (\boldsymbol{s}, \boldsymbol{s}) \mid \boldsymbol{s} \in P_{\mathit{flat}(S)} \ \! \}$.
For instance, $N \stackrel{\mathrm{df. }}{=} \mathit{flat}(\mathbb{N})$ is the game of natural numbers sketched in the introduction, and $\boldsymbol{2} \stackrel{\mathrm{df. }}{=} \mathit{flat}(\mathbb{B})$ is the game of booleans. 
Also, $\boldsymbol{0} \stackrel{\mathrm{df. }}{=} \mathit{flat}(\emptyset)$ is the \emph{\bfseries empty game}.
\end{example}

Also, let us define a substructure relation between games: 
\begin{definition}[Subgames]
\label{DefDynamicSubgames}
Given games $G$ and $H$, we say that $H$ is a \emph{\bfseries (dynamic) subgame} of $G$, written $H \trianglelefteqslant G$, iff $M_H \subseteq M_G$, $\lambda_H = \lambda_G \upharpoonright M_H$, $\vdash_H \ \subseteq \ \vdash_G \cap \ \! ((\{ \star \} \cup M_H) \times M_H)$, $P_H \subseteq P_G$, $\forall d \in \mathbb{N} \cup \{ \omega \} . \ \! \simeq_H^d \ = \ \simeq_G^d \cap \ \! (P_H \times P_H)$ and $\mu(H) = \mu(G)$.
\end{definition} 

For $H \trianglelefteqslant G$, the condition on the identifications of positions is required \emph{for all numbers $d \in \mathbb{N} \cup \{ \omega \}$} so that the dynamic subgame relation $\trianglelefteqslant$ is preserved under the hiding operation (Theorem~\ref{ThmExternalClosureOfDynamicGames}); the last equation $\mu(H) = \mu(G)$ is to preserve the relation $\trianglelefteqslant$ under \emph{concatenation} of dynamic games (Definition~\ref{DefConcatenationOfDynamicGames}).

We shall later focus on \emph{well-founded} games:
\begin{definition}[Well-founded games \cite{clairambault2010totality}]
\label{DefWellFoundedGames}
A game $G$ is \emph{\bfseries well-founded} if $\vdash_G$ is well-founded \emph{downwards}, i.e., there is no countably infinite sequence $(m_i)_{i \in \mathbb{N}}$ of moves $m_i \in M_G$ such that $\star \vdash_G m_0 \wedge \forall i \in \mathbb{N} . \ \! m_i \vdash_G m_{i+1}$.
\end{definition}

Now, let us define the \emph{hiding operation} on games:
\begin{definition}[Hiding operation on games]
\label{DefHidingOperationOnDynamicGames}
Given $d \in \mathbb{N} \cup \{ \omega \}$, the \emph{\bfseries $\boldsymbol{d}$-hiding operation (on games)} maps each game $G$ to its \emph{\bfseries $\boldsymbol{d}$-external game} $\mathcal{H}^d(G)$ given by:
\begin{itemize}

\item The triple $(M_{\mathcal{H}^d(G)}, \lambda_{\mathcal{H}^d(G)}, \vdash_{\mathcal{H}^d(G)})$ is the $d$-external arena $\mathcal{H}^d(G)$ of the underlying arena $G$ (Definition~\ref{DefExternalDynamicArenas});

\item $P_{\mathcal{H}^d(G)} \stackrel{\mathrm{df. }}{=} \{ \mathcal{H}^d_G(\boldsymbol{s}) \mid \boldsymbol{s} \in P_G \ \! \}$;

\item $\mathcal{H}^d_G(\boldsymbol{s}) \simeq_{\mathcal{H}^d(G)} \mathcal{H}^d_G(\boldsymbol{t}) \stackrel{\mathrm{df. }}{\Leftrightarrow} \boldsymbol{s} \simeq_G^d \boldsymbol{t}$.

\end{itemize}
\end{definition}

Now, we give the first main theorem of the present work:
\begin{theorem}[External closure of games]
\label{ThmExternalClosureOfDynamicGames}
Given $d \in \mathbb{N} \cup \{ \omega \}$, (resp. well-founded) games are closed under the operation $\mathcal{H}^d$, and $H \trianglelefteqslant G$ implies $\mathcal{H}^d(H) \trianglelefteqslant \mathcal{H}^d(G)$.
\end{theorem}
\begin{proof}
Let $G$ be a game, and assume $d \in \mathbb{N} \cup \{ \omega \}$; we have to show that $\mathcal{H}^d(G)$ is a game.
By Lemma~\ref{LemExternalClosureLemma}, it suffices to show that j-sequences in $P_{\mathcal{H}^d(G)}$ are legal positions of the arena $\mathcal{H}^d(G)$, the set $P_{\mathcal{H}^d(G)}$ satisfies the axioms P1, DP2 and DP3, and the relation $\simeq_{\mathcal{H}^d(G)}$ is an equivalence relation on $P_{\mathcal{H}^d(G)}$ that satisfies the axioms I1, I2 and DI3.
Since $\mu(G) \in \mathbb{N}$, we assume $d \in \mathbb{N}$.

For alternation, assume $\boldsymbol{s_1} m n \boldsymbol{s_2} \in P_{\mathcal{H}^d(G)}$; we have to show $\lambda_{\mathcal{H}^d(G)}^\mathsf{OP} (m) \neq \lambda_{\mathcal{H}^d(G)}^\mathsf{OP} (n)$. 
We have $\mathcal{H}^d_G(\boldsymbol{t}_1 m m_1 m_2 \dots m_k n \boldsymbol{t}_2) = \boldsymbol{s_1} m n \boldsymbol{s_2}$ for some $\boldsymbol{t}_1 m m_1 m_2 \dots m_k n \boldsymbol{t_2} \in P_G$, where $\mathcal{H}^d_G(\boldsymbol{t_1}) = \boldsymbol{s_1}$, $\mathcal{H}^d_G(\boldsymbol{t_2}) = \boldsymbol{s_2}$ and $\mathcal{H}^d_G(m_1 m_2 \dots m_k) = \boldsymbol{\epsilon}$. 
Note that $(\lambda_G^\mathbb{N} (m) = 0 \vee \lambda_G^\mathbb{N} (m) > d) \wedge (\lambda_G^\mathbb{N} (n) = 0 \vee \lambda_G^\mathbb{N} (n) > d)$ and $0 < \lambda_G^\mathbb{N} (m_i) \leqslant d$ for $i = 1, 2, \dots, k$. 
By the axioms E3 and E4 on $G$, $k$ must be an even number, and thus $\lambda_{\mathcal{H}^d(G)}^\mathsf{OP} (m) = \lambda_G^\mathsf{OP} (m) = \lambda_G^\mathsf{OP} (m_2) = \lambda_G^\mathsf{OP} (m_4) = \dots = \lambda_G^\mathsf{OP} (m_{k}) \neq \lambda_G^\mathsf{OP} (n) = \lambda_{\mathcal{H}^d(G)}^\mathsf{OP} (n)$.

For generalized visibility, let $\boldsymbol{t} m \boldsymbol{u} \in P_{\mathcal{H}^d(G)}$ with $m$ non-initial. 
We have to show, for each $e \in \mathbb{N} \cup \{ \omega \}$, that if $\boldsymbol{t}m$ is $e$-complete, then: 
\begin{itemize}
\item if $m$ is a P-move, then the justifier $(\mathcal{J}^{\circleddash d}_{\boldsymbol{s}})^{\circleddash e}(m)$ occurs in $\lceil \mathcal{H}_{\mathcal{H}^d(G)}^e(\boldsymbol{t}) \rceil_{\mathcal{H}^e(\mathcal{H}^d(G))}$;
\item if $m$ is an O-move, then the justifier $(\mathcal{J}^{\circleddash d}_{\boldsymbol{s}})^{\circleddash e}(m)$ occurs in $\lfloor \mathcal{H}_{\mathcal{H}^d(G)}^e(\boldsymbol{t}) \rfloor_{\mathcal{H}^e(\mathcal{H}^d(G))}$.
\end{itemize}
Again, for $\mu(G) \in \mathbb{N}$, we may assume without loss of generality that $e \in \mathbb{N}$.
Note that the condition is then equivalent to:
\begin{itemize}
\item if $m$ is a P-move, then the justifier $\mathcal{J}_{\boldsymbol{s}}^{\circleddash (d+e)}(m)$ occurs in $\lceil \mathcal{H}_{G}^{d+e}(\boldsymbol{t'}) \rceil_{\mathcal{H}^{d+e}(G)}$;
\item if $m$ is an O-move, then the justifier $\mathcal{J}_{\boldsymbol{s}}^{\circleddash (d+e)}(m)$ occurs in $\lfloor \mathcal{H}_{G}^{d+e}(\boldsymbol{t'}) \rfloor_{\mathcal{H}^{d+e}(G)}$
\end{itemize}
where $\boldsymbol{t'}m \in P_G$ such that $\mathcal{H}^d_G(\boldsymbol{t'}m) = \boldsymbol{t}m$.
It holds by generalized visibility on $G$.

For IE-switch, let $\boldsymbol{s_1} m n \boldsymbol{s_2} \in P_{\mathcal{H}^d(G)}$ such that $\lambda_{\mathcal{H}^d(G)}^{\mathbb{N}}(m) \neq \lambda_{\mathcal{H}^d(G)}^{\mathbb{N}}(n)$.
Then, there is some $\boldsymbol{t_1} m \boldsymbol{u} n \boldsymbol{t_2} \in P_G$ such that $\mathcal{H}^d_G(\boldsymbol{t_1} m \boldsymbol{u} n \boldsymbol{t_2}) = \boldsymbol{s_1} m n \boldsymbol{s_2}$, where note that $\lambda_{G}^{\mathbb{N}}(m) \neq \lambda_{G}^{\mathbb{N}}(n)$.
Therefore, if $\boldsymbol{u} = \boldsymbol{\epsilon}$, then we clearly have $\lambda_{\mathcal{H}^d(G)}^{\mathsf{OP}}(m) = \mathsf{O}$ by IE-switch on $G$; otherwise, i.e., $\boldsymbol{u} = l \boldsymbol{u'}$, then we have the same conclusion as $\lambda_G^{\mathbb{N}}(m) \neq \lambda_G^{\mathbb{N}}(l)$.

We have established $P_{\mathcal{H}^d(G)} \subseteq \mathscr{L}_{\mathcal{H}^d(G)}$. 
Next, we verify the axioms P1, DP2 and DP3:
\begin{itemize}

\item \textsc{(P1).} Because $\boldsymbol{\epsilon} \in P_G$, we have $\boldsymbol{\epsilon} = \mathcal{H}^d_G(\boldsymbol{\epsilon}) \in P_{\mathcal{H}^d(G)}$; thus, $P_{\mathcal{H}^d(G)}$ is non-empty. 
For prefix-closure, let $\boldsymbol{s} m \in P_{\mathcal{H}^d(G)}$; we have to show $\boldsymbol{s} \in P_{\mathcal{H}^d(G)}$. 
There must be some $\boldsymbol{t} m \in P_G$ such that $\boldsymbol{s} m = \mathcal{H}^d_G(\boldsymbol{t} m) = \mathcal{H}^d_G(\boldsymbol{t}) m$. 
Thus, $\boldsymbol{s} = \mathcal{H}^d_G(\boldsymbol{t}) \in P_{\mathcal{H}^d(G)}$.

\item \textsc{(DP2).} If $\boldsymbol{s} m n \in P_{\mathcal{H}^d(G)}^{\mathsf{Even}}$ and $\lambda_{\mathcal{H}^d(G)}^{\mathbb{N}}(n) > 0$, then there is some $\boldsymbol{t} m \boldsymbol{u} n \in P_G^{\mathsf{Even}}$ such that $\mathcal{H}^d_G(\boldsymbol{t}m\boldsymbol{u}n) = \boldsymbol{s} m n$ and $\lambda_G^{\mathbb{N}}(n) > d > 0$.
Hence, by the axiom DP2 on $G$, there is some $\boldsymbol{t} m \boldsymbol{u} n r \in P_G$ such that $\lambda_G^{\mathbb{N}}(r) = \lambda_G^{\mathbb{N}}(n) > d$ by IE-switch on $G$.
Therefore, we have found $\boldsymbol{s} m n r = \mathcal{H}^d(\boldsymbol{t}m\boldsymbol{u}nr) \in P_{\mathcal{H}^d(G)}$, establishing DP2 on $\mathcal{H}^d(G)$.

\item \textsc{(DP3).} Assume $\boldsymbol{t} r, \boldsymbol{t'} r' \in P_{\mathcal{H}^d(G)}^{\mathsf{Odd}}$ and $i \in \mathbb{N}$ such that $i < \lambda_{\mathcal{H}^d(G)}^{\mathbb{N}}(r) = \lambda_{\mathcal{H}^d(G)}^{\mathbb{N}}(r')$ and $\mathcal{H}_{\mathcal{H}^d(G)}^i(\boldsymbol{t}) = \mathcal{H}_{\mathcal{H}^d(G)}^i(\boldsymbol{t'})$. 
We have some $\boldsymbol{u} r, \boldsymbol{u'} r' \in P_G$ with $\mathcal{H}^d_G(\boldsymbol{u}) = \boldsymbol{t}$ and $\mathcal{H}^d_G(\boldsymbol{u'}) = \boldsymbol{t'}$. 
Then, $\mathcal{H}_{G}^{d+i}(\boldsymbol{u}) = \mathcal{H}^i_{\mathcal{H}^d(G)}(\mathcal{H}^d_G(\boldsymbol{u})) = \mathcal{H}^i_{\mathcal{H}^d(G)}(\boldsymbol{t}) = \mathcal{H}^i_{\mathcal{H}^d(G)}(\boldsymbol{t'}) = \mathcal{H}^i_{\mathcal{H}^d(G)}(\mathcal{H}^d_G(\boldsymbol{u'})) = \mathcal{H}_{G}^{d+i}(\boldsymbol{u'})$.
Hence, by the axiom DP3 on $G$, we have $r = r'$ and $\mathcal{J}_{\boldsymbol{t} r}^{\circleddash i}(r) = \mathcal{J}_{\boldsymbol{u} r}^{\circleddash (d+i)}(r) = \mathcal{J}_{\boldsymbol{u'} r'}^{\circleddash (d+i)}(r') = \mathcal{J}_{\boldsymbol{t'} r'}^{\circleddash i}(r')$, establishing DP3 on $\mathcal{H}^d(G)$.
\end{itemize}

Next, $\simeq_{\mathcal{H}^d(G)}$ is a well-defined relation on $P_{\mathcal{H}^d(G)}$ since $\mathcal{H}_G^d(\boldsymbol{s}) \simeq_{\mathcal{H}^d(G)} \mathcal{H}_G^d(\boldsymbol{t})$ does not depend on the choice of representatives $\boldsymbol{s}, \boldsymbol{t} \in P_G$.
Also, it is straightforward to see that $\simeq_{\mathcal{H}^d(G)}$ is an equivalence relation.
Now, we show that $\simeq_{\mathcal{H}^d(G)}$ satisfies the axioms I1, I2 and DI3.
Note that I1 and I2 on $\simeq_{\mathcal{H}^d(G)}$ immediately follow from those on $\simeq_{G}$.
For DI3 on $\simeq_{\mathcal{H}^d(G)}$, if $\mathcal{H}^d_G(\boldsymbol{s}) \simeq_{\mathcal{H}^d(G)}^{e} \mathcal{H}^d_G(\boldsymbol{t})$, and $\mathcal{H}^d_G(\boldsymbol{s}) . m \in P_{\mathcal{H}^d(G)}$, where we may assume $e \neq \omega$, then $\exists \boldsymbol{s'}m \in P_G . \ \! \mathcal{H}^d_G(\boldsymbol{s'}m) = \mathcal{H}^d_G(\boldsymbol{s}) . m$, and so $\mathcal{H}_G^{d+e}(\boldsymbol{s'}) = \mathcal{H}_G^{d+e}(\boldsymbol{s}) \simeq_{\mathcal{H}^{d+e}(G)} \mathcal{H}_G^{d+e}(\boldsymbol{t})$.
By DI3 on $\simeq_G$, we may conclude that $\exists \boldsymbol{t} n \in P_G . \ \! \boldsymbol{s'} m \simeq_G^{d+e} \boldsymbol{t}n$, whence we obtain $\mathcal{H}^d_G(\boldsymbol{t}) . n \in P_{\mathcal{H}^d(G)}$ such that $\mathcal{H}^d_G(\boldsymbol{s}).m = \mathcal{H}^d_G(\boldsymbol{s'}m) \simeq_{\mathcal{H}^d(G)}^{e} \mathcal{H}^d_G(\boldsymbol{t}n) = \mathcal{H}^d_G(\boldsymbol{t}).n$.

Finally, the preservation of the dynamic subgame relation $\trianglelefteqslant$ under the operation $\mathcal{H}^d$ is clear from the definition, completing the proof.
\end{proof}

\begin{corollary}[Stepwise hiding on games]
For any game $G$, we have $\mathcal{H}^1(\mathcal{H}^i(G)) = \mathcal{H}^{i+1}(G)$ for all $i \in \mathbb{N}$.
\end{corollary}
\begin{proof}
By Lemmata~\ref{LemStepwiseHidingOnDynamicArenas} and \ref{LemStepwiseHidingOnJSequences}, it suffices to show the equation $\simeq_{\mathcal{H}^1(\mathcal{H}^i(G))} \ \! = \ \simeq_{\mathcal{H}^{i+1}(G)}$.
Then, given $\boldsymbol{s}, \boldsymbol{t} \in P_G$, we have:
\begin{align*}
\Leftrightarrow & \ \mathcal{H}_{\mathcal{H}^i(G)}^1(\mathcal{H}_G^i(\boldsymbol{s})) \simeq_{\mathcal{H}^1(\mathcal{H}^i(G))} \mathcal{H}_{\mathcal{H}^i(G)}^1(\mathcal{H}_G^{i}(\boldsymbol{t})) \\
\Leftrightarrow & \ \exists \mathcal{H}^i(\boldsymbol{s'}), \mathcal{H}^i(\boldsymbol{t'}) \in P_{\mathcal{H}^i(G)} . \ \! \mathcal{H}^i(\boldsymbol{s'}) \simeq_{\mathcal{H}^i(G)} \mathcal{H}^i(\boldsymbol{t'}) \wedge \mathcal{H}^1_{\mathcal{H}^i(G)}(\mathcal{H}^i(\boldsymbol{s'})) = \mathcal{H}^1_{\mathcal{H}^i(G)}(\mathcal{H}^i(\boldsymbol{s})) \\ & \ \wedge \mathcal{H}^1_{\mathcal{H}^i(G)}(\mathcal{H}^i(\boldsymbol{t'})) = \mathcal{H}^1_{\mathcal{H}^i(G)}(\mathcal{H}^i(\boldsymbol{t})) \\
\Leftrightarrow & \ \exists \boldsymbol{s''}, \boldsymbol{t''} \in P_G . \ \! \boldsymbol{s''} \simeq_G \boldsymbol{t''} \wedge \mathcal{H}_G^{i+1}(\boldsymbol{s''}) = \mathcal{H}_G^{i+1}(\boldsymbol{s}) \wedge \mathcal{H}_G^{i+1}(\boldsymbol{t''}) = \mathcal{H}_G^{i+1}(\boldsymbol{t}) \\
\Leftrightarrow & \ \mathcal{H}_G^{i+1}(\boldsymbol{s}) \simeq_{\mathcal{H}^{i+1}(G)} \mathcal{H}_G^{i+1}(\boldsymbol{t}) 
\end{align*}
which establishes the required equation.
\end{proof}

By the corollary, we may just focus on $\mathcal{H}^1$: 
\begin{convention*}
We write $\mathcal{H}$ for $\mathcal{H}^1$ and call it the \emph{\bfseries hiding operation (on games)}; $\mathcal{H}^i$ denotes the $i$-times iteration of $\mathcal{H}$ for all $i \in \mathbb{N}$.
\end{convention*}

\begin{corollary}[Hiding operation on legal positions]
\label{CoroHidingOperationOnDynamicLegalPositions}
Given an arena $G$ and a number $d \in \mathbb{N} \cup \{ \omega \}$, we have $\{ \mathcal{H}_G^d(\boldsymbol{s}) \mid \boldsymbol{s} \in \mathscr{L}_G \ \! \} = \mathscr{L}_{\mathcal{H}^d(G)}$.
\end{corollary}
\begin{proof}
Since there is an upper bound $\mu(G) \in \mathbb{N}$, it suffices to consider the case $d \in \mathbb{N}$.
Then, by Lemmata~\ref{LemStepwiseHidingOnDynamicArenas} and \ref{LemStepwiseHidingOnJSequences}, we may just focus on the case $d = 1$.

The inclusion $\{ \mathcal{H}_G(\boldsymbol{s}) \mid \boldsymbol{s} \in \mathscr{L}_G \ \! \} \subseteq \mathscr{L}_{\mathcal{H}(G)}$ is immediate by Theorem~\ref{ThmExternalClosureOfDynamicGames}.
For the other inclusion, let $\boldsymbol{t} \in \mathscr{L}_{\mathcal{H}(G)}$; we shall find some $\boldsymbol{s} \in \mathscr{L}_G$ such that
\begin{enumerate}
\item $\mathcal{H}_G(\boldsymbol{s}) = \boldsymbol{t}$;
\item 1-internal moves in $\boldsymbol{s}$ occur as even-length consecutive segments $m_1 m_2 \dots m_{2k}$, where $m_i$ justifies $m_{i+1}$ for $i= 1, 2, \dots, 2k-1$; 
\item $\boldsymbol{s}$ is 1-complete.
\end{enumerate}
We proceed by induction on $|\boldsymbol{t}|$. The base case $\boldsymbol{t} = \boldsymbol{\epsilon}$ is trivial.
For the inductive step, let $\boldsymbol{t} m \in \mathscr{L}_{\mathcal{H}(G)}$. 
Then, $\boldsymbol{t} \in \mathscr{L}_{\mathcal{H}(G)}$, and by the induction hypothesis there is some $\boldsymbol{s} \in \mathscr{L}_G$ that satisfies the three conditions (n.b., the first one is for $\boldsymbol{t}$).

If $m$ is initial, then $\boldsymbol{s} m \in \mathscr{L}_{G}$, and $\boldsymbol{s} m$ satisfies the three conditions. 
Thus, assume that $m$ is non-initial; we may write $\boldsymbol{t} m = \boldsymbol{t_1} n \boldsymbol{t_2} m$, where $m$ is justified by $n$. 

We then need a case analysis:
\begin{itemize}

\item Assume $n \vdash_G m$. 
We take $\boldsymbol{s} m$, where $m$ points to $n$.
Then, $\boldsymbol{s} m \in \mathscr{L}_G$ since:
\begin{itemize}

\item \textsc{(Justification).} It is immediate because $n \vdash_G m$.

\item \textsc{(Alternation).} By the condition 3 on $\boldsymbol{s}$, the last moves of $\boldsymbol{s}$ and $\boldsymbol{t}$ just coincide. 
Thus, the alternation condition holds for $\boldsymbol{s} m$. 

\item \textsc{(Generalized visibility).} It suffices to establish the visibility on $\boldsymbol{s} m$, as the other cases are included as the generalized visibility on $\boldsymbol{t} m$. It is straightforward to see that, by the condition 2 on $\boldsymbol{s}$, if the view of $\boldsymbol{t}$ contains $n$, then so does the view of $\boldsymbol{s}$. And since $\boldsymbol{t} m \in \mathscr{L}_{\mathcal{H}(G)}$, the view of $\boldsymbol{t}$ contains $n$. Hence, the view of $\boldsymbol{s}$ contains $n$ as well.

\item \textsc{(IE-switch).} Again, the last moves of $\boldsymbol{s}$ and $\boldsymbol{t}$ coincide by the condition 3 on $\boldsymbol{s}$; thus, IE-switch on $\boldsymbol{t} m$ can be directly applied. 

\end{itemize}
Also, it is easy to see that $\boldsymbol{s} m$ satisfies the three conditions.

\item Assume $n \neq \star$ and $\exists k \in \mathbb{N}^+, m_1, m_2, \dots, m_{2k} \in M_G \setminus M_{\mathcal{H}(G)}$ such that 
\begin{equation*}
n \vdash_G m_1 \wedge \forall i \in \overline{k} . \ \! m_{2i-1} \vdash_G m_{2i} \wedge m_{2k} \vdash_G m.
\end{equation*}
We then take $\boldsymbol{s} m_1 m_2 \dots m_{2k} m$, in which $m_1$ points to $n$, $m_{i}$ points to $m_{i-1}$ for $i = 2, 3, \dots, 2k$, and $m$ points to $m_{2k}$. 
Then, $\boldsymbol{s} m_1 m_2 \dots m_{2k} m \in \mathscr{L}_G$ because:
\begin{itemize}

\item \textsc{(Justification).} Obvious.

\item \textsc{(Alternation).} By the condition 3 on $\boldsymbol{s}$, the last moves of $\boldsymbol{s}$ and $\boldsymbol{t}$ just coincide. Thus, the alternation condition holds for $\boldsymbol{s} m_1 m_2 \dots m_{2k} m$. 

\item \textsc{(Generalized visibility).} By the same argument as the above case. 

\item \textsc{(IE-switch).} It clearly holds by the axiom E4.
\end{itemize}
Finally, it is easy to see that $\boldsymbol{s} m_1 m_2 \dots m_{2k} m$ satisfies the three conditions.
\end{itemize}
We have completed the case analysis. 
\end{proof}

\subsection{Constructions on Dynamic Games}
\label{ConstructionsOnDynamicGames}
Next, we show that dynamic games accommodate all the standard constructions on static games \cite{abramsky1999game}, i.e., they preserve the additional axioms for dynamic games, as well as some new constructions. 
This result implies that the notion of dynamic games (Definition~\ref{DefDynamicGames}) is in some sense `correct'.

\begin{convention*}
For brevity, we usually omit `tags' for disjoint union of sets. 
For instance, we write $x \in A + B$ iff $x \in A$ or $x \in B$ (not both); also, given relations $R_A \subseteq A \times A$ and $R_B \subseteq B \times B$, we write $R_A + R_B$ for the relation on the disjoint union $A + B$ such that $(x, y) \in R_A + R_B \stackrel{\mathrm{df. }}{\Leftrightarrow} (x, y) \in R_A \vee (x, y) \in R_B$ (not both). 
\end{convention*}

Let us begin with \emph{tensor (product)} $\otimes$. 
Roughly, a position of the tensor $A \otimes B$ of games $A$ and $B$ is an interleaving mixture of a position of $A$ and a position of $B$, in which an $AB$-parity change is made always by Opponent. 
Formally:
\begin{definition}[Tensor of games \cite{abramsky1999game}]
\label{DefTensorOfGames}
Given games $A$ and $B$, the \emph{\bfseries tensor (product)} $A \otimes B$ of $A$ and $B$ is defined by:
\begin{itemize}

\item $M_{A \otimes B} \stackrel{\mathrm{df. }}{=} M_A + M_B$;

\item $\lambda_{A \otimes B} \stackrel{\mathrm{df. }}{=} [\lambda_A, \lambda_B]$;

\item $\vdash_{A \otimes B} \ \stackrel{\mathrm{df. }}{=} \ \vdash_A + \ \! \vdash_B$;

\item $P_{A \otimes B} \stackrel{\mathrm{df. }}{=} \{ \boldsymbol{s} \in \mathscr{L}_{A \otimes B} \mid \boldsymbol{s} \upharpoonright A \in P_A, \boldsymbol{s} \upharpoonright B \in P_B \ \! \}$;

\item $\boldsymbol{s} \simeq_{A \otimes B} \boldsymbol{t} \stackrel{\mathrm{df. }}{\Leftrightarrow} \boldsymbol{s} \upharpoonright A \simeq_A \boldsymbol{t} \upharpoonright A \wedge \boldsymbol{s} \upharpoonright B \simeq_B \boldsymbol{t} \upharpoonright B \wedge \forall i \in \mathbb{N} . \ \! \boldsymbol{s}(i) \in M_A \Leftrightarrow \boldsymbol{t}(i) \in M_A$

\end{itemize}
where $\boldsymbol{s} \upharpoonright A$ (resp. $\boldsymbol{s} \upharpoonright B$) denotes the j-subsequence of $\boldsymbol{s}$ that consists of occurrences of moves of $A$ (resp. $B$).
\end{definition}

In fact, as explained in \cite{abramsky1997semantics}, in a position of a tensor $A \otimes B$, only Opponent can switch between the component games $A$ and $B$ (by alternation). 
\begin{example}
\label{ExTensor}
Consider the tensor $N \otimes N$ of the natural number game $N$ with itself, whose maximal position is either of the following forms:
\begin{center}
\begin{tabular}{ccccccccc}
$N_{[0]}$ & $\otimes$ & $N_{[1]}$ &&&& $N_{[0]}$ & $\otimes$ & $N_{[1]}$ \\ \cline{1-3} \cline{7-9}
\tikzmark{ctensor1} $q_{[0]}$&&&&&&&&\tikzmark{ctensor2} $q_{[1]}$ \\
\tikzmark{dtensor1} $n_{[0]}$&&&&&&&&\tikzmark{dtensor2} $m_{[1]}$ \\
&&\tikzmark{ctensor3} $q_{[1]}$&&&&\tikzmark{ctensor4} $q_{[0]}$&& \\
&&\tikzmark{dtensor3} $m_{[1]}$&&&&\tikzmark{dtensor4} $n_{[0]}$&&
\end{tabular}
\begin{tikzpicture}[overlay, remember picture, yshift=.25\baselineskip]
\draw [->] ({pic cs:dtensor1}) [bend left] to ({pic cs:ctensor1});
\draw [->] ({pic cs:dtensor2}) [bend left] to ({pic cs:ctensor2});
\draw [->] ({pic cs:dtensor3}) [bend left] to ({pic cs:ctensor3});
\draw [->] ({pic cs:dtensor4}) [bend left] to ({pic cs:ctensor4});
\end{tikzpicture}
\end{center}
where $n, m \in \mathbb{N}$, and $(\_)_{[i]}$ ($i = 0, 1$) are again arbitrary, unspecified `tags' such that $[0] \neq [1]$ to distinguish the two copies of $N$, and the arrows represent pointers.
Henceforth, however, we usually omit `tags' $(\_)_{[i]}$ unless it is strictly necessary. 
\end{example}

\begin{theorem}[Well-defined tensor of games]
\label{ThmWellDefinedTensorOnDynamicGames}
(Resp. well-founded) games are closed under tensor $\otimes$.
\end{theorem}
\begin{proof}
Since static games are closed under tensor $\otimes$ \cite{abramsky1999game}, it suffices to show that $\otimes$ preserves the condition on labeling function and the axioms E1, E2, E4, DP2, DP3 and DI3 (n.b., $\otimes$ clearly preserves well-foundedness of games).
However, non-trivial ones are just DP3 and DI3; thus, we just focus on these two axioms.
Let $A$ and $B$ be any games.
To verify DP3 on $A \otimes B$, let $\boldsymbol{s} l m n, \boldsymbol{s'} l' m' n' \in P_{A \otimes B}^{\mathsf{Odd}}$ and $i \in \mathbb{N}$ such that $\mathcal{H}_{A \otimes B}^i(\boldsymbol{s}lm) = \mathcal{H}_{A \otimes B}^i(\boldsymbol{s'}l'm')$ and $i < \lambda_{A \otimes B}^{\mathbb{N}}(n) = \lambda_{A \otimes B}^{\mathbb{N}}(n')$.
Note that $\lambda_{A \otimes B}^{\mathbb{N}}(m) = \lambda_{A \otimes B}^{\mathbb{N}}(n) = \lambda_{A \otimes B}^{\mathbb{N}}(n') = \lambda_{A \otimes B}^{\mathbb{N}}(m')$ by IE-switch.
At a first glance, it seems that $A \otimes B$ does not satisfy DP3 as Opponent may choose to play in $A$ or $B$ at will. 
It is, however, not the case for internal moves for $\boldsymbol{s} l m n \in P_{A \otimes B}^{\mathsf{Odd}}$ with $m$ internal implies $m, n \in M_A$ or $m, n \in M_B$. 
This property immediately follows from Table~\ref{TableDoubleParityDiagram} which shows all the possible transitions of OP- and IE-parities for a play of $A \otimes B$, where a state $(X^Y, Z^W)$ indicates that the next move of $A$ (resp. $B$) has the OP-parity $X$ (resp. $Z$) and the IE-parity $Y$ (resp. $W$).
\begin{table}
\begin{diagram}
(\mathsf{P}^\mathsf{I}, \mathsf{O}^\mathsf{E}) & & \lTo^{A} & & (\mathsf{O}^\mathsf{E}, \mathsf{O}^\mathsf{E}) & & \rTo^{B} & & (\mathsf{O}^\mathsf{E}, \mathsf{P}^\mathsf{I}) \\
\dCorresponds_A & & & \ldCorresponds^A & & \rdCorresponds^B & & & \dCorresponds^B \\
(\mathsf{O}^\mathsf{I}, \mathsf{O}^\mathsf{E}) & \rTo^A & (\mathsf{P}^\mathsf{E}, \mathsf{O}^\mathsf{E}) & & & & (\mathsf{O}^\mathsf{E}, \mathsf{P}^\mathsf{E}) & \lTo^B & (\mathsf{O}^\mathsf{E}, \mathsf{O}^\mathsf{I})
\end{diagram}
\caption{The double parity diagram for the tensor $A \otimes B$.}
\label{TableDoubleParityDiagram}
\end{table}
Note that $m = m'$ and $\mathcal{J}^{\circleddash i}_{\boldsymbol{s}lm}(m) = \mathcal{J}^{\circleddash i}_{\boldsymbol{s'}l'm'}(m')$ as $\mathcal{H}_{A \otimes B}^i(\boldsymbol{s}l) . m = \mathcal{H}_{A \otimes B}^i(\boldsymbol{s}lm) = \mathcal{H}_{A \otimes B}^i(\boldsymbol{s'}l'm') = \mathcal{H}_{A \otimes B}^i(\boldsymbol{s'}l') . m'$. 
Thus, $m$, $n$, $m'$ and $n'$ belong to the same component game. 
If $m, n, m', n' \in M_A$, then $(\boldsymbol{s}l \upharpoonright A) . m n, (\boldsymbol{s'}l' \upharpoonright A) . m' n' \in P_A^{\mathsf{Odd}}$, $\mathcal{H}_A^i((\boldsymbol{s}l \upharpoonright A) . m) = \mathcal{H}_{A \otimes B}^i(\boldsymbol{s}lm) \upharpoonright \mathcal{H}^i(A) = \mathcal{H}_{A \otimes B}^i(\boldsymbol{s'}l'm') \upharpoonright \mathcal{H}^i(A) = \mathcal{H}_A^i((\boldsymbol{s'}l' \upharpoonright A) . m')$ and $i < \lambda_A^{\mathbb{N}}(n) = \lambda_A^{\mathbb{N}}(n')$; thus by DP3 on $A$, we conclude that $n = n'$ and $\mathcal{J}^{\circleddash i}_{\boldsymbol{s}lmn}(n) = \mathcal{J}^{\circleddash i}_{(\boldsymbol{s}l \upharpoonright A) . mn}(n) = \mathcal{J}^{\circleddash i}_{(\boldsymbol{s'}l' \upharpoonright A) . m'n'}(n') = \mathcal{J}^{\circleddash i}_{\boldsymbol{s'}l'm'n'}(n')$.
The other case is completely analogous, showing that $A \otimes B$ satisfies DP3.

Finally, to show that $A \otimes B$ satisfies DI3, assume $\boldsymbol{s} \simeq_{A \otimes B}^d \boldsymbol{t}$ and $\boldsymbol{s} m \in P_{A \otimes B}$ for some $d \in \mathbb{N}$; we have to find some $\boldsymbol{t} n \in P_{A \otimes B}$ such that $\boldsymbol{s} m \simeq_{A \otimes B}^d \boldsymbol{t} n$.
Assume $m \in M_A$ for the other case is symmetric. 
Since $\boldsymbol{s} \simeq_{A \otimes B}^d \boldsymbol{t}$, we have some $\boldsymbol{s'} \simeq_{A \otimes B} \boldsymbol{t'}$ such that $\mathcal{H}_{A \otimes B}^d(\boldsymbol{s'}) = \mathcal{H}_{A \otimes B}^d(\boldsymbol{s})$ and $\mathcal{H}_{A \otimes B}^d(\boldsymbol{t'}) = \mathcal{H}_{A \otimes B}^d(\boldsymbol{t})$.
Thus, $\boldsymbol{s'} \upharpoonright A \simeq_A \boldsymbol{t'} \upharpoonright A$, $\mathcal{H}_A^d(\boldsymbol{s} \upharpoonright A) = \mathcal{H}_{A \otimes B}^d(\boldsymbol{s}) \upharpoonright \mathcal{H}^d(A) = \mathcal{H}_{A \otimes B}^d(\boldsymbol{s'}) \upharpoonright \mathcal{H}^d(A) = \mathcal{H}_A^d(\boldsymbol{s'} \upharpoonright A)$ and $\mathcal{H}_A^d(\boldsymbol{t} \upharpoonright A) = \mathcal{H}_{A \otimes B}^d(\boldsymbol{t}) \upharpoonright \mathcal{H}^d(A) = \mathcal{H}_{A \otimes B}^d(\boldsymbol{t'}) \upharpoonright \mathcal{H}^d(A) = \mathcal{H}_A^d(\boldsymbol{t'} \upharpoonright A)$, whence $\boldsymbol{s} \upharpoonright A \simeq_A^d \boldsymbol{t} \upharpoonright A$.
Similarly, $\boldsymbol{s} \upharpoonright B \simeq_B^d \boldsymbol{t} \upharpoonright B$ with $\boldsymbol{s'} \upharpoonright B \simeq_B \boldsymbol{t'} \upharpoonright B$, $\mathcal{H}_B^d(\boldsymbol{s} \upharpoonright B) = \mathcal{H}_B^d(\boldsymbol{s'} \upharpoonright B)$ and $\mathcal{H}_B^d(\boldsymbol{t} \upharpoonright B) = \mathcal{H}_B^d(\boldsymbol{t'} \upharpoonright B)$.
Now, since $(\boldsymbol{s} \upharpoonright A) . m = \boldsymbol{s}m \upharpoonright A \in P_A$, we have some $(\boldsymbol{t} \upharpoonright A) . n \in P_A$ such that $(\boldsymbol{s} \upharpoonright A) . m \simeq_{A}^d (\boldsymbol{t} \upharpoonright A) . n$, i.e., some $\boldsymbol{u} \simeq_A \boldsymbol{v}$ such that $\mathcal{H}_A^d(\boldsymbol{u}) = \mathcal{H}_A^d((\boldsymbol{s} \upharpoonright A) . m)$ and $\mathcal{H}_A^d(\boldsymbol{v}) = \mathcal{H}_A^d((\boldsymbol{t} \upharpoonright A) . n)$.
By Table~\ref{TableDoubleParityDiagram}, we may obtain a unique $\boldsymbol{\tilde{s}} \in P_{A \otimes B}$ from $\boldsymbol{u}$ and $\boldsymbol{s'} \upharpoonright B$ and a unique $\boldsymbol{\tilde{t}} \in P_{A \otimes B}$ from $\boldsymbol{v}$ and $\boldsymbol{t'} \upharpoonright B$ such that $\boldsymbol{\tilde{s}} \simeq_{A \otimes B} \boldsymbol{\tilde{t}}$, $\mathcal{H}_{A \otimes B}^d(\boldsymbol{\tilde{s}}) = \mathcal{H}_{A \otimes B}^d(\boldsymbol{s} m)$ and $\mathcal{H}_{A \otimes B}^d(\boldsymbol{\tilde{t}}) = \mathcal{H}_{A \otimes B}^d(\boldsymbol{t} n)$, establishing $\boldsymbol{s} m \simeq_{A \otimes B}^d \boldsymbol{t} n$.
\end{proof}

Next, let us recall \emph{linear implication} $\multimap$, which has been illustrated by examples in Section~\ref{Introduction}.
The linear implication $A \multimap B$ is intended to be the `space' of \emph{linear functions} from $A$ to $B$ in the sense of linear logic \cite{girard1987linear}, i.e., they consume exactly one input in $A$ to produce an output in $B$ (strictly speaking, they consume \emph{at most one} input since it is possible that no moves of $A$ are performed at all during a play of $A \multimap B$).

One additional point for \emph{dynamic} games is that we need to apply the $\omega$-hiding operation $\mathcal{H}^\omega$ to the domain $A$ since otherwise the linear implication $A \multimap B$ may not satisfy the axiom DP2 or DP3.
It conceptually makes sense too for the roles of Player and Opponent in $A$ are exchanged, and thus Player should not be able to `see' internal moves of $A$.
\begin{definition}[Linear implication between games \cite{abramsky1999game}]
\label{DefLinearImplication}
The \emph{\bfseries linear implication} $A \multimap B$ from a game $A$ to another $B$ is defined by: 
\begin{itemize}

\item $M_{A \multimap B} \stackrel{\mathrm{df. }}{=} M_{\mathcal{H}^\omega(A)} + M_B$;

\item $\lambda_{A \multimap B} \stackrel{\mathrm{df. }}{=} [\overline{\lambda_{\mathcal{H}^\omega(A)}}, \lambda_B]$, where $\overline{\lambda_{\mathcal{H}^\omega(A)}} \stackrel{\mathrm{df. }}{=} \langle \overline{\lambda_{\mathcal{H}^\omega(A)}^\textsf{OP}}, \lambda_{\mathcal{H}^\omega(A)}^\textsf{QA}, \lambda_{\mathcal{H}^\omega(A)}^\mathbb{N} \rangle$, and $\overline{\lambda_{G}^\textsf{OP}} (m) \stackrel{\mathrm{df. }}{=} \begin{cases} \mathsf{P} \ &\text{if $\lambda_{G}^\textsf{OP} (m) = \textsf{O}$;} \\ \textsf{O} \ &\text{otherwise} \end{cases}$ for any game $G$;

\item $\star \vdash_{A \multimap B} m \stackrel{\mathrm{df. }}{\Leftrightarrow} \star \vdash_B m$;

\item $m \vdash_{A \multimap B} n \ (m \neq \star) \stackrel{\mathrm{df. }}{\Leftrightarrow} (m \vdash_{\mathcal{H}^\omega(A)} n) \vee (m \vdash_B n) \vee (\star \vdash_B m \wedge \star \vdash_{\mathcal{H}^\omega(A)} n)$;

\item $P_{A \multimap B} \stackrel{\mathrm{df. }}{=} \{ \boldsymbol{s} \in \mathscr{L}_{\mathcal{H}^\omega(A) \multimap B} \mid \boldsymbol{s} \upharpoonright \mathcal{H}^\omega(A) \in P_{\mathcal{H}^\omega(A)}, \boldsymbol{s} \upharpoonright B \in P_B \ \! \}$;

\item $\boldsymbol{s} \simeq_{A \multimap B} \boldsymbol{t} \stackrel{\mathrm{df. }}{\Leftrightarrow} \boldsymbol{s} \upharpoonright \mathcal{H}^\omega(A) \simeq_{\mathcal{H}^\omega(A)} \boldsymbol{t} \upharpoonright \mathcal{H}^\omega(A) \wedge \boldsymbol{s} \upharpoonright B \simeq_B \boldsymbol{t} \upharpoonright B \wedge \forall i \in \mathbb{N} . \ \! \boldsymbol{s}(i) \in M_{\mathcal{H}^\omega(A)} \Leftrightarrow \boldsymbol{t}(i) \in M_{\mathcal{H}^\omega(A)}$

\end{itemize}
where pointers from an initial occurrence of $\mathcal{H}^\omega(A)$ to that of $B$ in $\boldsymbol{s}$ are deleted.
\end{definition}

Dually to $A \otimes B$, it is easy to see that during a play of $A \multimap B$ only Player may switch between $\mathcal{H}^\omega(A)$ and $B$ (again by alternation); see \cite{abramsky1997semantics} for the details.

\begin{example}
See again the examples of linear implication in Section~\ref{Introduction} to see how Definition~\ref{DefLinearImplication} actually works. 
\end{example}

\begin{theorem}[Well-defined linear implication between games]
\label{ThmWellDefinedLinearImplicationOnDynamicGames}
(Resp. well-founded) games are closed under linear implication.
\end{theorem}
\begin{proof}
Again, it suffices to show the preservation property of the additional conditions on the labeling function and the axioms E1, E2, E4, DP2, DP3 and DI3.
For brevity, assume that $A$ is normalized and consider $A \multimap B$.
Again, non-trivial conditions are just DP3 and DI3, but DI3 may be shown in a way similar to the case of tensor.

To verify DP3, let $i \in \mathbb{N}$ and $\boldsymbol{s} l m n, \boldsymbol{s'} l' m' n' \in P_{A \multimap B}^{\mathsf{Odd}}$ such that $\mathcal{H}_{A \multimap B}^i(\boldsymbol{s}lm) = \mathcal{H}_{A \multimap B}^i(\boldsymbol{s'}l'm')$ and $i < \lambda_{A \multimap B}^{\mathbb{N}}(n) = \lambda_{A \multimap B}^{\mathbb{N}}(n')$.
Again, $m$ and $m'$ are both internal, and so $m$, $n$, $m'$ and $n'$ all belong to $B$.
Thus, $(\boldsymbol{s} l \upharpoonright B) . m n, (\boldsymbol{s'} l' \upharpoonright B) . m' n' \in P_{B}^{\mathsf{Odd}}$ such that $\mathcal{H}_B^i((\boldsymbol{s}l \upharpoonright B) . m) = \mathcal{H}_{A \multimap B}^i(\boldsymbol{s} l m) \upharpoonright \mathcal{H}^i(B) = \mathcal{H}_{A \multimap B}^i(\boldsymbol{s'} l' m') \upharpoonright \mathcal{H}^i(B) = \mathcal{H}_B^i((\boldsymbol{s'}l' \upharpoonright B) . m')$ and $i < \lambda_{B}^{\mathbb{N}}(n) = \lambda_{B}^{\mathbb{N}}(n')$; thus, by DP2 on $B$, we may conclude that $n = n'$ and $\mathcal{J}_{\boldsymbol{s} l m n}^{\circleddash i}(n) = \mathcal{J}_{(\boldsymbol{s}l \upharpoonright B) . mn}^{\circleddash i}(n) = \mathcal{J}_{(\boldsymbol{s}l \upharpoonright B) . mn}^{\circleddash i}(n') = \mathcal{J}_{\boldsymbol{s} l m n}^{\circleddash i}(n')$.
\end{proof}

Next, \emph{product} $\&$ forms the categorical product in the categories of static games and strategies \cite{abramsky1999game}.
A position of the product $A \& B$ is simply a position of $A$ or $B$:  
\begin{definition}[Product of games \cite{abramsky1999game}]
\label{DefProduct}
Given games $A$ and $B$, the \emph{\bfseries product} $A \& B$ of $A$ and $B$ is defined by:

\begin{itemize}

\item $M_{A \& B} \stackrel{\mathrm{df. }}{=} M_A + M_B$;

\item $\lambda_{A \& B} \stackrel{\mathrm{df. }}{=} [\lambda_A, \lambda_B]$;

\item $\vdash_{A \& B} \ \stackrel{\mathrm{df. }}{=} \ \vdash_A + \ \! \vdash_B$;

\item $P_{A \& B} \stackrel{\mathrm{df. }}{=} \{ \boldsymbol{s} \in \mathscr{L}_{A \& B} \mid (\boldsymbol{s} \upharpoonright A \in P_A \wedge \boldsymbol{s} \upharpoonright B = \boldsymbol{\epsilon}) \vee (\boldsymbol{s} \upharpoonright A = \boldsymbol{\epsilon} \wedge \boldsymbol{s} \upharpoonright B \in P_B) \ \! \}$;

\item $\boldsymbol{s} \simeq_{A \& B} \boldsymbol{t} \stackrel{\mathrm{df. }}{\Leftrightarrow} \boldsymbol{s} \simeq_A \boldsymbol{t} \vee \boldsymbol{s} \simeq_B \boldsymbol{t}$.

\end{itemize}
\end{definition}

\begin{example}
A maximal position of the product $\boldsymbol{2} \& N$ is either of the following forms:
\begin{center}
\begin{tabular}{ccccccccc}
$\boldsymbol{2}$ & $\&$ & $N$ &&&& $\boldsymbol{2}$ & $\&$ & $N$ \\ \cline{1-3} \cline{7-9}
$q$ \tikzmark{cproduct1} &&&&&&&&$q$ \tikzmark{cproduct2} \\
$b$ \tikzmark{dproduct1} &&&&&&&&$n$ \tikzmark{dproduct2}
\end{tabular}
\begin{tikzpicture}[overlay, remember picture, yshift=.25\baselineskip]
\draw [->] ({pic cs:dproduct1}) [bend right] to ({pic cs:cproduct1});
\draw [->] ({pic cs:dproduct2}) [bend right] to ({pic cs:cproduct2});
\end{tikzpicture}
\end{center}
where $b \in \mathbb{B}$ and $n \in \mathbb{N}$. 
\end{example}

Now, for our game-semantic CCBoC (given in Section~\ref{DynamicGameSemantics}), let us generalize product:
\begin{notation*}
Given a function $f : X \to Y$ and a subset $Z \subseteq X$, we write $f \downharpoonright Z : X \setminus Z \to Y$ for the restriction of $f$ to the subset $X \setminus Z \subseteq X$. 
\end{notation*}

\begin{definition}[Pairing of games]
\label{DefPairingOfDynamicGames}
The \emph{\bfseries pairing} $\langle L, R \rangle$ of games $L$ and $R$ such that $\mathcal{H}^\omega(L) \trianglelefteqslant C \multimap A$ and $\mathcal{H}^\omega(R) \trianglelefteqslant C \multimap B$ for some normalized games $A$, $B$ and $C$ is defined by:
\begin{itemize}

\item $M_{\langle L, R \rangle} \stackrel{\mathrm{df. }}{=} M_C + (M_L \setminus M_C) + (M_R \setminus M_C)$, where `tags' for the disjoint union is chosen in such a way that $\mathcal{H}^\omega(\langle L, R \rangle) \trianglelefteqslant C \multimap A \& B$ holds;

\item $\lambda_{\langle L, R \rangle} \stackrel{\mathrm{df. }}{=} [\overline{\lambda_C}, \lambda_L \downharpoonright M_C, \lambda_R \downharpoonright M_C]$;

\item $m \vdash_{\langle L, R \rangle} n \stackrel{\mathrm{df. }}{\Leftrightarrow} (\mathit{att}_{\langle L, R \rangle}(m) = \mathit{att}_{\langle L, R \rangle}(n) \vee \mathit{att}_{\langle L, R \rangle}(m) = C \vee \mathit{att}_{\langle L, R \rangle}(n) = C) \wedge (\mathit{peel}_{\langle L, R \rangle}(m) \vdash_L \mathit{peel}_{\langle L, R \rangle}(n) \vee \mathit{peel}_{\langle L, R \rangle}(m) \vdash_R \mathit{peel}_{\langle L, R \rangle}(n))$;

\item $P_{\langle L, R \rangle} \stackrel{\mathrm{df. }}{=} \{ \boldsymbol{s} \in \mathscr{L}_{L \& R} \mid (\boldsymbol{s} \upharpoonright L \in P_L \wedge \boldsymbol{s} \upharpoonright R = \boldsymbol{\epsilon}) \vee (\boldsymbol{s} \upharpoonright L = \boldsymbol{\epsilon} \wedge \boldsymbol{s} \upharpoonright R \in P_R) \ \! \}$;

\item $\boldsymbol{s} \simeq_{\langle L, R \rangle} \boldsymbol{t} \stackrel{\mathrm{df. }}{\Leftrightarrow} (\boldsymbol{s} \upharpoonright L = \boldsymbol{\epsilon} \Leftrightarrow \boldsymbol{t} \upharpoonright L = \boldsymbol{\epsilon}) \wedge \boldsymbol{s} \upharpoonright L \simeq_L \boldsymbol{t} \upharpoonright L \wedge \boldsymbol{s} \upharpoonright R \simeq_R \boldsymbol{t} \upharpoonright R$

\end{itemize}
where the map $\mathit{peel}_{\langle L, R \rangle} : M_{\langle L, R \rangle} \rightarrow M_L \cup M_R$ is the obvious left inverse of the `tagging' for $M_{\langle L, R \rangle}$, $\boldsymbol{s} \upharpoonright L$ (resp. $\boldsymbol{s} \upharpoonright R$) is the j-subsequence of $\boldsymbol{s}$ that consists of moves $x$ such that $\mathit{peel}_{\langle L, R \rangle}(x) \in M_L$ (resp. $\mathit{peel}_{\langle L, R \rangle}(x) \in M_R$) yet changed into $\mathit{peel}_{\langle L, R \rangle}(x)$, and the map $\mathit{att}_{\langle L, R \rangle} : M_{\langle L, R \rangle} \rightarrow \{ L, R, C \}$ is given by $\mathit{att}_{\langle L, R \rangle}(m) \stackrel{\mathrm{df. }}{=} \begin{cases} L &\text{if $\mathit{peel}_{\langle L, R \rangle}(m) \in M_L \setminus M_C$;} \\ R &\text{if $\mathit{peel}_{\langle L, R \rangle}(m) \in M_R \setminus M_C$;} \\ C &\text{otherwise (i.e., if $\mathit{peel}_{\langle L, R \rangle}(m) \in M_C$).} \end{cases}$
\end{definition}

Pairing of games is indeed a generalization of product for we have $\langle T \multimap A, T \multimap B \rangle = T \multimap A \& B$ for any games $A$ and $B$, where note that each game $G$ coincides with the linear implication $T \multimap G$ up to `tags'.
Also, we shall see that the \emph{(generalized) pairing $\langle \sigma, \tau \rangle$ of strategies} $\sigma : L$ and $\tau : R$ forms a strategy on the pairing $\langle L, R \rangle$ (Definition~\ref{DefPairingOfDynamicStrategies}).

\begin{theorem}[Well-defined pairing of games]
\label{ThmWellDefinedPairingOnDynamicGames}
If (resp. well-founded) games $L$ and $R$ satisfy $\mathcal{H}^\omega(L) \trianglelefteqslant C \multimap A$ and $\mathcal{H}^\omega(R) \trianglelefteqslant C \multimap B$ for normalized games $A$, $B$ and $C$, then the pairing $\langle L, R \rangle$ is a (resp. well-founded) game that satisfies $\mathcal{H}^\omega(\langle L, R \rangle) \trianglelefteqslant C \multimap A \& B$.
\end{theorem}
\begin{proof}
Similar to and simpler than the case of tensor.
\end{proof}

Now, let us recall \emph{exponential} $!$, which is essentially the countably-infinite iteration of tensor, i.e., $!A$ and $A \otimes A \otimes \dots$ coincide up to `tags'.
Precisely, it is defined as follows:
\begin{definition}[Exponential of games \cite{abramsky2000full,mccusker1998games}]
\label{DefExponential}
Given a game $A$, the \emph{\bfseries exponential} $!A$ of $A$ is defined by:
\begin{itemize}

\item $M_{!A} \stackrel{\mathrm{df. }}{=} M_A \times \mathbb{N}$;

\item $\lambda_{!A} : (a, i) \mapsto \lambda_A(a)$;

\item $\star \vdash_{!A} (a, i) \stackrel{\mathrm{df. }}{\Leftrightarrow} \star \vdash_A a$;

\item $(a, i) \vdash_{!A} (a', i') \stackrel{\mathrm{df. }}{\Leftrightarrow} i = i' \wedge a \vdash_A a'$;

\item $P_{!A} \stackrel{\mathrm{df. }}{=} \{ \boldsymbol{s} \in \mathscr{L}_{!A} \mid \forall i \in \mathbb{N} . \ \! \boldsymbol{s} \upharpoonright i \in P_A \ \! \}$;

\item $\boldsymbol{s} \simeq_{!A} \boldsymbol{t} \stackrel{\mathrm{df. }}{\Leftrightarrow} \exists \varphi \in \mathcal{P}(\mathbb{N}) . \ \! \forall i \in \mathbb{N} . \ \! \boldsymbol{s} \upharpoonright \varphi (i) \simeq_A \boldsymbol{t} \upharpoonright i \wedge \pi_2^\ast (\boldsymbol{s}) = (\varphi \circ \pi_2)^\ast(\boldsymbol{t})$

\end{itemize}
where $\boldsymbol{s} \upharpoonright i$ is the j-subsequence of $\boldsymbol{s}$ that consists of occurrences of moves of the form $(a, i)$ yet changed into $a$, and $\mathcal{P}(\mathbb{N})$ is the set of all permutations of natural numbers.
\end{definition}

\begin{example}
A typical position of the exponential $! \boldsymbol{2}$ is as follows:
\begin{center}
\begin{tabular}{c}
$! \boldsymbol{2}$ \\ \hline
\tikzmark{ExC1} $(q, 10)$ \\
\tikzmark{ExD1} $(\mathit{tt}, 10)$ \\
\tikzmark{ExC2} $(q, 100)$ \\
\tikzmark{ExD2} $(\mathit{ff}, 100)$ 
\end{tabular}
\begin{tikzpicture}[overlay, remember picture, yshift=.25\baselineskip]
\draw [->] ({pic cs:ExD1}) [bend left] to ({pic cs:ExC1});
\draw [->] ({pic cs:ExD2}) [bend left] to ({pic cs:ExC2});
\end{tikzpicture}
\end{center}
\end{example}

Now, it should be clear, from the definition of $\simeq_{!A}$, why we have equipped each game with an identification of positions: A particular choice of the `tag' $(\_, i)$ for an exponential $!A$ should not matter; since this identification may occur \emph{locally} in games in a \emph{nested} form, e.g., $!(!A \otimes B)$, $!A \multimap B$, etc., it gives a neat solution to define a \emph{tailored} identification $\simeq_G$ of positions as part of the structure of each game $G$.
It was first introduced by \cite{abramsky2000full} and also employed in \cite{mccusker1998games}.

Exponential enables us, via \emph{Girard's translation} \cite{girard1987linear} $A \Rightarrow B \stackrel{\mathrm{df. }}{=} \ !A \multimap B$, to model the construction $\Rightarrow$ of the usual \emph{\bfseries implication} (or the \emph{\bfseries function space}).

\begin{example}
\label{ExExponential}
In the linear implication $\boldsymbol{2} \& \boldsymbol{2} \multimap \boldsymbol{2}$, Player may play at most only in one $\boldsymbol{2}$ out of the domain $\boldsymbol{2} \& \boldsymbol{2}$:
\begin{center}
\begin{tabular}{ccccccccccccc}
$\boldsymbol{2}$ & $\&$ & $\boldsymbol{2}$ & $\multimap$ & $\boldsymbol{2}$ &&&& $\boldsymbol{2}$ & $\&$ & $\boldsymbol{2}$ & $\multimap$ & $\boldsymbol{2}$ \\ \cline{1-5} \cline{9-13}
&&&& \tikzmark{cimplication2} $q$ \tikzmark{cimplication1} &&&&&&&& \tikzmark{cimplication5} $q$ \tikzmark{cimplication4} \\
\tikzmark{cimplication3} $q$ \tikzmark{dimplication2} &&&&&&&&&& \tikzmark{cimplication6} $q$ \tikzmark{dimplication5} && \\
\tikzmark{dimplication3} $b^{(1)}$&&&&&&&&&& \tikzmark{dimplication6} $b^{(1)}$&& \\
&&&&$b^{(2)}$ \tikzmark{dimplication1} &&&&&&&&$b^{(2)}$ \tikzmark{dimplication4}
\end{tabular}
\begin{tikzpicture}[overlay, remember picture, yshift=.25\baselineskip]
\draw [->] ({pic cs:dimplication1}) [bend right] to ({pic cs:cimplication1});
\draw [->] ({pic cs:dimplication2}) to ({pic cs:cimplication2});
\draw [->] ({pic cs:dimplication3}) [bend left] to ({pic cs:cimplication3});
\draw [->] ({pic cs:dimplication4}) [bend right] to ({pic cs:cimplication4});
\draw [->] ({pic cs:dimplication5}) to ({pic cs:cimplication5});
\draw [->] ({pic cs:dimplication6}) [bend left] to ({pic cs:cimplication6});
\end{tikzpicture}
\end{center}
where $b^{(1)}, b^{(2)} \in \mathbb{B}$. 
On the other hand, however, positions of the implication $\boldsymbol{2} \& \boldsymbol{2} \Rightarrow \boldsymbol{2} = \ ! (\boldsymbol{2} \& \boldsymbol{2}) \multimap \boldsymbol{2}$ are of the expected form; for instance:
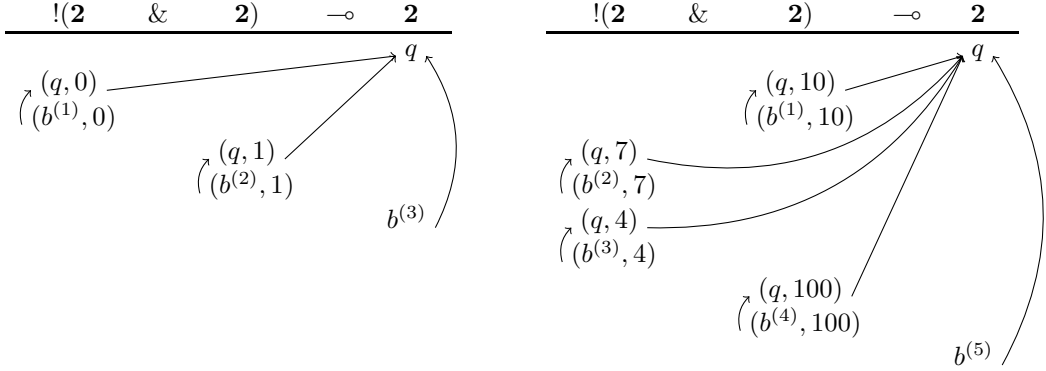
\begin{center}
\begin{tabular}{ccccccccccccc}
$!(\boldsymbol{2}$ & $\&$ & $\boldsymbol{2})$ & $\multimap$ & $\boldsymbol{2}$ &&&& $!(\boldsymbol{2}$ & $\&$ & $\boldsymbol{2})$ & $\multimap$ & $\boldsymbol{2}$ \\ \cline{1-5} \cline{9-13}
&&&& \tikzmark{cimplication17} $q$ \tikzmark{cimplication16} &&&&&&&& \tikzmark{cimplication8} $q$ \tikzmark{cimplication7} \\
\tikzmark{cimplication18} $(q, 0)$ \tikzmark{dimplication17} &&&&&&&&&& \tikzmark{cimplication9} $(q, 10)$ \tikzmark{dimplication8} && \\
\tikzmark{dimplication18} $(b^{(1)}, 0)$&&&&&&&&&& \tikzmark{dimplication9} $(b^{(1)}, 10)$&& \\
&& \tikzmark{cimplication20} $(q, 1)$ \tikzmark{dimplication19} &&&&&& \tikzmark{cimplication11} $(q, 7)$ \tikzmark{dimplication10} &&&& \\
&& \tikzmark{dimplication20} $(b^{(2)}, 1)$&&&&&& \tikzmark{dimplication11} $(b^{(2)}, 7)$&&&& \\
&&&&$b^{(3)}$ \tikzmark{dimplication16} &&&& \tikzmark{cimplication13} $(q, 4)$ \tikzmark{dimplication12} &&&& \\
&&&&&&&& \tikzmark{dimplication13} $(b^{(3)}, 4)$&&&& \\
&&&&&&&&&& \tikzmark{cimplication15} $(q, 100)$ \tikzmark{dimplication14} && \\
&&&&&&&&&& \tikzmark{dimplication15} $(b^{(4)}, 100)$&& \\
&&&&&&&&&&&& $b^{(5)}$ \tikzmark{dimplication7}
\end{tabular}
\begin{tikzpicture}[overlay, remember picture, yshift=.25\baselineskip]
\draw [->] ({pic cs:dimplication7}) [bend right] to ({pic cs:cimplication7});
\draw [->] ({pic cs:dimplication8}) to ({pic cs:cimplication8});
\draw [->] ({pic cs:dimplication9}) [bend left] to ({pic cs:cimplication9});
\draw [->] ({pic cs:dimplication10}) [bend right]  to ({pic cs:cimplication8});
\draw [->] ({pic cs:dimplication11}) [bend left] to ({pic cs:cimplication11});
\draw [->] ({pic cs:dimplication12}) [bend right] to ({pic cs:cimplication8});
\draw [->] ({pic cs:dimplication13}) [bend left] to ({pic cs:cimplication13});
\draw [->] ({pic cs:dimplication14}) to ({pic cs:cimplication8});
\draw [->] ({pic cs:dimplication15}) [bend left] to ({pic cs:cimplication15});
\draw [->] ({pic cs:dimplication16}) [bend right] to ({pic cs:cimplication16});
\draw [->] ({pic cs:dimplication17}) to ({pic cs:cimplication17});
\draw [->] ({pic cs:dimplication18}) [bend left] to ({pic cs:cimplication18});
\draw [->] ({pic cs:dimplication19}) to ({pic cs:cimplication17});
\draw [->] ({pic cs:dimplication20}) [bend left] to ({pic cs:cimplication20});
\end{tikzpicture}
\end{center}
where $b^{(1)}, b^{(2)}, b^{(3)}, b^{(4)}, b^{(5)} \in \mathbb{B}$. 
Hence, e.g., Player may play as conjunction $\wedge : \mathbb{B} \times \mathbb{B} \rightarrow \mathbb{B}$ or disjunction $\vee : \mathbb{B} \times \mathbb{B} \rightarrow \mathbb{B}$ on the implication $\boldsymbol{2} \& \boldsymbol{2} \Rightarrow \boldsymbol{2}$ in the obvious manner, but not on the linear implication $\boldsymbol{2} \& \boldsymbol{2} \multimap \boldsymbol{2}$.
This example illustrates why the standard notion of functions corresponds in game semantics to implication $\Rightarrow$, not linear one $\multimap$.
\end{example}

For the game-semantic CCBoC, let us generalize exponential of games as follows:
\begin{definition}[Promotion of games]
\label{DefPromotionOfDynamicGames}
Given a game $G$ such that $\mathcal{H}^\omega(G) \trianglelefteqslant \ !A \multimap B$ for some normalized games $A$ and $B$, the \emph{\bfseries promotion} $G^\dagger$ of $G$ is defined by:
\begin{itemize}

\item $M_{G^\dagger} \stackrel{\mathrm{df. }}{=} ((M_G \setminus M_{!A}) \times \mathbb{N}) + M_{!A}$;

\item $\lambda_{G^\dagger} : ((m, i) \in (M_G \setminus M_{!A}) \times \mathbb{N}) \mapsto \lambda_G(m), ((a, j) \in M_{!A}) \mapsto \lambda_G(a, j)$;

\item $\star \vdash_{G^\dagger} (m, i) \stackrel{\mathrm{df. }}{\Leftrightarrow} \star \vdash_G m$ for all $i \in \mathbb{N}$;

\item $(m, i) \vdash_{G^\dagger} (n, j) \stackrel{\mathrm{df. }}{\Leftrightarrow} (i = j \wedge m, n \in M_G \setminus M_{!A} \wedge m \vdash_G n) \\ \vee (i = j \wedge m \vdash_A n) \vee (m \in M_G \setminus M_{!A} \wedge (n, j) \in M_{!A} \wedge m \vdash_G (n, j))$;

\item $P_{G^\dagger} \stackrel{\mathrm{df. }}{=} \{ \boldsymbol{s} \in \mathscr{L}_{G^\dagger} \mid \forall i \in \mathbb{N} . \ \! \boldsymbol{s} \upharpoonright i \in P_G \ \! \}$;

\item $\boldsymbol{s} \simeq_{G^\dagger} \boldsymbol{t} \stackrel{\mathrm{df. }}{\Leftrightarrow} \exists \varphi \in \mathcal{P}(\mathbb{N}) . \ \! \forall i \in \mathbb{N} . \ \! \boldsymbol{s} \upharpoonright \varphi(i) \simeq_G \boldsymbol{t} \upharpoonright i \wedge \pi_2^\ast(\boldsymbol{s}) = (\varphi \circ \pi_2)^\ast(\boldsymbol{t})$

\end{itemize}
where $\boldsymbol{s} \upharpoonright i$ is the j-subsequence of $\boldsymbol{s}$ that consists of moves $(m, i)$ with $m \in M_G \setminus M_{!A}$, or $(a, \langle i, j \rangle)$ with $a \in M_A \wedge j \in \mathbb{N}$, yet changed into $m$ or $(a, j)$, respectively.
\end{definition}

Note that we have $(T \multimap A)^\dagger = \ !T \multimap \ !A$ for any game $A$, and therefore promotion of games is indeed a generalization of exponential. 
Also, we shall see later that the \emph{(generalized) promotion $\phi^\dagger$ of a strategy} $\phi : G$ forms a strategy on the promotion $G^\dagger$. 

\begin{example}
Let us consider the promotion $(!A \multimap B)^\dagger$, where $A$ and $B$ are arbitrary normalized games.
If there is the following position of $!A \multimap B$:
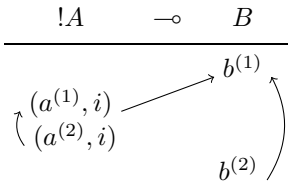
\begin{center}
\begin{tabular}{ccc}
$!A$ & $\multimap$ & $B$ \\ \hline 
&&\tikzmark{cpromotion71} $b^{(1)}$ \tikzmark{cpromotion73} \\
\tikzmark{cpromotion72} $(a^{(1)}, i)$ \tikzmark{dpromotion71}&& \\
\tikzmark{dpromotion72} $(a^{(2)}, i)$ && \\
&&$b^{(2)}$ \tikzmark{dpromotion73}
\end{tabular}
\begin{tikzpicture}[overlay, remember picture, yshift=.25\baselineskip]
\draw [->] ({pic cs:dpromotion71}) to ({pic cs:cpromotion71});
\draw [->] ({pic cs:dpromotion72}) [bend left] to ({pic cs:cpromotion72});
\draw [->] ({pic cs:dpromotion73}) [bend right] to ({pic cs:cpromotion73});
\end{tikzpicture}
\end{center}
then there is the following position of the promotion $(!A \multimap B)^\dagger$, where note that $j, j' \in \mathbb{N}$ are arbitrarily chosen by Opponent: 
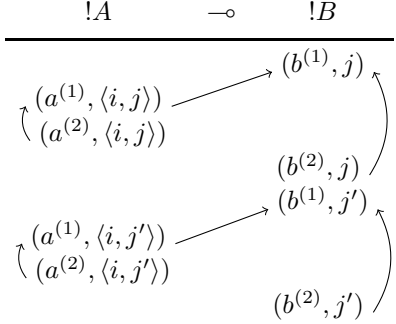
\begin{center}
\begin{tabular}{ccc}
$!A$ & $\multimap$ & $!B$ \\ \hline 
&&\tikzmark{cpromotion11} $(b^{(1)}, j)$ \tikzmark{cpromotion13} \\
\tikzmark{cpromotion12} $(a^{(1)}, \langle i, j \rangle)$ \tikzmark{dpromotion11} && \\
\tikzmark{dpromotion12} $(a^{(2)}, \langle i, j \rangle)$ && \\
&& $(b^{(2)}, j)$ \tikzmark{dpromotion13} \\
&& \tikzmark{cpromotion14} $(b^{(1)}, j')$ \tikzmark{cpromotion16} \\
\tikzmark{cpromotion15} $(a^{(1)}, \langle i, j' \rangle)$ \tikzmark{dpromotion14} && \\
\tikzmark{dpromotion15} $(a^{(2)}, \langle i, j' \rangle)$ && \\
&& $(b^{(2)}, j')$ \tikzmark{dpromotion16}
\end{tabular}
\begin{tikzpicture}[overlay, remember picture, yshift=.25\baselineskip]
\draw [->] ({pic cs:dpromotion11}) to ({pic cs:cpromotion11});
\draw [->] ({pic cs:dpromotion12}) [bend left] to ({pic cs:cpromotion12});
\draw [->] ({pic cs:dpromotion13}) [bend right] to ({pic cs:cpromotion13});
\draw [->] ({pic cs:dpromotion14}) to ({pic cs:cpromotion14});
\draw [->] ({pic cs:dpromotion15}) [bend left] to ({pic cs:cpromotion15});
\draw [->] ({pic cs:dpromotion16}) [bend right] to ({pic cs:cpromotion16});
\end{tikzpicture}
\end{center}
\end{example}

\begin{theorem}[Well-defined promotion of games]
\label{ThmWellDefinedPromotionOnDynamicGames}
If a (resp. well-founded) game $G$ satisfies $\mathcal{H}^\omega(G) \trianglelefteqslant \ !A \multimap B$ for some normalized games $A$ and $B$, then $G^\dagger$ is a (resp. well-founded) game that satisfies $\mathcal{H}^\omega(G)^\dagger \trianglelefteqslant \ !A \multimap \ !B$.
\end{theorem}
\begin{proof}
Similar to the case of tensor. 
\end{proof}

Now, let us introduce a new, central construction on games, which formalizes the construction for `internal communication' between strategies sketched in Section~\ref{Introduction}:
\begin{definition}[Concatenation and composition of games]
\label{DefConcatenationOfDynamicGames}
Given games $J$ and $K$ that satisfies $\mathcal{H}^\omega(J) \trianglelefteqslant A \multimap B$ and $\mathcal{H}^\omega(K) \trianglelefteqslant B \multimap C$ for some normalized games $A$, $B$ and $C$, the \emph{\bfseries concatenation} $J \ddagger K$ of $J$ and $K$ is defined by:
\begin{itemize}

\item $M_{J \ddagger K} \stackrel{\mathrm{df. }}{=} M_J + M_K$, where `tags' for the disjoint union is chosen in such a way that $\mathcal{H}^\omega(J \ddagger K) \trianglelefteqslant A \multimap C$ holds; 

\item $\lambda_{J \ddagger K} \stackrel{\mathrm{df. }}{=} [\lambda_J \downharpoonright M_{B_{[1]}}, \lambda^{+\mu}_J \upharpoonright M_{B_{[1]}}, \lambda^{+\mu}_K \upharpoonright M_{B_{[2]}}, \lambda_K \downharpoonright M_{B_{[2]}}]$, where $B_{[1]}$ and $B_{[2]}$ are the copies of $B$ that belong to $J$ and $K$, respectively, $\lambda_G^{+ \mu} \stackrel{\mathrm{df. }}{=}  \langle \lambda_G^{\mathsf{OP}}, \lambda_G^{\mathsf{QA}}, n \mapsto \lambda_G^{\mathbb{N}} (n) + \mu \rangle$ ($G$ is $J$ or $K$), and $\mu \stackrel{\mathrm{df. }}{=} \max(\mu(J), \mu(K)) + 1$; 

\item $\star \vdash_{J \ddagger K} m \stackrel{\mathrm{df. }}{\Leftrightarrow} \star \vdash_K m$;

\item $m \vdash_{J \ddagger K} n \ (m \neq \star) \stackrel{\mathrm{df. }}{\Leftrightarrow} m \vdash_J n \vee m \vdash_K n \vee (\star \vdash_{B_{[2]}} m \wedge \star \vdash_{B_{[1]}} n)$;

\item $P_{J \ddagger K} \stackrel{\mathrm{df. }}{=} \{ \boldsymbol{s} \in \mathscr{J}_{J \ddagger K} \mid \boldsymbol{s} \upharpoonright J \in P_J, \boldsymbol{s} \upharpoonright K \in P_K, \boldsymbol{s}  \upharpoonright B_{[1]}, B_{[2]} \in \mathit{pr}_B \ \! \}$;

\item $\boldsymbol{s} \simeq_{J \ddagger K} \boldsymbol{t} \stackrel{\mathrm{df. }}{\Leftrightarrow} (\forall i \in \mathbb{N} . \ \! \boldsymbol{s}(i) \in M_J \Leftrightarrow \boldsymbol{t}(i) \in M_J) \wedge \boldsymbol{s} \upharpoonright J \simeq_J \boldsymbol{t} \upharpoonright J \wedge  \boldsymbol{s} \upharpoonright K \simeq_K \boldsymbol{t} \upharpoonright K$

\end{itemize}
where $\mathit{pr}_B \stackrel{\mathrm{df. }}{=} \{ \bm{s} \in P_{B_{[1]} \multimap B_{[2]}} \mid \forall \bm{t} \preceq{\bm{s}}. \ \mathsf{Even}(\bm{t}) \Rightarrow \bm{t} \upharpoonright B_{[1]} = \bm{t} \upharpoonright B_{[2]} \ \! \}$.
Moreover, the \emph{\bfseries composition} $J ; K$ (or $K \circ J$) of $J$ and $K$ is given by:
\begin{equation*}
J ; K \stackrel{\mathrm{df. }}{=} \mathcal{H}^\omega(J \ddagger K).
\end{equation*}
\end{definition}

\begin{example}
A typical maximal position of the concatenation $(\boldsymbol{2} \multimap \boldsymbol{2}) \ddagger (\boldsymbol{2} \multimap \boldsymbol{2})$ is:
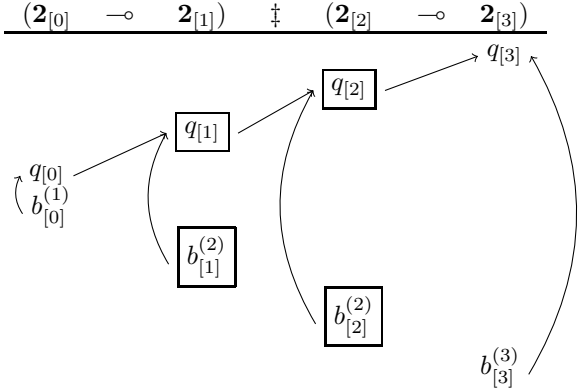
\begin{center}
\begin{tabular}{ccccccc}
$(\boldsymbol{2}_{[0]}$ & $\multimap$ & $\boldsymbol{2}_{[1]})$ & $\ddagger$ & $(\boldsymbol{2}_{[2]}$ & $\multimap$ & $\boldsymbol{2}_{[3]})$ \\ \cline{1-7} 
&&&&&&\tikzmark{cCon1} $q_{[3]}$ \tikzmark{cCon7} \\
&&&&\tikzmark{cCon2} \fbox{$q_{[2]}$} \tikzmark{dCon1}&& \\
&&\tikzmark{cCon3} \fbox{$q_{[1]}$} \tikzmark{dCon2} &&&& \\
\tikzmark{cCon4} $q_{[0]}$ \tikzmark{dCon3} &&&&&& \\
\tikzmark{dCon4} $b^{(1)}_{[0]}$ &&&&&& \\
&& \tikzmark{dCon5} \fbox{$b^{(2)}_{[1]}$} &&&& \\
&&&& \tikzmark{dCon6} \fbox{$b^{(2)}_{[2]}$} && \\
&&&&&& $b^{(3)}_{[3]}$ \tikzmark{dCon7}
\end{tabular}
\begin{tikzpicture}[overlay, remember picture, yshift=.25\baselineskip]
\draw [->] ({pic cs:dCon1}) to ({pic cs:cCon1});
\draw [->] ({pic cs:dCon2}) to ({pic cs:cCon2});
\draw [->] ({pic cs:dCon3}) to ({pic cs:cCon3});
\draw [->] ({pic cs:dCon4}) [bend left] to ({pic cs:cCon4});
\draw [->] ({pic cs:dCon5}) [bend left] to ({pic cs:cCon3});
\draw [->] ({pic cs:dCon6}) [bend left] to ({pic cs:cCon2});
\draw [->] ({pic cs:dCon7}) [bend right] to ({pic cs:cCon7});
\end{tikzpicture}
\end{center}
where $b^{(1)}, b^{(2)}, b^{(3)} \in \mathbb{B}$.
We have marked internal moves by a square box just for clarity.
\end{example}

We shall see that the `non-hiding composition' or \emph{concatenation} $\iota \ddagger \kappa$ of strategies $\iota : J$ and $\kappa : K$ forms a strategy on the concatenation $J \ddagger K$.
It generalizes the particular case, where $J = A \multimap B$ and $K = B \multimap C$, so that $\iota ; \kappa = \mathcal{H}^\omega(\iota \ddagger \kappa) : \mathcal{H}^\omega(J \ddagger K) = J ; K = A \multimap C$ (as we shall establish shortly), which reformulates conventional composition of static strategies as \emph{concatenation plus hiding} of dynamic strategies. 

\begin{theorem}[Well-defined concatenation and composition of games]
\label{ThmWellDefinedConcatenationOnDynamicGames}
(Resp. well-founded) games are closed under concatenation and composition.
\end{theorem}
\begin{proof}
By Theorem~\ref{ThmExternalClosureOfDynamicGames}, it suffices to focus on concatenation, where well-foundedness is clearly preserved under concatenation. 
We first show that the arena $J \ddagger K$ is well-defined.
The set $M_{J \ddagger K}$ and the function $\lambda_{J \ddagger K}$ are clearly well-defined, where the \emph{finite} upper bounds $\mu(J)$ and $\mu(K)$ are crucial.
For the relation $\vdash_{J \ddagger K}$, the axioms E1 and E3 clearly hold. 
For the axiom E2, if $m \vdash_{J \ddagger K} n$ and $\lambda_{J \ddagger K}^{\mathsf{QA}}(n) = \mathsf{A}$, then $m, n \in M_K \setminus M_{B_{[2]}}$, $m, n \in M_{B_{[2]}}$, $m, n \in M_{B_{[1]}}$ or $m, n \in M_J \setminus M_{B_{[1]}}$. 
In either case, $\lambda_{J \ddagger K}^{\mathsf{QA}}(m) = \mathsf{Q}$ and $\lambda_{J \ddagger K}^{\mathbb{N}}(m) = \lambda_{J \ddagger K}^{\mathbb{N}}(n)$.

For the axiom E4, let $m \vdash_{J \ddagger K} n$, $m \neq \star$ and $\lambda_{J \ddagger K}^{\mathbb{N}}(m) \neq \lambda_{J \ddagger K}^{\mathbb{N}}(n)$. We proceed by a case analysis. 
If $(m \vdash_K n) \wedge (m, n \in M_K \setminus M_{B_{[2]}} \vee m, n \in M_{B_{[2]}})$, then we may just apply E4 on $K$. 
It is similar if $(m \vdash_J n) \wedge (m, n \in M_J \setminus M_{B_{[1]}} \vee m, n \in M_{B_{[1]}})$.
Note that the case $\star \vdash_{B_{[2]}} m \wedge \star \vdash_{B_{[1]}} n$ cannot happen.
Now, consider the case $m \vdash_K n \wedge m \in M_K \setminus M_{B_{[2]}} \wedge n \in M_{B_{[2]}}$. 
If $m$ is external, then $m \in M_C$, and so E4 on $J \ddagger K$ is satisfied by the definition of $B \multimap C$; if $m$ is internal, then we may apply E4 on $K$. 
The case $m \vdash_K n \wedge n \in M_K \setminus M_{B_{[2]}} \wedge m \in M_{B_{[2]}}$ is simpler as $m$ must be internal.
The remaining cases $m \vdash_J n \wedge m \in M_J \setminus M_{B_{[1]}} \wedge n \in M_{B_{[1]}}$ and $m \vdash_J n \wedge n \in M_J \setminus M_{B_{[1]}} \wedge m \in M_{B_{[1]}}$ are analogous. 
Hence, we have shown that the arena $J \ddagger K$ is well-defined.

Next, we show that $P_{J \ddagger K} \subseteq \mathscr{L}_{J \ddagger K}$.
For justification, let $\boldsymbol{s}m \in P_{J \ddagger K}$ with $m$ non-initial. The non-trivial case is when $m$ is initial in $J$. But in this case, $m$ is initial in $B_{[1]}$, and so it has a justifier in $B_{[2]}$. 
For alternation and IE-switch, similarly to Table~\ref{TableDoubleParityDiagram} for tensor $\otimes$, we have Table~\ref{TableConcatenationDoubleParityDiagram} for $J \ddagger K$, in which the first (resp. the second) component of each state is about the OP- and IE-parities of the next move of $J$ (resp. $K$). 
For readability, some states are written twice, and the dotted arrow indicates two necessarily consecutive moves of $B$. 
Then, alternation and IE-switch on $J \ddagger K$ immediately follows from this diagram and the corresponding axioms on $J$ and $K$.
\begin{table}
\begin{diagram}
& & (\mathsf{O}^\mathsf{E}, \mathsf{O}^\mathsf{E}) & \rTo^C & (\mathsf{O}^\mathsf{E}, \mathsf{P}^\mathsf{I}) \\
& & \dTo^C \uTo_C & & \dTo^{K} \uTo_{K} \\
(\mathsf{P}^\mathsf{I}, \mathsf{O}^\mathsf{E}) & \lDotsto^{B_{[1]} B_{[2]}} & (\mathsf{O}^\mathsf{E}, \mathsf{P}^\mathsf{E}) & \lTo^{K} & (\mathsf{O}^\mathsf{E}, \mathsf{O}^\mathsf{I}) \\
\dTo^{J} \uTo_{J} & & \dDotsto^{\begin{matrix} B_{[2]} \\ B_{[1]} \end{matrix}} \uDotsto_{\begin{matrix} B_{[2]} \\ B_{[1]} \end{matrix}} & & \dTo^{K} \uTo_{K}  \\
(\mathsf{O}^\mathsf{I}, \mathsf{O}^\mathsf{E}) & \rTo^{J} & (\mathsf{P}^\mathsf{E}, \mathsf{O}^\mathsf{E}) & \rDotsto^{B_{[1]} B_{[2]}} & (\mathsf{O}^\mathsf{E}, \mathsf{P}^\mathsf{I}) \\
\dTo^{J} \uTo_{J} & & \dTo^A \uTo_A & & \\
(\mathsf{P}^\mathsf{I}, \mathsf{O}^\mathsf{E}) & \lTo^A & (\mathsf{O}^\mathsf{E}, \mathsf{O}^\mathsf{E}) & &
\end{diagram}
\caption{The double parity diagram for the concatenation $J \ddagger K$.}
\label{TableConcatenationDoubleParityDiagram}
\end{table}

For generalized visibility, let $\boldsymbol{s} m \in P_{J \ddagger K}$ with $m$ non-initial and $d \in \mathbb{N} \cup \{ \omega \}$ such that $\boldsymbol{s} m$ is $d$-complete. 
Without loss of generality, we may assume $d \in \mathbb{N}$ as $\boldsymbol{s}$ is finite.
It is not hard to see that $\mathcal{H}^d_{J \ddagger K}(\boldsymbol{s}m) \in P_{\mathcal{H}^d(J) \ddagger \mathcal{H}^d(K)}$ if $\mathcal{H}^d(J \ddagger K)$ is not normalized; thus, this case is reduced to the (usual) visibility on $\mathcal{H}^d(J) \ddagger \mathcal{H}^d(K)$. Otherwise, it is no harm to select the \emph{least} $d \in \mathbb{N}^+$ such that $\mathcal{H}^d(J \ddagger K)$ is normalized; then $\mathcal{H}_{J \ddagger K}^{d-1}(\boldsymbol{s}m) \in P_{(A \multimap B_{[1]}) \ddagger (B_{[2]} \multimap C)}$, and thus the visibility of $\mathcal{H}_{J \ddagger K}^d(\boldsymbol{s}m) = \mathcal{H}_{\mathcal{H}^{d-1}(J \ddagger K)}(\mathcal{H}_{J \ddagger K}^{d-1}(\boldsymbol{s}m))$ can be shown completely in the same way as the proof that shows the composition of strategies is well-defined (in particular it satisfies visibility) \cite{mccusker1998games,harmer2004innocent}.
Consequently, it suffices to consider the case $d = 0$, i.e., to show the (usual) visibility. 

For this, we need the following:
\begin{lemma}[Visibility lemma]
\label{LemViewLemma}
Assume that $\boldsymbol{t} \in P_{J \ddagger K}$ and $\boldsymbol{t} \neq \boldsymbol{\epsilon}$. 
\begin{enumerate}
\item If the last move of $\boldsymbol{t}$ is of $M_J \setminus M_{B_{[1]}}$, then $\lceil \boldsymbol{t} \upharpoonright J \rceil_J \preceq \lceil \boldsymbol{t} \rceil_{J \ddagger K} \upharpoonright J$ and $\lfloor \boldsymbol{t} \upharpoonright J \rfloor_J \preceq \lfloor \boldsymbol{t} \rfloor_{J \ddagger K} \upharpoonright J$;

\item If the last move of $\boldsymbol{t}$ is of $M_K \setminus M_{B_{[2]}}$, then $\lceil \boldsymbol{t} \upharpoonright K \rceil_K \preceq \lceil \boldsymbol{t} \rceil_{J \ddagger K} \upharpoonright K$ and $\lfloor \boldsymbol{t} \upharpoonright K \rfloor_K \preceq \lfloor \boldsymbol{t} \rfloor_{J \ddagger K} \upharpoonright K$;

\item If the last move of $\boldsymbol{t}$ is an O-move of $M_{B_{[1]}} \cup \ \! M_{B_{[2]}}$, then $\lceil \boldsymbol{t} \upharpoonright B_{[1]}, B_{[2]} \rceil_{B_{[1]} \multimap B_{[2]}} \preceq \lfloor \boldsymbol{t} \rfloor_{J \ddagger K} \upharpoonright B_{[1]}, B_{[2]}$ and 
$\lfloor \boldsymbol{t} \upharpoonright B_{[1]}, B_{[2]} \rfloor_{B_{[1]} \multimap B_{[2]}} \preceq \lceil \boldsymbol{t} \rceil_{J \ddagger K} \upharpoonright B_{[1]}, B_{[2]}$.
\end{enumerate}
\end{lemma}
\begin{proof}[Proof of the lemma]
By induction on $|\boldsymbol{t}|$ with case analysis on the last move of $\boldsymbol{t}$. 
\end{proof}

Note that we may write $\boldsymbol{s} m = \boldsymbol{s_1} n \boldsymbol{s_2} m$, where $n$ justifies $m$.
If $\boldsymbol{s_2} = \boldsymbol{\epsilon}$, then it is trivial; so assume $\boldsymbol{s_2} = \boldsymbol{s'_2} r$. 
We then proceed by a case analysis on $m$:
\begin{itemize}

\item Assume $m \in M_J \setminus M_{B_{[1]}}$. Then, $n \in M_J$ and $r \in M_J$ by Table \ref{TableConcatenationDoubleParityDiagram}.
By Lemma \ref{LemViewLemma}, $\lceil \boldsymbol{s} \upharpoonright J \rceil \preceq \lceil \boldsymbol{s} \rceil \upharpoonright J$ and $\lfloor \boldsymbol{s} \upharpoonright J \rfloor \preceq \lfloor \boldsymbol{s} \rfloor \upharpoonright J$. 
Also, for $(\boldsymbol{s} \upharpoonright J) . m \in P_J$ and visibility on $J$,
\begin{align*}
&\text{$n$ occurs in $\lceil \boldsymbol{s} \upharpoonright J \rceil$ if $m$ is a P-move}; \\
&\text{$n$ occurs in $\lfloor \boldsymbol{s} \upharpoonright J \rfloor$ if $m$ is an O-move}.
\end{align*}
Hence we may conclude that $n$ occurs in $\lceil \boldsymbol{s} \rceil$ (resp. $\lfloor \boldsymbol{s} \rfloor$) if $m$ is a P- (resp. O-) move. 

\item Assume $m \in M_K \setminus M_{B_{[2]}}$. This case can be handled in a completely analogous way to the above case. 

\item Assume $m \in M_{B_{[1]}}$. If $m$ is a P-move, then $n, r \in M_J$ and so it can be handled in the same way as the case $m \in M_J \setminus M_{B_{[1]}}$; thus, assume that $m$ is an O-move. Then, $r \in M_{B_{[2]}}$ and it is a `copy' of $m$. 
Since $r$ is an O-move of $B_{[1]} \multimap B_{[2]}$, by Lemma \ref{LemViewLemma}, $\lceil \boldsymbol{s} \upharpoonright B_{[1]}, B_{[2]} \rceil \preceq \lfloor \boldsymbol{s} \rfloor \upharpoonright B_{[1]}, B_{[2]}$.
Note that $n$ is a move of $B_{[1]}$ or an initial move of $B_{[2]}$. In either case, we have $(\boldsymbol{s} \upharpoonright B_{[1]}, B_{[2]}) . m \in P_{B_{[1]} \multimap B_{[2]}}$; thus, $n$ occurs in $\lceil \boldsymbol{s} \upharpoonright B_{[1]}, B_{[2]} \rceil$. Hence we may conclude that $n$ occurs in $\lfloor \boldsymbol{s} \rfloor$.

\item Assume $m \in M_{B_{[2]}}$. 
If $m$ is a P-move, then $n, r \in M_K$; so it can be dealt with in the same way as the case $m \in M_K \setminus M_{B_{[2]}}$. 
Thus, assume $m$ is an O-move. 
By Table \ref{TableConcatenationDoubleParityDiagram}, we have $r \in M_{B_{[1]}}$, and it is an O-move of $B_{[1]} \multimap B_{[2]}$. 
Thus by Lemma \ref{LemViewLemma}, $\lceil \boldsymbol{s} \upharpoonright B_{[1]}, B_{[2]} \rceil \preceq \lfloor \boldsymbol{s} \rfloor \upharpoonright B_{[1]}, B_{[2]}$.
Then again, $(\boldsymbol{s} \upharpoonright B_{[1]}, B_{[2]}) . m \in P_{B_{[1]} \multimap B_{[2]}}$; thus, $n$ occurs in $\lceil \boldsymbol{s} \upharpoonright B_{[1]}, B_{[2]} \rceil$, and so $n$ occurs in $\lfloor \boldsymbol{s} \rfloor$.
\end{itemize}

Next, we verify the axioms P1, DP2 and DP3.
For P1, $\boldsymbol{\epsilon} \in P_{J \ddagger K}$ is clear; for prefix-closure, let $\boldsymbol{s} m \in P_{J \ddagger K}$. If $m \in M_J \setminus M_{B_{[1]}}$, then $(\boldsymbol{s} \upharpoonright J) . m = \boldsymbol{s}m \upharpoonright J \in P_J$; thus, $\boldsymbol{s} \upharpoonright J \in P_J$, $\boldsymbol{s} \upharpoonright K = \boldsymbol{s}m \upharpoonright K \in P_K$ and $\boldsymbol{s} \upharpoonright B_{[1]}, B_{[2]} = \boldsymbol{s}m \upharpoonright B_{[1]}, B_{[2]} \in \mathit{pr}_B$, whence $\boldsymbol{s} \in P_{J \ddagger K}$.
The other cases may be handled similarly.
For DP2, assume $\boldsymbol{s} m n \in P_{J \ddagger K}^{\mathsf{Even}}$ and $\lambda_{J \ddagger K}^{\mathbb{N}}(n) > 0$.
If $n \not \in M_{B_{[1]}} \cup M_{B_{[2]}}$, then we may just apply DP2 on $J$ or $K$; and the remaining case is trivial by the definition of $J \ddagger K$.

For DP3, let $i \in \mathbb{N}$ and $\boldsymbol{s}m, \boldsymbol{s'}m' \in P_{J \ddagger K}^{\mathsf{Odd}}$ such that $i < \lambda_{J \ddagger K}^{\mathbb{N}}(m) = \lambda_{J \ddagger K}^{\mathbb{N}}(m')$ and $\mathcal{H}_{J \ddagger K}^i(\boldsymbol{s}) = \mathcal{H}_{J \ddagger K}^i(\boldsymbol{s'})$.
Without loss of generality, we may assume $i = 0$ and $\lambda_{J \ddagger K}^{\mathbb{N}}(m) = 1 = \lambda_{J \ddagger K}^{\mathbb{N}}(m')$ because if $\lambda_{J \ddagger K}^{\mathbb{N}}(m) = \lambda_{J \ddagger K}^{\mathbb{N}}(m') = j > 1$, then we may consider $\mathcal{H}_{J \ddagger K}^{j-1}(\boldsymbol{s}m), \mathcal{H}_{J \ddagger K}^{j-1}(\boldsymbol{s'}m') \in P_{\mathcal{H}^{j-1}(J) \ddagger \mathcal{H}^{j-1}(K)}$ (n.b., the justifiers of $m$ and $m'$ have the same priority order).
Thus, $\boldsymbol{s} = \boldsymbol{s'}$ and $m, m' \in M_J \vee m, m' \in M_K$. If $m, m' \in M_J$ (resp. $m, m' \in M_K$), then $(\boldsymbol{s} \upharpoonright J) . m, (\boldsymbol{s'} \upharpoonright J) . m' \in P_J^{\mathsf{Odd}}$ (resp. $(\boldsymbol{s} \upharpoonright K) . m, (\boldsymbol{s'} \upharpoonright K) . m' \in P_K^{\mathsf{Odd}}$), and so we may just apply DP3 on $J$ (resp. $K$).

Finally, the axioms I1, I2 and DI3 on $\simeq_{J \ddagger K}$ can be verified similarly to the case of tensor, completing the proof.
\end{proof}

For completeness, let us explicitly define the rather trivial \emph{currying} of games:
\begin{definition}[Currying of games]
\label{DefCurryingOfDynamicGames}
Given a game $G$ such that $\mathcal{H}^\omega(G) \trianglelefteqslant A \otimes B \multimap C$ for some normalized games $A$, $B$ and $C$, the \emph{\bfseries currying} $\Lambda(G)$ of $G$ is $G$ up to `tags' that satisfies $\mathcal{H}^\omega(\Lambda(G)) \trianglelefteqslant A \multimap (B \multimap C)$. 
\end{definition}

Trivially, (resp. well-founded) games are closed under currying.

Next, we show that these constructions as well as the hiding operation preserve the subgame relation $\trianglelefteqslant$ (Definition~\ref{DefDynamicSubgames}):

\begin{notation*}
We write $\clubsuit_{i \in I}$, where $I$ is $\{ 1 \}$ or $\{ 1, 2 \}$, for any of the constructions on games introduced so far, i.e., $\clubsuit_{i \in I}$ is either $\otimes$, $\multimap$, $\langle \_, \_ \rangle$, $(\_)^\dagger$, $\ddagger$ or $\Lambda$. 
\end{notation*}

\begin{lemma}[Preservation of subgames]
\label{LemPreservationOfDynamicSubgames}
Let $\clubsuit_{i \in I}$ be a construction on games, and assume $H_i \trianglelefteqslant G_i$ for all $i \in I$. Then, $\clubsuit_{i \in I} H_i \trianglelefteqslant \clubsuit_{i \in I} G_i$.
\end{lemma}
\begin{proof} 
Let us first consider tensor. 
It is trivial to check the conditions on the sets of moves and the labeling functions, and so we omit them.
For the enabling relations:
\begin{align*}
\vdash_{H_1 \otimes H_2} &= \ \vdash_{H_1} + \ \vdash_{H_2} \\
&\subseteq (\vdash_{G_1} \! \cap \ ((\{ \star \} \cup M_{H_1}) \times M_{H_1})) + (\vdash_{G_2} \! \cap \ ((\{ \star \} \cup M_{H_2}) \times M_{H_2})) \\
&= (\vdash_{G_1} \! \cap \ ((\{ \star \} \cup M_{H_1 \otimes H_2}) \times M_{H_1 \otimes H_2})) + (\vdash_{G_2} \! \cap \ ((\{ \star \} \cup M_{H_1 \otimes H_2}) \times M_{H_1 \otimes H_2})) \\
&= (\vdash_{G_1} \! + \vdash_{G_2}) \cap ((\{ \star \} \cup M_{H_1 \otimes H_2}) \times M_{H_1 \otimes H_2}) \\
&= \ \vdash_{G_1 \otimes G_2} \! \cap \ ((\{ \star \} \cup M_{H_1 \otimes H_2}) \times M_{H_1 \otimes H_2}).
\end{align*}

For the positions, we have:
\begin{align*}
P_{H_1 \otimes H_2} &= \{ \boldsymbol{s} \in \mathscr{L}_{H_1 \otimes H_2} \mid \forall i \in \{ 1, 2 \} . \ \! \boldsymbol{s} \upharpoonright H_i \in P_{H_i} \} \\
&\subseteq \{ \boldsymbol{s} \in \mathscr{L}_{G_1 \otimes G_2} \mid \forall i \in \{ 1, 2 \} . \ \! \boldsymbol{s} \upharpoonright G_i \in P_{G_i} \} \\
&= P_{G_1 \otimes G_2}.
\end{align*}

For the identifications of positions, given $d \in \mathbb{N} \cup \{ \omega \}$, we have:
\begin{align*}
\boldsymbol{s} \simeq_{H_1 \otimes H_2}^d \boldsymbol{t} &\Leftrightarrow \exists \boldsymbol{s'}, \boldsymbol{t'} \in P_{H_1 \otimes H_2} . \ \! \boldsymbol{s'} \simeq_{H_1 \otimes H_2} \boldsymbol{t'} \wedge \mathcal{H}_{H_1 \otimes H_2}^d(\boldsymbol{s'}) =  \mathcal{H}_{H_1 \otimes H_2}^d(\boldsymbol{s}) \\ & \ \ \ \ \wedge \mathcal{H}_{H_1 \otimes H_2}^d(\boldsymbol{t'}) =  \mathcal{H}_{H_1 \otimes H_2}^d(\boldsymbol{t}) \\
&\Leftrightarrow \forall j \in \{ 1, 2 \} . \ \!  \exists \boldsymbol{s'_j}, \boldsymbol{t'_j} \in P_{H_j} . \ \! \boldsymbol{s'_j} \simeq_{H_j} \boldsymbol{t'_j} \wedge \mathcal{H}_{H_j}^d(\boldsymbol{s'_j}) =  \mathcal{H}_{H_j}^d(\boldsymbol{s} \upharpoonright H_j) \\ & \ \ \ \ \wedge \mathcal{H}_{H_j}^d(\boldsymbol{t'_j}) =  \mathcal{H}_{H_j}^d(\boldsymbol{t} \upharpoonright H_j) \wedge \forall k \in \mathbb{N} . \ \! s_k \in M_{H_1} \Leftrightarrow t_k \in M_{H_1} \\
&\Leftrightarrow \forall j \in \{ 1, 2 \} . \ \! \boldsymbol{s} \upharpoonright H_j \simeq_{H_j}^d \boldsymbol{t} \upharpoonright H_j \wedge \forall k \in \mathbb{N} . \ \! s_k \in M_{H_1} \Leftrightarrow t_k \in M_{H_1} \\
&\Leftrightarrow \forall j \in \{ 1, 2 \} . \ \! \boldsymbol{s} \upharpoonright G_j, \boldsymbol{t} \upharpoonright G_j \in P_{H_j} \wedge \boldsymbol{s} \upharpoonright G_j \simeq_{G_j}^d \boldsymbol{t} \upharpoonright G_j \\ & \ \ \ \ \wedge \forall k \in \mathbb{N} . \ \! s_k \in M_{G_1} \Leftrightarrow t_k \in M_{G_1} \\
&\Leftrightarrow \boldsymbol{s}, \boldsymbol{t} \in P_{H_1 \otimes H_2} \wedge \boldsymbol{s} \simeq_{G_1 \otimes G_2}^d \boldsymbol{t}.
\end{align*}

Finally, we have $\mu(H_1 \otimes H_2) = \max(\mu(H_1), \mu(H_2)) = \max(\mu(G_1), \mu(G_2)) = \mu(G_1 \otimes G_2)$, showing that $H_1 \otimes H_2 \trianglelefteqslant G_1 \otimes G_2$.

Linear implication and promotion are similar, and pairing and currying are even simpler; thus, we omit them.
Next, let us consider concatenation.
Assume that $\mathcal{H}^\omega(H_1) \trianglelefteqslant A \multimap B$, $\mathcal{H}^\omega(H_2) \trianglelefteqslant B \multimap C$, $\mathcal{H}^\omega(G_1) \trianglelefteqslant D \multimap E$, $\mathcal{H}^\omega(G_2) \trianglelefteqslant E \multimap F$ for some normalized games $A$, $B$, $C$, $D$, $E$ and $F$; without loss of generality, we assume that these normalized games are the least ones with respect to $\trianglelefteqslant$.
By Theorem~\ref{ThmExternalClosureOfDynamicGames}, $\mathcal{H}^\omega(H_1) \trianglelefteqslant \mathcal{H}^\omega(G_1) \trianglelefteqslant D \multimap E$ and $\mathcal{H}^\omega(H_2) \trianglelefteqslant \mathcal{H}^\omega(G_2) \trianglelefteqslant E \multimap F$, which in turn implies $A \trianglelefteqslant D$, $B \trianglelefteqslant E$ and $C \trianglelefteqslant F$.
First, we clearly have $M_{H_1 \ddagger H_2} \subseteq M_{G_1 \ddagger G_2}$ and $\lambda_{G_1 \ddagger G_2} \upharpoonright M_{H_1 \ddagger H_2} = \lambda_{H_1 \ddagger H_2}$, where $\mu(H_i) = \mu(G_i)$ for $i = 1, 2$ ensures that the priority orders of moves of $B$ coincide. 

Next, for the enabling relations, we have:
\begin{align*}
\star \vdash_{H_1 \ddagger H_2} m \Leftrightarrow \star \vdash_{H_2} m \Leftrightarrow \star \vdash_{C} m \Rightarrow \star \vdash_{F} m \Leftrightarrow \star \vdash_{G_1 \ddagger G_2} m 
\end{align*}
as well as:
\begin{align*}
m \vdash_{H_1 \ddagger H_2} n & \Leftrightarrow m \vdash_{H_1} n \vee m \vdash_{H_2} n \vee (\star \vdash_{B_{[2]}} \! m \wedge \star \vdash_{B_{[1]}} n) \\
& \Rightarrow m \vdash_{G_1} n \vee m \vdash_{G_2} n \vee (\star \vdash_{E_{[2]}} m \wedge \star \vdash_{E_{[1]}} n) \\
& \Leftrightarrow m \vdash_{G_1 \ddagger G_2} n
\end{align*}
for any $m, n \in M_{H_1 \ddagger H_2}$.
For the positions, we have:
\begin{align*}
P_{H_1 \ddagger H_2} &= \{ \boldsymbol{s} \in \mathscr{J}_{H_1 \ddagger H_2} \mid \boldsymbol{s} \upharpoonright H_1 \in P_{H_1}, \boldsymbol{s} \upharpoonright H_2 \in P_{H_2}, \boldsymbol{s} \upharpoonright B_{[1]}, B_{[2]} \in \mathit{pr}_B \ \! \} \\
&\subseteq \{ \boldsymbol{s} \in \mathscr{J}_{G_1 \ddagger G_2} \mid \boldsymbol{s} \upharpoonright G_1 \in P_{G_1}, \boldsymbol{s} \upharpoonright G_2 \in P_{G_2}, \boldsymbol{s} \upharpoonright E_{[1]}, E_{[2]} \in \mathit{pr}_{E} \ \! \} \\
&= P_{G_1 \ddagger G_2}.
\end{align*}

Finally, we may show, in the same manner as in the case of tensor, the required condition on the identifications of positions, completing the proof.
\end{proof}

At the end of the present section, we establish the following useful lemma:
\begin{lemma}[Hiding lemma on games]
\label{LemHidingLemmaOnDynamicGames}
Let $\clubsuit_{i \in I}$ be a construction on games and $G_i$ a game for all $i \in I$. 
For each $d \in \mathbb{N} \cup \{ \omega \}$, we have:
\begin{enumerate}

\item $\mathcal{H}^d(\clubsuit_{i \in I} G_i) = \clubsuit_{i \in I} \mathcal{H}^d(G_i)$ if $\clubsuit_{i \in I} \neq \ddagger$;

\item $\mathcal{H}^d((G_1) \ddagger (G_2)) \trianglelefteqslant A \multimap C$ if $\mathcal{H}^d(G_1 \ddagger G_2)$ is normalized, where $A$, $B$ and $C$ are normalized games such that $\mathcal{H}^\omega(G_1) \trianglelefteqslant A \multimap B$ and $\mathcal{H}^\omega(G_2) \trianglelefteqslant B \multimap C$, and in particular $(A \multimap B) ; (B \multimap C) = A \multimap C$;

\item $\mathcal{H}^d(G_1 \ddagger G_2) = \mathcal{H}^d(G_1) \ddagger \mathcal{H}^d(G_2)$ otherwise.

\end{enumerate}
\end{lemma}
\begin{proof}
Since there is an upper bound of the priority orders of each game, it suffices to consider the case $d \in \mathbb{N}$.
But then, as $\mathcal{H}^{i+1} = \mathcal{H} \circ \mathcal{H}^i$ for all $i \in \mathbb{N}$, we may focus on $d = 1$.
We focus on tensor as the other constructions may be handled similarly. 

We have to show $\mathcal{H}(G_1 \otimes G_2) \trianglelefteqslant \mathcal{H}(G_1) \otimes \mathcal{H}(G_2)$. Their sets of moves and labeling functions clearly coincide. 
For the enabling relations, we have:
\begin{align*}
\star \vdash_{\mathcal{H}(G_1 \otimes G_2)} m &\Leftrightarrow \star \vdash_{G_1 \otimes G_2} m \Leftrightarrow \star \vdash_{G_1} m \vee \star \vdash_{G_2} m \\
&\Leftrightarrow \star \vdash_{\mathcal{H}(G_1)} m \vee \star \vdash_{\mathcal{H}(G_2)} m \\
&\Leftrightarrow \star \vdash_{\mathcal{H}(G_1) \otimes \mathcal{H}(G_2)} m
\end{align*}
as well as:
\begin{align*}
& m \vdash_{\mathcal{H}(G_1 \otimes G_2)} n \ (m \neq \star) \\
 \Leftrightarrow \ & (m \vdash_{G_1 \otimes G_2} \! n) \vee \exists k \in \mathbb{N}^+, m_1, m_2, \dots, m_{2k} \in M_{G_1 \otimes G_2} \setminus M_{\mathcal{H}(G_1 \otimes G_2)} .  \\ & \ m \vdash_{G_1 \otimes G_2} m_1 \wedge \forall i \in \overline{k} . \ \! m_{2i-1} \vdash_{G_1 \otimes G_2} m_{2i} \wedge \wedge m_{2k} \vdash_{G_1 \otimes G_2} n \\
\Leftrightarrow \ & (m \vdash_{G_1} n \vee m \vdash_{G_2} n) \vee \exists i \in \{ 1, 2 \}, k \in \mathbb{N}^+, m_1, m_2, \dots, m_{2k} \in M_{G_i} \setminus M_{\mathcal{H}(G_i)} . \\ & \ m \vdash_{G_i} m_1 \wedge \forall i \in \overline{k} . \ \! m_{2i-1} \vdash_{G_1 \otimes G_2} m_{2i} \wedge m_{2k} \vdash_{G_i} n \\
\Leftrightarrow \ & \exists i \in \{ 1, 2 \} . \ \! m \vdash_{G_i} n \vee \exists k \in \mathbb{N}^+, m_1, m_2, \dots, m_{2k} \in M_{G_i} \setminus M_{\mathcal{H}(G_i)} . \\ & \ m \vdash_{G_i} m_1 \wedge \forall i \in \overline{k} . \ \! m_{2i-1} \vdash_{G_1 \otimes G_2} m_{2i} \wedge m_{2k} \vdash_{G_i} n \\
\Leftrightarrow \ & m \vdash_{\mathcal{H}(G_1) \otimes \mathcal{H}(G_2)} n.
\end{align*} 
Thus, the arenas $\mathcal{H}(G_1 \otimes G_2)$ and $\mathcal{H}(G_1) \otimes \mathcal{H}(G_2)$ coincide.

For the positions, we have: 
\begin{align*}
& \boldsymbol{s} \in P_{\mathcal{H}(G_1 \otimes G_2)} \\
\Leftrightarrow \ & \exists \boldsymbol{t} \in \mathscr{L}_{G_1 \otimes G_2}. \ \! \mathcal{H}_{G_1 \otimes G_2}(\boldsymbol{t}) = \boldsymbol{s} \wedge \forall i \in \{ 1, 2 \} . \ \! \boldsymbol{t} \upharpoonright G_i \in P_{G_i} \\
\Leftrightarrow \ & \exists \boldsymbol{t} \in \mathscr{L}_{G_1 \otimes G_2}. \ \! \mathcal{H}_{G_1 \otimes G_2}(\boldsymbol{t}) = \boldsymbol{s} \wedge \forall i \in \{ 1, 2 \} . \ \! \mathcal{H}_{G_i}(\boldsymbol{t} \upharpoonright G_i) \in P_{\mathcal{H}(G_i)} \\
\Leftrightarrow \ & \exists \boldsymbol{t} \in \mathscr{L}_{G_1 \otimes G_2}. \ \! \mathcal{H}_{G_1 \otimes G_2}(\boldsymbol{t}) = \boldsymbol{s} \wedge \forall i \in \{ 1, 2 \} . \ \! \mathcal{H}_{G_1 \otimes G_2}(\boldsymbol{t}) \upharpoonright \mathcal{H}(G_i) \in P_{\mathcal{H}(G_i)} \\
& \text{(n.b., $\Leftarrow$ is by induction on $|\boldsymbol{t}|$)} \\
\Leftrightarrow \ & \boldsymbol{s} \in \mathscr{L}_{\mathcal{H}(G_1 \otimes G_2)} = \mathscr{L}_{\mathcal{H}(G_1) \otimes \mathcal{H}(G_2)} \wedge \forall i \in \{ 1, 2 \} . \ \! \boldsymbol{s} \upharpoonright \mathcal{H}(G_i) \in P_{\mathcal{H}(G_i)} \\
\Leftrightarrow \ & \boldsymbol{s} \in P_{\mathcal{H}(G_1) \otimes \mathcal{H}(G_2)}.
\end{align*}

Finally, for the identifications of positions, given $d \in \mathbb{N} \cup \{ \omega \}$, we have:
\begin{align*}
& \ \mathcal{H}_{G_1 \otimes G_2}(\boldsymbol{s}) \simeq_{\mathcal{H}(G_1 \otimes G_2)}^d \mathcal{H}_{G_1 \otimes G_2}(\boldsymbol{t}) \\
\Leftrightarrow & \ \exists \boldsymbol{s'}, \boldsymbol{t'} \in P_{G_1 \otimes G_2} . \ \! \mathcal{H}_{G_1 \otimes G_2}(\boldsymbol{s'}) \simeq_{\mathcal{H}(G_1 \otimes G_2)} \mathcal{H}_{G_1 \otimes G_2}(\boldsymbol{t'}) \wedge \mathcal{H}_{G_1 \otimes G_2}^{d+1}(\boldsymbol{s'}) = \mathcal{H}_{G_1 \otimes G_2}^{d+1}(\boldsymbol{s}) \\ &\wedge \mathcal{H}_{G_1 \otimes G_2}^{d+1}(\boldsymbol{t'}) = \mathcal{H}_{G_1 \otimes G_2}^{d+1}(\boldsymbol{t}) \\
\Leftrightarrow & \ \forall j \in \{ 1, 2 \} . \ \! \exists \boldsymbol{s'_j}, \boldsymbol{t'_j} \in P_{G_j} . \ \! \mathcal{H}_{G_j}(\boldsymbol{s'_j}) \simeq_{\mathcal{H}(G_j)} \mathcal{H}_{G_j}(\boldsymbol{t'_j}) \wedge \mathcal{H}_{G_j}^{d+1}(\boldsymbol{s'_j}) = \mathcal{H}_{G_j}^{d+1}(\boldsymbol{s} \upharpoonright G_j) \\ &\wedge  \mathcal{H}_{G_j}^{d+1}(\boldsymbol{t'_j}) = \mathcal{H}_{G_j}^{d+1}(\boldsymbol{t} \upharpoonright G_j) \wedge \forall k \in \mathbb{N} . \ \! \mathcal{H}_{G_1 \otimes G_2}^{d+1}(\boldsymbol{s}(k)) \in M_{\mathcal{H}^{d+1}(G_1)} \Leftrightarrow \mathcal{H}_{G_1 \otimes G_2}^{d+1}(\boldsymbol{t}(k)) \in M_{\mathcal{H}^{d+1}(G_1)} \\
\Leftrightarrow & \ \forall j \in \{ 1, 2 \} . \ \! \boldsymbol{s} \upharpoonright G_j \simeq_{G_j}^{d+1} \boldsymbol{t} \upharpoonright G_j \wedge \forall k \in \mathbb{N} . \ \! \mathcal{H}_{G_1 \otimes G_2}^{d+1}(\boldsymbol{s}(k)) \in M_{\mathcal{H}^{d+1}(G_1)} \Leftrightarrow \mathcal{H}_{G_1 \otimes G_2}^{d+1}(\boldsymbol{t}(k)) \in M_{\mathcal{H}^{d+1}(G_1)} \\
\Leftrightarrow & \ \mathcal{H}_{G_1 \otimes G_2}(\boldsymbol{s}) \simeq_{\mathcal{H}(G_1) \otimes \mathcal{H}(G_2)}^d \mathcal{H}_{G_1 \otimes G_2}(\boldsymbol{t}) \wedge \mathcal{H}_{G_1 \otimes G_2}(\boldsymbol{s}), \mathcal{H}_{G_1 \otimes G_2}(\boldsymbol{t}) \in P_{\mathcal{H}(G_1 \otimes G_2)}
\end{align*}
which completes the proof. 
\end{proof}


\subsection{Dynamic Strategies}
\emph{Dynamic strategies}, another central notion of the present work, is just static strategies \cite{abramsky1999game} \emph{on dynamic games}:
\begin{definition}[Dynamic strategies]
A \emph{\bfseries dynamic strategy} on a (dynamic) game $G$ is a subset $\sigma \subseteq P_G^{\mathsf{Even}}$, written $\sigma : G$, that satisfies:
\begin{itemize}

\item \textsc{(S1).} It is non-empty and \emph{even-prefix-closed} (i.e., $\boldsymbol{s} m n \in \sigma \Rightarrow \boldsymbol{s} \in \sigma$);

\item \textsc{(S2).} It is \emph{deterministic} on even-length positions (i.e., $\boldsymbol{s} m n, \boldsymbol{s'} m' n' \in \sigma \wedge \boldsymbol{s} m = \boldsymbol{s'} m' \Rightarrow \boldsymbol{s} m n = \boldsymbol{s'} m' n'$).

\end{itemize}
A dynamic strategy $\sigma : G$ is said to be \emph{\bfseries normalized} if $\forall \boldsymbol{s} \in \sigma, \forall i \in \overline{|\boldsymbol{s}|} . \ \! \lambda_G^{\mathbb{N}}(\boldsymbol{s}(i)) = 0$.
\end{definition}

Clearly, a normalized dynamic strategy on a normalized dynamic game is equivalent to a static strategy. 

\begin{convention*}
Henceforth, a \emph{\bfseries strategy} refers to a dynamic strategy by default. 
\end{convention*}

As positions of a game $G$ are identified up to $\simeq_G$, we must identify strategies on $G$ if they behave in the same manner up to $\simeq_G$, leading to: 
\begin{definition}[Identification of strategies \cite{abramsky2000full,mccusker1998games}]
The \emph{\bfseries identification of strategies} on a game $G$, written $\simeq_G$, is the relation between strategies $\sigma, \tau : G$ given by:
\begin{align*}
\sigma \simeq_G \tau \stackrel{\mathrm{df. }}{\Leftrightarrow} \ &\forall \boldsymbol{s} \in \sigma, \boldsymbol{t} \in \tau, \boldsymbol{s}m, \boldsymbol{t}l \in P_G . \ \! \boldsymbol{s} m \simeq_G \boldsymbol{t} l \Rightarrow \forall \boldsymbol{s} m n \in \sigma . \ \! \exists \boldsymbol{t} l r \in \tau . \ \! \boldsymbol{s} m n \simeq_G \boldsymbol{t} l r \\ &\wedge \forall \boldsymbol{t} l r \in \tau . \ \! \exists \boldsymbol{s} m n \in \sigma . \ \! \boldsymbol{t} l r \simeq_G \boldsymbol{s} m n.
\end{align*}
\end{definition}

We are particularly concerned with strategies identified with themselves:
\begin{definition}[Validity of strategies]
\label{DefValidityOfDynamicStrategies}
A strategy $\sigma : G$ is \emph{\bfseries valid} if $\sigma \simeq_G \sigma$. 
\end{definition}

Since internal moves are conceptually `invisible' to Opponent, a strategy $\sigma : G$ must be \emph{externally consistent}: If $\boldsymbol{s} m n, \boldsymbol{s}' m' n' \in \sigma$, $\lambda_G^{\mathbb{N}}(n) = \lambda_G^{\mathbb{N}}(n') = 0$ and $\mathcal{H}_G^\omega(\boldsymbol{s} m) = \mathcal{H}_G^\omega(\boldsymbol{s}' m')$, then $n = n'$ and $\mathcal{J}_{\boldsymbol{s}mn}^{\circleddash \omega}(n) = \mathcal{J}_{\boldsymbol{s}'m'n'}^{\circleddash \omega}(n')$. 
Moreover, external consistency of strategies should hold with respect to identification of positions as well.
In fact, we now proceed to establish a stronger property (Theorem~\ref{ThmExternalConsistency}).

\begin{lemma}[O-determinacy]
\label{LemODeterminacy}
Let $\sigma, \tau : G$ such that $\sigma \simeq_G \tau$, and $d \in \mathbb{N} \cup \{ \omega \}$.
\begin{enumerate}

\item If $\boldsymbol{s} m, \boldsymbol{s'} m' \in P_G$ are $d$-complete, $\boldsymbol{s}, \boldsymbol{s'} \in \sigma$, and $\mathcal{H}_G^d(\boldsymbol{s} m) = \mathcal{H}_G^d(\boldsymbol{s'} m')$, then $\boldsymbol{s} m = \boldsymbol{s'} m'$;

\item If $\boldsymbol{s}m, \boldsymbol{t}l \in P_G$ are $d$-complete, $\boldsymbol{s} \in \sigma$, $\boldsymbol{t} \in \tau$, and $\mathcal{H}_G^d(\boldsymbol{s}m) \simeq_{\mathcal{H}^d(G)} \mathcal{H}_G^d(\boldsymbol{t}l)$, then $\boldsymbol{s}m \simeq_G \boldsymbol{t}l$.

\end{enumerate}
\end{lemma}
\begin{proof}
Let us focus on the first statement for the second one can be proved similarly. 
We proceed by induction on $|\boldsymbol{s}|$. The base case $\boldsymbol{s} = \boldsymbol{\epsilon}$ is trivial: For any $d \in \mathbb{N} \cup \{ \omega \}$, if $\mathcal{H}_G^d(\boldsymbol{s} m) = \mathcal{H}_G^d(\boldsymbol{s'} m')$, then $\mathcal{H}_G^d(\boldsymbol{s'} m') = \mathcal{H}_G^d(\boldsymbol{s} m) = m$, and so $\boldsymbol{s'} m' = m = \boldsymbol{s} m$. 

For the induction step, let $d \in \mathbb{N} \cup \{ \omega \}$ be fixed, and assume $\mathcal{H}_G^d(\boldsymbol{s} m) = \mathcal{H}_G^d(\boldsymbol{s'} m')$. We may suppose that $\boldsymbol{s} m = \boldsymbol{t} l r m$, where $l$ is the rightmost O-move occurring on the left of $m$ in $\boldsymbol{s}$ such that $\lambda_G^{\mathbb{N}}(l) = 0 \vee \lambda_G^{\mathbb{N}}(l) > d$.
Then, $\mathcal{H}_G^d(\boldsymbol{s'} m') = \mathcal{H}_G^d(\boldsymbol{s}m) = \mathcal{H}_G^d(\boldsymbol{t}) . l . \mathcal{H}_G^d(r m)$, and so we may write $\boldsymbol{s'} m' = \boldsymbol{t'_1} . l . \boldsymbol{t'_2} . m'$.
Now, $\boldsymbol{t}, \boldsymbol{t'_1} \in \sigma$, $\boldsymbol{t} l, \boldsymbol{t'_1} l \in P_G$, $\mathcal{H}_G^d(\boldsymbol{t}l) = \mathcal{H}_G^d(\boldsymbol{t'_1}l)$, and $\boldsymbol{t}l$ and $\boldsymbol{t'}l'$ are both $d$-complete; thus, by the induction hypothesis, $\boldsymbol{t} l = \boldsymbol{t'_1} l$.
Thus, $\mathcal{H}_G^d(\boldsymbol{t}) . l . \mathcal{H}_G^d(\boldsymbol{t'_2} m') = \mathcal{H}_G^d(\boldsymbol{s'} m') = \mathcal{H}_G^d(\boldsymbol{s} m) = \mathcal{H}_G^d(\boldsymbol{t}) . l . \mathcal{H}_G^d(r m)$, whence $\boldsymbol{t'_2}$ is of the form $r \boldsymbol{t''_2}$ by the determinacy of $\sigma$. 
Hence, $\boldsymbol{s} m = \boldsymbol{t} l r m$ and $\boldsymbol{s'} m' = \boldsymbol{t} l r \boldsymbol{t''_2} m'$.
Finally, if $r$ is external, then so is $m$ by IE-switch, and so $\boldsymbol{s'} m' = \boldsymbol{s} m$; if $r$ is $j$-internal ($j > d$), then so is $m$, and we apply the axiom DP2 for $i = j-1$ to $\boldsymbol{s}$ and $\boldsymbol{s'}$, whence $\boldsymbol{s} m = \boldsymbol{s'} m'$.
\end{proof}

\begin{theorem}[External consistency]
\label{ThmExternalConsistency}
Let $\sigma, \tau : G$ such that $\sigma \simeq_G \tau$, and $d \in \mathbb{N} \cup \{ \omega \}$.
\begin{enumerate}

\item If $\boldsymbol{s} m n, \boldsymbol{s'} m' n' \in \sigma$ are $d$-complete, and $\mathcal{H}_G^d(\boldsymbol{s} m) = \mathcal{H}_G^d(\boldsymbol{s'} m')$, then $\boldsymbol{s}mn = \boldsymbol{s'}m'n'$;

\item If $\boldsymbol{s} m n \in \sigma, \boldsymbol{t} l r \in \tau$ are $d$-complete, and $\mathcal{H}_G^d(\boldsymbol{s} m) \simeq_{\mathcal{H}^d(G)} \mathcal{H}_G^d(\boldsymbol{t} l)$, then $\boldsymbol{s}mn \simeq_G \boldsymbol{t}lr$.

\end{enumerate}
\end{theorem}
\begin{proof}
Let us first prove the first statement.
Let $\sigma : G$ be a strategy, $\boldsymbol{s} m n, \boldsymbol{s'} m' n' \in \sigma$ and $d \in \mathbb{N} \cup \{ \omega \}$, and assume that $\boldsymbol{s} m n, \boldsymbol{s'} m' n'$ are both $d$-complete and $\mathcal{H}_G^d(\boldsymbol{s} m) = \mathcal{H}_G^d(\boldsymbol{s'} m')$.
By the first statement of Lemma~\ref{LemODeterminacy}, we have $\boldsymbol{s}m = \boldsymbol{s'} m'$; thus, by the axiom S2 on $\sigma$, we have $n = n'$ and $\mathcal{J}_{\boldsymbol{s}mn}(n) = \mathcal{J}_{\boldsymbol{s'}m'n'}(n')$, whence $\mathcal{J}^{\circleddash d}_{\boldsymbol{s}mn}(n) = \mathcal{J}^{\circleddash d}_{\boldsymbol{s'}m'n'}(n')$.

Similarly, the second statement is proved by the second statement of Lemma~\ref{LemODeterminacy}, completing the proof. 
\end{proof}

\begin{corollary}[Stepwise identification of strategies]
\label{CoroStepwiseIdentificationOfStrategies}
Any strategies $\sigma, \tau : G$ such that $\sigma \simeq_G \tau$ satisfy $\sigma \simeq_G^d \tau$ for all $d \in \mathbb{N} \cup \{ \omega \}$, where:
\begin{align*}
\sigma \simeq_G^d \tau \stackrel{\mathrm{df. }}{\Leftrightarrow} \ &\forall \boldsymbol{s} \in \sigma, \boldsymbol{t} \in \tau, \boldsymbol{s}m, \boldsymbol{t}l \in P_G . \ \! \boldsymbol{s} m \simeq_G^d \boldsymbol{t} l \Rightarrow \forall \boldsymbol{s} m n \in \sigma . \ \! \exists \boldsymbol{t} l r \in \tau . \ \! \boldsymbol{s} m n \simeq_G^d \boldsymbol{t} l r \\ &\wedge \forall \boldsymbol{t} l r \in \tau . \ \! \exists \boldsymbol{s} m n \in \sigma . \ \! \boldsymbol{t} l r \simeq_G^d \boldsymbol{s} m n.
\end{align*}
\end{corollary}
\begin{proof}
Immediate from Theorem~\ref{ThmExternalConsistency}. 
\end{proof}

Hence, for any strategies $\sigma, \tau : G$, we have:
\begin{equation*}
\sigma \simeq_G \tau \Leftrightarrow \forall d \in \mathbb{N} \cup \{ \omega \} . \ \! \sigma \simeq_G^d \tau
\end{equation*}
which will be useful later in the paper. 

Let us proceed to show that the relation $\simeq_G$ on strategies on any game $G$ is a PER.
\begin{lemma}[PER lemma]
\label{LemFirstPERLemma}
Given $\sigma, \tau : G$ such that $\sigma \simeq_G \tau$, we have:
\begin{equation*}
(\forall \boldsymbol{s} \in \sigma . \ \! \exists \boldsymbol{t} \in \tau . \ \! \boldsymbol{s} \simeq_G \boldsymbol{t}) \wedge (\forall \boldsymbol{t} \in \tau . \ \! \exists \boldsymbol{s} \in \sigma . \ \! \boldsymbol{t} \simeq_G \boldsymbol{s}).
\end{equation*}
\end{lemma}
\begin{proof}
By symmetry, it suffices to show $\forall \boldsymbol{s} \in \sigma . \ \! \exists \boldsymbol{t} \in \tau . \ \! \boldsymbol{s} \simeq_G \boldsymbol{t}$. We prove it by induction on $|\boldsymbol{s}|$. The base case is trivial; for the inductive step, let $\boldsymbol{s} m n \in \sigma$.
By the induction hypothesis, there exists some $\boldsymbol{t} \in \tau$ such that $\boldsymbol{s} \simeq_G \boldsymbol{t}$.
Then, by the axiom DI3 on $\simeq_G$, there exists some $\boldsymbol{t} l \in \tau$ such that $\boldsymbol{s} m \simeq_G \boldsymbol{t} l$.
Finally, since $\sigma \simeq_G \tau$, there exists some $\boldsymbol{t} l r \in \tau$ such that $\boldsymbol{s} m n \simeq_G \boldsymbol{t} l r$, completing the proof.
\end{proof}

\begin{proposition}[PERs on strategies]
Given a game $G$, the identification $\simeq_G$ of strategies on $G$ is a PER, i.e., a symmetric, transitive relation.
\end{proposition}
\begin{proof}
We just show the transitivity as the symmetry is obvious.
Let $\sigma, \tau, \mu : G$ such that $\sigma \simeq_G \tau$ and $\tau \simeq_G \mu$. 
Assume that $\boldsymbol{s}mn \in \sigma$, $\boldsymbol{u} \in \mu$ and $\boldsymbol{s} m \simeq_G \boldsymbol{u} p$. 
By Lemma~\ref{LemFirstPERLemma}, there exists some $\boldsymbol{t} \in \tau$ such that $\boldsymbol{s} \simeq_G \boldsymbol{t}$.
By the axiom DI3 on $\simeq_G$, there exists some $\boldsymbol{t} l \in P_G$ such that $\boldsymbol{s} m \simeq_G \boldsymbol{t} l$, whence $\boldsymbol{t} l \simeq_G \boldsymbol{u} p$.
Also, since $\sigma \simeq_G \tau$, there exists some $\boldsymbol{t} l r \in \tau$ such that $\boldsymbol{s} m n \simeq_G \boldsymbol{t} l r$.
Finally, since $\tau \simeq_G \mu$, there exists some $\boldsymbol{u} p q \in \mu$ such that $\boldsymbol{t} l r \simeq_G \boldsymbol{u} p q$, whence $\boldsymbol{s} m n \simeq_G \boldsymbol{u} p q$, completing the proof.
\end{proof}

Therefore, given a game $G$, we may take the equivalence classes $[\sigma] \stackrel{\mathrm{df. }}{=} \{ \tau : G \mid \sigma \simeq_G \tau \ \! \}$ of valid strategies $\sigma : G$; these equivalence classes, rather than strategies themselves, have interpreted proofs and programs \cite{abramsky2000full,mccusker1998games}.

At this point, let us note that even-length positions are not necessarily preserved under the hiding operation on j-sequences (Definition~\ref{DefHidingOnJSequences}). 
For instance, let $\boldsymbol{s} m n \boldsymbol{t}$ be an even-length position of a game $G$ such that $\boldsymbol{s} m$ (resp. $n \boldsymbol{t}$) consists of external (resp. internal) moves only. By IE-switch on $G$, $m$ is an O-move, and so $\mathcal{H}_G^\omega(\boldsymbol{s} m n \boldsymbol{t}) = \boldsymbol{s} m$ is of odd-length. 

Taking into account this fact, we define:
\begin{definition}[Hiding operation on strategies]
\label{DefHidingOperationOnDynamicStrategies}
Let $G$ be a game, and $d \in \mathbb{N} \cup \{ \omega \}$. 
Given $\boldsymbol{s} \in P_G$, we define:
\begin{equation*}
\boldsymbol{s} \natural \mathcal{H}_G^d \stackrel{\mathrm{df. }}{=} \begin{cases} \mathcal{H}_G^d(\boldsymbol{s}) &\text{if $\boldsymbol{s}$ is $d$-complete (Definition~\ref{DefDynamicArenas});} \\ \boldsymbol{t} &\text{otherwise, where $\mathcal{H}_G^d(\boldsymbol{s}) = \boldsymbol{t}m$.} \end{cases}
\end{equation*}
The \emph{\bfseries $\boldsymbol{d}$-hiding operation $\mathcal{H}^d$ (on strategies)} is then given by:
\begin{equation*}
\mathcal{H}^d : (\sigma : G) \mapsto \{ \boldsymbol{s} \natural \mathcal{H}_G^d \mid \boldsymbol{s} \in \sigma \ \! \}.
\end{equation*}
\end{definition}

Let us proceed to establish a beautiful fact: $\sigma : G \Rightarrow \mathcal{H}^d(\sigma) : \mathcal{H}^d(G)$ for all $d \in \mathbb{N} \cup \{ \omega \}$.
For this task, we need the following lemma:
\begin{lemma}[Asymmetry lemma]
\label{LemAsymmetryLemma}
Let $\sigma : G$ be a strategy, and $d \in \mathbb{N} \cup \{ \omega \}$. Assume that $\boldsymbol{s} m n \in \mathcal{H}^d(\sigma)$, where $\boldsymbol{s} m n = \boldsymbol{t} m \boldsymbol{u} n \boldsymbol{v} \natural \mathcal{H}_G^d$ with $\boldsymbol{t} m \boldsymbol{u} n \boldsymbol{v} \in \sigma$ not $d$-complete. 
Then, $\boldsymbol{s} m n = \mathcal{H}^d(\boldsymbol{t} m \boldsymbol{u} n)  = \mathcal{H}^d(\boldsymbol{t}) . m n$.
\end{lemma}
\begin{proof}
Since $\boldsymbol{t} m \boldsymbol{u} n \boldsymbol{v} \in \sigma$ is not $d$-complete, we may write $\boldsymbol{v} = \boldsymbol{v}_1 l \boldsymbol{v}_2 r$ with $\lambda_G^{\mathbb{N}}(l) = 0 \vee \lambda_G^{\mathbb{N}}(l) > d$, $0 < \lambda_G^{\mathbb{N}}(r) \leqslant d$ and $0 < \lambda_G^{\mathbb{N}}(x) \leqslant d$ for all moves $x$ in $\boldsymbol{v}_1$ or $\boldsymbol{v}_2$. Then, we have $\boldsymbol{s} m n = \boldsymbol{t} m \boldsymbol{u} n \boldsymbol{v}_1 l \boldsymbol{v}_2 r \natural \mathcal{H}_G^d = \mathcal{H}_G^d(\boldsymbol{t}) m \mathcal{H}_G^d(\boldsymbol{u}) n = \mathcal{H}_G^d(\boldsymbol{t}) m n$.
\end{proof}

We are now ready to establish:
\begin{theorem}[Hiding theorem]
\label{ThmHidingTheorem}
If $\sigma : G$, then $\mathcal{H}^d(\sigma) : \mathcal{H}^d(G)$ for all $d \in \mathbb{N} \cup \{ \omega \}$.
\end{theorem}
\begin{proof}
We first show $\mathcal{H}^d(\sigma) \subseteq P_{\mathcal{H}^d(G)}^{\mathsf{Even}}$. 
Let $\boldsymbol{s} \in \mathcal{H}^d(\sigma)$, i.e., $\boldsymbol{s} = \boldsymbol{t} \natural \mathcal{H}_G^d$ for some $\boldsymbol{t} \in \sigma$. Let us write $\boldsymbol{t} = \boldsymbol{t'} m$ as the case $\boldsymbol{t} = \boldsymbol{\epsilon}$ is trivial.
\begin{itemize}

\item If $\boldsymbol{t}$ is $d$-complete, then $\boldsymbol{s} = \boldsymbol{t} \natural \mathcal{H}_G^d = \mathcal{H}_G^d(\boldsymbol{t})  \in P_{\mathcal{H}^d(G)}$. Also, since $\boldsymbol{s} = \mathcal{H}_G^d(\boldsymbol{t'}) m$ and $m$ is a P-move, $\boldsymbol{s}$ must be of even-length by alternation on $\mathcal{H}^d(G)$.

\item If $\boldsymbol{t}$ is not $d$-complete, then we may write $\boldsymbol{t} = \boldsymbol{t''} m_0 m_1 \dots m_k$, where $m_k = m$, $\boldsymbol{t''} m_0$ is $d$-complete, and $0 < \lambda_G^{\mathbb{N}}(m_i) \leqslant d$ for $i = 1, 2, \dots, k$.
By IE-switch, $m_0$ is an O-move, and thus $\boldsymbol{s} = \mathcal{H}_G^d(\boldsymbol{t''}) \in P_{\mathcal{H}^d(G)}$ is of even-length.
\end{itemize}
It remains to verify the axioms S1 and S2. 
For S1, $\mathcal{H}^d(\sigma)$ is non-empty as $\boldsymbol{\epsilon} \in \mathcal{H}^d(\sigma)$.
For the even-prefix-closure, let $\boldsymbol{s} m n \in \mathcal{H}^d(\sigma)$; we have to show $\boldsymbol{s} \in \mathcal{H}^d(\sigma)$.
We have some $\boldsymbol{t} m \boldsymbol{u} n \boldsymbol{v} \in \sigma$ such that $\boldsymbol{t} m \boldsymbol{u} n \boldsymbol{v} \natural \mathcal{H}_G^d = \boldsymbol{s} m n$. By Lemma~\ref{LemAsymmetryLemma}, $\boldsymbol{s} m n = \mathcal{H}_G^d(\boldsymbol{t}) m n$, whence $\boldsymbol{s} = \mathcal{H}_G^d(\boldsymbol{t})$.
For $\boldsymbol{t}m$ is $d$-complete, so is $\boldsymbol{t}$ by IE-switch.
Thus, $\boldsymbol{s} = \mathcal{H}_G^d(\boldsymbol{t}) = \boldsymbol{t} \natural \mathcal{H}_G^d \in \mathcal{H}^d(\sigma)$.

Finally for S2, let $\boldsymbol{s} m n, \boldsymbol{s} m n' \in \mathcal{H}^d(\sigma)$; we have to show $n = n'$ and $\mathcal{J}^{\circleddash d}_{\boldsymbol{s}m}(n) = \mathcal{J}^{\circleddash d}_{\boldsymbol{s}m}(n')$.
Clearly, $\boldsymbol{s} m n = \boldsymbol{t} m \boldsymbol{u} n \boldsymbol{v} \natural \mathcal{H}_G^d$, $\boldsymbol{s} m n' = \boldsymbol{t'} m \boldsymbol{u'} n' \boldsymbol{v'} \natural \mathcal{H}_G^d$ for some $\boldsymbol{t} m \boldsymbol{u} n \boldsymbol{v}, \boldsymbol{t'} m \boldsymbol{u'} n' \boldsymbol{v'} \in \sigma$.
Then, by Lemma~\ref{LemAsymmetryLemma}, $\boldsymbol{s} m n = \mathcal{H}_G^d(\boldsymbol{t} m \boldsymbol{u}) n$ and $\boldsymbol{s} m n' = \mathcal{H}_G^d(\boldsymbol{t'} m \boldsymbol{u'}) n'$. Therefore, by Theorem~\ref{ThmExternalConsistency}, $n = n'$ and $\mathcal{J}^{\circleddash d}_{\boldsymbol{s}mn}(n) = \mathcal{J}^{\circleddash d}_{\boldsymbol{s}mn'}(n')$, completing the proof.
\end{proof}

Next, let us review standard constraints on strategies.
First, recall that a programming language is \emph{total} if its computation always terminates in a finite period of time. 
This programming concept is interpreted in game semantics by \emph{totality} of strategies in the sense similar to totality of partial functions: 
\begin{definition}[Totality of strategies \cite{abramsky1997semantics}]
\label{DefTotalityOfStrategies}
A strategy $\sigma : G$ is \emph{\bfseries total} if it satisfies $\forall \boldsymbol{s} \in \sigma, \boldsymbol{s} m \in P_G . \ \! \exists \boldsymbol{s} m n \in \sigma$.
\end{definition}

Nevertheless, it is well-known that totality of strategies is \emph{not} preserved under composition due to the problem of `infinite chattering' \cite{abramsky1997semantics,clairambault2010totality}. 
For this point, one usually imposes a condition on strategies stronger than totality, e.g., \emph{winning} \cite{abramsky1997semantics}, that is preserved under composition.
We may certainly just apply the winning condition of \cite{abramsky1997semantics}, but it requires an additional structure on games, which may be criticized as extrinsic and/or ad-hoc; thus, we prefer another, simpler solution.
A natural idea is then to require that strategies should not contain any strictly increasing (with respect to $\preceq$) infinite sequence of positions. 
However, we have to relax this constraint: The dereliction $\mathit{der}_A$ (Definition~\ref{DefDerelictions}), the $\beta$-identity on a game $A$ in the game-semantic CCBoC given in Section~\ref{DynamicGameSemantics}, satisfies it iff so does $A$, but we cannot impose it on games as the operation $\Rightarrow \ = \ !(\_) \multimap (\_)$ on games, which is the $\beta$-exponential construction in the CCBoC, does not preserve it.

Thus, instead, we apply the same idea to \emph{P-views}, arriving at:
\begin{definition}[Noetherianity of strategies \cite{clairambault2010totality}]
\label{DefNoetherianityOfStrategies}
A strategy $\sigma : G$ is \emph{\bfseries noetherian} if it does not contain any strictly increasing (with respect to $\preceq$) infinite sequence of P-views of $G$.
\end{definition}
It has been shown in \cite{clairambault2010totality} that total, noetherian static strategies are closed under composition. 

Next, recall that one of the highlights of \emph{HO-games} \cite{hyland2000full} is to give a one-to-one correspondence between PCF B\"{o}hm trees and \emph{innocent}, \emph{well-bracketed} static strategies (on static games modeling types of PCF).
That is, the two conditions narrow down the hom-sets of the codomain of the interpretation functor, i.e., the category of HO-games, so that the interpretation becomes \emph{full}.
Roughly, a strategy is innocent if its computation depends only on P-views, and well-bracketed if every `question-answering' by the strategy is achieved in the `last-question-first-answered' fashion. Formally:
\begin{definition}[Innocence of strategies \cite{hyland2000full}]
\label{DefInnocenceOfStrategies}
A strategy $\sigma : G$ is \emph{\bfseries innocent} if $\forall \boldsymbol{s}mn, \boldsymbol{t} \in \sigma, \boldsymbol{t} m \in P_G . \ \! \lceil \boldsymbol{t} m \rceil = \lceil \boldsymbol{s} m \rceil \Rightarrow \boldsymbol{t}mn \in \sigma \wedge \lceil \boldsymbol{t} m n \rceil = \lceil \boldsymbol{s} m n \rceil$.
\end{definition}

\begin{definition}[Well-bracketing of strategies \cite{hyland2000full}]
\label{DefWellBracketingOfStrategies}
A strategy $\sigma : G$ is \emph{\bfseries well-bracketed (wb)} if, given $\boldsymbol{s} q \boldsymbol{t} a \in \sigma$, where $\lambda_G^{\mathsf{QA}}(q) = \mathsf{Q}$, $\lambda_G^{\mathsf{QA}}(a) = \mathsf{A}$ and $\mathcal{J}_{\boldsymbol{s}q\boldsymbol{t}a}(a) = q$, each occurrence of a question in $\boldsymbol{t'}$, defined by $\lceil \boldsymbol{s} q \boldsymbol{t} \rceil_G = \lceil \boldsymbol{s} q \rceil_G . \boldsymbol{t'}$, justifies an occurrence of an answer in $\boldsymbol{t'}$.
\end{definition}


Now, let us show that the standard constraints on strategies except totality are all preserved under the hiding operation, which implies that dynamic strategies are a reasonable generalization of static strategies in a certain sense.
\begin{corollary}[Preservation of constraints on strategies under hiding]
\label{CoroPreservationOfConstraintsOnDynamicStrategiesUnderHiding}
If a strategy $\sigma : G$ is valid, innocent, wb or noetherian, then so is $\mathcal{H}^d(\sigma) : \mathcal{H}^d(G)$, and if another $\tau : G$ satisfies $\sigma \simeq_G \tau$, then $\mathcal{H}^d(\sigma) \simeq_{\mathcal{H}^d(G)} \mathcal{H}^d(\tau)$, for all $d \in \mathbb{N} \cup \{ \omega \}$.
\end{corollary}
\begin{proof}
Let $d \in \mathbb{N} \cup \{ \omega \}$ be arbitrarily fixed. 
We have $\mathcal{H}^d(\sigma) : \mathcal{H}^d(G)$ by Theorem~\ref{ThmHidingTheorem}.
\begin{itemize}

\item Preservation of validity is by Lemma~\ref{LemODeterminacy}, Corollary~\ref{CoroStepwiseIdentificationOfStrategies} and the axiom DI3 on $\simeq_G$;

\item Preservation of innocence and noetherianity holds because $\lceil \mathcal{H}_G^d(\boldsymbol{s}m) \rceil_{\mathcal{H}^d(G)}$ is a j-subsequence of $\mathcal{H}_G^d(\lceil \boldsymbol{s} m \rceil_G)$ for any $\boldsymbol{s} m \in P_G^{\mathsf{Odd}}$;

\item Well-bracketing is preserved under the $d$-hiding operation $\mathcal{H}^d$ because \emph{both} of the question and the answer of each `QA-pair' are either deleted or retained.

\end{itemize}
Finally, preservation of identification of strategies is proved similarly to that of validity, completing the proof. 
\end{proof}

\begin{remark*}
Totality of strategies is \emph{not} preserved under the $d$-hiding operation $\mathcal{H}^d$ on strategies for all $d \in \mathbb{N} \cup \{ \omega \}$.
For instance, consider any total strategy that always performs a 1-internal P-move, which is no longer total when $\mathcal{H}$ is applied.
As we shall see shortly, it is why totality is preserved under concatenation of strategies but not under composition (i.e., composition coincides with \emph{concatenation plus hiding}).
\end{remark*}

At the end of the present section, we establish an inductive property of the $d$-hiding operation on strategies for each $d \in \mathbb{N} \cup \{ \omega \}$:
\begin{notation*}
Given $\sigma : G$ and $d \in \mathbb{N} \cup \{ \omega \}$, we define $\sigma^d_\downarrow \stackrel{\mathrm{df. }}{=} \{ \boldsymbol{s} \in \sigma \mid \text{$\boldsymbol{s}$ is $d$-complete} \ \! \}$ and $\sigma^d_\uparrow \stackrel{\mathrm{df. }}{=} \sigma \setminus \sigma^d_{\downarrow}$.
\end{notation*}

\begin{lemma}[Hiding and complete positions]
\label{LemHidingAndCompletePositions}
Let $\sigma : G$. 
Given $i, d \in \mathbb{N}$ such that $i \geqslant d$, we have $\mathcal{H}^i(\sigma) = \mathcal{H}^i(\sigma^d_\downarrow) \stackrel{\mathrm{df. }}{=} \{ \boldsymbol{s} \natural \mathcal{H}_G^i \mid \boldsymbol{s} \in \sigma^d_\downarrow \ \! \}$.
\end{lemma}
\begin{proof}
$\mathcal{H}^i(\sigma^d_\downarrow) \subseteq \mathcal{H}^i(\sigma)$ is obvious.
For the opposite inclusion, let $\boldsymbol{s} \in \mathcal{H}^i(\sigma)$, i.e., $\boldsymbol{s} = \boldsymbol{t} \natural \mathcal{H}_G^i$ for some $\boldsymbol{t} \in \sigma$; we have to show $\boldsymbol{s} \in \mathcal{H}^i(\sigma^d_\downarrow)$.
If $\boldsymbol{t} \in \sigma^d_\downarrow$, then we are done; thus, assume otherwise. 
If there is no external or $j$-internal move with $j > i$ other than the first move $m_0$ in $\boldsymbol{t}$, then $\boldsymbol{s} = \boldsymbol{\epsilon} \in \mathcal{H}^i(\sigma^d_\downarrow)$; so assume otherwise.
As a result, we may write $\boldsymbol{t} = m_0 \boldsymbol{t_1} m n \boldsymbol{t_2} r$, where $\boldsymbol{t_2} r$ consists only of $j$-internal moves with $0 < j \leqslant i$, and $m$ and $n$ are P- and O-moves, respectively, such that $\lambda_G^{\mathbb{N}}(m) = \lambda_G^{\mathbb{N}}(n) = 0 \vee \lambda_G^{\mathbb{N}}(m) = \lambda_G^{\mathbb{N}}(n) > i$.
Take $m_0 \boldsymbol{t_1} m \in \sigma^d_\downarrow$ such that $m_0 \boldsymbol{t_1} m \natural \mathcal{H}_G^i = m_0 \mathcal{H}_G^i(\boldsymbol{t_1}) m = \boldsymbol{t} \natural \mathcal{H}_G^i = \boldsymbol{s}$, whence $\boldsymbol{s} \in \mathcal{H}^i(\sigma^d_\downarrow)$.
\end{proof}

We are now ready to show:
\begin{lemma}[Stepwise hiding on strategies]
Given $\sigma : G$, we have $\mathcal{H}^{i+1}(\sigma) = \mathcal{H}^1(\mathcal{H}^i(\sigma))$ for all $i \in \mathbb{N}$.
\end{lemma}

\begin{proof}
We first show the inclusion $\mathcal{H}^{i+1}(\sigma) \subseteq \mathcal{H}^1(\mathcal{H}^i(\sigma))$.
By Lemma~\ref{LemHidingAndCompletePositions}, we may write any element of the set $\mathcal{H}^{i+1}(\sigma)$ as $\boldsymbol{s} \natural \mathcal{H}_G^{i+1}$ for some $\boldsymbol{s} \in \sigma^{i+1}_\downarrow$.
Then observe that:
\begin{align*}
\boldsymbol{s} \natural \mathcal{H}_G^{i+1} = \mathcal{H}_G^{i+1}(\boldsymbol{s}) = \mathcal{H}_{\mathcal{H}^i(G)}(\mathcal{H}_G^i(\boldsymbol{s})) = (\boldsymbol{s} \natural \mathcal{H}_G^i) \natural \mathcal{H}^1_{\mathcal{H}^i(G)} \in \mathcal{H}^1(\mathcal{H}^i(\sigma)).
\end{align*}

For the opposite inclusion $\mathcal{H}^1(\mathcal{H}^i(\sigma)) \subseteq \mathcal{H}^{i+1}(\sigma)$, again by Lemma~\ref{LemHidingAndCompletePositions}, we may write any element of $\mathcal{H}^1(\mathcal{H}^i(\sigma))$ as $(\boldsymbol{s} \natural \mathcal{H}_G^i) \natural \mathcal{H}_{\mathcal{H}^i(G)}^1$ for some $\boldsymbol{s} \in \sigma^i_\downarrow$.
We have to show that $(\boldsymbol{s} \natural \mathcal{H}_G^i) \natural \mathcal{H}_{\mathcal{H}^i(G)}^1 \in \mathcal{H}^{i+1}(\sigma)$. 
If $\boldsymbol{s} \in \sigma_\downarrow^{i+1}$, then it is completely analogous to the above argument; so assume otherwise. Also, if an external or $j$-internal move with $j > i+1$ in $\boldsymbol{s}$ is only the first move $m_0$, then $(\boldsymbol{s} \natural \mathcal{H}_G^i) \natural \mathcal{H}_{\mathcal{H}^i(G)}^1 = \boldsymbol{\epsilon} \in \mathcal{H}^{i+1}(\sigma)$; thus assume othewise. 
Now, we may write: 
\begin{equation*}
\boldsymbol{s} = \boldsymbol{s'} m n m_1 m_2 \dots m_{2k} r
\end{equation*}
where $\lambda_G^{\mathbb{N}}(r) = i+1$, $m_1, m_2, \dots, m_{2k}$ are $j$-internal with $0 < j \leqslant i+1$, and $m$ and $n$ are external or $j$-internal P- and O-moves with $j > i+1$, respectively. Then, 
\begin{align*}
(\boldsymbol{s} \natural \mathcal{H}_G^i) \natural \mathcal{H}^1_{\mathcal{H}^i(G)} &= \mathcal{H}_G^i(\boldsymbol{s}) \natural \mathcal{H}^1_{\mathcal{H}^i(G)} \\
&= \mathcal{H}_{\mathcal{H}^i(G)}(\mathcal{H}_G^i(\boldsymbol{s'})) . m \\
&= \mathcal{H}_G^{i+1}(\boldsymbol{s'}) . m \ \text{(by Lemma~\ref{LemStepwiseHidingOnJSequences})} \\
&= \boldsymbol{s} \natural \mathcal{H}_G^{i+1} \in \mathcal{H}^{i+1}(\sigma)
\end{align*}
which completes the proof.
\end{proof}

Thus, as in the case of games, we may focus on the operation $\mathcal{H}^1$:
\begin{convention*}
Henceforth, we write $\mathcal{H}$ for $\mathcal{H}^1$ and call it the \emph{\bfseries hiding operation (on strategies)}; $\mathcal{H}^i$ denotes the $i$-times iteration of $\mathcal{H}$ for all $i \in \mathbb{N}$.
\end{convention*}


\subsection{Constructions on Dynamic Strategies}
\label{ConstructionsOnDynamicStrategies}
Next, let us recall standard constructions on strategies \cite{abramsky1999game}. Note that since (dynamic) strategies are simply `static strategies on (dynamic) games', they are clearly closed under all the constructions on static strategies. 

Nevertheless, the CCBoC of games and strategies given in Section~\ref{DynamicGameSemantics} has normalized games as 0-cells and strategies $\phi : G$ such that $\mathcal{H}^\omega(G) \trianglelefteqslant A \Rightarrow B$ as 1-cells $A \rightarrow B$, and therefore we need to generalize \emph{pairing} and \emph{promotion} of static strategies; in fact, we have generalized product and exponential of static games respectively to pairing and promotion of dynamic games for this purpose.
Also, we shall decompose and generalize composition of static strategies as \emph{concatenation plus hiding} of dynamic strategies, for which we have introduced concatenation of dynamic games. 

Let us begin with recalling \emph{tensor} $\otimes$ of strategies:
\begin{definition}[Tensor of strategies \cite{abramsky1999game}]
\label{DefTensorOfStrategies}
Given games $A$, $B$, $C$ and $D$, and strategies $\phi : A \multimap C$ and $\psi : B \multimap D$, the \emph{\bfseries tensor (product)} $\phi \otimes \psi$ of $\phi$ and $\psi$ is given by:
\begin{equation*}
\phi \otimes \psi \stackrel{\mathrm{df. }}{=} \{ \boldsymbol{s} \in \mathscr{L}_{A \otimes B \multimap C \otimes D} \mid \boldsymbol{s} \upharpoonright A, C \in \phi, \boldsymbol{s} \upharpoonright B, D \in \psi \ \! \}.
\end{equation*}
\end{definition}

Intuitively the tensor $\phi \otimes \psi : A \otimes B \multimap C \otimes D$ of $\phi : A \multimap C$ and $\psi : B \multimap D$ plays by $\phi$ if the last O-move is of $A$ or $C$, and by $\psi$ otherwise.

\begin{example}
The tensor $\mathit{succ} \otimes\mathit{double} : N \otimes N \multimap N \otimes N$, where $\mathit{succ}, \mathit{double} : N \multimap N$ are given in Section~\ref{Introduction}, plays, e.g., as follows:
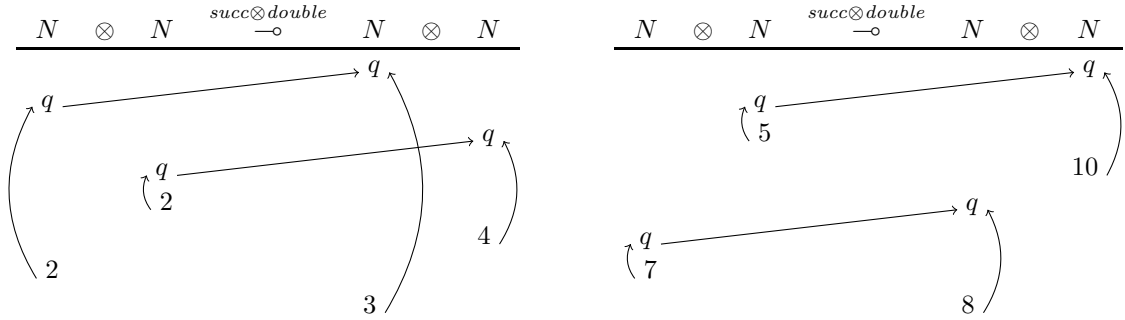
\begin{center}
\begin{tabular}{ccccccccccccccccc}
$N$&$\otimes$&$N$&$\stackrel{\mathit{succ} \otimes\mathit{double}}{\multimap}$&$N$ & $\otimes$ & $N$ &&&&$N$&$\otimes$&$N$&$\stackrel{\mathit{succ} \otimes\mathit{double}}{\multimap}$& $N$ & $\otimes$ & $N$ \\ \cline{1-7} \cline{11-17}
&&&& \tikzmark{ctensor57} $q$ \tikzmark{ctensor59} &&&&&&&&&&&& \tikzmark{ctensor51} $q$ \tikzmark{ctensor53} \\
\tikzmark{ctensor58} $q$ \tikzmark{dtensor57} &&&&&&&&&&&& \tikzmark{ctensor52} $q$ \tikzmark{dtensor51} &&&& \\
&&&&&& \tikzmark{ctensor60} $q$ \tikzmark{ctensor62} &&&&&& \tikzmark{dtensor52} $5$&&&& \\
&& \tikzmark{ctensor61} $q$ \tikzmark{dtensor60} &&&&&&&&&&&&&&$10$ \tikzmark{dtensor53} \\
&& \tikzmark{dtensor61} $2$&&&&&&&&&&&& \tikzmark{ctensor54} $q$ \tikzmark{ctensor56} && \\
&&&&&&$4$ \tikzmark{dtensor62} &&&& \tikzmark{ctensor55} $q$ \tikzmark{dtensor54} &&&&&& \\
\tikzmark{dtensor58} $2$&&&&&&&&&& \tikzmark{dtensor55} $7$ &&&&&& \\
&&&&$3$ \tikzmark{dtensor59} &&&&&&&&&&$8$ \tikzmark{dtensor56} &&
\end{tabular}
\begin{tikzpicture}[overlay, remember picture, yshift=.25\baselineskip]
\draw [->] ({pic cs:dtensor51}) to ({pic cs:ctensor51});
\draw [->] ({pic cs:dtensor52}) [bend left] to ({pic cs:ctensor52});
\draw [->] ({pic cs:dtensor53}) [bend right] to ({pic cs:ctensor53});
\draw [->] ({pic cs:dtensor54}) to ({pic cs:ctensor54});
\draw [->] ({pic cs:dtensor55}) [bend left] to ({pic cs:ctensor55});
\draw [->] ({pic cs:dtensor56}) [bend right] to ({pic cs:ctensor56});
\draw [->] ({pic cs:dtensor57}) to ({pic cs:ctensor57});
\draw [->] ({pic cs:dtensor58}) [bend left] to ({pic cs:ctensor58});
\draw [->] ({pic cs:dtensor59}) [bend right] to ({pic cs:ctensor59});
\draw [->] ({pic cs:dtensor60}) to ({pic cs:ctensor60});
\draw [->] ({pic cs:dtensor61}) [bend left] to ({pic cs:ctensor61});
\draw [->] ({pic cs:dtensor62}) [bend right] to ({pic cs:ctensor62});
\end{tikzpicture}
\end{center}
\end{example}

\begin{lemma}[Well-defined tensor of strategies]
\label{LemWellDefinedTensorOfStrategies}
Given games $A$, $B$, $C$ and $D$, and strategies $\phi : A \multimap C$ and $\psi : B \multimap D$, $\phi \otimes \psi$ is a  strategy on  $A \otimes B \multimap C \otimes D$. 
If $\phi$ and $\psi$ are innocent (resp. wb, total, noetherian), then so is $\phi \otimes \psi$.
Given $\phi' : A \multimap C$ and $\psi' : B \multimap D$ with $\phi \simeq_{A \multimap C} \phi'$ and $\psi \simeq_{B \multimap D} \psi'$, $\phi \otimes \psi \simeq_{A \otimes B \multimap C \otimes D} \phi' \otimes \psi'$.
\end{lemma}
\begin{proof}
Straightforward; see \cite{mccusker1998games,abramsky2000full}.
\end{proof}

We proceed to recall \emph{pairing} of strategies:
\begin{definition}[Pairing of strategies \cite{abramsky1999game}]
\label{DefPairingOfStrategies}
Given games $A$, $B$ and $C$, and strategies $\phi : C \multimap A$ and $\psi : C \multimap B$, the \emph{\bfseries pairing} $\langle \phi, \psi \rangle$ of $\phi$ and $\psi$ is defined by:
\begin{align*}
\langle \phi, \psi \rangle \stackrel{\mathrm{df. }}{=} \{ \boldsymbol{s} \in \mathscr{L}_{C \multimap A \& B} \mid (\boldsymbol{s} \upharpoonright C, A \in \phi \wedge \boldsymbol{s} \upharpoonright B = \boldsymbol{\epsilon}) \vee (\boldsymbol{s} \upharpoonright C, B \in \psi \wedge \boldsymbol{s} \upharpoonright A = \boldsymbol{\epsilon}) \ \! \}.
\end{align*}
\end{definition}

That is, the pairing $\langle \phi, \psi \rangle : C \multimap A \& B$ of $\phi : C \multimap A$ and $\psi : C \multimap B$ plays by $\phi$ if the play is of $C \multimap A$, and by $\psi$ otherwise.

\begin{example}
The pairing $\langle \mathit{succ}, \mathit{double} \rangle : N \multimap N \& N$ plays as either of the following:
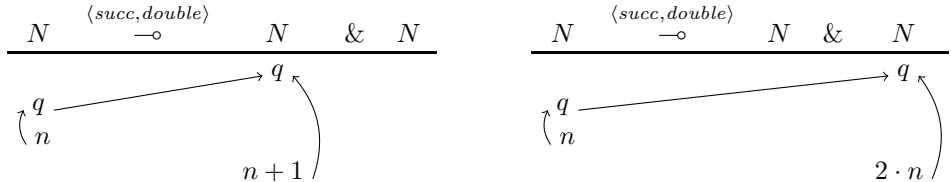
\begin{center}
\begin{tabular}{ccccccccccccccc}
$N$&$\stackrel{\langle \mathit{succ}, \mathit{double} \rangle}{\multimap}$&$N$ & $\&$ & $N$ &&&& $N$&$\stackrel{\langle \mathit{succ}, \mathit{double} \rangle}{\multimap}$& $N$ & $\&$ & $N$ \\ \cline{1-5} \cline{9-13}
&& \tikzmark{cpairing14} $q$ \tikzmark{cpairing16} &&&&&&&&&& \tikzmark{cpairing11} $q$ \tikzmark{cpairing13} \\
\tikzmark{cpairing15} $q$ \tikzmark{dpairing14} &&&&&&&& \tikzmark{cpairing12} $q$ \tikzmark{dpairing11} &&&& \\
\tikzmark{dpairing15} $n$&&&&&&&& \tikzmark{dpairing12} $n$&&&& \\
&&$n+1$ \tikzmark{dpairing16} &&&&&&&&&&$2 \cdot n$ \tikzmark{dpairing13}
\end{tabular}
\begin{tikzpicture}[overlay, remember picture, yshift=.25\baselineskip]
\draw [->] ({pic cs:dpairing11}) to ({pic cs:cpairing11});
\draw [->] ({pic cs:dpairing12}) [bend left] to ({pic cs:cpairing12});
\draw [->] ({pic cs:dpairing13}) [bend right] to ({pic cs:cpairing13});
\draw [->] ({pic cs:dpairing14}) to ({pic cs:cpairing14});
\draw [->] ({pic cs:dpairing15}) [bend left] to ({pic cs:cpairing15});
\draw [->] ({pic cs:dpairing16}) [bend right] to ({pic cs:cpairing16});
\end{tikzpicture}
\end{center}
where $n \in \mathbb{N}$, depending on the first O-move.
\end{example}

\begin{lemma}[Well-defined pairing of strategies]
\label{LemWellDefinedPairingOfStrategies}
Given games $A$, $B$ and $C$, and strategies $\phi : C \multimap A$ and $\psi : C \multimap B$, $\langle \phi, \psi \rangle$ is a  strategy on $C \multimap A \& B$. 
If $\phi$ and $\psi$ are innocent (resp. wb, total, noetherian), then so is $\langle \phi, \psi \rangle$.
Given $\phi' : C \multimap A$ and $\psi' : C \multimap B$ with $\phi \simeq_{C \multimap A} \phi'$ and $\psi \simeq_{C \multimap B} \psi'$,  $\langle \phi, \psi \rangle \simeq_{C \multimap A \& B} \langle \phi', \psi' \rangle$.
\end{lemma}
\begin{proof}
Straightforward; see \cite{mccusker1998games,abramsky2000full}.
\end{proof}

Next, let us recall \emph{promotion} of strategies:
\begin{definition}[Promotion of strategies \cite{mccusker1998games}]
\label{DefPromotionOfStrategies}
Given games $A$ and $B$, and a strategy $\varphi : \ !A \multimap B$, the \emph{\bfseries promotion} $\varphi^{\dagger}$ of $\varphi$ is defined by: 
\begin{equation*}
\varphi^{\dagger} \stackrel{\mathrm{df. }}{=} \{ \boldsymbol{s} \in \mathscr{L}_{!A \multimap !B} \mid \forall i \in \mathbb{N} . \ \! \boldsymbol{s} \upharpoonright i \in \varphi \ \! \}.
\end{equation*}
\end{definition}

That is, the promotion $\varphi^\dagger : \ !A \multimap \ !B$ of $\varphi : A \Rightarrow B$ plays, during a play $\boldsymbol{s}$ of $!A \multimap \ !B$, as $\varphi$ for each j-subsequence $\boldsymbol{s} \upharpoonright i$ or \emph{thread}.  
We could have defined noetherianity of strategies in terms of positions, but then it would not be preserved under promotion by the obvious reason; it is why we have defined it in terms of P-views (Definition~\ref{DefNoetherianityOfStrategies}).

\begin{example}
Let $\mathit{succ} : N \Rightarrow N$ be the successor strategy (n.b., it is on the implication $\Rightarrow$, not the linear implication $\multimap$), which specifically selects, say, the `tag' $(\_, 0)$ in the domain $!N$. 
Then, the promotion $\mathit{succ}^\dagger : \ !N \multimap \ !N$ plays, e.g., as follows:
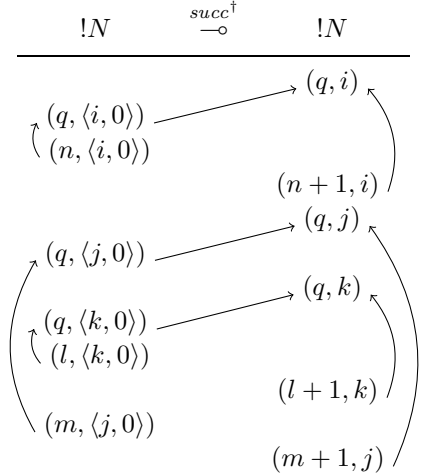
\begin{center}
\begin{tabular}{ccc}
$!N$ & $\stackrel{\mathit{succ}^\dagger}{\multimap}$ & $!N$ \\ \hline 
&& \tikzmark{cpromotion1} $(q, i)$ \tikzmark{cpromotion3} \\
\tikzmark{cpromotion2} $(q, \langle i, 0 \rangle)$ \tikzmark{dpromotion1} && \\
\tikzmark{dpromotion2} $(n, \langle i, 0 \rangle)$&& \\
&&$(n+1, i)$ \tikzmark{dpromotion3} \\
&& \tikzmark{cpromotion4} $(q, j)$ \tikzmark{cpromotion6} \\
\tikzmark{cpromotion5} $(q, \langle j, 0 \rangle)$ \tikzmark{dpromotion4} && \\
&& \tikzmark{cpromotion64} $(q, k)$ \tikzmark{cpromotion66} \\
\tikzmark{cpromotion65} $(q, \langle k, 0 \rangle)$ \tikzmark{dpromotion64} && \\
\tikzmark{dpromotion65} $(l, \langle k, 0 \rangle)$&& \\
&&$(l+1, k)$ \tikzmark{dpromotion66} \\
\tikzmark{dpromotion5} $(m, \langle j, 0 \rangle)$&& \\
&&$(m+1, j)$ \tikzmark{dpromotion6} 
\end{tabular}
\begin{tikzpicture}[overlay, remember picture, yshift=.25\baselineskip]
\draw [->] ({pic cs:dpromotion1}) to ({pic cs:cpromotion1});
\draw [->] ({pic cs:dpromotion2}) [bend left] to ({pic cs:cpromotion2});
\draw [->] ({pic cs:dpromotion3}) [bend right] to ({pic cs:cpromotion3});
\draw [->] ({pic cs:dpromotion4}) to ({pic cs:cpromotion4});
\draw [->] ({pic cs:dpromotion5}) [bend left] to ({pic cs:cpromotion5});
\draw [->] ({pic cs:dpromotion6}) [bend right] to ({pic cs:cpromotion6});
\draw [->] ({pic cs:dpromotion64}) to ({pic cs:cpromotion64});
\draw [->] ({pic cs:dpromotion65}) [bend left] to ({pic cs:cpromotion65});
\draw [->] ({pic cs:dpromotion66}) [bend right] to ({pic cs:cpromotion66});
\end{tikzpicture}
\end{center}
where $i, j, k, n, m, l \in \mathbb{N}$ such that $i \neq j$, $i \neq k$ and $j \neq k$, and they are all selected by Opponent.
Note that $\mathit{succ}^\dagger$ consistently plays as $\mathit{succ}$ for each thread.
\end{example}

\begin{lemma}[Well-defined promotion of strategies]
\label{LemWellDefinedPromotionOfStrategies}
Given games $A$ and $B$, and a strategy $\varphi : \ !A \multimap B$, the promotion $\varphi^{\dagger}$ is a strategy on $!A \multimap \ !B$. 
If $\varphi$ is innocent (resp. wb, total, noetherian), then so is $\varphi^\dagger$.
Given $\tilde{\varphi} : \ !A \multimap B$ with $\varphi \simeq_{!A \multimap B} \tilde{\varphi}$, $\varphi^\dagger \simeq_{!A \multimap !B} \tilde{\varphi}^\dagger$.
\end{lemma}
\begin{proof}
Straightforward; see \cite{mccusker1998games,abramsky2000full}.
\end{proof}

We proceed to recall a simple kind strategies, which are $\beta$-identities of our game-semantic CCBoC given in Section~\ref{DynamicGameSemantics}:
\begin{definition}[Derelictions \cite{abramsky2000full,mccusker1998games}]
\label{DefDerelictions}
The \emph{\bfseries dereliction} $\mathit{der}_A : \ ! A \multimap A$ on a normalized game $A$ is defined by:
\begin{equation*}
\mathit{der}_A \stackrel{\mathrm{df. }}{=} \{ \boldsymbol{s} \in P_{!A \multimap A}^{\mathsf{Even}} \mid \forall \boldsymbol{t} \preceq \boldsymbol{s} . \ \! \mathsf{Even}(\boldsymbol{t}) \Rightarrow (\boldsymbol{t} \upharpoonright \ !A) \upharpoonright 0 = \boldsymbol{t} \upharpoonright A \ \! \}.
\end{equation*}
\end{definition}

Note that any `tag' $(\_, i)$ such that $i \in \mathbb{N}$ would work; our choice $(\_, 0)$ does not matter.
\begin{lemma}[Well-defined derelictions]
\label{LemWellDefinedDerelictions}
Given a normalized game $A$, $\mathit{der}_A$ is a valid, innocent, wb, total strategy on $!A \multimap A$. 
It is noetherian if $A$ is well-founded.
\end{lemma}
\begin{proof}
We just show that $\mathit{der}_A$ is noetherian if $A$ is well-founded for the other points are trivial, e.g., validity of $\mathit{der}_A$ is immediate from the definition of $\simeq_{!A \multimap A}$. 
Given $\boldsymbol{s}mm \in \mathit{der}_A$, it is easy to see by induction on $|\boldsymbol{s}|$ that  the P-view $\lceil \boldsymbol{s} m \rceil$ is of the form $m_1 m_1 m_2 m_2 \dots m_k m_k m$, and thus there is a sequence $\star \vdash_A m_1 \vdash_A m_2 \dots \vdash_A m_k \vdash_A m$ of enabling pairs. 
Therefore, if $A$ is well-founded, then $\mathit{der}_A$ must be noetherian.
\end{proof}

Let us proceed to introduce some generalizations of existing constructions.
Note that tensor, pairing and promotion of static strategies have been already generalized slightly because they allow non-normalized dynamic games and strategies.
However, for the game-semantic CCBoC in Section~\ref{DynamicGameSemantics}, we need further generalizations:
\begin{definition}[Generalized pairing of strategies]
\label{DefPairingOfDynamicStrategies}
Given strategies $\phi : L$ and $\psi : R$ such that $\mathcal{H}^\omega(L) \trianglelefteqslant C \multimap A$ and $\mathcal{H}^\omega(R) \trianglelefteqslant C \multimap B$ for some normalized games $A$, $B$ and $C$, the \emph{\bfseries (generalized) pairing} $\langle \phi, \psi \rangle$ of $\phi$ and $\psi$ is defined by:
\begin{equation*}
\langle \phi, \psi \rangle \stackrel{\mathrm{df. }}{=} \{ \boldsymbol{s} \in \mathscr{L}_{\langle L, R \rangle} \mid (\boldsymbol{s} \upharpoonright L \in \phi \wedge \boldsymbol{s} \upharpoonright R = \boldsymbol{\epsilon}) \vee (\boldsymbol{s} \upharpoonright R \in \psi \wedge \boldsymbol{s} \upharpoonright L = \boldsymbol{\epsilon}) \ \! \}.
\end{equation*}
\end{definition}

\begin{theorem}[Well-defined generalized pairing of strategies]
\label{ThmWellDefinedPairingOfDynamicStrategies}
Given strategies $\phi : L$ and $\psi : R$ such that $\mathcal{H}^\omega(L) \trianglelefteqslant C \multimap A$ and $\mathcal{H}^\omega(R) \trianglelefteqslant C \multimap B$ for some normalized games $A$, $B$ and $C$, $\langle \phi, \psi \rangle$ is a strategy on $\langle L, R \rangle$.
If $\phi$ and $\psi$ are innocent (resp. wb, total, noetherian), then so is $\langle \phi, \psi \rangle$.
Given $\phi' : L$ and $\psi' : R$ such that $\phi \simeq_L \phi'$ and $\psi \simeq_R \psi'$, we have $\langle \phi, \psi \rangle \simeq_{\langle L, R \rangle} \langle \phi', \psi' \rangle$.
\end{theorem}
\begin{proof}
Straightforward. 
\end{proof}

\begin{convention*}
Henceforth, \emph{\bfseries pairing of strategies} refers to the generalized one. 
\end{convention*}

\begin{definition}[Generalized promotion of strategies]
\label{DefPromotionOfDynamicStrategies}
Given a strategy $\varphi : G$ such that $\mathcal{H}^\omega(G) \trianglelefteqslant \ !A \multimap B$ for some normalized games $A$ and $B$, the \emph{\bfseries (generalized) promotion} $\varphi^\dagger$ of $\varphi$ is defined by:
\begin{equation*}
\varphi^\dagger \stackrel{\mathrm{df. }}{=} \{ \boldsymbol{s} \in \mathscr{L}_{G^\dagger} \mid \forall i \in \mathbb{N} . \ \! \boldsymbol{s} \upharpoonright i \in \varphi \ \! \}.
\end{equation*}
\end{definition}

\begin{theorem}[Well-defined generalized promotion on strategies]
\label{ThmWellDefinedPromotionOnDynamicStrategies}
Given a strategy $\varphi : G$ such that $\mathcal{H}^\omega(G) \trianglelefteqslant \ !A \multimap B$ for some normalized games $A$ and $B$, $\varphi^\dagger$ is a strategy on $G^\dagger$.
If $\varphi$ is innocent (resp. wb, total, noetherian), then so is $\varphi^\dagger$.
Given $\varphi' : G$ such that $\varphi \simeq_G \varphi'$, we have $\varphi^\dagger \simeq_{G^\dagger} \varphi'^\dagger$.
\end{theorem}
\begin{proof}
Straightforward. 
\end{proof}

\begin{convention*}
Henceforth, \emph{\bfseries promotion of strategies} refers to the generalized one. 
\end{convention*}

Next, let us introduce a new construction on strategies, which plays a fundamental role in the present work:
\begin{definition}[Concatenation of strategies]
\label{DefConcatenationOnDynamicStrategies}
Let $\iota : J$ and $\kappa : K$ be strategies such that $\mathcal{H}^\omega(J) \trianglelefteqslant A \multimap B$ and $\mathcal{H}^\omega(K) \trianglelefteqslant B \multimap C$ for some normalized games $A$, $B$ and $C$. 
The \emph{\bfseries concatenation} $\iota \ddagger \kappa$ of $\iota$ and $\kappa$ is defined by: 
\begin{equation*}
\iota \ddagger \kappa \stackrel{\mathrm{df. }}{=} \{ \boldsymbol{s} \in \mathscr{J}_{J \ddagger K} \mid \boldsymbol{s} \upharpoonright J \in \iota, \boldsymbol{s} \upharpoonright K \in \kappa, \boldsymbol{s} \upharpoonright B_{[1]}, B_{[2]} \in \mathit{pr}_B \ \! \}.
\end{equation*}
\end{definition}

\begin{theorem}[Well-defined concatenation of strategies]
\label{ThmWellDefinedConcatenationOnDynamicStrategies}
Let $\iota : J$ and $\kappa : K$ be strategies such that $\mathcal{H}^\omega(J) \trianglelefteqslant A \multimap B$ and $\mathcal{H}^\omega(K) \trianglelefteqslant B \multimap C$, where $A$, $B$ and $C$ are normalized games. 
Then, $\iota \ddagger \kappa : J \ddagger K$ and $\mathcal{H}^\omega(\iota) ; \mathcal{H}^\omega(\kappa) = \mathcal{H}^\omega(\iota \ddagger \kappa) : A \multimap C$, where $\mathcal{H}^\omega(\iota) ; \mathcal{H}^\omega(\kappa)$ is the \emph{composition} of $\mathcal{H}^\omega(\iota) : A \multimap B$ and $\mathcal{H}^\omega(\kappa) : B \multimap C$ \cite{abramsky1999game}.
If $\iota$ and $\kappa$ are innocent (resp. wb, noetherian, winning), then so is $\iota \ddagger \kappa$.
Given $\iota' : J$ and $\kappa' : K$ with $\iota \simeq_J \iota'$ and $\kappa \simeq_K \kappa'$, we have $\iota \ddagger \kappa \simeq_{J \ddagger K} \iota' \ddagger \kappa'$.
\end{theorem}

\begin{proof}
We just show the first statement as the other ones are straightforward. 
It then suffices to prove $\iota \ddagger \kappa : J \ddagger K$ and $\mathcal{H}^\omega(\iota \ddagger \kappa) = \iota ; \kappa$ since it implies $\iota ; \kappa = \mathcal{H}^\omega(\iota \ddagger \kappa) : \mathcal{H}^\omega(J \ddagger K) \trianglelefteqslant A \multimap C$ by Lemmata~\ref{LemHidingLemmaOnDynamicGames} and \ref{ThmHidingTheorem}.
However, $\mathcal{H}^\omega(\iota \ddagger \kappa) = \iota ; \kappa$ is immediate from the definition of concatenation; thus, we focus on $\iota \ddagger \kappa : J \ddagger K$.

First, we have $\iota \ddagger \kappa \subseteq P_{J \ddagger K}$ as any $\boldsymbol{s} \in \iota \ddagger \kappa$ satisfies $\boldsymbol{s} \in \mathscr{J}_{J \ddagger K}$, $\boldsymbol{s} \upharpoonright J \in \iota \subseteq P_J$, $\boldsymbol{s} \upharpoonright K \in \kappa \subseteq P_K$ and $\boldsymbol{s} \upharpoonright B_{[1]}, B_{[2]} \in \mathit{pr}_B$. It is also immediate that such $\boldsymbol{s}$ is of even-length. 
It remains to verify the axioms S1 and S2. For this, we need:
\begin{center}
($\diamondsuit$) Each $\boldsymbol{s} \in \iota \ddagger \kappa$ consists of adjacent pairs $mn$ such that $m, n \in M_J$ or $m, n \in M_K$.
\end{center}
\begin{proof}[Proof of the claim $\diamondsuit$] By induction on $|\boldsymbol{s}|$. The base case is trivial. For the inductive step, let $\boldsymbol{s} m n \in \iota \ddagger \kappa$. If $m \in M_J$, then $(\boldsymbol{s} \upharpoonright J) . m . (n \upharpoonright J) \in \sigma$, where $\boldsymbol{s} \upharpoonright J$ is of even-length by the induction hypothesis. Thus, we must have $n \in M_J$. If $m \in M_K$, then $n \in M_K$ by the same argument.
\end{proof}

\begin{itemize}

\item \textsc{(S1).} Since $\boldsymbol{\epsilon} \in \iota \ddagger \kappa$, we have $\iota \ddagger \kappa \neq \emptyset$. For even-prefix-closure, assume $\boldsymbol{s} m n \in \iota \ddagger \kappa$. By the claim $\diamondsuit$, either $m, n \in M_J$ or $m, n \in M_K$. In either case, it is straightforward to see that $\boldsymbol{s} \in P_{J \ddagger K}$, $\boldsymbol{s} \upharpoonright J \in \iota$, $\boldsymbol{s} \upharpoonright K \in \kappa$ and $\boldsymbol{s} \upharpoonright B_{[1]}, B_{[2]} \in \mathit{pr}_B$, i.e., $\boldsymbol{s} \in \iota \ddagger \kappa$.

\item \textsc{(S2).} Assume $\boldsymbol{s} m n, \boldsymbol{s} m n' \in \iota \ddagger \kappa$. 
By the claim $\diamondsuit$, either $m, n, n' \in M_J$ or $m, n, n' \in M_K$. In the former case, 
$(\boldsymbol{s} \upharpoonright J) . m n, (\boldsymbol{s} \upharpoonright J) . m n' \in \iota$.
Thus, $n = n'$ and $\mathcal{J}_{\boldsymbol{s}mn}(n) = \mathcal{J}_{(\boldsymbol{s} \upharpoonright J).mn}(n) =\mathcal{J}_{(\boldsymbol{s} \upharpoonright J).mn'}(n') = \mathcal{J}_{\boldsymbol{s}mn'}(n')$ by S2 on $\iota$, where note that $n$ and $n'$ are both P-moves and thus non-initial in $J$. The latter case may be handled similarly.
\end{itemize}
Therefore, we have shown that $\iota \ddagger \kappa : J \ddagger K$.
\end{proof}

Note that totality of (dynamic) strategies is \emph{not} preserved under composition, but it is preserved under concatenation.
This phenomenon is essentially because totality is not preserved under the hiding operation as already remarked above.

For completeness, let us explicitly define the rather trivial \emph{currying} of strategies: 
\begin{definition}[Currying of strategies]
\label{DefCurryingOfDynamicStrategies}
Given $\sigma : G$ with $\mathcal{H}^\omega(G) \trianglelefteqslant A \otimes B \multimap C$ for some normalized  games $A$, $B$ and $C$, the \emph{\bfseries currying} $\Lambda(\sigma) : \Lambda(G)$ of $\sigma$ is $\sigma$ up to `tags'.
\end{definition}

\begin{proposition}[Well-defined currying of strategies]
\label{PropWellDefinedCurryingOnDynamicStrategies}
Strategies are closed under currying, and currying preserves totality, innocence, well-bracketing, noetherianity and identification of strategies. 
\end{proposition}
\begin{proof}
Obvious. 
\end{proof}

Now, as in the case of games, we establish the \emph{hiding lemma} on strategies (Lemma~\ref{LemHidingLemmaOnDynamicStrategies}).
We first need the following:
\begin{lemma}[Hiding on legal positions in the second form]
\label{LemHidingOnLegalPositionsInSecondForm}
For any arena $G$ and number $d \in \mathbb{N} \cup \{ \omega \}$, we have $\mathscr{L}_{\mathcal{H}^d(G)} = \{ \boldsymbol{s} \natural \mathcal{H}_G^d \mid \boldsymbol{s} \in \mathscr{L}_G \ \! \}$.
\end{lemma}
\begin{proof}
Observe that:
\begin{align*}
\{ \boldsymbol{s} \natural \mathcal{H}_G^d \mid \boldsymbol{s} \in \mathscr{L}_G \ \! \} &= \{ \boldsymbol{s} \natural \mathcal{H}_G^d \mid \boldsymbol{s} \in \mathscr{L}_G, \text{$\boldsymbol{s}$ is $d$-complete} \ \! \} \\ 
&= \{ \mathcal{H}_G^d(\boldsymbol{s}) \mid \boldsymbol{s} \in \mathscr{L}_G, \text{$\boldsymbol{s}$ is $d$-complete} \ \! \} \\
&= \{ \mathcal{H}_G^d(\boldsymbol{s}) \mid \boldsymbol{s} \in \mathscr{L}_G \ \! \} \ \text{(by the same argument as above)} \\
&= \mathscr{L}_{\mathcal{H}^d(G)} \ \text{(by Corollary~\ref{CoroHidingOperationOnDynamicLegalPositions})}
\end{align*}
completing the proof. 
\end{proof}

\begin{notation*}
We write $\spadesuit_{i \in I}$, where $I$ is $\{ 1 \}$ or $\{ 1, 2 \}$, for any of the constructions on strategies introduced so far, i.e., $\spadesuit_{i \in I}$ is either $\otimes$, $(\_)^\dagger$,  $\langle \_, \_ \rangle$, $\ddagger$, $;$ or $\Lambda$. 
\end{notation*}

\begin{lemma}[Hiding lemma on strategies]
\label{LemHidingLemmaOnDynamicStrategies}
Let $\spadesuit_{i \in I}$ be a construction on strategies, and $\sigma_i : G_i$ for each $i \in I$. 
Then, for all $d \in \mathbb{N} \cup \{ \omega \}$, we have:
\begin{enumerate}

\item $\mathcal{H}^d(\spadesuit_{i \in I} \sigma_i) = \spadesuit_{i \in I} \mathcal{H}^d(\sigma_i)$ if $\spadesuit_{i \in I}$ is $\otimes$, $(\_)^\dagger$, $\langle \_, \_ \rangle$ or $\Lambda$;

\item $\mathcal{H}^d(\sigma_1 \ddagger \sigma_2) = \mathcal{H}^d(\sigma_1) \ddagger \mathcal{H}^d(\sigma_2)$ if $\mathcal{H}^d(\sigma_1 \ddagger \sigma_2)$ is not normalized; 

\item $\mathcal{H}^d(\sigma_1 \ddagger \sigma_2) = \mathcal{H}^d(\sigma_1) ; \mathcal{H}^d(\sigma_2)$ otherwise.

\end{enumerate}
\end{lemma}

\begin{proof}
As in the case of games, it suffices to assume $d = 1$.
Here, we just focus on pairing since the other constructions may be handled analogously.

Let $\sigma_i : G_i$, $i = 1, 2$, be strategies such that $\mathcal{H}^\omega(G_1) \trianglelefteqslant C \multimap A$, $\mathcal{H}^\omega(G_2) \trianglelefteqslant C \multimap B$ for some normalized games $A$, $B$ and $C$. 
For $\mathcal{H}(\langle \sigma_1, \sigma_2 \rangle) \subseteq \langle \mathcal{H}(\sigma_1), \mathcal{H}(\sigma_2) \rangle$, observe that:
\begin{align*}
\boldsymbol{s} \in \mathcal{H}(\langle \sigma_1, \sigma_2 \rangle) &\Rightarrow \exists \boldsymbol{t} \in \langle \sigma_1, \sigma_2 \rangle . \ \! \boldsymbol{t} \natural \mathcal{H}^1_{\langle G_1, G_2 \rangle} = \boldsymbol{s} \\
&\Rightarrow \exists \boldsymbol{t} \in \mathscr{L}_{\langle G_1, G_2 \rangle} . \ \! \boldsymbol{t} \natural \mathcal{H}^1_{\langle G_1, G_2 \rangle} = \boldsymbol{s} \wedge ((\boldsymbol{t} \upharpoonright G_1 \in \sigma_1 \wedge \boldsymbol{t} \upharpoonright G_2 = \boldsymbol{\epsilon}) \vee (\boldsymbol{t} \upharpoonright G_2 \in \sigma_2 \wedge \boldsymbol{t} \upharpoonright G_1 = \boldsymbol{\epsilon})) \\
&\Rightarrow \boldsymbol{s} \in \mathscr{L}_{\mathcal{H}(\langle G_1, G_2 \rangle)} \wedge (\boldsymbol{s} \upharpoonright \mathcal{H}(G_1) \in \mathcal{H}(\sigma_1) \wedge \boldsymbol{s} \upharpoonright \mathcal{H}(G_2) = \boldsymbol{\epsilon}) \\ & \ \ \ \ \vee (\boldsymbol{s} \upharpoonright \mathcal{H}(G_2) \in \mathcal{H}(\sigma_2) \wedge \boldsymbol{s} \upharpoonright \mathcal{H}(G_1) = \boldsymbol{\epsilon})) \ \text{(by Lemma~\ref{LemHidingOnLegalPositionsInSecondForm})} \\
&\Rightarrow \boldsymbol{s} \in \langle \mathcal{H}(\sigma_1), \mathcal{H}(\sigma_2) \rangle.
\end{align*}
Next, we show the converse:
\begin{align*}
\boldsymbol{s} \in \langle \mathcal{H}(\sigma_1), \mathcal{H}(\sigma_2) \rangle &\Rightarrow \boldsymbol{s} \in \mathscr{L}_{\mathcal{H}(\langle G_1, G_2 \rangle)} \wedge (\boldsymbol{s} \upharpoonright \mathcal{H}(G_1) \in \mathcal{H}(\sigma_1) \wedge \boldsymbol{s} \upharpoonright \mathcal{H}(G_2) = \boldsymbol{\epsilon}) \\ & \ \ \ \ \vee (\boldsymbol{s} \upharpoonright \mathcal{H}(G_2) \in \mathcal{H}(\sigma_2) \wedge \boldsymbol{s} \upharpoonright \mathcal{H}(G_1) = \boldsymbol{\epsilon})) \\
&\Rightarrow (\exists \boldsymbol{u} \in \sigma_1 . \ \! \boldsymbol{u} \natural \mathcal{H}_{G_1}^1 = \boldsymbol{s} \upharpoonright \mathcal{H}(G_1) \wedge \boldsymbol{u} \upharpoonright G_2 = \boldsymbol{\epsilon}) \\ & \ \ \ \ \vee (\exists \boldsymbol{v} \in \sigma_2 . \ \! \boldsymbol{v} \natural \mathcal{H}_{G_2}^1 = \boldsymbol{s} \upharpoonright \mathcal{H}(G_2) \wedge \boldsymbol{v} \upharpoonright \mathcal{H}(G_1) = \boldsymbol{\epsilon}) \\
&\Rightarrow \exists \boldsymbol{w} \in \langle \sigma_1, \sigma_2 \rangle . \ \! \boldsymbol{w} \natural \mathcal{H}_{\langle G_1, G_2 \rangle}^1 = \boldsymbol{s} \\
&\Rightarrow \boldsymbol{s} \in \mathcal{H}(\langle \sigma_1, \sigma_2 \rangle)
\end{align*}
which completes the proof. 
\end{proof}


Finally, as a technical preparation for the next section, let us define:
\begin{definition}[Dereliction games]
\label{DefDerelictionGames}
The \emph{\bfseries dereliction game} on a game $G$ is the subgame $\varXi_G \trianglelefteqslant G \Rightarrow G$ given by $M_{\varXi_G} \stackrel{\mathrm{df. }}{=} M_{G \Rightarrow G}$, $\lambda_{\varXi_G} \stackrel{\mathrm{df. }}{=} \lambda_{G \Rightarrow G}$, $\vdash_{\varXi_G} \stackrel{\mathrm{df. }}{=} \ \! \vdash_{G \Rightarrow G}$, $P_{\varXi_G} \stackrel{\mathrm{df. }}{=} \{ \boldsymbol{s} \in P_{G_{[0]} \Rightarrow G_{[1]}} \mid \forall \boldsymbol{t} \preceq \boldsymbol{s} . \ \! \mathsf{Even}(\boldsymbol{t}) \Rightarrow \boldsymbol{t} \upharpoonright G_{[0]} = \boldsymbol{t} \upharpoonright G_{[1]} \}$, and $\simeq_{\varXi_G} \stackrel{\mathrm{df. }}{=} \ \! \simeq_{G \Rightarrow G} \ \upharpoonright P_{\varXi_G} \times P_{\varXi_G}$.
Given normalized games $A$ and $B$, we define:
\begin{itemize}

\item $\varPi^{A, B}_1 \trianglelefteqslant A \& B \Rightarrow A$ to be $\varXi_A$ up to `tags', where we often abbreviate it as $\varPi_1$;

\item $\varPi^{A, B}_2 \trianglelefteqslant A \& B \Rightarrow B$ to be $\varXi_B$ up to `tags', where we often abbreviate it as $\varPi_2$;

\item $\varUpsilon_{A, B} \trianglelefteqslant B^A \& A \Rightarrow B$ to be $\varXi_{A \Rightarrow B}$ up to `tags', where we often abbreviate it as $\varUpsilon$.

\end{itemize}
\end{definition}

That is, the dereliction game $\varXi_G$ on a game $G$ is the subgame of $G \Rightarrow G$, in which only plays by the dereliction $\mathit{der}_G$ are possible.

\begin{lemma}[D-lemma]
\label{LemDerelictionLemma}
Given normalized games $A$, $B$, $C$, $L \trianglelefteqslant C \Rightarrow A$, $R \trianglelefteqslant C \Rightarrow B$, $P \trianglelefteqslant C \Rightarrow A \& B$, $U \trianglelefteqslant A \& B \Rightarrow C$ and $V \trianglelefteqslant A \Rightarrow C^B$, we have:
\begin{align*}
\langle L, R \rangle^\dagger ; \varPi^{A, B}_1 &= L \\
\langle L, R \rangle^\dagger ; \varPi^{A, B}_2 &= R \\
\langle P^\dagger ; \varPi^{A, B}_1, P^\dagger ; \varPi^{A, B}_2 \rangle &= P \\
\langle (\varPi^{A, B}_1)^\dagger ; \Lambda(U), \varPi^{A, B}_2 \rangle^\dagger ; \varUpsilon_{B, C} &= U \\
\Lambda(\langle (\varPi^{A, B}_1)^\dagger ; V, \varPi^{A, B}_2 \rangle^\dagger ; \varUpsilon_{B, C}) &= V.
\end{align*}
\end{lemma}
\begin{proof}
Straightforward.
\end{proof}

\section{Dynamic Game Semantics of Finitary PCF}
\label{DynamicGameSemantics}
This section is the climax of the present work. 
We first define a game-semantic CCBoC $\mathcal{LDG}$ (Definition~\ref{DefCCBoCCDG}) and a standard structure $\mathcal{S_G}$ for FPCF in $\mathcal{LDG}$ (Definition~\ref{DefGameSemanticStructureForFPCF}) in Section~\ref{DynamicGameSemanticsOfFPCF}. 
Then, as the main result, we show that the induced interpretation $\llbracket \_ \rrbracket_{\mathcal{LDG}}^{\mathcal{S_G}}$ satisfies the PDCP (Theorem~\ref{ThmDynamicTheorem}), and thus the DCP by Theorem~\ref{ThmDynamicSemanticsOfSystemTVartheta}, in Section~\ref{MainResult}, giving the first instance of dynamic game semantics.

\subsection{Dynamic Game Semantics of Finitary PCF}
\label{DynamicGameSemanticsOfFPCF}
Let us give the CCBoC $\mathcal{LDG}$ of dynamic games and strategies:
\begin{definition}[The CCBoC $\boldsymbol{\mathcal{LDG}}$]
\label{DefCCBoCCDG}
The CCBoC $\mathcal{LDG} = (\mathcal{LDG}, \mathcal{H})$ is defined by:
\begin{itemize}

\item Objects are normalized, well-founded games;

\item A $\beta$-morphisms $A \to B$ is a pair $(J, [\phi]_{\mathsf{W}})$ of a game $J$ such that $\mathcal{H}^\omega(J) \trianglelefteqslant A \Rightarrow B$ and the equivalence class $[\phi]_{\mathsf{W}} \stackrel{\mathrm{df. }}{=} \{ \psi : J \mid \text{$\psi$ is winning}, \psi \simeq_J \phi \ \! \}$ of a valid, winning strategy $\phi : J$; 

\item The $\beta$-composition $A \stackrel{(J, [\phi]_{\mathsf{W}})}{\to} B \stackrel{(K, [\psi]_{\mathsf{W}})}{\to} C$ is the pair $(J^\dagger \ddagger K, [\phi^\dagger \ddagger \psi]_{\mathsf{W}})$;

\item The $\beta$-identity $\mathit{id}_A : A \to A$ on each object $A$ is the pair $(\varXi_A, [\mathit{der}_A]_{\mathsf{W}})$;

\item The evaluation $\mathcal{H}$ maps morphisms $(J, [\phi]_{\mathsf{W}}) : A \to B$ to $\mathcal{H}(J, [\phi]_{\mathsf{W}}) \stackrel{\mathrm{df. }}{=} (\mathcal{H}(J), [\mathcal{H}(\phi)]_{\mathsf{W}})$;

\item The $\beta$-terminal object is the terminal game $T$ (Example~\ref{ExTerminalGame});

\item $\beta$-product and $\beta$-exponential are respectively given by $A \times B \stackrel{\mathrm{df. }}{=} A \& B$ and $B^A \stackrel{\mathrm{df. }}{=} A \Rightarrow B = \ !A \multimap B$ for any objects $A, B \in \mathcal{LDG}$;

\item $\beta$-pairing is given by $\langle (L, [\alpha]_{\mathsf{W}}), (R, [\beta]_{\mathsf{W}}) \rangle \stackrel{\mathrm{df. }}{=} (\langle L, R \rangle, [\langle \alpha, \beta \rangle]_{\mathsf{W}}) : C \to A \& B$ for any objects $A, B, C \in \mathcal{LDG}$, and morphisms $(L, [\alpha]_{\mathsf{W}}) : C \to A$ and $(R, [\beta]_{\mathsf{W}}) : C \to B$;

\item The $\beta$-projections $\pi_1 : A \& B \to A$ and $\pi_2 : A \& B \to B$ are respectively the pairs $(\varPi^{A, B}_1, [\varpi^{A, B}_1]_{\mathsf{W}})$ and $(\varPi^{A, B}_2, [\varpi^{A, B}_2]_{\mathsf{W}})$ for any objects $A, B \in \mathcal{LDG}$, where $\varpi^{A, B}_1 : \varPi^{A, B}_1$ and $\varpi^{A, B}_2 : \varPi^{A, B}_2$ are respectively the derelictions $\mathit{der}_A$ and $\mathit{der}_B$ up to `tags';

\item $\beta$-currying is given by $\Lambda(G, [\varphi]_{\mathsf{W}}) \stackrel{\mathrm{df. }}{=} (\Lambda(G), [\Lambda(\varphi)]_{\mathsf{W}}) : A \to (B \Rightarrow C)$ for any objects $A, B, C \in \mathcal{LDG}$, and morphism $(G, [\varphi]_{\mathsf{W}}) : A \& B \to C$; 

\item The $\beta$-evaluation $\mathit{ev}_{B, C} : C^B \& B \to C$ for any objects $B, C \in \mathcal{LDG}$ is the pair $(\varUpsilon_{B, C}, [\upsilon_{B, C}]_{\mathsf{W}})$, where $\upsilon_{B, C} : \varUpsilon_{B, C}$ is the dereliction $\mathit{der}_{B \Rightarrow C}$ up to `tags'.

\end{itemize}
\end{definition}

Note that we have made the \emph{underlying game} of each $\beta$-morphism in $\mathcal{LDG}$ explicit in order to take the equivalence class of strategies.
Also, we have focused on \emph{well-founded} games and \emph{winning} strategies for the full completeness result (Corollary~\ref{CoroDynamicFullCompleteness}), where note that games must be well-founded for derelictions to be noetherian (Lemma~\ref{LemWellDefinedDerelictions}).

\begin{theorem}[Well-defined $\boldsymbol{\mathcal{LDG}}$]
\label{ThmLDG}
The structure $\mathcal{LDG}$ forms a CCBoC.
\end{theorem}
\begin{proof}
First, for $\beta$-composition, let $A, B, C \in \mathcal{LDG}$, $(J, [\phi]_{\mathsf{W}}) : A \to B$ and $(K, [\psi]_{\mathsf{W}}) : B \to C$ in $\mathcal{LDG}$. 
Then, $\phi^\dagger : J^\dagger$ by Theorem~\ref{ThmWellDefinedPromotionOnDynamicStrategies}, and $\mathcal{H}^\omega(J^\dagger) \trianglelefteqslant \ !A \multimap \ !B$ by Theorem~\ref{ThmWellDefinedPromotionOnDynamicGames}; thus, we may form $\phi^\dagger \ddagger \psi : J^\dagger \ddagger K$ such that $\mathcal{H}^\omega(J^\dagger \ddagger K) \trianglelefteqslant A \Rightarrow C$ by Theorem~\ref{ThmWellDefinedConcatenationOnDynamicStrategies}.
Also, promotion and concatenation both preserve validity and winning of strategies (by Theorems~\ref{ThmWellDefinedPromotionOnDynamicStrategies} and \ref{ThmWellDefinedConcatenationOnDynamicStrategies}). 
Hence, the pair $(J^\dagger \ddagger K, [\phi^\dagger \ddagger \psi]_{\mathsf{W}})$ is a $\beta$-morphism $A \to C$ in $\mathcal{LDG}$.
Note that the composition does not depend on the representatives $\phi$ and $\psi$.

Moreover, $\beta$-composition preserves $\simeq$: For any $A, B, C \in \mathcal{LDG}$, $(J, [\iota]_{\mathsf{W}}), (\tilde{J}, [\tilde{\iota}]_{\mathsf{W}}) : A \rightarrow B$ and $(K, [\kappa]_{\mathsf{W}}), (\tilde{K}, [\tilde{\kappa}]_{\mathsf{W}}) : B \rightarrow C$ in $\mathcal{LDG}$, if $\mathcal{H}^\omega(J, [\iota]_{\mathsf{W}}) = \mathcal{H}^\omega(\tilde{J}, [\tilde{\iota}]_{\mathsf{W}})$ and $\mathcal{H}^\omega(K, [\kappa]_{\mathsf{W}}) = \mathcal{H}^\omega(\tilde{K}, [\tilde{\kappa}]_{\mathsf{W}})$, then $\mathcal{H}^\omega(J^\dagger \ddagger K) = \mathcal{H}^\omega(J)^\dagger ; \mathcal{H}^\omega(K) = \mathcal{H}^\omega(\tilde{J})^\dagger ; \mathcal{H}^\omega(\tilde{K}) = \mathcal{H}^\omega(\tilde{J}^\dagger \ddagger \tilde{K})$ by Lemma~\ref{LemHidingLemmaOnDynamicGames}, and $\mathcal{H}^\omega(\iota^\dagger \ddagger \kappa) \simeq_{\mathcal{H}^\omega(J^\dagger \ddagger K)} \mathcal{H}^\omega(\tilde{\iota}^\dagger \ddagger \tilde{\kappa})$ by Corollary~\ref{CoroPreservationOfConstraintsOnDynamicStrategiesUnderHiding}, whence $\mathcal{H}^\omega(J^\dagger \ddagger K, [\iota^\dagger \ddagger \kappa]_{\mathsf{W}}) = \mathcal{H}^\omega(\tilde{J}^\dagger \ddagger \tilde{K}, [\tilde{\iota}^\dagger \ddagger \tilde{\kappa}]_{\mathsf{W}})$.

Then clearly, associativity of $\beta$-composition up to $\simeq$ holds: Given $D \in \mathcal{LDG}$, and $(G, [\varphi]) : C \to D$ in $\mathcal{LDG}$, by Lemma~\ref{LemHidingLemmaOnDynamicGames} we have: 
\begin{align*}
\mathcal{H}^\omega((J^\dagger \ddagger K)^\dagger \ddagger G) &= (\mathcal{H}^\omega(J^\dagger) ; \mathcal{H}^\omega(K))^\dagger ; \mathcal{H}^\omega(G) \\
&= (\mathcal{H}^\omega(J)^\dagger ; \mathcal{H}^\omega(K)^\dagger) ; \mathcal{H}^\omega(G) \\
&= \mathcal{H}^\omega(J)^\dagger ; (\mathcal{H}^\omega(K)^\dagger ; \mathcal{H}^\omega(G)) \\
&= \mathcal{H}^\omega(J^\dagger) ; (\mathcal{H}^\omega(K^\dagger) ; \mathcal{H}^\omega(G)) \\
&= \mathcal{H}^\omega(J^\dagger \ddagger (K^\dagger \ddagger G))
\end{align*}
as well as by Lemma~\ref{LemHidingLemmaOnDynamicStrategies}:
\begin{align*}
 \mathcal{H}^\omega((\phi^\dagger \ddagger \psi)^\dagger \ddagger \varphi) &= (\mathcal{H}^\omega(\phi^\dagger) ; \mathcal{H}^\omega(\psi))^\dagger ; \mathcal{H}^\omega(\varphi) \\
&= (\mathcal{H}^\omega(\phi^\dagger) ; \mathcal{H}^\omega(\psi)^\dagger) ; \mathcal{H}^\omega(\varphi) \\
&= (\mathcal{H}^\omega(\phi^\dagger) ; (\mathcal{H}^\omega(\psi^\dagger) ; \mathcal{H}^\omega(\varphi)) \\
&= \mathcal{H}^\omega(\phi^\dagger \ddagger (\psi^\dagger \ddagger \varphi)) 
\end{align*}
whence $((J, [\phi]_{\mathsf{W}}) ; (K, [\psi]_{\mathsf{W}})) ; (G, [\varphi]_{\mathsf{W}}) \simeq (J, [\phi]_{\mathsf{W}}) ; ((K, [\psi]_{\mathsf{W}}) ; (G, [\varphi]_{\mathsf{W}}))$.

Similarly, unit law up to $\simeq$ holds; we leave the details to the reader. 

Also, $\mathcal{H}$ clearly satisfies the four axioms of BoC (Definition~\ref{DefBoCs}), having shown that $\mathcal{LDG}$ is a BoC.
It remains to verify its cartesian closed structure up to $\simeq$.

The universal property of the $\beta$-terminal game $T$ up to $\simeq$ is obvious, where we define $!_A \stackrel{\mathrm{df. }}{=} (A \Rightarrow T, [\{ \boldsymbol{\epsilon} \}]_{\mathsf{W}}) : A \rightarrow T$ for each $A \in \mathcal{LDG}$.
The $\beta$-projections are clearly values in $\mathcal{LDG}$.
Given $\beta$-morphisms $(L, [\alpha]_{\mathsf{W}}) : C \to A$ and $(R, [\beta]_{\mathsf{W}}) : C \rightarrow B$ in $\mathcal{LDG}$, i.e., $\alpha : L$, $\beta : R$, $\mathcal{H}^\omega(L) \trianglelefteqslant C \Rightarrow A$ and $\mathcal{H}^\omega(R) \trianglelefteqslant C \Rightarrow B$, we may obtain the valid, winning pairing $\langle \alpha, \beta \rangle : \langle L, R \rangle$ such that $\mathcal{H}^\omega(\langle L, R \rangle) \trianglelefteqslant C \Rightarrow A \& B$ by Theorem~\ref{ThmWellDefinedPairingOnDynamicGames}.
Hence, the pair $(\langle L, R \rangle, [\langle \alpha, \beta \rangle]_{\mathsf{W}})$ is a $\beta$-morphism $C \rightarrow A \& B$ in $\mathcal{LDG}$, which does not depend on the representatives $\alpha$ and $\beta$.
Note also that the $\beta$-pairing clearly preserves values in $\mathcal{LDG}$.

Also, we have by Lemmata~\ref{LemHidingLemmaOnDynamicGames} and \ref{LemDerelictionLemma}:
\begin{align*}
\mathcal{H}^\omega(\langle L, R \rangle^\dagger \ddagger \varPi^{A, B}_1) = \langle \mathcal{H}^\omega(L), \mathcal{H}^\omega(R) \rangle^\dagger ; \varPi^{A, B}_1 = \mathcal{H}^\omega(L)
\end{align*}
as well as by Lemma~\ref{LemHidingLemmaOnDynamicStrategies}:
\begin{align*}
\mathcal{H}^\omega(\langle \alpha, \beta \rangle^\dagger \ddagger \varpi^{A, B}_1) &= \langle \mathcal{H}^\omega(\alpha)^\dagger, \mathcal{H}^\omega(\beta)^\dagger \rangle ; \varpi^{A, B}_1 \\
&= \mathcal{H}^\omega(\alpha).
\end{align*}
Similarly, $\mathcal{H}^\omega(\langle L, R \rangle^\dagger \ddagger \varPi^{A, B}_2) = \mathcal{H}^\omega(R)$ and $\mathcal{H}^\omega(\langle \alpha, \beta \rangle^\dagger \ddagger \varpi^{A, B}_2) = \mathcal{H}^\omega(\beta)$.
Hence, $\langle (L, [\alpha]_{\mathsf{W}}), (R, [\beta]_{\mathsf{W}}) \rangle ; \pi_1 \simeq (L, [\alpha]_{\mathsf{W}})$ and $\langle (L, [\alpha]_{\mathsf{W}}), (R, [\beta]_{\mathsf{W}}) \rangle ; \pi_2 = (R, [\beta]_{\mathsf{W}})$ hold.

Next, given any $\beta$-morphism $(P, [\rho]_{\mathsf{W}}) : C \rightarrow A \& B$ in $\mathcal{LDG}$, we have: 
\begin{align*}
\mathcal{H}^\omega(\langle P^\dagger \ddagger \varPi^{A, B}_1, P^\dagger \ddagger \varPi^{A, B}_2 \rangle) &= \langle \mathcal{H}^\omega(P)^\dagger ; \varPi^{A, B}_1, \mathcal{H}^\omega(P)^\dagger ; \varPi^{A, B}_2 \rangle \\
&= \mathcal{H}^\omega(P) 
\end{align*}
again by Lemmata~\ref{LemHidingLemmaOnDynamicGames} and \ref{LemDerelictionLemma}, as well as by Lemma~\ref{LemHidingLemmaOnDynamicStrategies}:
\begin{align*}
 \mathcal{H}^\omega(\langle \rho^\dagger \ddagger \varpi^{A, B}_1, \rho^\dagger \ddagger \varpi^{A, B}_2 \rangle) &= \langle \mathcal{H}^\omega(\rho)^\dagger ; \varpi^{A, B}_1, \mathcal{H}^\omega(\rho)^\dagger ; \varpi^{A, B}_2 \rangle \\
&= \mathcal{H}^\omega(\rho).
\end{align*}
Hence, $\langle (P, [\rho]) ; \pi_1, (P, [\rho]) ; \pi_2 \rangle \simeq (P, [\rho])$ holds.

It is also straightforward to check that $\beta$-pairing in $\mathcal{LDG}$ preserves $\simeq$: Given any $\beta$-morphisms $(L, [\alpha]_{\mathsf{W}}), (\tilde{L}, [\tilde{\alpha}]_{\mathsf{W}}) : C \rightarrow A$ and $(R, [\beta]_{\mathsf{W}}), (\tilde{R}, [\tilde{\beta}]_{\mathsf{W}}) : C \rightarrow B$ in $\mathcal{LDG}$ such that $\mathcal{H}^\omega(L, [\alpha]_{\mathsf{W}}) = \mathcal{H}^\omega(\tilde{L}, [\tilde{\alpha}]_{\mathsf{W}})$ and $\mathcal{H}^\omega(R, [\beta]_{\mathsf{W}}) = \mathcal{H}^\omega(\tilde{R}, [\tilde{\beta}]_{\mathsf{W}})$, we have:
\begin{align*}
\mathcal{H}^\omega(\langle (L, [\alpha]_{\mathsf{W}}), (R, [\beta]_{\mathsf{W}}) \rangle) &= (\mathcal{H}^\omega(\langle L, R \rangle), [\mathcal{H}^\omega(\langle \alpha, \beta \rangle)]_{\mathsf{W}}) \\
&= (\langle \mathcal{H}^\omega(L), \mathcal{H}^\omega(R) \rangle, [\langle \mathcal{H}^\omega(\alpha), \mathcal{H}^\omega(\beta) \rangle]_{\mathsf{W}}) \\
&= (\langle \mathcal{H}^\omega(\tilde{L}), \mathcal{H}^\omega(\tilde{R}) \rangle, [\langle \mathcal{H}^\omega(\tilde{\alpha}), \mathcal{H}^\omega(\tilde{\beta}) \rangle]_{\mathsf{W}}) \\
&= (\mathcal{H}^\omega(\langle \tilde{L}, \tilde{R} \rangle), [\mathcal{H}^\omega(\langle \tilde{\alpha}, \tilde{\beta} \rangle)]_{\mathsf{W}}) \\
&= \mathcal{H}^\omega(\langle (\tilde{L}, [\tilde{\alpha}]_{\mathsf{W}}), (\tilde{R}, [\tilde{\beta}]_{\mathsf{W}}) \rangle).
\end{align*}

Finally, the requirements for $\beta$-exponentials, $\beta$-currying and $\beta$-evaluations are proved more or less similarly to the case of $\beta$-products, $\beta$-pairing and $\beta$-projections, and thus we leave the details to the leader.
\end{proof}

We proceed to give a standard structure (Definition~\ref{DefStructuresForFPCF}) for FPCF in $\mathcal{LDG}$:
\begin{definition}[Standard structure in $\boldsymbol{\mathcal{LDG}}$]
\label{DefGameSemanticStructureForFPCF}
The standard structure 
\begin{equation*}
\mathcal{S}_{\mathcal{G}} = (\boldsymbol{2}, T, \&, \pi, \Rightarrow, \mathit{ev}, \underline{\mathit{tt}}, \underline{\mathit{ff}}, \vartheta)
\end{equation*}
of games and strategies for FPCF in $\mathcal{LDG}$ is given by:
\begin{itemize}

\item $\boldsymbol{2}$ is the game of booleans (Example~\ref{ExFlatGames}), and $T$ is the terminal game (Example~\ref{ExTerminalGame});

\item $\&$ is product of games, and $\pi_1^{A, B} \stackrel{\mathrm{df. }}{=} (\varPi^{A, B}_1, [\varpi^{A, B}_1]_{\mathsf{W}})$ ($i = 1, 2$) for any $A, B \in \mathcal{LDG}$;

\item $\Rightarrow$ is function space of games, and $\mathit{ev}_{A, B} \stackrel{\mathrm{df. }}{=} (\varUpsilon_{A, B}, [\upsilon_{A, B}]_{\mathsf{W}})$ for any $A, B \in \mathcal{LDG}$; 

\item $\underline{\mathit{tt}} \stackrel{\mathrm{df. }}{=} (T \Rightarrow \boldsymbol{2}, [\mathsf{Pref}(\{ q . \mathit{tt} \})^{\mathsf{Even}}]_{\mathsf{W}}), \underline{\mathit{ff}} \stackrel{\mathrm{df. }}{=} (T \Rightarrow \boldsymbol{2}, [\mathsf{Pref}(\{ q . \mathit{ff} \})^{\mathsf{Even}}]_{\mathsf{W}}) : T \rightarrow \boldsymbol{2}$;

\item $\vartheta \stackrel{\mathrm{df. }}{=} (\boldsymbol{2} \& (\boldsymbol{2} \& \boldsymbol{2}) \Rightarrow \boldsymbol{2}, [\mathit{case}]_{\mathsf{W}}) : \boldsymbol{2} \& (\boldsymbol{2} \& \boldsymbol{2}) \rightarrow \boldsymbol{2}$, where $\mathit{case} : \boldsymbol{2} \& (\boldsymbol{2} \& \boldsymbol{2}) \Rightarrow \boldsymbol{2},$ is the standard game semantics of the $\mathsf{case}$-construction \cite{hyland2000full,abramsky1999game} modified to a normalized  strategy in the obvious manner.

\end{itemize}
\end{definition}

\begin{lemma}[Standardness of $\boldsymbol{\mathcal{S}_{\mathcal{G}}}$]
\label{LemStandardnessOfS_G}
The structure $\mathcal{S}_{\mathcal{G}}$ for FPCF in $\mathcal{LDG}$ is standard in the sense defined in Definition~\ref{DefStructuresForFPCF}. 
\end{lemma}
\begin{proof}
Straightforward. 
\end{proof}

\subsection{Game-Semantic Dynamic Correspondence Property for FPCF}
\label{MainResult}
At last, we are now ready to prove that our game semantics satisfies a DCP:
\begin{theorem}[PDCP-theorem]
\label{ThmDynamicTheorem}
The interpretation $\llbracket \_ \rrbracket_{\mathcal{LDG}}^{\mathcal{S_G}}$ of FPCF (Definitions~\ref{DefStructuresForFPCF} and \ref{DefCCBoCCDG}) satisfies the PDCP (Definition~\ref{DefPDCP}).
\end{theorem}
\begin{proof}
To establish the PDCP, the only non-trivial case is to show for any reduction of FPCF of the form $\mathsf{(\lambda x^A . \ \! V)W \rightarrow U}$, where $\mathsf{V}$, $\mathsf{W}$ and  $\mathsf{U}$ are values, $\mathcal{H}(\llbracket \mathsf{(\lambda x^A . \ \! V)W} \rrbracket_{\mathcal{LDG}}^{\mathcal{S_G}}) = \llbracket \mathsf{U} \rrbracket_{\mathcal{LDG}}^{\mathcal{S_G}}$ (n.b., $\llbracket \mathsf{(\lambda x^A . \ \! V)W} \rrbracket_{\mathcal{LDG}}^{\mathcal{S_G}} \neq \llbracket \mathsf{U} \rrbracket_{\mathcal{LDG}}^{\mathcal{S_G}}$ is immediate from the first component of each $\beta$-morphism in $\mathcal{LDG}$ and the third axiom on standardness of $\mathcal{S_G}$); the other conditions for the PDCP follow from Lemmata~\ref{LemHidingLemmaOnDynamicGames} and \ref{LemHidingLemmaOnDynamicStrategies}. 
Let us focus on the non-trivial case, for which we define the \emph{\bfseries height} $\mathit{Ht}(\mathsf{B}) \in \mathbb{N}$ of each type $\mathsf{B}$ by $\mathit{Ht}(o) \stackrel{\mathrm{df. }}{=} 0$ and $\mathit{Ht}(\mathsf{B_1 \Rightarrow B_2}) \stackrel{\mathrm{df. }}{=} \max(\mathit{Ht}(\mathsf{B_1})\!+\!1, \mathit{Ht}(\mathsf{B_2}))$.
We proceed by induction on the height of the type $\mathsf{A}$ of $\mathsf{W}$.

Below, given $\beta$-morphisms $(H, [\tau]_{\mathsf{W}}) : C \rightarrow (A \Rightarrow B)$ and $(G, [\sigma]_{\mathsf{W}}) : C \rightarrow A$ in $\mathcal{LDG}$, we define the $\beta$-morphism $(H, [\tau]_{\mathsf{W}}) \lfloor (G, [\sigma]_{\mathsf{W}}) \rfloor \stackrel{\mathrm{df. }}{=} (\langle G, H \rangle^\dagger \ddagger \varUpsilon_{A, B}, [\langle \tau, \sigma \rangle^\dagger \ddagger \upsilon_{A, B}]_{\mathsf{W}}) : C \rightarrow B$ in $\mathcal{LDG}$. 
If $(H, [\tau]_{\mathsf{W}}) : C \to (A_1 \Rightarrow A_2 \Rightarrow \dots \Rightarrow A_k \Rightarrow B)$ and $(G_i, [\sigma_i]_{\mathsf{W}}) : C \to A_i$ for $i = 1, 2, \dots, k$, then we write $(H, [\tau]_{\mathsf{W}}) \lfloor (G_1, [\sigma_1]_{\mathsf{W}}), (G_2, [\sigma_2]_{\mathsf{W}}), \dots, (G_k, [\sigma_k]_{\mathsf{W}}) \rfloor$ for $(H, [\tau]_{\mathsf{W}}) \lfloor (G_1, [\sigma_1]_{\mathsf{W}}) \rfloor \lfloor (G_2, [\sigma_2]_{\mathsf{W}}) \rfloor \dots \lfloor (G_k, [\sigma_k]_{\mathsf{W}}) \rfloor : C \rightarrow B$.
We abbreviate in this proof the interpretation $\llbracket \_ \rrbracket_{\mathcal{LDG}}^{\mathcal{S_G}}$ as $\llbracket \_ \rrbracket$.
Let $\mathsf{\Gamma}$ be the context of $\mathsf{(\lambda x^A . \ \! V)W}$ (as well as $\mathsf{U}$).
In the following, we abbreviate each $\beta$-morphism $(G, [\sigma]_{\mathsf{W}})$ in $\mathcal{LDG}$ as $[\sigma]$ for brevity, and focus on the second components (i.e., the equivalence classes of strategies); the corresponding equations on the first components (i.e., games) may be obtained, thanks to Lemmata~\ref{LemHidingLemmaOnDynamicGames} and \ref{LemDerelictionLemma}, similarly to the ways for the first components shown below.

For the base case, assume $\mathit{Ht}(\mathsf{A}) = 0$, i.e., $\mathsf{A} \equiv o$. 
By induction on $|\mathsf{V}|$, we have:
\begin{itemize}

\item If $\mathsf{V} \equiv \mathsf{tt}$, then $\mathsf{(\lambda x^A . \ \! \mathsf{tt})W \to \mathsf{tt}}$, and clearly $\mathcal{H}(\llbracket \mathsf{(\lambda x^A . \ \! \mathsf{tt}) W} \rrbracket) = \langle \Lambda (\llbracket \mathsf{tt} \rrbracket), \llbracket \mathsf{W} \rrbracket \rangle^\dagger ; [\upsilon] = \llbracket \mathsf{tt}\rrbracket$.
The case of $\mathsf{V} \equiv \mathsf{ff}$ is analogous. 

\item If $\mathsf{V \equiv \lambda y^C . \ \! V'}$, then $\mathsf{(\lambda x^A y^C . \ \! V')W \to \lambda y^C . U'}$ with $\mathsf{(\lambda x^A . \ \! V')W \to U'}$ (for $\mathit{nf}((\mathsf{\lambda x^A y^C . \ \! V')W}) \equiv \mathit{nf} (\mathsf{\lambda y^C . \ \! V'[W/x]}) \equiv \mathsf{\lambda y^C. \mathit{nf} (V'[W/x]) \equiv \lambda y^C . \ \! \mathit{nf} ((\lambda x^A . V')W)}$). 
Then, we have, by the induction hypothesis, $\mathcal{H}(\llbracket (\mathsf{\lambda x^A . \ \! V')W} \rrbracket) = \llbracket \mathsf{U'} \rrbracket$. Hence, we get:
\begin{align*}
\mathcal{H}(\llbracket \mathsf{V W} \rrbracket) &= \mathcal{H}(\langle \Lambda_{\llbracket \mathsf{A} \rrbracket}(\Lambda_{\llbracket \mathsf{C} \rrbracket}(\llbracket \mathsf{V'} \rrbracket)), \llbracket \mathsf{W} \rrbracket \rangle^\dagger \ddagger [\upsilon]) \\
&= \langle \Lambda_{\llbracket \mathsf{A} \rrbracket}(\Lambda_{\llbracket \mathsf{C} \rrbracket}(\llbracket \mathsf{V'} \rrbracket)), \llbracket \mathsf{W} \rrbracket \rangle^\dagger ; [\upsilon] \ \text{(by Lemma~\ref{LemHidingLemmaOnDynamicStrategies})} \\
&= \Lambda_{\llbracket \mathsf{C} \rrbracket}(\langle \Lambda_{\llbracket \mathsf{A} \rrbracket}(\llbracket \mathsf{V'} \rrbracket), \llbracket \mathsf{W} \rrbracket \rangle^\dagger ; [\upsilon]) \\
&= \Lambda_{\llbracket \mathsf{C} \rrbracket}(\mathcal{H}(\langle \Lambda_{\llbracket \mathsf{A} \rrbracket}(\llbracket \mathsf{V'} \rrbracket), \llbracket \mathsf{W} \rrbracket \rangle^\dagger \ddagger [\upsilon])) \\
&= \Lambda_{\llbracket \mathsf{C} \rrbracket}(\mathcal{H}(\llbracket (\mathsf{\lambda x^A . V')W} \rrbracket)) \\
&= \Lambda_{\llbracket \mathsf{C} \rrbracket}(\llbracket \mathsf{U'} \rrbracket) \\
&= \llbracket \mathsf{\lambda y^C . U'} \rrbracket.
\end{align*}

\item If $\mathsf{V \equiv \mathsf{case} (y V_1 \dots V_k)[\tilde{V}_1 ; \tilde{V}_2]}$ with $\mathsf{x \neq y}$, then $\mathsf{(\lambda x^A . V)W \to U}$, where
\begin{equation*}
\mathsf{U \equiv \mathsf{case} (y \mathit{nf} (V_1[W/x]) \dots \mathit{nf} (V_k[W/x]))[\mathit{nf} (\tilde{V}_1[W/x]) ; \mathit{nf} (\tilde{V}_1[W/x])]}.
\end{equation*}
By the induction hypothesis and the interpretation of the variable $\mathsf{y}$, we have:
\begin{align*}
& \ \llbracket \mathsf{(\lambda x^A . V)W} \rrbracket \\
= & \ \mathcal{H}^\omega(\Lambda_{\llbracket \mathsf{A} \rrbracket}(\langle \llbracket \mathsf{y} \rrbracket \lfloor \llbracket \mathsf{V_1} \rrbracket, \dots, \llbracket \mathsf{V_k} \rrbracket \rfloor, \langle \llbracket \mathsf{\tilde{V}_1} \rrbracket, \llbracket \mathsf{\tilde{V}_2} \rrbracket \rangle \rangle^\dagger \ddagger [\mathit{case}]) \lfloor \llbracket \mathsf{W} \rrbracket \rfloor) \\
= & \ \mathcal{H}^\omega(\langle \Lambda_{\llbracket \mathsf{A} \rrbracket}(\llbracket \mathsf{y} \rrbracket)\lfloor \llbracket \mathsf{W} \rrbracket \rfloor \lfloor \Lambda_{\llbracket \mathsf{A} \rrbracket}(\llbracket \mathsf{V_1} \rrbracket)\lfloor \llbracket \mathsf{W} \rrbracket \rfloor, \dots, \Lambda_{\llbracket \mathsf{A} \rrbracket}(\llbracket \mathsf{V_k} \rrbracket)\lfloor \llbracket \mathsf{W} \rrbracket \rfloor \rfloor, \langle \Lambda_{\llbracket \mathsf{A} \rrbracket}(\llbracket \mathsf{\tilde{V}_1} \rrbracket)\lfloor \llbracket \mathsf{W} \rrbracket \rfloor, \\& \ \Lambda_{\llbracket \mathsf{A} \rrbracket}(\llbracket \mathsf{\tilde{V}_2} \rrbracket)\lfloor \llbracket \mathsf{W} \rrbracket \rfloor \rangle \rangle^\dagger \ddagger [\mathit{case}]) \\
= & \ \mathcal{H}^\omega(\langle \llbracket \mathsf{(\lambda x . \ \! y) W} \rrbracket \lfloor \llbracket \mathsf{(\lambda x . \ \! V_1) W} \rrbracket, \dots, \llbracket \mathsf{(\lambda x . \ \! V_k) W} \rrbracket \rfloor, \langle \llbracket \mathsf{(\lambda x . \ \! \tilde{V}_1) W} \rrbracket, \llbracket \mathsf{(\lambda x . \ \! \tilde{V}_2) W} \rrbracket \rangle \rangle^\dagger \ddagger [\mathit{case}]) \\
= & \ \mathcal{H}^\omega(\langle \llbracket \mathsf{y} \rrbracket \lfloor \llbracket \mathsf{\mathit{nf}(V_1[W/x])} \rrbracket, \dots, \llbracket \mathsf{\mathit{nf}(V_k[W/x])} \rrbracket \rfloor, \langle \llbracket \mathsf{\mathit{nf}(\tilde{V}_1[W/x])} \rrbracket, \llbracket \mathsf{\mathit{nf}(\tilde{V}_2[W/x])} \rrbracket \rangle \rangle^\dagger \ddagger [\mathit{case}]) \\ 
= & \ \llbracket \mathsf{U} \rrbracket.
\end{align*}

\item If $\mathsf{V \equiv \mathsf{case} (x)[\tilde{V}_1 ; \tilde{V}_2]}$, then $\mathsf{(\lambda x^A . V)W \to U}$, where
\begin{equation*}
\mathsf{U \equiv \mathsf{case} (W)[\mathit{nf} (\tilde{V}_1[W/x]) ; \mathit{nf} (\tilde{V}_2[W/x])]}.
\end{equation*}
By the same reasoning as the above case, we get $\mathcal{H}(\llbracket (\mathsf{\lambda x^A . V)W} \rrbracket) = \llbracket \mathsf{U} \rrbracket$.

\end{itemize}

Next, for the inductive step, assume $\mathit{Ht}(\mathsf{A}) = h + 1$. We may proceed in the same way as the base case, i.e., by induction on $|\mathsf{V}|$, except that the last case is generalized to $\mathsf{V \equiv \mathsf{case} (x V_1 \dots V_k)[\tilde{V}_1 ; \tilde{V}_2]}$, where $\mathsf{A} \equiv \mathsf{A_1} \Rightarrow \mathsf{A_2} \Rightarrow \dots \Rightarrow \mathsf{A_k} \Rightarrow o$ ($k \geqslant 0$). 
We have to consider the additional case of $k \geqslant 1$;  then we have $\mathsf{(\lambda x^A . \ \! V)W \to U}$, where
\begin{equation*}
\mathsf{U \equiv \mathsf{case} (\mathit{nf} (W (V_1[W/x]) \dots (V_k[W/x]))) [\mathit{nf} (\tilde{V}_1[W/x]) ; \mathit{nf} (\tilde{V}_2[W/x])]}.
\end{equation*}
We then have the following chain of equations:
\begin{align*}
& \mathcal{H} \llbracket \mathsf{(\lambda x . V)W}  \rrbracket \\
= \ & \mathcal{H} (\Lambda (\mathcal{H}^{\omega}(\langle \llbracket \mathsf{x} \rrbracket \lfloor \llbracket \mathsf{V_1} \rrbracket, \dots, \llbracket \mathsf{V_k} \rrbracket \rfloor, \langle \llbracket \mathsf{\tilde{V}_1} \rrbracket, \llbracket \mathsf{\tilde{V}_2} \rrbracket \rangle \rangle^\dagger \ddagger [\mathit{case}])) \lfloor \llbracket \mathsf{W} \rrbracket \rfloor) \\
= \ & \mathcal{H}^{\omega} (\langle \Lambda (\llbracket \mathsf{x} \rrbracket) \lfloor \llbracket \mathsf{W} \rrbracket \rfloor \lfloor \Lambda (\llbracket \mathsf{V_1} \rrbracket) \lfloor \llbracket \mathsf{W} \rrbracket \rfloor, \dots, \Lambda (\llbracket \mathsf{V_k} \rrbracket) \lfloor \llbracket \mathsf{W} \rrbracket \rfloor \rfloor, \langle \Lambda (\llbracket \mathsf{\tilde{V}_1} \rrbracket) \lfloor \llbracket \mathsf{W} \rrbracket \rfloor, \Lambda (\llbracket \mathsf{\tilde{V}_2} \rrbracket) \lfloor \llbracket \mathsf{W} \rrbracket \rfloor \rangle \rangle^\dagger \ddagger [\mathit{case}]) \\
= \ & \mathcal{H}^{\omega} (\langle \llbracket \mathsf{(\lambda x . \ \! x) W} \rrbracket \lfloor \llbracket \mathsf{(\lambda x . \ \! V_1) W} \rrbracket, \dots, \llbracket \mathsf{(\lambda x . \ \! V_k) W} \rrbracket \rfloor, \langle \llbracket \mathsf{(\lambda x . \ \! \tilde{V}_1) W} \rrbracket, \llbracket \mathsf{(\lambda x . \ \! \tilde{V}_2) W} \rrbracket \rangle \rangle^\dagger \ddagger [\mathit{case}]) \\
= \ & \mathcal{H}^{\omega} (\langle \llbracket \mathsf{W} \rrbracket \lfloor \llbracket \mathit{nf}(\mathsf{V_1[W/x]}) \rrbracket, \dots, \llbracket \mathit{nf}(\mathsf{V_k[W/x]}) \rrbracket \rfloor, \langle \llbracket \mathit{nf} (\mathsf{\tilde{V}_1[W/x]}) \rrbracket, \llbracket \mathit{nf} (\mathsf{\tilde{V}_2[W/x]}) \rrbracket \rangle \rangle^\dagger \ddagger [\mathit{case}]) \\
&\text{(by the induction hypothesis with respect to $|\mathsf{V}|$)} \\
= \ & \mathcal{H}^{\omega} (\langle \llbracket \mathit{nf} (\mathsf{W (V_1[W/x]}) \dots (\mathsf{V_k[W/x])}) \rrbracket, \langle \llbracket \mathit{nf} (\mathsf{\tilde{V}_1[W/x]}) \rrbracket, \llbracket \mathit{nf} (\mathsf{\tilde{V}_2[W/x]}) \rrbracket \rangle \rangle^\dagger \ddagger [\mathit{case}]) \\
& \text{(by the induction hypothesis (applied $k$-times) with respect to the hight of types $\mathsf{A}$)} \\
= \ & \llbracket \mathsf{case} (\mathit{nf} (\mathsf{W (V_1[W/x]}) \dots (\mathsf{V_k[W/x])})) [\mathit{nf} (\mathsf{\tilde{V}_1[W/x]}) ;  \mathit{nf} (\mathsf{\tilde{V}_2[W/x]})] \rrbracket \\
= \ & \llbracket \mathsf{U} \rrbracket
\end{align*}
which completes the proof.
\end{proof}

\begin{corollary}[Dynamic game semantics of FPCF]
\label{CoroFirstDynamicGameSemantics}
The interpretation $\llbracket \_ \rrbracket_{\mathcal{LDG}}^{\mathcal{S_G}}$ of FPCF and the hiding operation $\mathcal{H}$ satisfy the DCP in the sense of Definition~\ref{DefDCP}.
\end{corollary}
\begin{proof}
By Lemma~\ref{LemStandardnessOfS_G}, and Theorems~\ref{ThmDynamicSemanticsOfSystemTVartheta}, \ref{ThmLDG} and \ref{ThmDynamicTheorem}.
\end{proof}

The relation between the syntax and the semantics of FPCF is actually tighter than Corollary~\ref{CoroFirstDynamicGameSemantics}: Exploiting the strong definability result \cite{amadio1998domains,hyland2000full}, FPCF can be seen as a \emph{formal calculus} for computations in the CCBoC $\mathcal{LDG}$.
In addition, FPCF represents every computation in $\mathcal{LDG}$ by the following full completeness result \cite{curien2007definability}: Any strategy on a game that interprets a type of FPCF is the denotation of some term of FPCF:
\begin{corollary}[Dynamic full completeness]
\label{CoroDynamicFullCompleteness}
Let $G$ be a game such that for some strategy $\sigma : G$ the pair $(G, [\sigma]_{\mathsf{W}})$ is the interpretation $\llbracket \mathsf{\Gamma \vdash M : B} \rrbracket_{\mathcal{LDG}}^{\mathcal{S_G}}$ of a program $\mathsf{\Gamma \vdash M : B}$ of FPCF.
Then, for any strategy $\tilde{\sigma} : G$, there is a program $\mathsf{\Gamma \vdash \tilde{M} : B}$ of FPCF such that $\llbracket \mathsf{\Gamma \vdash \tilde{M} : B} \rrbracket_{\mathcal{LDG}}^{\mathcal{S_G}} = (G, [\tilde{\sigma}]_{\mathsf{W}})$.
\end{corollary}
\begin{proof}
Note that the game $G$ is constructed along with the construction of type $\mathsf{B}$ of FPCF.
We proceed by induction on the construction of $G$ (or $\mathsf{B}$).

First, since values of FPCF are PCF B\"{o}hm trees except that the natural number type $\iota$ is replaced with the boolean type $o$, and the bottom term $\mathsf{\bot}$ is deleted, the conventional full completeness and the strong definability hold for values of FPCF in the same way as that of the conventional game semantics of PCF, where the winning condition on strategies excludes the denotation of the bottom term $\mathsf{\bot}$; see \cite{abramsky1999game,curien2006notes} for the details. 

It remains to consider the rule A for applications, i.e., the case where $G$ is of the form $\langle U, V \rangle^\dagger \ddagger \varUpsilon$.
But then, note that only plays by the dereliction (up to `tags') are possible in $\varUpsilon$ (Definition~\ref{DefDerelictionGames}), and therefore we may just apply the induction hypothesis.
\end{proof}

\section{Conclusion and Future Work}
\label{ConclusionAndFutureWorkOfDynamicGameSemantics}
We have presented a \emph{mathematical} (and \emph{syntax-independent}) formulation of dynamics and intensionality of computation in terms of bicategories as well as games and strategies. 
From the opposite angle, we have developed bicategorical and game-semantic frameworks for dynamic, intensional computation with a convenient formal calculus.  

Let us emphasize that the dynamic, intensional nature of our semantics stands in sharp contrast to the static, extensional nature of conventional (categorical or game) semantics.
In particular, our semantics satisfies the highly non-trivial DCP with respect to FPCF.

Note also that the present work refines and generalizes standard categorical and game semantics of type theories.
For instance, composition of static strategies is decomposed and generalized as \emph{concatenation plus hiding} of dynamic strategies.
Also, standard constructions and constraints on static games and strategies are naturally accommodated in the framework of dynamic games and strategies.
Moreover, from the category-theoretic point, the present work refines the standard CCC-interpretation of type theories by the CCBoC-interpretation.
In this sense, our approach is natural and general, achieving \emph{mathematics} of dynamics and intensionality of computation as promised in Section~\ref{Introduction}. 

Let us remark that our result does not contradict the standard result \cite{danos1996game}, i.e., the correspondence between the execution of linear head reduction (LHR) and the step-by-step `internal communication' between conventional strategies.
In fact, LHR is a finer reduction strategy than the operational semantics of FPCF (Definition~\ref{DefFPCF}), and the work by Danos et al. implies that LHR corresponds in conventional game semantics what should be called a `move-wise' execution of the hiding operation.
On the other hand, our operational semantics is executed in a much coarser, `type-wise' fashion, and thus it may be seen as executing at a time a certain `chunk' of LHR in a specific order.
Our dynamic game semantics captures such a coarser dynamics of computation, and therefore it does not contradict the work \cite{danos1996game}.
Of course, it is highly interesting to refine the present work to capture LHR or another, finer reduction strategy such as \emph{explicit substitution} \cite{rose1996explicit} and the \emph{differential $\lambda$-calculus} \cite{ehrhard2003differential}, which we leave as future work.

More generally, the most immediate future work is to apply the framework of dynamic game semantics to various logics and computations as in the case of conventional game semantics.
Also, it would be interesting to see how accurately our game-semantic approach can measure the computational complexity of (higher-order) programming.

Finally, the notion of (CC)BoCs can be a concept of interest in its own right. 
For instance, it might be fruitful to develop it further to accommodate various models of computations in the same spirit of \cite{longley2015higher} but on \emph{computation}, not computability. 
Also, it might be interesting to consider their relation with \emph{computations as monads} in the sense introduced by Eugenio Moggi \cite{moggi1991notions}.

\section*{Acknowledgements}
The first author acknowledges the financial support from Funai Overseas Scholarship, and Luke Ong and Sam Staton for fruitful discussions. 
The second author acknowledges support from the EPSRC grant EP/K015478/1 on Quantum Mathematics and Computation,
and U.S. AFOSR FA9550-12-1-0136.
\bibliographystyle{apalike}
\bibliography{CategoricalLogic,GamesAndStrategies,RecursionTheory,PCF,TypeTheoriesAndProgrammingLanguages,HoTT,GoI,LinearLogic}

\begin{thebibliography}{}

\bibitem[Abramsky et~al., 1997]{abramsky1997semantics}
Abramsky, S. et~al. (1997).
\newblock Semantics of interaction: An introduction to game semantics.
\newblock {\em Semantics and Logics of Computation, Publications of the Newton
  Institute}, pages 1--31.

\bibitem[Abramsky et~al., 2000]{abramsky2000full}
Abramsky, S., Jagadeesan, R., and Malacaria, P. (2000).
\newblock Full abstraction for {PCF}.
\newblock {\em Information and Computation}, 163(2):409--470.

\bibitem[Abramsky and Jung, 1994]{abramsky1994domain}
Abramsky, S. and Jung, A. (1994).
\newblock Domain theory.
\newblock In {\em Handbook of logic in computer science}. Oxford University
  Press.

\bibitem[Abramsky and McCusker, 1999]{abramsky1999game}
Abramsky, S. and McCusker, G. (1999).
\newblock Game semantics.
\newblock In {\em Computational logic}, pages 1--55. Springer.

\bibitem[Abramsky and Melli{\`e}s, 1999]{abramsky1999concurrent}
Abramsky, S. and Melli{\`e}s, P.-A. (1999).
\newblock Concurrent games and full completeness.
\newblock In {\em Proceedings of the 14th Annual IEEE Symposium on Logic in
  Computer Science}, page 431. IEEE Computer Society.

\bibitem[Amadio and Curien, 1998]{amadio1998domains}
Amadio, R.~M. and Curien, P.-L. (1998).
\newblock {\em Domains and Lambda-Calculi}.
\newblock Number~46. Cambridge University Press, Cambridge.

\bibitem[Church, 1940]{church1940formulation}
Church, A. (1940).
\newblock A formulation of the simple theory of types.
\newblock {\em The journal of symbolic logic}, 5(02):56--68.

\bibitem[Clairambault and Harmer, 2010]{clairambault2010totality}
Clairambault, P. and Harmer, R. (2010).
\newblock Totality in arena games.
\newblock {\em Annals of Pure and Applied Logic}, 161(5):673--689.

\bibitem[Crole, 1993]{crole1993categories}
Crole, R.~L. (1993).
\newblock {\em Categories for {T}ypes}.
\newblock Cambridge University Press.

\bibitem[Curien, 2006]{curien2006notes}
Curien, P.-L. (2006).
\newblock Notes on game semantics.
\newblock {\em From the author's web page}.

\bibitem[Curien, 2007]{curien2007definability}
Curien, P.-L. (2007).
\newblock Definability and full abstraction.
\newblock {\em Electronic Notes in Theoretical Computer Science}, 172:301--310.

\bibitem[Danos et~al., 1996]{danos1996game}
Danos, V., Herbelin, H., and Regnier, L. (1996).
\newblock Game semantics and abstract machines.
\newblock In {\em Proceedings of the 11th Annual IEEE Symposium on Logic in
  Computer Science}, page 394. IEEE Computer Society.

\bibitem[Danos and Regnier, 2004]{danos2004head}
Danos, V. and Regnier, L. (2004).
\newblock Head linear reduction.
\newblock {\em Unpublished}.

\bibitem[Dimovski et~al., 2005]{dimovski2005data}
Dimovski, A., Ghica, D.~R., and Lazi{\'c}, R. (2005).
\newblock Data-abstraction refinement: A game semantic approach.
\newblock In {\em International Static Analysis Symposium}, pages 102--117.
  Springer.

\bibitem[Ehrhard and Regnier, 2003]{ehrhard2003differential}
Ehrhard, T. and Regnier, L. (2003).
\newblock The differential lambda-calculus.
\newblock {\em Theoretical Computer Science}, 309(1):1--41.

\bibitem[Gierz et~al., 2003]{gierz2003continuous}
Gierz, G., Hofmann, K.~H., Keimel, K., Lawson, J.~D., Mislove, M., and Scott,
  D.~S. (2003).
\newblock {\em Continuous lattices and domains}, volume~93.
\newblock Cambridge University Press.

\bibitem[Girard, 1987]{girard1987linear}
Girard, J.-Y. (1987).
\newblock Linear logic.
\newblock {\em Theoretical computer science}, 50(1):1--101.

\bibitem[Girard, 1989]{girard1989geometry}
Girard, J.-Y. (1989).
\newblock Geometry of {I}nteraction {I}: {I}nterpretation of {S}ystem {F}.
\newblock {\em Studies in Logic and the Foundations of Mathematics},
  127:221--260.

\bibitem[Girard, 1990]{girard1990geometry}
Girard, J.-Y. (1990).
\newblock Geometry of interaction {II}: Deadlock-free algorithms.
\newblock In {\em COLOG-88}, pages 76--93. Springer.

\bibitem[Girard, 1995]{girard1995geometry}
Girard, J.-Y. (1995).
\newblock Geometry of interaction {III}: accommodating the additives.
\newblock {\em London Mathematical Society Lecture Note Series}, pages
  329--389.

\bibitem[Girard, 2003]{girard2003geometry}
Girard, J.-Y. (2003).
\newblock Geometry of interaction {IV}: the feedback equation.
\newblock In {\em Logic Colloquium}, volume~3, pages 76--117. Citeseer.

\bibitem[Girard, 2011]{girard2011geometry}
Girard, J.-Y. (2011).
\newblock Geometry of interaction {V}: logic in the hyperfinite factor.
\newblock {\em Theoretical Computer Science}, 412(20):1860--1883.

\bibitem[Girard, 2013]{girard2013geometry}
Girard, J.-Y. (2013).
\newblock Geometry of interaction {VI}: a blueprint for transcendental syntax.
\newblock {\em preprint}.

\bibitem[Girard et~al., 1989]{girard1989proofs}
Girard, J.-Y., Taylor, P., and Lafont, Y. (1989).
\newblock {\em Proofs and Types}, volume~7.
\newblock Cambridge University Press Cambridge.

\bibitem[Greenland, 2005]{greenland2005game}
Greenland, W.~E. (2005).
\newblock {\em Game Semantics for Region Analysis}.
\newblock PhD thesis, University of Oxford.

\bibitem[Gunter, 1992]{gunter1992semantics}
Gunter, C.~A. (1992).
\newblock {\em Semantics of Programming Languages: Structures and Techniques}.
\newblock MIT press, Cambridge, MA.

\bibitem[Hankin, 1994]{Hankin1994-HANLCA-2}
Hankin, C. (1994).
\newblock {\em Lambda Calculi: A Guide for the Perplexed}.
\newblock Oxford University Press.

\bibitem[Harmer, 2004]{harmer2004innocent}
Harmer, R. (2004).
\newblock Innocent game semantics.
\newblock {\em Lecture notes}, 2007.

\bibitem[Hilken, 1996]{hilken1996towards}
Hilken, B.~P. (1996).
\newblock Towards a proof theory of rewriting: the simply typed
  2$\lambda$-calculus.
\newblock {\em Theoretical Computer Science}, 170(1-2):407--444.

\bibitem[Hyland and Ong, 2000]{hyland2000full}
Hyland, J. M.~E. and Ong, C.-H. (2000).
\newblock On {F}ull {A}bstraction for {PCF}: {I}, {II}, and {III}.
\newblock {\em Information and computation}, 163(2):285--408.

\bibitem[Hyland, 1997]{hyland1997game}
Hyland, M. (1997).
\newblock Game semantics.
\newblock {\em Semantics and logics of computation}, 14:131.

\bibitem[Jacobs, 1999]{jacobs1999categorical}
Jacobs, B. (1999).
\newblock {\em {Categorical Logic and Type Theory}}, volume 141.
\newblock Elsevier.

\bibitem[Lambek and Scott, 1988]{lambek1988introduction}
Lambek, J. and Scott, P.~J. (1988).
\newblock {\em Introduction to {H}igher-order {C}ategorical {L}ogic}, volume~7.
\newblock Cambridge University Press.

\bibitem[Laurent, 2004]{laurent2004polarized}
Laurent, O. (2004).
\newblock Polarized games.
\newblock {\em Annals of Pure and Applied Logic}, 130(1-3):79--123.

\bibitem[Longley and Normann, 2015]{longley2015higher}
Longley, J. and Normann, D. (2015).
\newblock {\em Higher-Order Computability}.
\newblock Springer, Heidelberg.

\bibitem[McCusker, 1998]{mccusker1998games}
McCusker, G. (1998).
\newblock {\em Games and Full Abstraction for a Functional Metalanguage with
  Recursive Types}.
\newblock Springer Science \& Business Media, London.

\bibitem[Mellies, 2005]{mellies2005axiomatic}
Mellies, P.-A. (2005).
\newblock Axiomatic rewriting theory {I}: A diagrammatic standardization
  theorem.
\newblock In {\em Processes, Terms and Cycles: steps on the road to infinity},
  pages 554--638. Springer.

\bibitem[Moggi, 1991]{moggi1991notions}
Moggi, E. (1991).
\newblock Notions of computation and monads.
\newblock {\em Information and computation}, 93(1):55--92.

\bibitem[Ong, 2006]{ong2006model}
Ong, C.-H. (2006).
\newblock On model-checking trees generated by higher-order recursion schemes.
\newblock In {\em 21st Annual IEEE Symposium on Logic in Computer Science
  (LICS'06)}, pages 81--90. IEEE.

\bibitem[Ouaknine, 1997]{ouaknine1997two}
Ouaknine, J. (1997).
\newblock {\em A {T}wo-{D}imensional {E}xtension of {L}ambek's {C}ategorical
  {P}roof {T}heory}.
\newblock PhD thesis, McGill University, Montr{\'e}al.

\bibitem[Pitts, 2001]{pitts2001categorical}
Pitts, A.~M. (2001).
\newblock Categorical {L}ogic.
\newblock In {\em Handbook of logic in computer science}, pages 39--123. Oxford
  University Press.

\bibitem[Plotkin, 1977]{plotkin1977lcf}
Plotkin, G.~D. (1977).
\newblock {LCF} considered as a programming language.
\newblock {\em Theoretical Computer Science}, 5(3):223--255.

\bibitem[Rose, 1996]{rose1996explicit}
Rose, K.~H. (1996).
\newblock {\em Explicit Substitution: Tutorial \& Survey}.
\newblock Computer Science Department.

\bibitem[Scott, 1976]{scott1976data}
Scott, D. (1976).
\newblock Data types as lattices.
\newblock {\em Siam Journal on Computing}, 5(3):522--587.

\bibitem[Scott, 1993]{scott1993type}
Scott, D.~S. (1993).
\newblock A type-theoretical alternative to {ISWIM}, {CUCH}, {OWHY}.
\newblock {\em Theoretical Computer Science}, 121(1):411--440.

\bibitem[Seely, 1987]{seely1987modelling}
Seely, R.~A. (1987).
\newblock Modelling computations: A 2-categorical framework.
\newblock In {\em Proceedings of the 2nd Annual IEEE Symposium on Logic in
  Computer Science}, pages 65--71. IEEE Computer Society.

\bibitem[S{\o}rensen and Urzyczyn, 2006]{sorensen2006lectures}
S{\o}rensen, M.~H. and Urzyczyn, P. (2006).
\newblock {\em Lectures on the Curry-Howard isomorphism}, volume 149.
\newblock Elsevier.

\bibitem[Winskel, 1993]{winskel1993formal}
Winskel, G. (1993).
\newblock {\em The Formal Semantics of Programming Languages: An Introduction}.
\newblock MIT press, Cambridge, MA.

\end{thebibliography}


\begin{thebibliography}{}
 \bibitem[Augustsson and Johnsson 1987]{AJ1}
   Augustsson,~L. and Johnsson,~T. (1987) LML users' manual. PMG
   Report, Department of Computer Science, Chalmers University of
   Technology, Goteborg, Sweden.
 \bibitem[Conklin 1987]{JC}
   Conklin,~J. (1987) Hypertext: an introduction and survey.
   \emph{IEEE Computer} \textbf{20}~(9), 17--41.
\bibitem[Dijkstra 1976]{EWD}
   Dijkstra,~E.\,W. (1976) \emph{A Discipline of Programming}.
   Prentice-Hall.
\bibitem[Knuth 1984]{DEK}
   Knuth,~D.\,E. (1984) Literate programming. \emph{BCS Comput.\ J.}
   \textbf{27}~(2), 97--111 (May).
\bibitem[Danvy and Filinski 1992]{DO}
   Danvy, O. and Filinski, A. (1992) Representing control: a study
   of the CPS transformation \emph{Math. Struct. in Comp. Science}
   \textbf{2}~(2), 361--391.
\bibitem[Reynolds 1968]{JCR}
   Reynolds,~J.\,C. (1969) Transformation systems and the
   algebraic structure of atomic formulas. In B. Meltzer and
   D. Michie (editors), \emph{Machine Intelligence} \textbf{5}, 135--151.
   Edinburgh University Press.
\bibitem[Toyn \emph{et al.} 1987]{TDR}
   Toyn,~I., Dix,~A. and Runciman,~C. (1987) Performance
   polymorphism. In \emph{Functional Programming Languages and
   Computer Architecture, Lecture Notes in Computer Science}
   \textbf{274}, 325--346. Springer-Verlag.
\end{thebibliography}

\newpage
\appendix

\label{lastpage}

\end{document}